\documentclass[11pt]{article}

\pdfoutput=1

\usepackage[letterpaper, portrait, margin=1in]{geometry}
\usepackage[colorlinks=true,linkcolor=blue,citecolor=ForestGreen]{hyperref}
\usepackage{algorithm}
\usepackage[noend]{algpseudocode}
\usepackage{url}
\usepackage{amsmath,amssymb,amsthm}
\usepackage{thmtools,thm-restate}
\usepackage[noabbrev,capitalise,nameinlink]{cleveref}
\usepackage{mathtools}
\usepackage{xspace}
\usepackage{verbatim}
\usepackage{mathrsfs}
\usepackage[usenames,dvipsnames,svgnames,table]{xcolor}
\usepackage{pgf}
\usepackage[dvipsnames]{xcolor}
\usepackage{todo}
\usepackage{tabularx}
\usepackage{enumitem}

\usepackage{dsfont}

\newcommand*\ie{i.\kern.1em e.\ }
\newcommand*\eg{e.\kern.1em g.\ }

\theoremstyle{plain}
\newtheorem{theorem}{Theorem}[section]
\newtheorem{lemma}[theorem]{Lemma}
\newtheorem{fact}[theorem]{Fact}
\newtheorem{proposition}[theorem]{Proposition}
\newtheorem{claim}[theorem]{Claim}
\newtheorem{corollary}[theorem]{Corollary}

\newtheorem{observation}[theorem]{Observation}

\theoremstyle{definition}

\newtheorem{definition}[theorem]{Definition}
\newtheorem{remark}[theorem]{Remark}

\theoremstyle{plain}

\newcommand{\ignore}[1]{}

\DeclareMathOperator{\supp}{supp}   
\DeclareMathOperator{\poly}{poly}
\DeclareMathOperator{\sign}{sign}

\newcommand{\dist}{\mathsf{dist}}

\newcommand{\Var}[1]{\mathrm{Var} \left[ #1 \right]}
\newcommand{\Varu}[2]{\underset{ #1 } {\mathrm{Var}} \left[ #2 \right]}
\newcommand{\Varuc}[3]{\underset{ #1 } {\mathrm{Var}} \left[ #2 \;\; \left| \;\; #3 \right. \right]}

\newcommand{\Ex}[1]{\bE \left[ #1 \right]}
\newcommand{\Exu}[2]{\underset{#1} \bE \left[ #2 \right] }
\newcommand{\Exuc}[3]{\underset{#1} \bE \left[ #2 \;\; \left| \;\; #3
\right.\right] }

\renewcommand{\Pr}[1]{\bP \left[ #1 \right]} 
\newcommand{\Pru}[2]{\underset{ #1 }\bP \left[ #2 \right]}
\newcommand{\Pruc}[3]{\underset{ #1 }\bP \left[ #2 \;\; \mid \;\; #3 \right]}

\newcommand{\define}{\vcentcolon=}

\newcommand{\floor}[1]{\ensuremath{\lfloor #1 \rfloor}}
\newcommand{\ceil}[1]{\ensuremath{\lceil #1 \rceil}}

\newcommand{\inn}[1]{\langle #1 \rangle}

\DeclarePairedDelimiter{\abs}{\lvert}{\rvert}

\newcommand{\ind}[1]{\mathds{1} \left[ #1 \right] }

\newcommand{\zo}{\{0,1\}}

\newcommand{\cC}{\ensuremath{\mathcal{C}}}
\newcommand{\cD}{\ensuremath{\mathcal{D}}}
\newcommand{\cE}{\ensuremath{\mathcal{E}}}
\newcommand{\cF}{\ensuremath{\mathcal{F}}}
\newcommand{\cG}{\ensuremath{\mathcal{G}}}
\newcommand{\cH}{\ensuremath{\mathcal{H}}}
\newcommand{\cI}{\ensuremath{\mathcal{I}}}

\newcommand{\cM}{\ensuremath{\mathcal{M}}}

\newcommand{\cO}{\ensuremath{\mathcal{O}}}

\newcommand{\cS}{\ensuremath{\mathcal{S}}}
\newcommand{\cT}{\ensuremath{\mathcal{T}}}

\newcommand{\cX}{\ensuremath{\mathcal{X}}}

\newcommand{\bE}{\ensuremath{\mathbb{E}}}

\newcommand{\bN}{\ensuremath{\mathbb{N}}}

\newcommand{\bR}{\ensuremath{\mathbb{R}}}

\newcommand{\bZ}{\ensuremath{\mathbb{Z}}}

\usepackage{float}
\usepackage{afterpage}
\usepackage{accents}
\usepackage{bm}
\usepackage{centernot}
\usepackage[nottoc]{tocbibind}
\usepackage{tocloft}
\restylefloat{table}
\usepackage{etoolbox}

\newcommand\lequestion{\stackrel{\mathclap{\normalfont\mbox{?}}}{\le}}
\newcommand\gtquestion{\stackrel{\mathclap{\normalfont\mbox{?}}}{>}}
\newcommand\gequestion{\stackrel{\mathclap{\normalfont\mbox{?}}}{\ge}}
\newcommand\eqquestion{\stackrel{\mathclap{\normalfont\mbox{?}}}{=}}

\newcommand{\parity}{\mathsf{par}}

\newcommand{\Poi}{\mathsf{Poi}}
\newcommand{\Bin}{\mathsf{Bin}}
\newcommand{\Ber}{\mathsf{Ber}}
\newcommand{\Multinomial}{\mathsf{Multinomial}}
\newcommand{\TV}{\mathsf{TV}}
\newcommand{\edit}{\mathsf{edit}}
\newcommand{\str}{\mathsf{str}}

\newcommand{\close}{\textsc{close}}
\newcommand{\far}{\textsc{far}}

\newcommand{\samp}{\mathsf{samp}}

\DeclareMathOperator{\smallinterval}{\mathsf{small}}
\DeclareMathOperator{\largeinterval}{\mathsf{large}}

\newcommand{\bPhi}{\boldsymbol{\Phi}}

\DeclarePairedDelimiterX{\infdivx}[2]{(}{)}{%
  #1\;\delimsize\|\;#2%
}
\newcommand{\RelativeConcentrationT}[3]{\rho_{#3}\infdivx{#1}{#2}}

\newtoggle{anonymous}
\togglefalse{anonymous}

\newcommand{\reledit}{\mathsf{rel-edit}}
\newcommand{\ham}{\mathsf{ham}}
\newcommand{\trace}{\mathsf{trace}}

\title{Distribution Testing Under the Parity Trace}

\iftoggle{anonymous}{%
\author{Anonymous Author}
}{%
\author{%
  Renato Ferreira Pinto Jr.\thanks{Partly funded by an NSERC Canada Graduate Scholarship Doctoral
  Award.}\\
  University of Waterloo\\
  \texttt{r4ferrei@uwaterloo.ca}
\and Nathaniel Harms\thanks{Partly funded by an NSERC Postdoctoral Fellowship. Much of this work was
done while the author was at the University of Waterloo.}\\
  EPFL \\
  \texttt{nathaniel.harms@epfl.ch}}
}
\date{}

\begin{document}
\maketitle

\begin{abstract}
    \emph{Distribution testing} is a fundamental statistical task with many applications, but we
    are interested in a variety of problems where systematic mislabelings of the sample prevent us
from applying the existing theory.
    To apply distribution testing to these problems, we introduce
    \emph{distribution testing under the parity trace}, where the algorithm receives an ordered
sample $S$ that reveals only the least significant bit of each element. This abstraction reveals
connections between the following three problems of interest, allowing new upper and lower bounds:
    \begin{enumerate}
        \item In distribution testing with a \emph{confused collector}, the collector of the sample
may be incapable of distinguishing between nearby elements of a domain (\eg a machine learning
classifier). We prove bounds for distribution testing with a confused collector on domains
structured as a cycle or a path.
\item Recent work on the fundamental \emph{testing vs.~learning} question established tight
lower bounds on \emph{distribution-free sample-based} property testing by reduction from
distribution testing, but the tightness is limited to \emph{symmetric} properties. The parity trace
allows a broader family of equivalences to \emph{non-symmetric} properties, while recovering and
strengthening many of the previous results with a different technique.
        \item We give the first results for property testing in the well-studied
            \emph{trace reconstruction} model, where the goal is to test whether an unknown string
            $x$ satisfies some property or is far from satisfying that property, given only
            independent random \emph{traces} of $x$.
    \end{enumerate}
    Our main technical result is a tight bound of $\widetilde \Theta\left((n/\epsilon)^{4/5} + \sqrt
n/\epsilon^2\right)$ for testing uniformity of distributions over $[n]$ under the parity trace,
leading also to results for the problems above.
\end{abstract}

\thispagestyle{empty}
\setcounter{page}{0}
\newpage

\setlength{\cftbeforesecskip}{0.125em}

\thispagestyle{empty}
\setcounter{page}{0}
\newpage
{\small
\setcounter{tocdepth}{2} 
\tableofcontents
}
\thispagestyle{empty}
\setcounter{page}{0}
\newpage
\setcounter{page}{1}

\section{Introduction}

Making decisions about an unknown probability distribution $\cD$, using only random samples, is a
basic type of statistical task. Deciding whether $\cD$ satisfies some property, or is \emph{far}
(according to some distance metric) from all distributions satisfying that property, is the purpose
of a \emph{distribution testing} algorithm. Distribution testing is well-studied and interesting on
its own, and also has many useful applications. But we are interested in some problems where
systematic mislabelings of data prevent us from applying the existing theory.  So we define
\emph{distribution testing under the parity trace} to help understand these problems.  Before
defining this abstraction, let us explain these problems.

\newcommand{\profAnonymous}{\iftoggle{anonymous}{Anonymous }{Julie Messier }}

\paragraph*{1. Distribution testing with a confused collector.} We wish to make a decision about an
unknown distribution $\cD$ over some domain $\cX$, using only a random sample $S$ from $\cD$.
Unfortunately, $S$ has been collected or labeled by an entity who does not know the difference
between some elements of $\cX$. Perhaps our sample of woodland flora was tabulated by a research
assistant who cannot differentiate between black spruce and white spruce, or between red maple and
sugar maple, and has counted the spruces together and the maples together by mistake\footnote{We
thank ecologist Prof. \profAnonymous for these examples of species that are easily confused by
students.}. Or, the sample was labeled by a machine learning classifier, and for each pair of
elements $x,y \in \cX$ there is some chance that it has not learned to distinguish $x$ from $y$ and
lumps together all the samples of $x$ and $y$. Or, the sample labels have been hashed by a function
that introduces collisions between nearby elements of $\cX$.  Or, we wish to know about the
distribution of fossils by year, but it is not possible to distinguish between fossils from year $x$
and fossils from year $y$, unless a random geological event leaves a mark in the rock between years
$x$ and $y$. Recent work in learning theory notes that this type of problem is common in the applied
literature, but little is known theoretically \cite{FKKT21}. We introduce a model for this type of
problem, called \emph{distribution testing with a confused collector}. 

\paragraph*{2. Distribution-free sample-based property testing.} The \emph{testing vs.~learning}
question of \cite{GGR98} is one of the fundamental questions in property testing. It asks which
properties can be \emph{tested} more efficiently than they can be \emph{learned}.
\emph{Distribution-free sample-based property testing} is the property testing model
corresponding to standard PAC learning, so understanding testing vs.~learning in this model is
essential for many of the standard motivations for property testing \cite{GGR98}.  Recent progress
on testing vs.~learning used connections to distribution testing to get new upper bounds
\cite{GR16} and lower bounds \cite{ES20,BFH21,CP22} on property testing, exhibiting in particular
an equivalence between property and distribution testing for \emph{symmetric} properties of
functions $[n] \to \zo$ (\ie properties which are closed under permutations on $[n]$)
\cite{GR16,BFH21}. But these techniques fall short of answering the testing vs.~learning question
for important properties like $k$-alternating functions and halfspaces, because these properties are
non-symmetric and do not allow the same equivalences, which leaves a gap between the lower bounds of
\cite{BFH21} and the upper bounds from PAC learning that requires new techniques to resolve.

Distribution testing under the parity trace overcomes some of the limitations of \cite{BFH21} by
giving us the first equivalence between distribution testing and testing \emph{non-symmetric}
properties. We recover many of the lower bounds of \cite{BFH21} using a different
technique with stronger consequences for testing vs.~learning, and we also get new tight
positive results for distribution-free sample-based testing of joint function-distribution
properties, adding to the short list of
positive results on distribution-free sample-based testing \cite{GR16,RR20,RR21,BFH21}.

\paragraph*{3. Property testing for trace reconstruction.} 
Trace reconstruction is a beautiful problem posed in \cite{Lev01,BKKM04}.
Relevant to us is trace reconstruction under the
\emph{deletion channel}, which has recently received significant attention (\eg
\cite{HMPW08,DOS17,NP17,HL20,KMMP21,GSZ22,CDL+22,Rubi22,CDL+23}).  The problem is this: There is a
string $x \in \zo^N$ and a deletion rate $\delta \in (0,1)$. A random \emph{trace} is obtained from
$x$ by deleting each character independently with probability $\delta$ to produce a substring.  The
algorithm is given a sequence of independent traces and it must reconstruct the string $x$ using as
few traces as possible. The problem is often motivated by computational biology, where this is a
simplified model of the way biologists typically have access to DNA strings: the ``true'' DNA is not
available, but instead there are a number of corrupted copies. See \cite{BPRS20} for a survey on
biological applications.

Trace reconstruction is notoriously difficult to analyze, with a huge gap between the best known
lower bound of $\widetilde \Omega(N^{3/2})$ \cite{Chas21lower} and upper bound of $\exp(\widetilde
O(N^{1/5}))$ \cite{Chas21upper} (where the hidden constants depend on $\delta$).  However, if the
goal is to make a decision about the unknown string $x$, complete \emph{reconstruction} may be
unnecessary. We propose \emph{property testing} in the trace reconstruction model, which, to our
knowledge, has not yet been studied. The goal is simply to make a decision about $x$ from its
traces, without reconstructing $x$ completely.  In terms of the standard biological application, we
wish to make a decision about a DNA string, from a number of corrupted copies, \emph{without}
reconstructing it.  We prove the first non-trivial property testing results for trace
reconstruction, which follow from an equivalence to distribution testing under the parity trace.

\paragraph*{This paper.} Standard distribution testing algorithms make their decisions based on the
\emph{histogram}, which counts the number of times each element of the domain occurs in the sample.
The common challenge in each of the problems above is that, to apply distribution testing, the
tester needs to handle a certain structured mislabeling of the sample that prevents it from
constructing the histogram. \emph{Distribution testing under the parity trace} is an abstraction of
this challenge. The purpose of this paper is to relate this model to the problems above, and to
begin understanding the model by proving tight bounds on the most fundamental distribution testing
task, \emph{testing uniformity}. This is significantly more difficult to analyze than in the
standard model, and we believe it is necessary before advancing to some more difficult problems that
we will discuss.

\subsection{Distribution Testing under the Parity Trace}

Let us now define distribution testing under the parity trace.  Let $\Pi$ be a property (\ie set) of
probability distributions over $\bN$.  As in standard distribution testing, for a distribution $\pi$
over $\bN$, a distribution tester under the parity trace must accept (with probability $2/3$) any
input $\pi \in \Pi$, and reject (with probability $2/3$) any input $\pi$ that is $\epsilon$-far from
$\Pi$, meaning that its distance to any $\pi' \in \Pi$ is at least $\epsilon$. (Standard
distribution testing often uses the total variation distance, but we will see that this is not the
natural choice in this case.)  Instead of receiving a sample $S$ from the distribution $\pi$, the
tester receives the \emph{parity trace} of $S$, denoted by $\trace(S)$, defined as follows. For any
multiset $S \subset \bN$ of size $m$, put $S = \{ x_1, x_2, \dotsc, x_m \}$ in sorted order $x_1
\leq x_2 \leq \dotsm \leq x_m$, and write $\parity(x) \define (x \mod 2)$ for the parity of $x$.
Then
\[
  \trace(S) \define (\parity(x_1), \parity(x_2), \dotsc, \parity(x_m)) \,.
\]
For example, on sample $S = \{ 5, 1, 6, 2, 4, 2\}$, the algorithm receives $\trace(S) = 100010$,
which is the string of parities of $(1,2,2,4,5,6)$. Notice that, for example, the uniform
distribution over $\{1, \dotsc, n\}$ and the uniform distribution over $\{n+1, \dotsc, 2n\}$ are
indistinguishable under the parity trace when $n$ is even, although they have total variation (TV)
distance 1, so it is not obvious \emph{a priori} which distribution testing tasks are even possible
under the parity trace.

\paragraph*{Testing uniformity.}
To begin understanding the parity trace, consider the problem of
\emph{testing uniformity} (\eg \cite{GR00,Pan08,ADK15,DKN15b}, see \cite{Can22}). The goal is
to accept the uniform distribution over $[2n]$ and reject the distributions over $[2n]$ that are
$\epsilon$-far from uniform.  One may see that testing uniformity under the parity trace is indeed
possible, even with respect the TV distance, using a coupon-collector argument. After receiving a
trace of size $\Theta(n \log n)$, with high probability the trace either included every element of
the domain $[2n]$, or it can safely reject.  If the trace included every element of the domain,
the algorithm can deduce the exact identity of each sample point, and simulate the standard
distribution tester, giving a bound of $O(n \log n + \sqrt{n}/\epsilon^2 )$, which
follows from the tight $\Theta(\sqrt n / \epsilon^2)$ bound in the standard model
\cite{Pan08,VV17,DGPP18,DGPP19}.

It is not immediately clear whether a sample of size $o(n)$ suffices. The main technical
contribution of this paper is to establish tight bounds on this problem: sublinear sample size is
indeed achievable, but the problem exhibits a gap between the standard model and the parity trace
model.  We discuss the proof in \cref{section:intro-confused-collector,section:proof-overview},
as the confused collector model will serve as a warm-up.

\begin{theorem}[Informal; see \cref{thm:intro-main}]
\label{thm:intro-main-informal}
Testing uniformity of distributions on domain $[2n]$ under the parity trace, with respect to the TV
distance, requires sample size $\widetilde \Theta\left(\left(\frac{n}{\epsilon}\right)^{4/5} +
\frac{\sqrt n}{\epsilon^2}\right)$.
\end{theorem}

\paragraph*{Edit distance.}
\cref{thm:intro-main-informal} uses the TV distance, but this is not always possible.  Two
distributions may have TV distance 1 while being indistinguishable under the parity trace, so TV
distance is not the most natural metric, and we require a new one in order to relate the parity
trace model to the other problems discussed in this paper.  We define the \emph{edit distance}
pseudo-metric for distributions, which has the desired property that two distributions $\pi$ and
$\pi'$ are indistinguishable under the parity trace, if and only if the edit distance is 0.  We
think of a distribution $\pi$ over $\bN$ as an alternating ``fractional string''\!\!,
\[
  1^{\pi(1)} 0^{\pi(2)} 1^{\pi(3)} \dotsm 1^{\pi(2i-1)} 0^{\pi(2i)} \dotsm
\]
where $b^{p}$ indicates that $b$ is repeated $p$ times (which is fractional). Then the parity trace
of size $m$ from $\pi$ is obtained by sampling $m$ independent random characters proportional to
their fractional number of repetitions $p$, and concatenating them in order. The
distribution of the parity trace is invariant under certain ``free'' edit operations, like $b^p \to
b^{p/2} b^{p/2}$, $a^p c^q \to a^p b^0 c^q$, or $a^p b^0 c^q \to a^p c^q$, while other ``expensive''
edit operations like $b^p c^q \to b^{p-\delta} c^{q+\delta}$ may change the distribution of the
trace. The edit distance is the cost of transforming one distribution into another; see
\cref{def:distribution-edit-distance}.

\subsection{Distribution Testing with a Confused Collector}
\label{section:intro-confused-collector}

We introduce the \emph{confused collector} to model distribution testing problems where the
algorithm receives a random sample $S$ that has been systematically mislabeled; recall the examples
on the first page. To formalize the problem, imagine that for each two elements $x$ and $y$ in the
domain, there is some probability that all appearances of $x$ and $y$ in the sample $S$ have been
joined and counted together.  These joins must be transitive, so the probabilities that the pairs
$(x,y)$, $(y,z)$, or $(x,z)$ are joined are not independent. That means there must be some
structured random process that joins the domain elements, which we choose to model as follows.

Let $p$ be a distribution over a (finite) structured domain $V$, whose structure is given
by a ``base graph'' $G = (V,E)$. For example, $G$ could be a tree representing the taxonomy of a
collection of fauna.  The distribution testing algorithm has a parameter $\eta \in (0,1]$ called the
\emph{resolution} (representing the accuracy of the classifier), and it receives a random sample $S$
of size $m$ produced as follows. First sample a subgraph $H$ of $G$ by including each edge $uv \in
E$ with probability $1-\eta$, and let $C_1, \dotsc, C_t$ be its connected components.  For each
$C_i$, let $c_i \in C_i$ be an arbitrary representative of the component. Then sample a set $S'$ of
$m$ independent points from $p$ and label each $s \in S'$ with the representative $c_i$ of its
component. The resulting sample $S$ is given to the algorithm.  Note that, with resolution $\eta =
1$, the graph $H$ is an independent set and each element is given its proper label, so the model
becomes the standard distribution testing model.

Given a property $\Pi$ of distributions over $V$ and parameter $\epsilon$, a tester for $\Pi$, with
resolution $\eta$ and sample complexity $m$, must accept (with probability\footnote{Unlike standard
testers, we cannot simply repeat the tester to boost the success probability, which depends partly
on the resolution $\eta$.} $2/3$) any distribution in $\Pi$, and reject (with probability $2/3$)
any distribution that is $\epsilon$-far in TV distance from all distributions in $\Pi$.

\paragraph*{Results.}
We get results when the base graph $G$ is a cycle or path, which capture situations where the domain
is $[n]$ and domain elements are distinguishable only if a random ``separator'' occurs between them
(like the fossil example given on the first page, or if the sample labels have been randomly hashed
by a monotone hash function\footnote{A monotone hash function is one that preserves the order of the
keys, see \eg \cite{AFK23}}).

\begin{theorem}[Informal; see \cref{thm:confused-collector-main}.]
\label{thm:intro-confused-collector-main}
Let $G = (V,E)$ be a path or cycle on $n$ vertices, let $\epsilon \geq \widetilde
\Omega(n^{-1/4})$ and $\eta \geq \widetilde \Omega(n^{-1/5} \epsilon^{-4/5})$. Then testing
uniformity requires $\widetilde O\left(\frac{\sqrt n}{\epsilon^2
\eta^{3/2}}\right)$ samples.
\end{theorem}

This interpolates between the optimal $\Theta(\sqrt n /\epsilon^2)$ bound for uniformity testing
with resolution $\eta = 1$, and $\widetilde O\left((n/\epsilon)^{4/5}\right)$ when $\eta$ is as
small as the theorem allows\footnote{Note that a lower bound on $\eta$ in the theorem is necessary.
For example, a sample of woodland fauna labeled by the authors would have resolution $\eta = 0$ and
no decisions could be made based on this, regardless of sample size.}.

\paragraph*{Techniques.} We describe our techniques for
\cref{thm:intro-main-informal,thm:intro-confused-collector-main} in more detail in
\cref{section:proof-overview}, but briefly mention the main idea here.
\cref{thm:intro-confused-collector-main} serves as a sort of warm-up to
\cref{thm:intro-main-informal}, which is proved by considering a similar problem on the
\emph{weighted} cycle, although the confused collector poses its own separate challenges in handling
the resolution parameter $\eta$.

Let $\bm T_i$ denote the multiplicity of element $i$ in the sample.  A standard simplification is to
assume that $\bm T_i$ is distributed as the independent Poisson $\Poi( m \cdot p(i) )$. The random
graph $\bm H$ introduces dependencies in the observed variables, and we let $\bm \Phi$ be the random
Boolean matrix describing these dependences, with $\bm \Phi_{i,j} = 1$ iff vertices $i,j$ belong to
the same connected component. Our proof boils down to an analysis of the random quadratic form $\bm
T^\top \bm \Phi \bm T$.  While concentration bounds for quadratic forms $\bm X^\top A \bm X$ have
been studied (including Hanson-Wright type inequalities, see \eg \cite{GSS21}), we are not aware of
bounds when the matrix $A$ is itself random, and inequalities of the type we require may be of
independent interest. See \cref{section:proof-overview} for more details.

\paragraph*{Related Work.} Similar models have been proposed independently in the recent literature
on machine learning and distribution testing, with similar motivations. See \cite{FKKT21} and
references therein for a discussion of the applied literature. \cite{FKKT21} propose a different
model from ours, where the partition of the domain is more general, but it is resampled
independently for each sample point\footnote{The reason for the difference is, briefly, that
\cite{FKKT21} assume sample points may be labeled by different entities with different
classifications, while we assume sample points are labeled by one entity with imperfect
classification.}, and they study questions of learning.  In distribution testing,
\cite{GR22,CFGMS22} study a model where the sample contains ``huge objects''\!\!, which themselves
support queries, so again the algorithm is not given the histogram, and must perform
queries on its sample.  \cite{CW21} studies systematic mislabelings that are guaranteed to be
permutations. In \cite{CW20}, the goal is to test if there exists a partition into intervals that
makes the input distribution $p$ equal to a reference distribution $q$.  Other models with imperfect
information about the samples include locally private testing \cite{GR18,She18,ACFT19} and inference
under information constraints \cite{ACT19,ACT20,ACFST21}.

\subsection{Distribution-Free Sample-Based Property Testing}
\label{section:intro-sample-based-testing}

\newcommand{\VC}{\mathsf{VC}}
We are interested in the fundamental \emph{testing vs.~learning} question of \cite{GGR98},
especially in the distribution-free sample-based property testing model corresponding to standard
PAC learning. This is essential for some proposed applications of property testing, like model
selection (\ie selecting an appropriate hypothesis class $\cH$ for learning) \cite{GGR98}.  Formal
connections between property testing and distribution testing, which we believe are essential for
understanding the testing vs.~learning question, were first articulated by \cite{GR16}, but their
results applied only to \emph{symmetric} properties of functions (\ie properties closed under
permutations on the domain).

As noted in \cite{GGR98,BFH21}, testing vs.~learning is essentially \emph{testing vs.~VC dimension},
since the sample size required for PAC learning a hypothesis class $\cH$ (ignoring the error
$\epsilon$) is $\Theta(\VC)$, where $\VC$ is the VC dimension of $\cH$.  Therefore, the goal is to
determine which classes $\cH$ can be tested using $o(\VC)$ samples.  For many important hypothesis
classes, including halfspaces over $\bR^n$, and $k$-alternating functions over $\bR$, \cite{BFH21}
show a lower bound of $\Omega\left(\frac{\VC}{\log \VC}\right)$ by defining the ``lower VC
dimension'' and using it to construct a reduction from \emph{support-size distinction} (see
\cite{RRSS09,VV11,WY19}), which is the problem of deciding whether a distribution on $[n]$ has
support size at most $\alpha n$ or at least $\beta n$. The bound is tight in some cases, due to
an $O\left(\frac{\VC}{\log \VC}\right)$ bound of \cite{GR16} for some symmetric properties, reducing
in the other direction to testing support size.

This leaves a gap between the sample size required for testing and learning many of the most
important hypothesis classes, like halfspaces. As in \cite{GR16}, we consider the gap between
$\Omega\left(\frac{\VC}{\log \VC}\right)$ and $O(\VC)$ to be significant; firstly because it leaves
open the question of whether testing can be done with sample size \emph{sublinear} in the sample
size required for PAC learning, and secondly because of the relationship to distribution testing,
especially support-size estimation, where this log factor is surprising and important
\cite{RRSS09,VV11,WY19}. Unfortunately, the technique of \cite{BFH21} cannot close this gap,
because, informally speaking, the tightness of the relationship between distribution testing and
property testing reaches its limit at the symmetric properties.

Our goal is to develop a stronger relationship between distribution testing and property testing
that surpasses this limitation. Distribution testing under the parity trace is a step towards this
goal. Consider the (non-symmetric) property of $k$-alternating functions, which are the functions
$\bR \to \zo$ which alternate between 0 and 1 at most $k$ times (equivalently, the class of unions
of $k$ intervals), studied in \cite{KR00,Nee14,BBBY12,BH18,CGG+19,BFH21}, for which the testing
vs.~learning question remains open. A first example of our technique is the following:

\begin{theorem}[Informal; see \cref{thm:testing-support-k}]
\label{thm:testing-support-k-informal}
Let $m_1(k,\epsilon)$ be the sample size required to test if a distribution has support size $k$, or
is $\epsilon$-far in \emph{edit distance} from having support size $k$, under the parity trace. Let
$m_2(k,\epsilon)$ be the sample size required to test if a function is $k$-alternating in the
distribution-free sample-based model. Then $m_1(k,\epsilon) = \Theta(m_2(k, \epsilon))$.
\end{theorem}

This is the first tight relationship between distribution testing and property testing for a
\emph{non-symmetric} property, and it is only a special case of a more general equivalence between
distribution testing and testing \emph{density properties}, explained below, which is required for
our results in the trace reconstruction model. The appearance of the edit distance highlights its
importance for applications of the parity trace. The authors disagree on what the correct sample
size $m_1(k,\epsilon)$ in this theorem should be\footnote{In fact, this doesn't seem to be known
even in the \emph{standard} model: the best upper bounds we could find are
$O\left(\tfrac{n}{\epsilon^2 \log n}\right)$ and $O(n/\epsilon)$, compared to a lower bound of
$\Omega\left(\tfrac{n}{\epsilon \log n}\right)$.}, which we will study in future work; the current
paper focuses on the simpler problem of testing uniformity, which is already significantly more
challenging to analyze in the parity trace model than the standard model. But, even without knowing
$m_1(k,\epsilon)$, we use \cref{thm:testing-support-k-informal} to recover many of the bounds of
\cite{BFH21} using a different proof that has stronger consequences for the testing vs.~learning
question.  We state the bounds for $k$-alternating functions and halfspaces, but we also recover the
bounds for intersections of halfspaces, and decision trees\footnote{Our statement includes a
dependence on $\epsilon$, which \cite{BFH21} does not. Ours follows from bounds on the distribution
testing problem, but the $\epsilon$ dependence can be appended to the \cite{BFH21} results in a
standard way (as in \cite{ES20}).}. (See \cite{MORS10,BBBY12,Har19,CP22} for other prior work on
testing halfspaces.)

\begin{theorem}[See \cite{BFH21}]
\label{thm:bfh-simple}
Distribution-free sample-based testing
$k$-alternating functions on domain $\bR$ requires $\Omega(\tfrac{k}{\epsilon \log k})$
samples, and testing halfspaces on domain $\bR^n$ requires $\Omega( \tfrac{n}{\epsilon \log n})$
samples.
\end{theorem}

Unlike the technique of \cite{BFH21}, it is possible that our technique can lead to better answers
for testing vs.~learning for $k$-alternating functions, halfspaces, and others.  Better lower bounds
on distribution testing under the parity trace would imply better lower bounds for $k$-alternating
functions, halfspaces, intersections of halfspaces, and decision trees. On the other hand, an
$o(\VC)$ \emph{upper bound} on (say) testing halfspaces, would imply an analog of the surprising
$o(n)$ bounds of \cite{VV11,WY19} for distinguishing support size $\leq \alpha n$ from $\geq n$,
which would hold \emph{even under the parity trace}, where the tester does not know the identities
of the sample points.

To clarify the connection between distribution testing and distribution-free sample-based property
testing, we expand our view of distribution-free property testing to allow properties of
\emph{labeled distributions}. A labeled distribution on is a pair $(f, \cD)$ of a function $f$ and a
distribution $\cD$.  The idea is that one may wish to test not only a property of a function $f$,
but a joint property of the function $f$ and probability distribution $\cD$. (We also point the
reader to a different interesting type of joint function-distribution testing in \cite{RV23,GKK23}.)

For example, we may wish to test not only whether $f$ is $k$-alternating, but that it also evenly
partitions $\cD$ into uniform monochromatic intervals. We call these the \emph{uniformly
$k$-alternating} functions, and we get a tight result for testing uniformly $2k$-alternating
functions (assuming the input is promised to be $2k$ alternating). This adds to the short list of
positive results in distribution-free sample-based testing \cite{GR16,BFH21,RR20,RR21}.

\begin{theorem}[Informal; see \cref{thm:intro-main-testing}]
\label{thm:intro-main-testing-informal}
Let $f$ be $2k$-alternating. Then testing if it is \emph{uniformly} $2k$-alternating with respect to
the TV distance requires $\widetilde \Theta((k/\epsilon)^{4/5} +\sqrt{k}/\epsilon^2)$ samples.
\end{theorem}

(For the more challenging non-promise version of this problem, we get a bound of
$O(\tfrac{k}{\epsilon} + \tfrac{k}{\epsilon^2 \log k})$ by defining a suitable
``testing-by-learning'' reduction for labeled distributions and using the tolerant uniformity tester
of \cite{VV17}; see \cref{thm:uniform-k-alternating-learning-bound}). The proof of
\cref{thm:intro-main-testing-informal,thm:testing-support-k-informal} use an equivalence to
distribution testing under the parity trace that holds in general for a natural class of labeled
distributions that we call \emph{density properties}.

\paragraph*{Density properties.} Informally\footnote{For simplicity, this discussion ignores the
possibility of infinitely many alternation points.}, every Boolean function $f : \bR \to \zo$ has a
unique set of \emph{alternation points} in $\bR$ where it changes value from 0 to 1 or vice-versa. A
\emph{density property} is a set of labeled distributions where membership of $(f, \cD)$ is
determined by its \emph{density sequence}: the sequence of probability masses $\cD(a,b]$ where $a,b$
are consecutive alternation points of $f$.

$k$-Alternating and uniformly $k$-alternating functions are both definable as density properties,
but there are many other interesting examples. The difficulty in testing density properties is that
the tester does not know \emph{which} interval of alternation points a sample belongs to.  Given two
sample points $x, y \in \bR$, the tester does not know if $x,y$ belong to the same interval or
different intervals, unless $f(x) \neq f(y)$, or $f(x) = f(y)$ and there is another sample point $z$
between $x$ and $y$ with $f(z) \neq f(x)$. Prior work has used \emph{queries} to overcome this
difficulty \cite{CGG+19}, but this is not possible in the sample-based model. 

Distribution testing under the parity trace captures this difficulty: testing density properties is
essentially equivalent to testing distributions under the parity trace. For any density property
$\Xi$, let $\Pi(\Xi)$ be the set of density sequences (\ie probability distributions) that define
$\Xi$. Using Ramsey theory (inspired by \cite{Fis04, DKN15a}, see also \cite{CW20}), we prove:

\begin{lemma}[Informal; see \cref{lemma:labeled-to-linear-trace-reductions}]
\label{lemma:labeled-tester-to-parity-tester-informal}
Testing $\Xi$ in the labeled-distribution model, with respect to an appropriate analogue of
\emph{edit distance}, is equivalent to testing $\Pi(\Xi)$ under the parity trace with respect to the
edit distance.
\end{lemma}

\paragraph*{Techniques.}
The key contributions here are the definitions of edit distance and the parity trace, which allow
sample-based property testing to be related to distribution testing using an application of Ramsey
theory in \cref{lemma:labeled-tester-to-parity-tester-informal}. The main results in this section
(including the recovery of the results of \cite{BFH21}) then follow by reductions that mainly rely
on properties of the edit distance.

\subsection{Property Testing for Trace Reconstruction}
\label{section:intro-trace-reconstruction}
\newcommand{\del}{\mathsf{del}}

We now turn to property testing for trace reconstruction, which had interested us separately from
the other problems in this paper, and the formal connection we present here was unexpected. In the
trace reconstruction problem (with the deletion channel), there is a string $x \in \zo^N$ and a
deletion rate $\delta \in (0,1)$.  A \emph{trace} $\bm T$ of $x$ is obtained by deleting each
character of $x$ independently with probability $\delta$ and taking the resulting substring. For
example, a trace of $x = 110011001100$ might look like $11110000$ or $101010$. The goal is to
reconstruct $x$ using as few independent traces as possible (see references above).

We are interested in making decisions about $x$ \emph{without} completely reconstructing it, so we
propose property testing in the trace reconstruction model. For a property $\Psi$ of strings
$\zo^N$, the algorithm should accept (with probability $2/3$) strings $x \in \Psi$, and reject
(with probability $2/3$) strings that are \emph{far} from $\Psi$. A natural choice of metric is
the (relative\footnote{The relative edit distance between two strings of length $N$ is $\frac{1}{N}$
times the edit distance.}) edit distance on strings, which is the standard choice for
\emph{approximate} trace reconstruction \cite{CDL+22}. The edit distance on strings is closely
related to our notion of edit distance on probability distributions.

To measure the complexity of a trace tester, we consider both the \emph{number of traces}, and the
\emph{expected size} of each trace.  Trace reconstruction is usually studied with \emph{constant}
deletion rate $\delta$, corresponding to traces of expected size $\Theta(N)$.  For testing, we hope
to permit extremely high deletion rates, so that traces have expected size $\rho N = o(N)$ where
$\rho = 1-\delta$ is the \emph{retention rate} (which puts our study in the same low-retention-rate
regime as the recent independent work \cite{CDL+23} on trace reconstruction). This is consistent
with conventional property testing problems, where the goal is to make a decision while seeing less
than a constant fraction of the input. We relate this problem to distribution testing under the
parity trace and labeled-distribution testing, and give results for three trace testing problems,
which appear to be the first results on this type of problem (the most similar problem in prior work
is distinguishing between two arbitrary strings $x, y$ that are close in Hamming- or edit distance
\cite{GSZ22}).

\paragraph*{Results.}

To initiate the study of property testing for trace reconstruction, we prove bounds on testing three
basic properties of strings.  For $n \in \bN$, we say that $x \in \zo^N$ is an \emph{$n$-block
string} if $x$ consists of at most $n$ consecutive \emph{blocks}, where a block is a (maximal)
all-1s string or all-0s string. The \emph{uniform} $n$-block strings are those consisting of $n$
blocks of equal length. We give results for:
\begin{enumerate}[itemsep=-2pt]
\item Testing if an $n$-block string is a \emph{uniform} $n$-block string;
\item Testing if an arbitrary string is a \emph{uniform} $n$-block string; and
\item Testing if an arbitrary string is an $n$-block string.
\end{enumerate}
These results use general two-way reductions between trace testing and distribution testing under
the parity trace
(\cref{lemma:trace-tester-from-parity-trace-tester,lemma:trace-testing-general-lower-bound}).  The
na\"ive application of our reduction gives an upper bound for testing with a \emph{single} trace,
which corresponds to the single-trace approximate reconstruction problem whose study was initiated
in concurrent and independent work \cite{CDL+23}. Our main application uses an additional trick to
apply the reduction, which gives a bound for an arbitrary number of traces.  Observe that even when
the number of blocks $n$ is large, \eg $n = \Theta(N)$, we can still test the property with a
\emph{single} trace of sublinear size (\ie deletion rate $\delta = 1 - o(1)$).

\begin{theorem}[Informal; see \cref{thm:trace-testing-uniform-k,thm:multiple-trace-lower-bound}]
\label{thm:intro-trace-testing-uniform}
Suppose $x \in \zo^N$ is promised to be an $n$-block string. For any $k$, trace testing
whether $x$ is a uniform $n$-block string, or $\epsilon$-far from a uniform $n$-block string in
relative edit distance, can be done with $k$ traces of expected size $\rho N = \widetilde
O\left(\frac{n^{4/5}}{k^{1/5}\epsilon^{4/5}} + \frac{\sqrt n}{\sqrt k \epsilon^2} \right)$.
Meanwhile, for large enough $N$, we must have $k \rho N =
\widetilde \Omega\left(\frac{n^{4/5}}{\epsilon^{4/5}} + \frac{\sqrt n}{\epsilon^2}\right)$.
\end{theorem}

We find it convenient to measure complexity using the expected size of each trace, but one may
rephrase our result in more conventional trace reconstruction language by saying that for fixed $k$,
if $\rho_k \cdot N$ is the bound on expected trace size, then for all retention rates $\rho \geq
\rho_k$, the number of traces required for testing is at most $k$. Increasing the number of traces
$k$ allows the tester to handle smaller retention rates, but the total number of observed bits $k
\rho N$ will increase. 

For the final two results, we do not have tight bounds for the corresponding distribution testing
problems under the parity trace, but we get non-trivial bounds that beat the coupon-collector
argument, almost ``for free'' from the theory we have developed.  For the labeled-distribution
testing model (\cref{section:intro-sample-based-testing}), we show that a ``testing-by-learning''
reduction holds, similar to the standard reduction of \cite{GGR98}, by defining a ``proper
learner-and-verifier pair'' that uses a distribution testing task instead of the ``verification
step'' of \cite{GGR98}. We then use \cref{lemma:labeled-tester-to-parity-tester-informal}, and the
relationship to trace testing, to get a general ``testing-by-learning'' technique for trace testing. 

\begin{theorem}[Informal; see \cref{thm:trace-testing-uniform-k-no-promise}]
\label{thm:trace-testing-uniform-k-no-promise-informal}
Testing whether $x \in \zo^N$ is a uniform $n$-block string, or $\epsilon$-far in relative edit
distance from the uniform $n$-block strings, can be done with a single trace of expected size $\rho
N = O\left(\frac{n}{\epsilon} + \frac{n}{\epsilon^2 \log n }\right)$.
\end{theorem}

\begin{theorem}[Informal; see
\cref{thm:trace-testing-support-n,thm:trace-testing-support-n-lower-bound}.]
\label{thm:intro-trace-testing-support-n-informal}
Testing whether $x \in \zo^N$ is an $n$-block string, or $\epsilon$-far in relative edit distance
from all $n$-block strings, can be done with a single trace of expected size $\rho N =
O(n/\epsilon)$, while for large enough $N$, any trace tester using $k$ traces must satisfy $k \rho N =
\Omega(n / \log n)$.
\end{theorem}

\cref{thm:trace-testing-uniform-k-no-promise-informal} uses the tolerant tester for uniformity from
\cite{VV17a} in the ``verification step'' of the testing-by-learning reduction, while
\cref{thm:intro-trace-testing-support-n-informal} uses the $O(k/\epsilon)$ upper bound for testing
$k$-alternating functions which follows from the VC dimension. We find these bounds somewhat
mysterious, because our testing-by-learning reduction for trace testing goes through the
non-constructive Ramsey theory argument of \cref{lemma:labeled-tester-to-parity-tester-informal} and
therefore the trace testers, which do not know the positions of the characters of the trace,
are obtained non-constructively from a labeled-distribution learner and verifier that strongly rely
on knowing the absolute positions of the sample points.

\subsection{Proof Overview}
\label{section:proof-overview}

We briefly describe our proofs for testing uniformity under the parity trace and with a confused
collector, \cref{thm:intro-main-informal,thm:intro-confused-collector-main}.

\paragraph*{Upper bounds.}
Let us review the standard uniformity tester \cite{GR00,DGPP19} (see also \cite{Can22}).  Let $p$ be
the input distribution over $[n]$. For a sample $S$ of size $m$, let $X_i$ be the multiplicity of
element $i$ in $S$. The tester counts the number of ``collisions'' in the sample: it
computes $Y \define \frac{1}{m(m-1)} \sum_{i=1}^n X_i (X_i - 1)$, and rejects if this is too
large. This works because $\Ex{\bm{Y}} = p^\top p = \|p \|_2^2$, which is large when $p$ is far from
uniform. Now we describe the tester for the confused collector.  For input distribution $p$ on
domain $\bZ_n$ (which are the vertices of the path or cycle), we use the standard simplification
that element $j$ occurs in the sample with multiplicity $\bm T_j \sim \Poi(m p_j)$ independently of
the other elements. Now redefine $\bm X_i$ as the number of sample points contained in the
$i^{th}$ connected component of $\bm H$, which the tester cannot distinguish: the $\bm X_i$
variables remain Poisson, but they are not independent.  The tester computes a ``collision count'',
as in the standard algorithm:
\[
  \bm Y \define \frac{1}{m} \sum_i \bm X_i (\bm X_i - 1)
    = \frac{1}{m}\left(\bm T^\top \bm \Phi \bm T  - \|\bm T\|_1 \right) \,,
\]
where $\bm \Phi$ is the random Boolean matrix with $\bm \Phi_{i,j} = 1$ iff $i,j$ belong to the same
connected component of $\bm H$. The expected value is $\Ex{\bm Y} = m \cdot p^\top \phi p$ where
$\phi = \Ex{\bm \Phi}$, and we show that this is again large when $p$ is far from uniform, using
spectral analysis of the matrix $\phi$ which is either Toeplitz (for paths) or circulant (for
cycles). To complete the analysis, we require a concentration inequality for the random quadratic
form $\bm T^\top \bm \Phi \bm T - \|\bm T\|_1$, which we obtain as long as $p$ is not too ``highly
concentrated'' in any interval (which the algorithm can test separately); see
\cref{lemma:confused-collector-concentration}:
\begin{equation}
\label{eq:intro-confused-collector-concentration}
  \Pr{ | \bm Y - \Ex{\bm Y} | \geq t } \leq \frac{\|p\|_2^2}{\eta t^2} \cdot \poly\log n \,.
\end{equation}
Extending the result to the parity trace is more challenging. On domain $[2n]$, we separate the
input distribution $\pi$ into the ``odd part'' $p$ and ``even part'' $q$, so 
$\pi = \pi(p,q) = (p_1, q_1, p_2, q_2, \dotsc, p_n, q_n)$.  The tester receives a trace of the form
\[
  \trace(S) = 1^{X_1} 0^{Z_1} 1^{X_2} 0^{Z_2} \dotsc 1^{X_t} 0^{Z_t} \,,
\]
where each $X_i, Z_i$ is the length of a consecutive ``run'' of 1s or 0s in the trace (\ie $X_i, Z_i
> 0$ except we may have $X_0=0$ or $Z_t=0$). By analogy to the standard tester, the natural thing to
try is to compute the number of ``collisions'' $\sum_{i=1}^n X_i(X_i-1) + \sum_{i=1}^n Z_i(Z_i-1)$
and pray that it works, which it does, more or less. Our tester considers the runs of 1s and 0s
separately: first, we think of $p$ as being a distribution over the vertices of a cycle, with $q$
giving weights to the edges. If $q$ was uniform, the analysis for the confused collector would now
apply, but it may not be.

To handle this, we define the \emph{uniform conjugate} of $q$ and denote it by $\tilde p$.
Informally, $\tilde p$ is the ``worst case'' instance of $p$ that makes every connected component
of $\bm H$ (sampled according to the weights determined by $q$) have the same expected mass $\tau$,
which would minimize $\Ex{ \bm Y }$. We essentially calculate a closed form solution for $\tilde p$
with $\tau = \frac{1-\|q\|_1}{\sum_{i=1}^n \tanh(mq_i/2)}$ by approximating the process of sampling
components of $\bm H$ with a Markov process (for which the use of a cycle instead of a path is
helpful). Then we write $p = \tilde p + z$ and, crucially, use the deviation $z$ from the uniform
conjugate to control both the mean and variance of $\bm Y$. We get an analog of equation
\eqref{eq:intro-confused-collector-concentration} that holds under some conditions on $p,q$ that the
algorithm can test separately; see \cref{lemma:concentration-parity-trace}:
\begin{equation}
\label{eq:intro-parity-trace-concentration}
  \Pr{ |\bm Y - \Ex{\bm Y}| \geq t } \leq \frac{\tfrac{1}{m} + z^\top \phi z}{t^2} \cdot \poly\log n
\,.
\end{equation}
The main condition that the algorithm must test separately is that $p$ is not too ``highly
concentrated'' relative to $q$, meaning that there is no interval where $p$ and $q$ are both
sufficiently large but $p$ is much larger than $q$. The algorithm repeats these tests with the roles
of $p$ and $q$ reversed.

\paragraph*{Lower bound.}
To get the lower bound in \cref{thm:intro-main-informal}, consider an adversary who flips a random
bit $\bm Z$ and gives the algorithm a input distribution sampled from ``meta-distribution''
$\cD_{\bm Z}$, where $\cD_0$ and $\cD_1$ are constructed out of \emph{dominoes} as follows.  A
domino is a 4-element piece $(p_i, q_i, p_{i+1}, q_{i+1})$ of a distribution $\pi(p,q) = (p_1, q_1,
p_2, q_2, \dotsc, p_n, q_n)$, so that $\pi(p,q)$ on domain $[2n]$ is made of $n/2$ dominoes. We use
the dominoes $(\tfrac{1}{2n}, \tfrac{1}{2n}, \tfrac{1}{2n}, \tfrac{1}{2n})$,
$(\tfrac{1-\epsilon}{2n}, \tfrac{1}{2n}, \tfrac{1+\epsilon}{2n}, \tfrac{1}{2n})$, and
$(\tfrac{1+\epsilon}{2n}, \tfrac{1}{2n}, \tfrac{1-\epsilon}{2n}, \tfrac{1}{2n})$. $\cD_0$ contains
only the uniform distribution ($n/2$ copies of the first domino), while $\cD_1$ is obtained by a
sequence of $n/2$ random choices from the last two dominoes.

We use an information-theoretic argument inspired by \cite{DK16}, to show that the algorithm
receives insufficient information about $\bm Z$ unless it receives $\widetilde
\Omega((n/\epsilon)^{4/5})$ samples. The tester gains no information about $\bm Z$ from any domino
receiving fewer than 3 sample points. We use the chain rule of information over small-enough
sequences of dominoes, and use bounds on the number of dominoes receiving 3 sample points to bound
the information from each small-enough sequence.

\subsection{Discussion \& Open Problems}
\label{section:intro-discussion}

The reader may notice three unfortunate negative qualities of this paper: The upper bounds have
$\widetilde O(\cdot)$ instead of $O(\cdot)$; the testing algorithms have more than 1 step; and the
number of pages seems excessive. We suspect that these three birds can be killed with one stone, if
one could prove tighter, unconditional concentration bounds on the quadratic forms $\bm T^\top \bm
\Phi \bm T$.


Regarding the testing vs.~learning question, the next step is to prove tight bounds on testing
support size under the parity trace, which would either give better lower bounds for $k$-alternating
functions (and therefore halfspaces and intersections of halfspaces) or possibly a surprising
$O(k / \log k)$ upper bound for $k$-alternating functions. We intend to study this in follow-up
work.

Our results for the confused collector were limited to paths and cycles, due to the connection to
the parity trace, but we suspect that a similar upper bound holds for trees, which we think would be
the next most natural problem in this model, given the ubiquity of tree-structured data.

Density properties are a natural class of properties where property testing is equivalent to
distribution testing under the parity trace. Adapting other distribution testing results, like
testing monotonicity (\cite{BKR04,CDGR18}), to the parity trace model, would imply new results in
distribution-free sample-based testing (for labeled-distributions), and the trace reconstruction
model.

We consider property testing in the trace reconstruction model to be one of the main conceptual
contributions of this paper. We have shown that testing properties of $n$-block strings is related
to distribution testing under the parity trace and testing density properties in the
labeled-distribution testing model. Other interesting properties to study would be
subsequence-freeness (with non-binary alphabet), which could possibly build on recent work
in sample-based testing \cite{RR21}, or testing regular languages, which are testable in the
standard query model (\eg \cite{AKNS01, BS21}) and which are already related to trace reconstruction
\cite{Chas21upper}.

\section{Preliminaries and Common Framework for Upper Bounds}
\label{sec:collision-based-testing-general}

In this section, we give the formal definitions for the parity trace and confused collector models
of distribution testing, and we introduce a common terminology and framework for analyzing our
algorithms in these models. The section is organized as follows: \cref{section:framework-notation}
introduces notation we use throughout the paper.
\cref{section:common-framework-confused-collector-def,section:parity-trace} define the confused
collector and parity trace models of distribution testing, respectively.
\cref{section:path-cycle-structured-rvs} introduces unifying vocabulary that views these two models
as outcomes from Poisson random variables on certain path- and cycle-structured domains. Then
\cref{section:shared-analysis} uses this vocabulary to establish generic results that will be
specialized into our upper bounds for the confused collector and parity trace models in the
subsequent sections.

\subsection{Notation}
\label{section:framework-notation}

In this paper, $\log x$ denotes the natural logarithm of $x$. $\bN$ denotes the set of positive
integers, \ie it does not include 0. For any $x$, we write $\bZ_{> x}$ for the set of integers
greater than $x$, and $\bZ_{< x}, \bZ_{\geq x}, \bZ_{\leq x}$ are defined similarly.
We denote random variables by boldface symbols, \eg $\bm{X}$.
We write $x = a \pm b$ as a shorthand for $a-b \le x \le a+b$. For an event $E$, $\ind{E}$ is the
indicator variable for $E$, which takes value 1 if and only if $E$ occurs.

For a distance metric $\dist(\cdot, \cdot)$ on a domain $\cX$, an element $y \in \cX$, and a set $X
\subseteq \cX$, we write
\[
  \dist( y, X ) \define \inf_{x \in X} \dist(y,x) \,.
\]
For a probability distribution $\cD$ over (countable) domain $\cX$ and any set $S \subseteq \cX$, we
write $\cD(S) = \sum_{x \in S} \cD(x)$.

Given a probability distribution $\pi$ and $m \in \bN$, we will write $\bm{S} \sim \samp(\pi, m)$
for the distribution over multisets $S$ obtained by drawing $m$ independent samples from $\pi$.

For a fixed domain $\cX$ and set $\Pi$ of probability distributions over $\cX$,
we will write $\far^\TV_\epsilon(\Pi)$ to denote the
set of distributions $\pi$ over $\cX$ such that $\dist_\TV(\pi, \Pi) > \epsilon$.
We will use a similar notation for other domains such as classes of labeled distributions $\Xi$
and strings $\Psi$, and for other applicable (pseudo)-metrics (\eg $\far^\edit_\epsilon(\Pi)$ for
distributions that are far from $\Pi$ in edit distance).

We will often use the notations $\gequestion$, \;$\lequestion$, \;$\eqquestion$\;, etc., within
proofs, when stating an (in)equality that will be established later on in the proof.

\subsection{Confused Collector: Definition \& Terminology}
\label{section:common-framework-confused-collector-def}

We will introduce the general confused collector model, although for this paper we will be
interested only in path- and cycle-structured domains.
The confused collector model on these domains also serves as a warm-up to the parity trace, so we
introduce and analyze it first.
Standard practice in distribution testing is to analyze a ``Poissonized'' version of the
algorithms, where instead of receiving a $m$ independent random sample points from the input
distribution $\pi$, the algorithm first chooses $\bm m \sim \Poi(m)$ and then samples $\bm m$
independent random points from $\pi$; this means that each point $x$ of the domain appears in the
sample $\Poi(m \cdot \pi(x))$-many times, independently of the other points.
For simplicity, we will define the Poissonized version of the confused collector
(See \cref{section:poissonization} and references therein for more details).

\begin{definition}[Confused Collector Sampling]
Let $G = (V,E)$ be a graph, and let $w : E \to [0,1]$ be a vector of non-negative weights.  We
define the following sampling process. A random subgraph $\bm H$ of $G$ is chosen by including each
edge $e$ independently with probability $1-w(e)$. Let $\bm C_1, \dotsc, \bm C_k$ be the connected
components of $\bm H$; assign to each $\bm{C}_i$ an arbitrary representative vertex
$\bm{c}_i$ of $\bm{C}_i$.

For a probability distribution $\pi$ (or indeed any non-negative vector $\pi : V \to \bR_{\geq 0}$)
and sample-size parameter $m$, we define a \emph{confused collector sample} $\bm S$ from $\pi$ as
follows.  $\bm H$ is chosen as above. For each vertex $v \in V$, we sample an independent Poisson
random variable $\bm s(v) \sim \Poi(m \pi(v))$, and add $\bm c(v)$ to the sample $\bm S$ with
multiplicity $\bm s(v)$, where $\bm c(v)$ is the representative of the connected component $\bm C_i$
that contains vertex $v$.
\end{definition}

For the moment, we are interested only in the case where the weights $w$ are constant, so that there
is some $\eta \in [0,1]$ such that $w(e) = \eta$ for all edges $e$. We call $\eta$ the
\emph{resolution}.

\begin{definition}[Distribution Testing with a Confused Collector]
\label{def:testing-confused-collector}
Fix a graph $G = (V,E)$ and a resolution parameter $\eta$. Let $\Pi_1, \Pi_2$ be properties of
probability distributions over $V$, and let $\alpha \in (0,1)$. A $(\Pi_1, \Pi_2,
\alpha)$-distribution tester with resolution $\eta$ and sample complexity $m$ is an algorithm $A$ that
receives a confused collector sample $\bm S$ from the input distribution $\pi$ and satisfies:
\begin{enumerate}
\item If $\pi \in \Pi_1$ then $\Pr{ A( \bm S ) \text{ accepts }} \geq \alpha$; and
\item If $\pi \in \Pi_2$ then $\Pr{ A( \bm S ) \text{ rejects }} \geq \alpha$.
\end{enumerate}
We will drop $\alpha$ from the notation when we assume $\alpha = 2/3$. However, we remark that the
confused collector does not allow to boost the success probability in the same way as a standard
distribution tester, due to the modified sampling process.
\end{definition}

\subsection{Parity Trace: Definition \& Terminology}
\label{section:parity-trace}

In this section we will formally define distribution testing under the parity trace and introduce
the notation and terminology that we will use to analyze our tester and prove \cref{thm:intro-main-informal}.
For a multiset $S \subset \bN$, recall the definition of the trace $\trace(S)$ from the
introduction. Then we define our testing model:

\begin{definition}
    \label{def:testing-parity-trace}
    Let $\Pi_1$ and $\Pi_2$ be any properties of distributions over domain $\bN$. A $(\Pi_1, \Pi_2,
    \alpha)$-distribution tester under the parity trace, with sample complexity $m$, is an algorithm $A$
    which satisfies the following.
    \begin{enumerate}
        \item If $\pi \in \Pi_1$, then
            $\Pru{\bm{S} \sim \samp(\pi,m)}{ A( \trace(\bm{S}) ) \text{ accepts }} \geq \alpha$.
        \item If $\pi \in \Pi_2$,
            then $\Pru{\bm{S} \sim \samp(\pi,m)}{ A( \trace(\bm{S}) ) \text{ rejects }} \geq \alpha$.
    \end{enumerate}
    The canonical version of this problem will have $\Pi_2 = \far^\edit_\epsilon(\Pi_1)$ or, in some
    cases, $\Pi_2 = \far^\TV_\epsilon(\Pi_1)$.
\end{definition}

We say that a vector $r \in \bR^{\bN}$ is a \textbf{partial distribution} if all of its entries are
non-negative, and $\sum_i r_i \leq 1$.

In the parity trace model, we think of a probability distribution $\pi$ over $\bN$ as defined by two
partial distributions $p, q \in \bR^{\bN}_{\geq 0}$, so that $\pi = \pi(p,q)$ where
\[
  \pi(p,q) \define (p_1, q_1, p_2, q_2, p_3, q_3, \dotsc ) \,,
\]
so that $p$ defines the part of the distribution over the odd elements, and $q$ defines the part of
the distribution over the even elements. We will always use the letters $p$ and $q$ for the partial
distributions over the odd and even elements, respectively.

In the parity trace model, the algorithm receives a trace $\trace(S)$ containing 1s and 0s, and it
will separately consider the statistics of the 1s and of the 0s. In the analysis, we will treat only
the statistics of the 1s, because the statistics for the 0s may be handled symmetrically. For the
purpose of analyzing the 1s, we may write the trace received by the algorithm (in regular expression
notation) in the form
\[
  \trace(S) = 1^{Z_1} \; 0^+ \; 1^{Z_2} \; 0^+ \dotsm 1^{Z_t} \; 0^* \,,
\]
for some $t$, where $Z_2, \dotsc, Z_t > 0$ and we allow $Z_1 = 0$. A contiguous sequence of 1s is
called a ``run''\!\!, and we call the values $Z_i$ the ``run-lengths''\!\!.

It will be convenient for our tester to actually use the ``circular trace''\!\!, obtained from
$\trace(S)$ string by stitching the ends of the string together, to form a necklace. If the trace
begins and ends with the same symbol, the first and last ``run'' are combined. So the algorithm sees
a \textbf{circular trace} of the form
\[
  1^{X_1} \; 0^+ \; 1^{X_2} \; 0^+ \dotsm 1^{X_b} \; 0^+ \,,
\]
where we might have $X_1 = Z_1 + Z_t$. (Here, $\sigma^+$ indicates that symbol $\sigma$ occurs at
least once.) For the purpose of testing uniformity, we are concerned only with the domain $[2n]$,
with the partial distributions $p, q$ being over $[n]$, so we may think of the domain itself as being
stitched into a necklace. Equivalently, we think of the domain $[n]$ of $p$ as being the vertices of
a cycle.

More precisely, we think of a cycle on vertices $\bZ_n$ with a partial distribution $p$ over the
vertices, and we define a weight vector $w$ on the edges, where the edge between vertex $i$ and
$i+1$ has weight $1-e^{-mq_i}$. Then, sampling a subgraph $\bm H$ as in the confused collector
sampling process, we see that vertices $i$ and $i+1$ in the cycle are adjacent in
$\bm H$ with probability $e^{-mq_i} = \Pr{ \Poi(mq_i) = 0 }$, which is the probability that
these vertices will contribute to the same run-length $X_j$ in the trace.

\subsection{Path- and Cycle-Structured Poisson Random Variables.}
\label{section:path-cycle-structured-rvs}

It is convenient to introduce a shared vocabulary for analyzing Poisson random variables on the
cycle and on the path. We will label the $n$ vertices of the cycle with the set $\bZ_n$ of integers
mod $n$, and we will also label the \emph{edges} of the cycle with the set $\bZ_n$, so that edge $i$
connects vertices $i$ and $i+1$ (with arithmetic mod $n$).  We will treat the path on $n$ vertices
as the subgraph of the cycle that excludes edge $n-1$ connecting vertices labeled $0$ and $n-1$.
When the subgraph $H$ contains edge $e$, we will sometimes abuse notation and write $e \in H$.

\newcommand{\llangle}{\langle\!\langle}
\newcommand{\rrangle}{\rangle\!\rangle}

A \textbf{circular interval} is a tuple $\llangle i, d \rrangle$ where $i \in \bZ_n$ and $d \in
\bZ$. If $d \geq 0$, we define the elements $\cE\llangle i, d \rrangle$ as the multiset of elements
starting at vertex $i \in \bZ_n$ and containing the $d-1$ elements ``clockwise'' from $i$, \ie the
multiset $\{ i, i+1, i+2, \dotsc, i+d-1 \}$, where addition is mod $n$. Note that for $d=1$ this
contains only $i$, while for $d > n$ this contains some elements with multiplicity greater than
1. If $d < 0$, we define the elements $\cE\llangle i, d \rrangle$ as the multiset of elements
starting at vertex $i \in \bZ_n$ and containing the $|d|-1$ elements ``counter-clockwise'', \ie the
multiset $\{i, i-1, i-2, \dotsc, i-|d|+1\}$.

The \textbf{endpoints} of $\llangle i, d \rrangle$ are the integers $i$ and $i + d - 1$ if $d \geq
0$, or $i+d+1$ and $i$ if $d < 0$.  We will often drop the $\cE$ from the notation, and equivocate
between the tuple $\llangle i, d \rrangle$ and its multiset of elements, so that we write $x \in
\llangle i, d \rrangle$ instead of $\cE \llangle i, d \rrangle$. However, a circular interval is
\emph{not} identified with its multiset of elements; for example, the circular intervals $\llangle
i, n \rrangle$ and $\llangle i+1, n \rrangle$ both contain the same elements $\bZ_n$, but they have
different endpoints.

For a circular interval $I$ and a vector $u : \bZ \to \bR$, we define
\[
  u[I] \define \sum_{s \in I} u_s \,,
\]
where we note that $s$ may occur multiple times in $I$ and $u_s$ is counted each time.

For a circular interval $\llangle i, d \rrangle$, we will define the circular interval $\llangle i,
d \rrangle^*$ to be the integers corresponding to the \emph{edges} induced by the vertices
$\llangle i, d \rrangle$; specifically
\[
\llangle i, d \rrangle^* = \begin{cases}
  \emptyset &\text{ if } d = 0 \\
  \llangle i, d-1 \rrangle &\text{ if } d \geq 1 \\
  \llangle i-1, 1-|d| \rrangle &\text{ if } d \leq -1 \,.
\end{cases}
\]
For any $s \in \bZ_n$, we say that a circular interval $I$ \textbf{crosses} $s$ if $s \in I^*$; \ie
$s$ is an edge between two vertices in $I$.

Fix any subgraph $H$ of the cycle (or path), and suppose that $H$ has $b$ connected components; note
that each connected component is a circular interval. We define the \textbf{buckets} induced by $H$
as $\Gamma_1, \dotsc, \Gamma_n$ such that $\Gamma_1, \dotsc, \Gamma_b$ are the connected components
of $H$, while $\Gamma_{b+1}, \dotsc, \Gamma_n = \emptyset$. For each vertex $i \in \bZ_n$, we define
\[
  \gamma(i) \define t \text{ such that } i \in \Gamma_t \,.
\]
We say that two vertices $i,j$ are \textbf{joined} if $\gamma(i) = \gamma(j)$, and we define the
\textbf{join matrix} $\Phi = \Phi(H)$ as
\[
  \Phi_{i,j} \define \begin{cases}
    1 &\text{ if } \gamma(i) = \gamma(j) \\
    0 &\text{ otherwise.}
  \end{cases}
\]
We define a \textbf{join function} $J$ such that for any circular interval
$I = \llangle i,d \rrangle$,
\[
    J(I) \define \ind{ \forall e \in \llangle i,d \rrangle^* : e \in H } \,.
\]
Thus if $J(I) = 1$, then for every $i, j \in \cE(I)$ we have $\Phi_{i,j} = 1$.

For a fixed sample (\ie multiset) $S \subset \bZ_n$ and for $i \in \bZ_n$, write $T_i$ for the
multiplicity of element $i$ in $S$. We then define for each $i \in [n]$ the variable
\[
  X_i \define \sum_{j \in \bZ_n : \gamma(j) = i} T_j \,,
\]
which is the total multiplicity of elements from bucket $\Gamma_i$ that occur in $S$.

Observe that the above variables depend on the subgraph $H$ and the sample $S$. For a fixed weight
vector $w$ and random subgraph $\bm H$ chosen according to the confused collector sampling
procedure, and a random sample $\bm S$ of vertices, we write the above variables in bold to denote
the random variables depending on $\bm H$ and $\bm S$. We will then write
\[
  \phi \define \Ex{\bm \Phi} \,,
\]
and observe that
\[
  \phi_{i,j} = \Pr{ \bm{\gamma}(i) = \bm{\gamma}(j) } \,.
\]
In our analysis of the confused collector and the parity trace, we have a sample-size parameter $m$
and an input (partial) distribution $p : \bZ_n \to [0,1]$. We will then have
\[
  \bm{T}_j \sim \Poi(m p_j)
\]
for each $j \in \bZ_n$, and therefore
\[
  \bm{X}_i \sim \Poi( m \cdot p[\bm \Gamma_i ] ) 
\]
for each $i \in [n]$. We will also have the random Boolean matrix $\bm \Phi$ which indicates the
connected components of $\bm H$. Our testing algorithms will rely on an analysis of the following
\emph{test statistic}.

\begin{definition}[Test Statistic]
For a fixed parameter $m$ and weight vector $w$, and random variables defined as above, we define
the test statistic
\[
  \bm Y \define \frac{1}{m} \sum_{i=1}^n \bm{X}_i (\bm{X}_i - 1) \,.
\]
By expanding the variables $\bm{X}_i$, the test statistic may be written as the quadratic form
\[
  \bm Y = \frac{1}{m} \left( \bm{T}^\top \bm{\Phi} \bm{T} - \| \bm T \|_1 \right) \,.
\]
\end{definition}

\subsection{Shared Analysis}
\label{section:shared-analysis}

We now proceed with a part of the analysis that is shared between our confused collector and parity
trace results, reflecting common challenges presented by each model. The
application of these results in the subsequent sections will then exploit
the particularities of each model---essentially, that the resolution $\eta$ is fixed
in the confused collector model, whereas in the parity trace model the partial distributions and
the selected sample size affect the sampling rate of both vertices and edges.

\subsubsection{Circular Intervals}
\label{section:circular-intervals}

\newcommand{\cyc}{\mathsf{cycle}}
\newcommand{\pth}{\mathsf{path}}

Our analysis will handle the cases where $G$ is a cycle or a path.
For the path, the circular intervals that cross the edge between vertices $0$ and $n-1$ 
are irrelevant, so it is convenient to define $\cI^\cyc$ as the set of all circular intervals, and
$\cI^\pth$ as the set of all circular intervals that do not cross edge $n-1$.

We will use $\cI \in \{ \cI^\cyc, \cI^\pth \}$ to denote the set of circular intervals relevant to
the analysis. In the case $\cI = \cI^\cyc$, each pair of vertices $i \le j$ has two disjoint paths
connecting them and therefore may be joined together in two ways. We define $\smallinterval(i,j)$
and $\largeinterval(i,j)$ as the two circular intervals defined as follows. Let
\[
  I_1 \define \llangle i, j - i + 1 \rrangle
  \qquad\text{ and }\qquad
  I_2 \define \llangle j, n - (j-i) + 1 \rrangle
\]
as the circular intervals corresponding to the two separate paths between $i$ and $j$. Then we
define
\[
  \smallinterval(i,j) \define \arg \max_{I \in \{I_1, I_2\}} \Ex{ \bm{J}[ I ] }
  \qquad\text{ and }\qquad
  \largeinterval(i,j) \define \arg \min_{I \in \{I_1, I_2\}} \Ex{ \bm{J}[ I ] } \,,
\]
breaking ties arbitrarily. Note that, in the case of the path, we will have only one way of joining
$i$ and $j$, so that $\Ex{ \bm{J}[ \largeinterval(i,j) ] } = 0$ in this case. Symmetrically, when
$i > j$ we define $\smallinterval(i,j) \define \smallinterval(j,i)$ and
$\largeinterval(i,j) \define \largeinterval(j,i)$.

For $\cI \in \{\cI^\cyc, \cI^\pth\}$, we define
\[
  \zeta(\cI) \define \max_{i,j} \Ex{ \bm{J}[ \largeinterval(i,j) ] } \,.
\]
The analysis proceeds in two cases. For the confused collector, we assume that the weight vector is
constant, so that $w(j) = \eta$ for each edge $j$, where $\eta$ is the resolution parameter; then
the probability that edge $j$ appears in $\bm H$ is $1-w(j) = 1-\eta$. For the parity trace, we have two
partial distributions, $p$ and $q$, which are the parts of the input distribution corresponding to
the odd and even elements of the domain, respectively. We treat $p$ as the distribution over the
vertices of the cycle, and we define the weight vector $w(j) \define 1 - e^{-mq_j}$, so that the
probability of edge $j$ appearing in $\bm H$ is $1-w(j) = e^{-mq_j} = \Pr{ \Poi(mq_j) = 0 }$.

\begin{proposition}
\label{prop:zeta-cycle-bound}
$\zeta(\cI^\pth) = 0$. For constant weights $w(j) = \eta$, we have
\[
    \zeta(\cI^\cyc) \leq (1-\eta)^{n/2} \,,
\]
and for weights $w(j) = 1-e^{-mq_j}$, we have
\[
  \zeta(\cI^\cyc) \leq e^{-\frac{m \|q\|_1}{2}} \,.
\]
\end{proposition}
\begin{proof}
For distinct $i, j \in \bZ_n$, define $I_1$ and $I_2$ as above, and note that for $a \in \{1,2\}$,
\[
  \Ex{ \bm{J}[I_a] } = \prod_{j \in I_a^*} (1-w(j)) \,.
\]
In the case $w(j) = \eta$, this is $(1-\eta)^{|I_a^*|}$, while in the case $w(j) = 1-e^{-mq_j}$, this
is $e^{-mq[I_a^*]}$. Note that $I_1, I_2$ partition $\bZ_n$, so in the first case we have either
$|I_1| \geq n/2$ or $|I_2| \geq n/2$, so the minimum is at most $(1-\eta)^{n/2}$. In the second case we
have either $q[I^*_1] \geq \|q\|_1/2$ or $q[I^*_2] \geq \|q\|_1/2$, so the minimum is at most
$e^{-\frac{m\|q\|_1}{2}}$.
\end{proof}

\subsubsection{Expectation of the Test Statistic}
\label{section:common-expectation}

We start by giving an expression for the expectation of the statistic $\bm{Y}$. Recall that we write
$p : \bZ_n \to [0,1]$ for the (partial) distribution over the vertices (of either the path or the
cycle), $m$ is the sample-size parameter, and $\phi = \Ex{\bm \Phi}$.

\begin{proposition}
    \label{prop:expectation-y-quadratic-form}
    The statistic $\bm{Y}$ satisfies
    \[
        \Ex{\bm{\bm{Y}}} = m p^\top \phi p \,.
    \]
\end{proposition}
\begin{proof}
    We use the facts that $\bm{T}$ and $\bPhi$ are independent and that, for $i \neq j$,
    $\bm{T}_i$ and $\bm{T}_j$ are independent. We will also use the property that, for
    $\bm{Z} \sim \Poi(\lambda)$, we have $\Ex{\bm{Z}} = \Var{\bm{Z}} = \lambda$ and, therefore,
    $\Ex{\bm{Z}^2} = \Ex{\bm{Z}} + \Ex{\bm{Z}}^2$. We obtain:
    \begin{align*}
        \Ex{\bm{Y}} &= \frac{1}{m} \left( \Ex{\bm{T}^\top \bPhi \bm{T}} - \Ex{\|\bm{T}\|_1} \right)
        = \frac{1}{m} \sum_{i=0}^{n-1} \sum_{j=0}^{n-1} \Ex{\bm{T}_i \bm{T}_j} \Ex{\bPhi_{i,j}}
            - \frac{1}{m} \sum_{i=0}^{n-1} \Ex{\bm{T}_i} \\
        &= \frac{1}{m} \sum_{i=0}^{n-1} \sum_{j=0}^{n-1} \Ex{\bm{T}_i} \Ex{\bm{T}_j} \phi_{i,j}
            + \frac{1}{m} \sum_{i=0}^{n-1} \Ex{\bm{T}_i} - \frac{1}{m} \sum_{i=0}^{n-1} \Ex{\bm{T}_i}
        = \frac{1}{m} \sum_{i=0}^{n-1} \sum_{j=0}^{n-1} (m p_i) (m p_j) \phi_{i,j} \\
        &= m p^\top \phi p \,. \qedhere
    \end{align*}
\end{proof}

It will sometimes be useful to write $p = p^* + z$ where $p^*$ is a reference partial distribution,
in which case we require:

\begin{proposition}
    \label{prop:quadratic-form-decomposition}
    Write $p = p^* + z$. Then $\bm{Y}$ satisfies
    \[
        \Ex{\bm{Y}} =
        m \left(p^*\right)^\top \phi p^* + 2m \left(p^*\right)^\top \phi z + m z^\top \phi z \,.
    \]
\end{proposition}
\begin{proof}
    This follows immediately from \cref{prop:expectation-y-quadratic-form} by expanding the
    quadratic form and recalling that $\bPhi$ is always a symmetric matrix, and hence so is
    $\phi = \Ex{\bPhi}$.
\end{proof}

\subsubsection{Variance of the Test Statistic: First Component}
\label{section:shared-analysis-variance-1}
\label{section:common-variance}

In this section, we will establish upper bounds for the variance of $\bm{Y}$ that are
general to both the parity trace and confused collector models. The sections dealing with
each particular model will proceed from here. 

Recall that the (random) partition of vertices into buckets $\bm \Gamma_1, \dotsc, \bm \Gamma_n$
depends on the random subgraph $\bm H$.  We start by noting that we can break down the variance of
$\bm{Y}$ into two components by the law of total variance:
\[
    \Var{\bm{Y}} = \Varu{\bm{H}}{\Exuc{\bm{T}}{\bm{Y}}{\bm{H}}}
                 + \Exu{\bm{H}}{\Varuc{\bm{T}}{\bm{Y}}{\bm{H}}} \,.
\]
We will handle the first term here, and the second term in
\cref{section:shared-analysis-variance-2}.

Recall that the weight vector $w$ is either the constant $\eta$ vector, or $w(j) = 1-e^{-mq_j}$.

\begin{proposition}
\label{prop:negligible-interval-bounds-new}
\label{prop:negligible-interval-bounds}
Let $\cI \in \{\cI^\cyc, \cI^\pth\}$.
For every $i, j, k, \ell \in \bZ_n$, the following hold:
\begin{enumerate}
\item $\Ex{\Phi(\bm J)_{i,j}} \ge \Ex{\bm J [\smallinterval(i,j)]}$;
\item $\Ex{\Phi(\bm{J})_{i,j} \cdot \Phi(\bm{J})_{k,\ell}} \le
        \Ex{\bm J [\smallinterval(i,j)] \cdot \bm J [\smallinterval(k,\ell)]} + 4 \cdot \zeta(\cI)$.
\end{enumerate}
\end{proposition}
\begin{proof}
Recall that $\Phi(\bm J)_{i,j} = 1$ if and only if $\bm\gamma(i) = \bm\gamma(j)$. This will occur
if $\smallinterval(i,j) \subseteq \bm \Gamma_{\gamma(i)}$, which happens when
$\bm{J}[\smallinterval(i,j)]=1$, yielding the first conclusion. Next, observe
\begin{align*}
    \bm\Phi_{i,j}
    = \max\{ \bm{J}[\smallinterval(i,j)], \bm{J}[\largeinterval(i,j)] \}
    \le \bm{J}[\smallinterval(i,j)] + \bm{J}[\largeinterval(i,j)] \,.
\end{align*}
    To prove the second statement, expand the product and use the fact that
    $\bm{J}$ is a Boolean vector:
    \begin{align*}
        \bm\Phi_{i,j} \bm \Phi_{k,\ell}
        &\le \left( \bm{J}[\smallinterval(i,j)] + \bm{J}[\largeinterval(i,j)] \right)
            \left( \bm{J}[\smallinterval(k,\ell)] + \bm{J}[\largeinterval(k,\ell)] \right) \\
        &\le \bm{J}[\smallinterval(i,j)]\bm{J}[\smallinterval(k,\ell)]
            + 2\left( \bm{J}[\largeinterval(i,j)] + \bm{J}[\largeinterval(k,\ell)] \right) \,.
    \end{align*}
    We have $\Ex{\bm{J}[\largeinterval(i,j)]}, \Ex{\bm{J}[\largeinterval(k,\ell)]} \leq \zeta(\cI)$
by definition, so the conclusion follows from taking the expectation.
\end{proof}

\begin{lemma}
\label{prop:general-variance-first-component}

Let $\cI \in \{\cI^\cyc, \cI^\pth\}$ and let $w$ be the weights on the edges.  There exists an
absolute constant $c > 0$ such that the first component of the variance of $\bm{Y}$ satisfies
\[
\Varu{\bm{H}}{\Exuc{\bm{T}}{\bm{Y}}{\bm{H}}}
  \le 5m^2 \zeta(\cI) \|p\|_1^4
    + c m^2 \cdot \sum_{\substack{I=\llangle i,d \rrangle \in \cI \\ 1 \le d \le n}}
        p_{i} p_{i+d-1} p[I]^2 \Ex{\bm{J}(I)}  \,.
\]
\end{lemma}
\begin{proof}
    Fix any subgraph $H$. Conditional on $\bm H = H$,
    \begin{align*}
        \Exuc{\bm{T}}{\bm{Y}}{\bm{H} = H}
        &= \frac{1}{m} \Ex{\bm{T}^\top \Phi \bm{T}} - \frac{1}{m} \Ex{\|\bm{T}\|_1}
        = \frac{1}{m} \sum_{i,j \in \bZ_n} \Ex{\bm{T}_i \bm{T}_j} \Phi_{i,j} - \frac{1}{m} \Ex{\Poi(m)} \\
        &= \frac{1}{m} \left( \sum_{i \in \bZ_n} \Ex{\bm{T}_i}
            + \sum_{i,j \in \bZ_n} \Ex{\bm{T}_i} \Ex{\bm{T}_j} \Phi_{i,j} \right)
            - 1
        = \frac{1}{m} \left( m + m^2 p^\top \Phi p \right) - 1 \\
        &= m p^\top \Phi p \,,
    \end{align*}
    and therefore the desired variance is
    \[
        \Varu{\bm{H}}{\Exuc{\bm{T}}{\bm{Y}}{\bm{H}}}
        = m^2 \Varu{\bm H}{p^\top \bm\Phi p} \,.
    \]
    Then, recalling that $\phi = \Ex{\bPhi}$, we expand
    $\Varu{\bm H}{p^\top \bm\Phi p}$ as follows:
    \begin{align*}
        \Var{p^\top \bm\Phi p}
        &= \Ex{\left( p^\top \bm\Phi p \right)^2}
            - \left( \Ex{p^\top \bm\Phi p} \right)^2
        = \Ex{\left( p^\top \bm\Phi p \right)^2} - \left( p^\top \phi p \right)^2 \\
        &= \Ex{\left( \sum_{i,j \in \bZ_n} p_i p_j \bm\Phi_{i,j} \right)^2}
            - \left( \sum_{i,j \in \bZ_n} p_i p_j \phi_{i,j} \right)^2 \\
        &= \Ex{\sum_{i,j,k,\ell \in \bZ_n}
                    p_i p_j p_k p_\ell \bm\Phi_{i,j} \bm\Phi_{k,\ell}}
            - \sum_{i,j,k,\ell \in \bZ_n} p_i p_j p_k p_\ell \phi_{i,j} \phi_{k,\ell} \\
        &= \sum_{i,j,k,\ell \in \bZ_n} p_i p_j p_k p_\ell \left(
            \Ex{\bm\Phi_{i,j} \bm\Phi_{k,\ell}} - \phi_{i,j} \phi_{k,\ell} \right) \,.
    \end{align*}
    We now use \cref{prop:negligible-interval-bounds} to simplify the quantity
    $\Ex{\bm\Phi_{i,j} \bm\Phi_{k,\ell}} - \phi_{i,j} \phi_{k,\ell}$:
    \begin{align*}
        &\Ex{\bm\Phi_{i,j} \bm\Phi_{k,\ell}} - \phi_{i,j} \phi_{k,\ell} \\
        &\qquad\le \Ex{\bm{J}[\smallinterval(i,j)] \bm{J}[\smallinterval(k,\ell)]}
            - \Ex{\bm{J}[\smallinterval(i,j)]} \Ex{\bm{J}[\smallinterval(k,\ell)]}
            + 4 \cdot \zeta(\cI) \,.
    \end{align*}
    If the intervals $\smallinterval(i,j)$ and $\smallinterval(k,\ell)$ are disjoint, then
    \[
      \Ex{\bm{J}[\smallinterval(i,j)] \bm{J}[\smallinterval(k,\ell)]}
    - \Ex{\bm{J}[\smallinterval(i,j)]} \Ex{\bm{J}[\smallinterval(k,\ell)]} = 0 \,.
    \]
    On the other hand, if these intervals are not disjoint, we will employ the simple upper bound
    \[
      \Ex{\bm{J}[\smallinterval(i,j)] \bm{J}[\smallinterval(k,\ell)]}
    - \Ex{\bm{J}[\smallinterval(i,j)]} \Ex{\bm{J}[\smallinterval(k,\ell)]} \le
    \Ex{\bm{J}[\smallinterval(i,j)] \bm{J}[\smallinterval(k,\ell)]} \,.
    \]
    We then consider two cases.

First, suppose that for every edge $s \in \bZ_n$, $\smallinterval(i,j)$ crosses $s$ or
$\smallinterval(k,\ell)$ crosses $s$. Then $\bZ_n \subseteq \smallinterval(i,j)^* \cup
\smallinterval(k,\ell)^*$, so $\bm J(\smallinterval(i,j)) \cdot \bm J(\smallinterval(k,\ell)) = 1$
only when every edge appears in $\bm{H}$, which happens with probability at most $\zeta(\cI)$
(since this event implies that every large interval is joined).
In this case, $\Ex{\bm{J}[\smallinterval(i,j)] \bm{J}[\smallinterval(k,\ell)]} \le \zeta(\cI)$.

As for the second case, let $s \in \bZ_n$ be such that neither $\smallinterval(i,j)$ nor
$\smallinterval(k,\ell)$ crosses $s$. Since $\smallinterval(i,j)$ and $\smallinterval(k,\ell)$ are not
disjoint,  it follows that there exists an interval $I = I_{i,j,k,\ell} \in \cI$ satisfying the
following:
    \begin{enumerate}
        \item The set $\smallinterval(i,j) \cup \smallinterval(k,\ell)$ is equal to
            the set of elements of $I_{i,j,k,\ell}$, where we are here taking the union as sets (not as
multisets); 
        \item The endpoints of $I_{i,j,k,\ell}$ are two of the indices $i,j,k,\ell$; and,
        \item $|I_{i,j,k,\ell}| \le n$ (because, in particular, $I_{i,j,k,\ell}$ does not cross $s$).
    \end{enumerate}
    It follows that $\bm{J}[\smallinterval(i,j)] \bm{J}[\smallinterval(k,\ell)] = 1$ if and only if
    $\bm{J}[I] = 1$, and hence we have the upper bound
    $\Ex{\bm{J}[\smallinterval(i,j)] \bm{J}[\smallinterval(k,\ell)]}
    - \Ex{\bm{J}[\smallinterval(i,j)]} \Ex{\bm{J}[\smallinterval(k,\ell)]} \le \Ex{\bm{J}[I]}$.
    Therefore, 
\begin{align*}
  &\sum_{i,j,k,\ell \in \bZ_n} p_i p_j p_k p_\ell \left(
    \Ex{\bm\Phi_{i,j} \bm\Phi_{k,\ell}} - \phi_{i,j} \phi_{k,\ell} \right) \\
  &\leq \sum_{i,j,k,\ell \in \bZ_n} p_i p_j p_k p_\ell \left(
    \Ex{\bm{J}[\smallinterval(i,j)] \bm{J}[\smallinterval(k,\ell)]}
      - \Ex{\bm{J}[\smallinterval(i,j)]} \Ex{\bm{J}[\smallinterval(k,\ell)]}
      + 4\cdot\zeta \right) \\
  &\leq 4 \cdot \|p\|_1^4 \cdot \zeta
    + \sum_{i,j,k,\ell \in \bZ_n} p_ip_jp_kp_\ell \left(
        \ind{ \bZ_n \subseteq \smallinterval(i,j)^* \cup \smallinterval(k,\ell)^* } \cdot \zeta(\cI)
      + \ind{I_{i,j,k,\ell} \text{ exists}} \cdot \Ex{ \bm{J}[ I_{i,j,k,\ell} ] } \right) \\
  &\leq 5 \zeta(\cI) \| p \|_1^4
    + \sum_{i,j,k,\ell \in \bZ_n} p_ip_jp_kp_\ell 
      \ind{I_{i,j,k,\ell} \text{ exists}} \cdot \Ex{ \bm{J}[ I_{i,j,k,\ell} ] } \,.
\end{align*}
The latter is bounded by summing over all intervals $I = \llangle i, d \rrangle \in \cI$
with $d \leq n$ and for each one taking the expression $c \cdot p_i p_{i+d-1} \sum_{j,k \in I}
p_j p_k \cdot \Ex{\bm{J}[I]}$, where $i$ and $i+d-1$ are the endpoints of $I$, and $c$ is a constant
counting the number of ways to get intersecting intervals with endpoints in $i,(i+d-1),k,\ell$.
Now, using $\sum_{j,k \in I} p_j p_k = p[I]^2$, we obtain
\[
    \Var{p^\top \Phi(\bm{J}) p}
    \leq 5 \zeta(\cI) \|p\|_1^4
    + c \cdot \sum_{\substack{I = \llangle i, d \rrangle \in \cI \\ 1 \le d \le n}}
      \Ex{\bm{J}[I]} p_i p_{i+d-1} p[I]^2 \,. \qedhere
\]
\end{proof}

\subsubsection{Variance of the Test Statistic: Second Component}
\label{section:shared-analysis-variance-2}

We introduce some notation for the partial distribution over the buckets $\Gamma_i$ (\ie connected
components of $H$) induced by $p$.

\begin{definition}[Bucketed Vector]
Let $\Gamma = (\Gamma_1, \dotsc, \Gamma_b)$ be the buckets resulting from a subgraph $H$, and
let $u : \bZ_n \to \bR$. Then \emph{$\Gamma$-bucketing of $u$} is the vector
$u_{|\Gamma} \in \bR^b$ given by
\[
    \left(u_{|\Gamma}\right)_i \define u\left[\Gamma_i\right] = \sum_{j \in \Gamma_i} u_j
    \quad \text{for all } i \in [b] \,.
\]
\end{definition}

We now show that the second component of the variance is captured by 2- and 3-norms of the bucketed
vector $p$. Recall that $\bm T_i \sim \Poi(mp_i)$ is the number of occurrences of vertex $i \in
\bZ_n$ in the sample. We first compute the variance of the terms $\bm{X}_i (\bm{X}_i - 1)$
that make up the test statistic:

\begin{proposition}
    \label{prop:poisson-expression-variance}
    If $\bm X \sim \Poi(\lambda)$, then $\Var{\bm X(\bm X-1)} = 4\lambda^3 + 2\lambda^2$.
\end{proposition}
\begin{proof}
    The Poisson random variable $\bm{X}$ has the following raw moments (see \eg \cite{Rio37}):
    \begin{align*}
        \Ex{X}   &= \lambda \,, \\
        \Ex{X^2} &= \lambda + \lambda^2 \,, \\
        \Ex{X^3} &= \lambda + 3\lambda^2 + \lambda^3 \,, \\
        \Ex{X^4} &= \lambda + 7\lambda^2 + 6\lambda^3 + \lambda^4 \,.
    \end{align*}
    Therefore we have
    \begin{align*}
        \Var{X(X-1)}
        &= \Ex{(X(X-1))^2} - \Ex{X(X-1)}^2
        = \Ex{(X^2 - X)^2} - (\Ex{X^2} - \Ex{X})^2 \\
        &= \Ex{X^4} - 2\Ex{X^3} + \Ex{X^2} - \Ex{X^2}^2 + 2\Ex{X^2}\Ex{X} - \Ex{X}^2 \\
        &= (\lambda + 7\lambda^2 + 6\lambda^3 + \lambda^4)
            - 2(\lambda + 3\lambda^2 + \lambda^3)
            + (\lambda + \lambda^2)
            - (\lambda + \lambda^2)^2
            + 2(\lambda + \lambda^2)\lambda
            - \lambda^2 \\
        &= 4\lambda^3 + 2\lambda^2 \,. \qedhere
    \end{align*}
\end{proof}

\begin{lemma}
\label{prop:general-variance-second-component}
Let $H$ be a subgraph
with induced buckets $\Gamma = (\Gamma_1, \dotsc, \Gamma_b)$, and let $p$ be a measure on $\bZ_n$
Then the conditional variance of $\bm{Y}$ given $\bm{H} = H$ satisfies
    \begin{align*}
        \Varuc{\bm{T}}{\bm{Y}}{\bm{H} = H} = 2 \|p_{|\Gamma}\|_2^2 + 4m \|p_{|\Gamma}\|_3^3 \,.
    \end{align*}
\end{lemma}
\begin{proof}
Using \cref{prop:poisson-expression-variance}, the desired variance is
\begin{align*}
      \Varuc{\bm{T}}{\bm{Y}}{\bm{H} = H}
      &= \Var{\frac{1}{m} \sum_{i=1}^b \bm{X}_i (\bm{X}_i - 1)}
      = \frac{1}{m^2} \sum_{i=1}^b \left[
          4\left(m p\left[\Gamma_i\right]\right)^3
          + 2\left(m p\left[\Gamma_i\right]\right)^2
          \right] \\
      &= 4m \|p_{|\Gamma}\|_3^3 + 2\|p_{|\Gamma}\|_2^2 \,. \qedhere
\end{align*}
\end{proof}

\subsubsection{Relative Concentration}
\label{sec:relative-concentration}

One of the main tools in our analysis will be ``relative concentration''\!\!, which compares the
probability mass of $p$ inside the circular intervals $I$, to another measure $q$ on the edges.

\begin{definition}[Relative Concentration]
Let $\cI \in \{\cI^\cyc, \cI^\pth\}$ and let $p, q : \bZ_n \to \bR_{\geq 0}$ be partial
distributions. Let $t \in \bR$. Then we define
\[
    \RelativeConcentrationT{p}{q}{t} \define
        \max_{I \in \cI : |I| \le n}
            \frac{p[I]}{ \max\left\{ q[I^*], t \right\} } \,.
\]
\end{definition}

We will require the following lemma, which allows us to find an interval $I$ exhibiting a large
difference between $p[I]$ and $q[I^*]$ if we assume high relative concentration
$\RelativeConcentrationT{p}{q}{t}$. 

\begin{lemma}
\label{lemma:relative-concentration-structural}
Let $\cI \in \{ \cI^\cyc, \cI^\pth \}$, and let $p, q : \bZ_n \to \bR_{\geq 0}$ be partial
distributions. Then there exists $I \in \cI$ of size at most
$n$ satisfying the following:
    \begin{enumerate}
        \item $q[I^*] \le t$; and
        \item $p[I] \ge \frac{1}{2} \cdot t \cdot \RelativeConcentrationT{p}{q}{t}$.
    \end{enumerate}
\end{lemma}
\begin{proof}
    By definition of relative concentration, there exists an interval $I \in \cI$ of size at most
    $n$ such that either
    \begin{enumerate}
    \item $q[I^*] \le t$ and $p[I] = t \RelativeConcentrationT{p}{q}{t}$; or
    \item $q[I^*] \ge t$ and $p[I] = q[I^*] \RelativeConcentrationT{p}{q}{t}$.
    \end{enumerate}
    In the former case, $I$ satisfies the required conditions and we are done.

    Therefore, we may assume that the second condition holds.
    Let $I$ be an interval of minimum size satisfying
    $q[I^*] \ge t$ and $p[I] \ge q[I^*] \RelativeConcentrationT{p}{q}{t}$ (in particular,
    equality will hold). Note that we must have $|I| \ge 2$, since otherwise $I^*$ would be empty,
    contradicting the assumption that $q[I^*] \ge t > 0$. We now consider two cases.

    \textbf{Case 1.} Suppose that we may partition $I = L \cup R$ where $L, R$ are nonempty circular
intervals such that one of the following two conditions hold\footnote{For two intervals $I_1 =
\llangle i, d_1 \rrangle$ and $I_2 = \llangle i, d_2 \rrangle$ with $d_1 > d_2 > 0$ (\ie $I_2$ is a
nonempty prefix of $I_1$), we will write $I_1 \setminus I_2$ to denote the interval $\llangle i+d_2,
d_1-d_2\rrangle$.}, call this pair of conditions ($*$):
    \begin{enumerate}
        \item $q[L^*], q[I^* \setminus L^*] \le t$; or
        \item $q[L^*], q[I^* \setminus L^*] \ge t$.
    \end{enumerate}
    If the first condition holds, we conclude the proof as follows. Since $p[I] = p[L] + p[R]$ and
    $p[I] = q[I^*] \RelativeConcentrationT{p}{q}{t} \ge t \RelativeConcentrationT{p}{q}{t}$,
    it must be that either $p[L] \ge \frac{1}{2} t \RelativeConcentrationT{p}{q}{t}$ or
    $p[R] \ge \frac{1}{2} t \RelativeConcentrationT{p}{q}{t}$. In the first case, $L$ satisfies the
    required conditions. In the second case, since $q[R^*] \le q[I^* \setminus L^*] \le t$,
    we conclude that $R$ satisfies the required conditions and we are done.

    If the second condition holds (which in particular implies that $L^*, I^* \setminus L^*$
    are nonempty), note that since $p[I] = p[L] + p[R]$ and
    $q[I^*] = q[L^*] + q[I^* \setminus L^*]$, we have
    \[
        \RelativeConcentrationT{p}{q}{t}
        = \frac{p[I]}{q[I^*]}
        = \frac{p[L] + p[R]}{q[L^*] + q[I^* \setminus L^*]}
        \le
        \max\left\{ \frac{p[L]}{q[L^*]}, \frac{p[R]}{ q[I^* \setminus L^*] } \right\}  \,.
    \]
    If $\frac{p[L]}{q[L^*]} \ge \RelativeConcentrationT{p}{q}{t}$, then since $q[L^*] \ge t$,
    we conclude that $L$ contradicts the minimality of $I$. Therefore we must have
    $\frac{p[R]}{q[R^*]} \ge \frac{p[R]}{q[I^* \setminus L^*]} \ge \RelativeConcentrationT{p}{q}{t}$.
    Now, if $q[R^*] \le t$, then $R$ satisfies the required conditions, since we have
    $p[R] \ge q[I^* \setminus L^*] \RelativeConcentrationT{p}{q}{t}
    \ge t \RelativeConcentrationT{p}{q}{t}$.
    Otherwise, if $q[R^*] > t$, then $R$ contradicts the minimality of $I$.
    This completes the proof in the first case.

\textbf{Case 2.} In the second case, we have that every partition $I = L \cup R$ with nonempty $L$ and $R$ fails
both of the conditions in ($*$). Write $I = \llangle i,d \rrangle$ where, as previously remarked, $d
\ge 2$.  Consider the sequences of circular intervals $L_1, \dotsc, L_{d-1}$ given by $L_j \define
\llangle i,j \rrangle$, and write $R_j \define \llangle i+j, d-j \rrangle = I \setminus L_i$, so
that $L^*_j = \llangle i, j-1 \rrangle$ and $R^*_j
= \llangle i+j, d-j-1 \rrangle$. Then each $L_j \cup R_j$ for $j \in [d-1]$ is a partition of $I$ with nonempty
$L_j$ and $R_j$, which therefore must fail the two conditions in ($*$).

    Now observe that $q[L^*_j]$ is non-decreasing with $j$ and $q[I^* \setminus L^*_j], q[R^*_j]$
    are non-increasing with $j$, while $q[L^*_1] = q[\emptyset] = 0$ and hence
    $q[L^*_1] \le q[I^* \setminus L^*_1]$. Fix the maximum index $j \in [d-1]$ satisfying
    $q[L^*_j] \le q[I^* \setminus L^*_j]$. We claim that $q[L^*_j], q[R^*_j] \le t$.

    Assume for the sake of contradiction that $q[L^*_j] > t$. By the selection of $j$, we have
    $q[I^* \setminus L^*_j] > t$. But then the partition $I = L_j \cup R_j$ satisfies the second
    condition in ($*$), a contradiction. So we have verified that $q[L^*_j] \le t$.

    Now assume for the sake of contradiction that $q[R^*_j] > t$. Then $R^*_j$ is nonempty, implying
    that $j < d-1$, and one can verify that $R^*_j = I^* \setminus L^*_{j+1}$. Therefore
    $q[I^* \setminus L^*_{j+1}] > t$. If $q[L^*_{j+1}] \ge t$, then the partition
    $I = L_{j+1} \cup R_{j+1}$ satisfies the second condition in ($*$), again a contradiction.
    We may therefore assume that $q[L^*_{j+1}] < t$. But this implies that
    $q[L^*_{j+1}] < q[I^* \setminus L^*_{j+1}]$, contradicting the maximality of our choice of $j$.
    Therefore we have verified that $q[R^*_j] \le t$.

    Finally, recall that $p[L_j] + p[R_j] = p[I] \ge q[I^*] \RelativeConcentrationT{p}{q}{t}
    \ge t \RelativeConcentrationT{p}{q}{t}$, and therefore either
    $p[L_j] \ge \frac{1}{2} t \RelativeConcentrationT{p}{q}{t}$ or
    $p[R_j] \ge \frac{1}{2} t \RelativeConcentrationT{p}{q}{t}$.
    Since $q[L^*_j], q[R^*_j] \le t$, it follows that either $L_j$ or $R_j$ satisfies the
    required conditions.
\end{proof}

\section{Testing Uniformity in the Confused Collector Model}
\label{sec:upper-bound-confused-collector}

Following the setup from \cref{sec:collision-based-testing-general}, we consider the task of
testing uniformity of an unknown distribution $p$ over the vertices $\bZ_n$ of a base graph
$G = (\bZ_n, E)$ in the confused collector model. Here, $G$ is the path or the cycle and every
edge $(i, i+1) \in E$ has weight $w(e) = \eta$, where $\eta$ is the resolution parameter.

Our analysis will treat the cases of the cycle and the path in a unified presentation. In the case
of the cycle, $E$ contains all $n$ edges connecting each vertex $i$ to $i+1$ (mod $n$), and the
set of relevant intervals is $\cI = \cI^\cyc$. We will write $\bPhi = \bPhi^{\cyc}$ for the
corresponding join matrix.
In the case of the path, $E$ does not contain an edge between $0$ and $n-1$, the set of relevant
intervals is $\cI = \cI^\pth$, and the join matrix is $\bPhi = \bPhi^{\pth}$. When a
result depends on the choice of domain, we will explicitly state the domain under consideration.

The tester is \cref{alg:tester-confused-collector}, and consists of two steps:
\begin{enumerate}
    \item Concentration test: checks whether any count in the sample is too large;
        this case corresponds to highly concentrated distributions, which can be rejected.
    \item Collision-based test: accept or reject depending on whether the test statistic $\bm Y$ is
        below a certain threshold.
\end{enumerate}

\begin{algorithm}[H]
    \caption{Uniformity tester in the confused collector model.}
    \label{alg:tester-confused-collector}
      
    \hspace*{\algorithmicindent}
        Set $m \gets c \cdot \frac{\sqrt{n}}{\epsilon^2} \cdot \frac{\log^2 n}{\eta^{3/2}}$. \\
    \hspace*{\algorithmicindent} \textbf{Constants:}
        $\alpha, \beta, L, c > 0$ to be defined later. \\
    \hspace*{\algorithmicindent} \textbf{Requires:}
        $\eta \ge \frac{L \log^{4/5} n}{n^{1/5} \epsilon^{4/5}}$. \\
    \begin{algorithmic}[1]
        \Procedure{UniformityTester-ConfusedCollector}{$p, n, \epsilon, \eta$}
            \State Let $X = (X_1, \dotsc, X_n)$ be the variables defined in
                                            \cref{sec:collision-based-testing-general}
                for a sample of size $\Poi(m)$ from $p$.
            \State \textbf{If} $\max_i X_i \ge \alpha \log n$ then \textbf{reject}.
            \State $Y \gets \frac{1}{m} \sum_i X_i(X_i-1)$.
            \State \textbf{If}
                $Y \ge \frac{m}{n^2}\sum_{i,j} \phi_{i,j} + \beta \frac{m}{n} \epsilon^2 \eta$
                then \textbf{reject}.
            \State \textbf{Accept}.
        \EndProcedure
    \end{algorithmic}
\end{algorithm}

\paragraph{Remark on the optimality of the collision-based tester.} Considering that we give a
Poissonized tester whose main statistic $\bm{Y}$ is equivalent to the collision-based statistic of
\cite{GR00} when $\eta=1$, it may seem surprising that we claim a sample complexity of
$\widetilde O\left( \sqrt{n}/\epsilon^2 \right)$ --- as opposed to
$O\left(\sqrt{n}/\epsilon^4\right)$--- when it is known that, for an analysis based on
bounding the variance of $\bm{Y}$ and applying Chebyshev's inequality, establishing the optimal
sample complexity is only possible with a different test statistic (\eg the modified chi-squared
statistic~\cite{CDVV14,DKN15b,VV17}) or a careful analysis of the non-Poissonized
tester~\cite{DGPP19} (see also the Remark in Section~2 therein). Our analysis implicitly avoids this
issue via our relative concentration test, which upper bounds $\|p\|_\infty$,
but another way to resolve the apparent conflict is to notice that dropping this extra test and
specializing our proof to the case $\eta=1$ would only incur a dependence on $1/\epsilon^4$, rather
than $\sqrt{n}/\epsilon^4$; and since our analysis only handles the case $\epsilon \ge
\widetilde \Omega\left( n^{-1/4} \right)$, the term $\sqrt{n}/\epsilon^2$ dominates $1/\epsilon^4$.

\subsection{Easy Case: Highly Concentrated Distributions}
\label{subsec:confused-collector-highly-concentrated}

We would like to call distribution $p$ ``highly concentrated'' if it contains too much mass in a
small contiguous range of the vertices $V$.
The tester will detect the highly concentrated distributions and reject,
while non-highly concentrated distributions are well-behaved in our analysis of the variance of the
main test statistic.
Concretely, we define highly concentrated distributions by imposing
a threshold on the relative concentration $\RelativeConcentrationT{p}{w}{t}$ introduced in
\cref{sec:collision-based-testing-general}, where $w$ is the constant vector of edge weights given
by $w(e) = \eta$.

\begin{definition}[Highly concentrated distributions]
Given a constant $C > 0$, positive integer $m$, resolution parameter $\eta$, and probability
distribution $p$ over $\bZ_n$, we say that $p$ is \emph{$C$-highly concentrated (under resolution
$\eta$ with respect to $m$)} if $\RelativeConcentrationT{p}{w}{t} \ge \frac{C \log^2 n}{m}$, where
$t = \frac{1}{\log n}$.
\end{definition}

One may think of this definition as follows: the average size of a bucket (connected component)
with resolution $\eta$ is $\approx 1/\eta$. In such an interval $I = \llangle i,d \rrangle$ with
$I^* = \llangle i,d-1 \rrangle$, we obtain $w[I^*] \approx 1$. If $p[I]/w[I^*] \ge \poly\log(n)/m$,
then the sample $\Poi(m p[I])$ ought to produce a large entry $\bm{X}_i$, so the
algorithm should be able to reject such distributions. On the other hand, for intervals $I$ that are
very small, the only way to ensure that the algorithm will likely reject is if $p[I]$ is still large
regardless of how small $w[I^*]$ is, which motivates the choice of $t$ in the definition.

\begin{remark}
    \label{remark:non-concentrated-l-infty-confused-collector}
    If $p$ is not $C$-highly concentrated, then in particular $\|p\|_\infty < \frac{C\log n}{m}$,
    as can be seen by taking intervals $I = \llangle i,1 \rrangle$ for each $i \in \bZ_n$.
\end{remark}

We now show that the first step of the tester correctly accepts the uniform distribution and
rejects highly concentrated distributions with good probability. Therefore, we will be able
to assume that $p$ is not highly concentrated when analyzing the second step of the tester.
We will need the following auxiliary result.

\begin{proposition}[Buckets are almost always small]
    \label{prop:bucket-sizes-small} Let $K \ge 2$ be a constant, and
    suppose $n \ge 3$. Then the buckets $\bm{\Gamma} = (\bm{\Gamma}_i, \dotsc, \bm{\Gamma}_{\bm{b}})$
    induced by $\bm{H}$ satisfy
    \[
        \abs*{\bm{\Gamma}_i} \le \frac{2K\log n}{\eta} \qquad \forall i \in [\bm{b}]
    \]
    except with probability at most $1/n^K$.
\end{proposition}
\begin{proof}
    Let $d$ be the smallest integer satisfying $d \ge \frac{2K\log n}{\eta}$, and fix any interval
    $I = \llangle i, d \rrangle$. The probability that all vertices in this interval are joined is
    \[
        \Pr{\bm{J}(I)=1} = (1-\eta)^{d-1}
        \le (1-\eta)^{\frac{2K\log n}{\eta} - 1}
        \le (1-\eta)^{\frac{(2K-1)\log n}{\eta}}
        \le e^{-(2K-1)\log n}
        = 1/n^{2K-1}
        \le 1/n^{K+1} \,,
    \]
    where we used the facts that $n \ge 3 \implies \frac{\log n}{\eta} \ge 1$ and that
    $K \ge 2 \implies 2K-1 \ge K+1$. Now, if any bucket has size at least $d$, then some interval
    $I = \llangle i,d \rrangle$ satisfies $J(I)=1$. Since there are at most
    $n$ such intervals, the probability of this event is at most $1/n^K$ by the union bound.
\end{proof}

\noindent
We will need the following tail bounds for the Poisson distribution, as stated in \cite{Can17}.

\begin{fact}
    \label{fact:poisson-concentration}
    Let $\bm{X} \sim \Poi(\lambda)$ for some $\lambda > 0$. Then for any $t > 0$,
    \[
        \Pr{\bm{X} \le \lambda - t}, \Pr{\bm{X} \ge \lambda + t}
        \le e^{-\frac{t^2}{2(\lambda+t)}} \,.
    \]
\end{fact}

The result below makes the assumption that $m \le n \eta$, which simplifies the analysis
and intuitively corresponds to the sublinear sample complexity regime in the standard uniformity
testing model. This assumption turns out to hold for the range of parameters we consider,
but not necessarily in more extreme regimes (see \cref{remark:confused-collector-parameter-ranges}).

\begin{lemma}
    \label{lemma:concentration-test-eta}
    For sufficiently large constant $\alpha > 0$ and all sufficiently large $n$, the following holds.
    Suppose $m \le n \eta$. Then for any distribution $p$ over $\bZ_n$, we have:
    \begin{enumerate}
        \item If $p$ is uniform, the first step of the tester only rejects with probability
            at most $1/100$; and
        \item If $p$ is $4\alpha$-highly concentrated, the first step of the tester rejects
            with probability at least $99/100$.
    \end{enumerate}
\end{lemma}
\begin{proof}
    \textbf{Completeness.} Suppose $p$ is the uniform distribution over $\bZ_n$.
    From \cref{prop:bucket-sizes-small}, we obtain that every bucket has size at most
    $\frac{4\log n}{\eta}$ except with probability $o(1)$. Assume that this is the case, and fix some
    particular bucket $\Gamma_i$. The number of elements sampled from this bucket is distributed as
    $\bm{X}_i \sim \Poi\left(m p\left[\Gamma_i\right]\right) = \Poi\left( m|\Gamma_i|/n \right)$.
    Then, using \cref{fact:poisson-concentration} and for $\alpha > 16$, the probability that
    $\bm{X}_i$ is so large that the tester rejects is
    \begin{align*}
        \Pr{\bm{X}_i \ge \alpha \log n}
        &\le \Pr{\Poi\left( m \cdot \frac{4\log n}{\eta} \cdot \frac{1}{n} \right) \ge \alpha \log n} \\
        &\le \Pr{\Poi (4\log n) \ge \alpha \log n} \qquad & \text{(Since $m \le n \eta$)} \\
        &= \Pr{\Poi(4\log n) - 4\log n \ge (\alpha-4)\log n} \\
        &\le e^{-\frac{(\alpha-4)^2 \log^2 n}{2(4\log n + (\alpha-4)\log n)}}
        \le e^{-\frac{(\alpha/2)^2 \log n}{2\alpha}}
        \le e^{-2\log n}
        = 1/n^2 \,.
    \end{align*}
    Hence, the probability that this happens for any $\bm{X}_i$ is at most $1/n = o(1)$.

    \textbf{Soundness.} Suppose $p$ is $4\alpha$-highly concentrated.
    Using \cref{lemma:relative-concentration-structural} and the definition of high concentration,
    we may choose some interval $I=\llangle i,d \rrangle \in \cI$ satisfying
    \begin{enumerate}
        \item $(d-1)\eta = w[\llangle i,d-1 \rrangle] \le \frac{1}{\log n}$,
            and thus $|I| = d \le 1 + \frac{1}{\eta \log n}$; and
        \item $p[I] \ge \frac{1}{2} \cdot \frac{1}{\log n} \cdot \frac{4\alpha \log^2 n}{m}
            = \frac{2\alpha \log n}{m}$.
    \end{enumerate}
    We first claim that all the elements in $I$ will be joined with high probability, \ie every
    edge in $I^*$ will be sampled into $\bm{H}$. Indeed, by the union bound, we have
    \[
        \Pr{J(I) = 0}
        = \Pr{\exists e \in I^* : e \not\in \bm{H}} \le
        (|I|-1) \cdot \eta \le \frac{1}{\eta \log n} \cdot \eta = \frac{1}{\log n} = o(1) \,.
    \]
    Now, suppose every element in $I$ belongs to the same bucket, say $\Gamma_i$. Recall that the
    random variable $\bm{X}_i \sim \Poi\left(m p\left[\Gamma_i\right]\right)$ represents the number
    of elements drawn from this bucket, and by our assumption on $I$, we have
    \[
        m p\left[\Gamma_i\right]
        \ge m \cdot \frac{2\alpha \log n}{m}
        = 2\alpha \log n \,.
    \]
    We now claim that, with high probability, $\bm{X}_i > \alpha \log n$ and hence the tester will
    reject. Indeed, using \cref{fact:poisson-concentration}, the probability that this does not
    occur is
    \begin{align*}
        \Pr{\bm{X}_i \le \alpha \log n}
        &\le \Pr{\Poi\left( 2\alpha \log n \right) \le \alpha \log n}
        = \Pr{\Poi\left(2\alpha\log n\right) \le 2\alpha\log n - \alpha \log n} \\
        &\le e^{-\frac{\alpha^2 \log^2 n}
                       {2(2\alpha\log n + \alpha\log n)}}
        = e^{-\frac{\alpha \log n}{6}}
        = o(1) \,. \qedhere
    \end{align*}
\end{proof}

\noindent
We now proceed to the second step of the tester, and analyze the test statistic $\bm{Y}$.

\subsection{Expected Value of the Test Statistic}

\paragraph{Notation.} Let $\mu$ denote the uniform distribution
over $\bZ_n$. We will write $\bPhi$ for the random join matrix and $\phi$ for its expectation when
statements hold for both the path and the cycle. Otherwise, we will specify $\bPhi^\pth$ or
$\bPhi^\cyc$.

Start by recalling that, as shown in \cref{prop:expectation-y-quadratic-form}, we may write the
expected value of $\bm{Y}$ as
\begin{equation}
    \label{eq:expectation-y-quadratic-form-confused-collector}
    \Ex{\bm{Y}} = m p^\top \phi p \,.
\end{equation}
When $G$ is the path, the expected join matrix $\phi^\pth \define \Ex{\bPhi^\pth}$ has a simple
formulation in terms of $\eta$.
It will be useful to name the quantity $1 - \eta$, \ie the probability of including each edge
in $\bm{H}$:
\[
    \nu \define 1 - \eta \,.
\]

\begin{proposition}
    \label{prop:phi-pth}
    The matrix $\phi^\pth$ is given by
    \[
        \phi^\pth_{i,j} = \nu^{\abs{i-j}}
    \]
    for each $i, j \in \bZ_n$.
\end{proposition}
\begin{proof}
    Here, the relevant intervals are $\cI = \cI^\pth$. Hence,
    for any $i < j$, we have that $i$ and $j$ are in the same bucket if and only if every edge
    between them is in $\bm{H}$:
    \[
        \phi^\pth_{i,j} = \Pr{\bPhi^\pth_{i,j}=1}
        = \Pr{\forall e \in \llangle i, j-i+1 \rrangle^* : e \in \bm{H}}
        = \nu^{j-i} \,. \qedhere
    \]
\end{proof}

When $G$ is the cycle, so that the expected join matrix is $\phi^\cyc \define \Ex{\bPhi^\cyc}$,
we need to account for the small and large intervals (in the notation of
\cref{section:circular-intervals}) connecting $i$ and $j$, as follows.

\begin{proposition}
    \label{prop:phi-cyc}
    The matrix $\phi^\cyc$ is given by
    \[
        \phi^\cyc_{i,j} = \nu^{\abs{i-j}} + \nu^{n-\abs{i-j}} - \nu^n \,.
    \]
\end{proposition}
\begin{proof}
    The sets $\smallinterval(i,j)^*$ and $\largeinterval(i,j)^*$ have sizes
    $\min\{\abs{i-j}, n-\abs{i-j}\}$ and $\max\{\abs{i-j}, n-\abs{i-j}\}$, respectively
    (recall they partition the edges of the cycle). By the principle of inclusion-exclusion,
    \begin{align*}
        \phi^\cyc_{i,j}
        &= \Pr{\bPhi^\cyc_{i,j} = 1} \\
        &= \Pr{\bm{J}(\smallinterval(i,j))=1} + \Pr{\bm{J}(\largeinterval(i,j))=1}
            - \Pr{\forall e \in E : e \in \bm{H}} \\
        &= \nu^{\abs{i-j}} + \nu^{n-\abs{i-j}} - \nu^n \,. \qedhere
    \end{align*}
\end{proof}

We would like to show that $\Ex{\bm{Y}}$ is large when $p$ is far from uniform.
Write $p = \mu + z$ where $z \in \bR^{\bZ_n}$. As shown in \cref{prop:quadratic-form-decomposition},
we may decompose the expectation as
\begin{equation}
    \label{eq:y-decomposition-confused-collector}
    \Ex{\bm{Y}} = m \mu^\top \phi \mu + 2m \mu^\top \phi z + m z^\top \phi z \,.
\end{equation}
Letting $\bm{Y}^{(\mu)}$ denote the test statistic when $p$ is the uniform
distribution, we notice that the first term above is precisely the baseline against which
\cref{alg:tester-confused-collector} thresholds the test statistic:

\begin{proposition}[Expectation of $\bm{Y}$ in the uniform case]
    \label{prop:expectation-y-uniform-confused-collector}
    When $p = \mu$, $\bm{Y} = \bm{Y}^{(\mu)}$ satisfies
    \[
        \Ex{\bm{Y}^{(\mu)}} = m \mu^\top \phi \mu = \frac{m}{n^2} \sum_{i,j} \phi_{i,j} \,.
    \]
\end{proposition}
\begin{proof}
    The claim follows from \eqref{eq:expectation-y-quadratic-form-confused-collector} and the
    assumption that $p = \mu = \vec 1 / n$.
\end{proof}

Therefore, our strategy will be to show that 1)~the minimum eigenvalue of $\phi$ is large, and hence
so is $z^\top \phi z$ when $\|z\|_2^2$ is large; and 2)~the term $\mu^\top \phi z$ is small in
absolute value (in fact zero when $G$ is the cycle),
so it does not affect the sum too much.  We start with the first goal. Both $\phi^\pth$
and $\phi^\cyc$ enjoy nice properties (they are a Toeplitz and a circulant matrix, respectively),
and we bound the minimum eigenvalue of each in turn. Let
${\lambda_{\min}}(\cdot), {\lambda_{\max}}(\cdot)$
denote the minimum and maximum eigenvalues of a (real symmetric) matrix, respectively.

\begin{lemma}[Minimum eigenvalue of $\phi^\pth$]
    \label{lemma:min-eigenvalue-pth}
    Let $\eta \in (0,1]$. Then $\lambda_{\min}(\phi^\pth) > \eta/2$.
\end{lemma}
\begin{proof}
The matrix $\phi^\pth$ is a symmetric Toeplitz matrix,
and its inverse may be found as in \cite{Sra}.
Recall that $\nu = 1 - \eta$, so that $0 \le \nu < 1$ and $\phi^\pth_{i,j} =
\nu^{\abs{i-j}}$ by \cref{prop:phi-pth}.
Then the inverse of $\phi^\pth$ (written $\phi^{-1}$ for short)
is the following tridiagonal matrix:
    \[
        \phi^{-1} = \frac{1}{1-\nu^2} \cdot
            \begin{bmatrix}
                1       & -\nu         & 0            & 0            & \dotsm    & 0       \\
                -\nu    & 1 + \nu^2    & -\nu         & 0            & \dotsm    & 0       \\
                0       & -\nu         & 1 + \nu^2    & -\nu         & \dotsm    & 0       \\
                        & \ddots       & \ddots       & \ddots       & \ddots    &         \\
                0       & \dotsc       & \dotsc       & -\nu         & 1 + \nu^2 & -\nu    \\
                0       & \dotsc       & \dotsc       & 0            & -\nu      & 1       \\
            \end{bmatrix}
        \,.
    \]
Now, we may upper bound the maximum eigenvalue of $\phi^{-1}$
using the Gershgorin circle theorem:
    \[
        \lambda_{\max}\left(\phi^{-1}\right)
        \le \max_{i \in \bZ_n} \left\{ \phi^{-1}_{i,i}
            + \sum_{j \neq i} \abs*{\phi^{-1}_{i,j}} \right\}
        \le \left(\frac{1}{1-\nu^2}\right) \cdot \left(1 + \nu^2 + 2\nu\right)
        = \frac{1+\nu}{1-\nu} \,.
    \]
    Hence we obtain
    \[
        \lambda_{\min}(\phi^\pth)
        = \frac{1}{\lambda_{\max}\left(\phi^{-1}\right)}
        \ge \frac{1-\nu}{1+\nu}
        > \frac{1-\nu}{2}
        = \eta/2 \,. \qedhere
    \]
\end{proof}

When $G$ is the cycle, it is convenient to work with a simplified close approximation for
$\phi^\cyc$. Essentially, we wish to ignore the large intervals and instead work with the
matrix $\phi^{\smallinterval}$ given by
\[
    \phi^{\smallinterval}_{i,j} \define \nu^{\abs*{\smallinterval(i,j)}-1} \,.
\]
We will need the observation that $\zeta(\cI)$ is negligibly small in our range of parameter $\eta$.

\begin{proposition}
    \label{prop:zeta-is-small-confused-collector}
    Suppose $\eta \ge \Omega(n^{-1/5})$, and let $K > 0$ be any constant.
    Then for all sufficiently large $n$,
    \[
        \zeta(\cI) \le \nu^{n/2} = o(n^{-K}) \,.
    \]
\end{proposition}
\begin{proof}
    The first inequality is \cref{prop:zeta-cycle-bound}. The second one is easy to check:
    \[
        \nu^{n/2} = (1-\eta)^{n/2} \le e^{-\eta \cdot n/2}
        \le e^{-\Omega(n^{-1/5} \cdot n)}
        = e^{-\Omega(n^{4/5})}
        = o(n^{-K}) \,. \qedhere
    \]
\end{proof}

We are now ready to lower bound the eigenvalues of $\phi^{\smallinterval}$ and
$\phi^\cyc$. We first lower bound $\lambda_{\min}(\phi^{\smallinterval})$, and then show that the
approximation error is negligible.

\begin{fact}[Eigenvalues of circulant matrices; see \cite{Gra06}]
    \label{fact:eigenvalues-circulant}
    Let $c_0, c_1, \dotsc, c_{n-1} \in \bR$. Then the matrix
    \[
        M = \begin{bmatrix}
            c_0      & c_1     & c_2     & \cdots & c_{n-1} \\
            c_{n-1}  & c_0     & c_1     & \cdots & c_{n-2} \\
            \vdots   & \vdots  & \vdots  & \ddots & \vdots \\
            c_2      & c_3     & c_4     & \cdots & c_1 \\
            c_1      & c_2     & c_3     & \cdots & c_0
        \end{bmatrix}
    \]
    given by $M_{j,k} = c_{(k-j) \mod n}$ has eigenvalues
    \[
        \lambda_\ell = \sum_{k=0}^{n-1} c_k \omega^{\ell k} \qquad \ell = 0, 1, \dotsc, n-1 \,,
    \]
    where $\omega = e^{-\frac{2\pi i}{n}}$ is a primitive $n$-th root of unity.
\end{fact}

\begin{lemma}[Minimum eigenvalue of $\phi^{\smallinterval}$]
    \label{lemma:min-eigenvalue-smallinterval}
    Let $\eta \in (0,1]$ satisfy $\eta \ge \Omega(n^{-1/5})$. Then for all sufficiently large $n$,
    $\lambda_{\min}(\phi^{\smallinterval}) > \eta/3$.
\end{lemma}
\begin{proof}
    First assume $n$ is odd. For each $k = 0, 1, \dotsc, n-1$, let
    $c_k \define \phi^{\smallinterval}_{0,k}$, so that $\phi^{\smallinterval}$ is a symmetric
    circulant matrix of the form stated in \cref{fact:eigenvalues-circulant}. In particular,
    letting $h \define \lfloor n/2 \rfloor$ for convenience, we have
    \[
        c_k = \begin{cases}
            \nu^k & \text{if } k \le h \\
            \nu^{n-k} & \text{if } k > h \,.
        \end{cases}
    \]
    Therefore for each $\ell = 0, 1, \dotsc, n-1$, the eigenvalue $\lambda_\ell$ is
    \begin{align*}
        \lambda_\ell
        &= \sum_{k=0}^{n-1} c_k \omega^{\ell k}
        = c_0 + \sum_{k=1}^h \nu^k \omega^{\ell k} + \sum_{k=h+1}^{n-1} \nu^{n-k} \omega^{\ell k}
        = 1 + \sum_{k=1}^h \nu^k \omega^{\ell k} + \sum_{k=1}^h \nu^{n - (n-k)} \omega^{\ell (n-k)} \\
        &= 1 + \sum_{k=1}^h \nu^k \omega^{\ell k} + \sum_{k=1}^h \nu^k \omega^{-\ell k}
        = 1 + \frac{\nu \omega^\ell - \nu^{h+1} \omega^{\ell(h+1)}}{1 - \nu \omega^\ell}
            + \frac{\nu \omega^{-\ell} - \nu^{h+1} \omega^{-\ell (h+1)}}{1 - \nu \omega^{-\ell}} \\
        &= \frac{1 - \nu(\omega^\ell + \omega^{-\ell}) + \nu^2 
                +  \nu \omega^\ell - \nu^2
                        - \nu^{h+1} \omega^{\ell(h+1)} + \nu^{h+2} \omega^{\ell h} 
                +  \nu \omega^{-\ell} - \nu^2
                        - \nu^{h+1} \omega^{-\ell(h+1)} + \nu^{h+2} \omega^{-\ell h}}
                {1 - \nu(\omega^\ell + \omega^{-\ell}) + \nu^2} \\
        &= \frac{1  - \nu^2 -\nu^{h+1}(\omega^{\ell(h+1)} + \omega^{-\ell(h+1)})
                    + \nu^{h+2}(\omega^{\ell h} + \omega^{-\ell h})}
                {1 - \nu(\omega^\ell + \omega^{-\ell}) + \nu^2}
        = \frac{(1 - \nu)(1 + \nu) \pm O(\nu^{n/2})}{1 - 2 \nu \cos(2\pi\ell/n) + \nu^2} \,,
    \end{align*}
    where we used the identity $e^{i \theta} + e^{-i \theta} = 2\cos(\theta)$ in the last step.
    Thus, recalling that $\nu = 1-\eta \in [0,1)$, we conclude that $\lambda_\ell$ is lower bounded
    by
    \[
        \frac{(1 - \nu)(1 + \nu) \pm O(\nu^{n/2})}{1 - 2 \nu \cos(2\pi\ell/n) + \nu^2}
        \ge \frac{(1-\nu)(1+\nu)}{(1+\nu)^2} - \frac{O(\nu^{n/2})}{(1-\nu)^2}
        \ge \frac{\eta}{2} - \frac{O(\nu^{n/2})}{(1-\nu)^2} \,.
    \]
    Then, using \cref{prop:zeta-is-small-confused-collector},
    \[
        \frac{O(\nu^{n/2})}{(1-\nu)^2}
        \le \frac{o(n^{-3/5})}{\eta^2}
        \le o(n^{-3/5} \cdot n^{2/5})
        = o(n^{-1/5})
        = o(\eta) \,,
    \]
    and thus $\lambda_\ell > \eta/3$. When $n$ is even, the same argument applies with an extra term
    of order $O(\nu^{n/2})$, which leaves the asymptotic analysis unaffected.
\end{proof}

\newcommand{\err}{\mathsf{err}}

\begin{lemma}
    \label{lemma:min-eigenvalue-cyc}
    Let $\eta \in (0,1]$ satisfy $\eta \ge \Omega(n^{-1/5})$. Then for all sufficiently large $n$,
    $\lambda_{\min}(\phi^{\cyc}) > \eta/4$.
\end{lemma}
\begin{proof}
    Let $\phi^\err \define \phi^\cyc - \phi^{\smallinterval}$. It is standard to check that
    $\lambda_{\min}(\phi^\cyc) \ge
    \lambda_{\min}(\phi^{\smallinterval}) + \lambda_{\min}(\phi^\err)$.
    Since $\lambda_{\min}(\phi^{\smallinterval}) > \eta / 3$ by
    \cref{lemma:min-eigenvalue-smallinterval}, it suffices to show that
    $\lambda_{\min}(\phi^\err) > -o(\eta)$.
    Since $\phi^\cyc_{i,j} =
    \nu^{\abs{\smallinterval(i,j)}-1} + \nu^{\abs{\largeinterval(i,j)}-1} - \nu^n$ by
    \cref{prop:phi-cyc}, we obtain
    \[
        \phi^\err_{i,j} = \nu^{\abs{\largeinterval(i,j)}-1} - \nu^n
    \]
    for all $i,j \in \bZ_n$. By definition of $\zeta(\cI)$ and recalling
    \cref{prop:zeta-cycle-bound}, we conclude that
    \[
        \|\phi^\err\|_\infty \le \zeta(\cI) \le \nu^{n/2} \,.
    \]
    By the Gershgorin circle theorem and \cref{prop:zeta-is-small-confused-collector},
    \begin{align*}
        \lambda_{\min}(\phi^\err)
        &\ge \min_{i \in \bZ_n} \Big\{ \phi^\err_{i,i} - \sum_{j \ne i} \abs*{\phi^\err_{i,j}} \Big\}
        > -n \cdot \nu^{n/2}
        \ge -n \cdot o(n^{-6/5})
        = -o(n^{-1/5})
        \ge -o(\eta) \,. \qedhere
    \end{align*}
\end{proof}

We use the minimum eigenvalue of $\phi$ to show that, if $\|p - \mu\|_2^2$ is large, then
$\Ex{\bm{Y}}$ is large. The following intermediate formulation of the expected value will
be useful.

\begin{proposition}
    \label{cor:separation-general-form}
    For all sufficiently large $n$, the following holds.
    Let $p$ be a distribution over $\bZ_n$ such that $\dist_\TV(p, \mu) > \epsilon$, and write
    $p = \mu + z$. Then
    \[
        \Ex{\bm{Y}} > \Ex{\bm{Y}^{(\mu)}} + m \frac{\eta}{8} \|z\|_2^2
                + \frac{m}{2n} \epsilon^2 \eta + 2m \mu^\top \phi z \,.
    \]
\end{proposition}
\begin{proof}
    First, combine \eqref{eq:y-decomposition-confused-collector} and
    \cref{prop:expectation-y-uniform-confused-collector}, along with the fact that
    $x^\top M x \ge \lambda_{\min}(M) \|x\|_2^2$ for any symmetric matrix $M$ and vector $x$
    and the eigenvalue bounds \cref{lemma:min-eigenvalue-pth,lemma:min-eigenvalue-cyc}
    to obtain
    \begin{align*}
        \Ex{\bm{Y}}
        &= m \mu^\top \phi \mu + 2m \mu^\top \phi z + m z^\top \phi z \\
        &\ge \Ex{\bm{Y}^{(\mu)}} + m \lambda_{\min}(\phi) \|z\|_2^2 + 2m \mu^\top \phi z \\
        &\ge \Ex{\bm{Y}^{(\mu)}} + m \frac{\eta}{4} \|z\|_2^2 + 2m \mu^\top \phi z \\
        &= \Ex{\bm{Y}^{(\mu)}}
            + m \frac{\eta}{8} \|z\|_2^2
            + m \frac{\eta}{8} \|z\|_2^2
            + 2m \mu^\top \phi z \,.
    \end{align*}
    Then, since $\|z\|_1 = 2\dist_\TV(p,\mu) > 2\epsilon$, we have
    $\|z\|_2^2 > \left(\frac{2\epsilon}{n}\right)^2 \cdot n = 4\epsilon^2 / n$, concluding the
    proof.
\end{proof}

Now we show that the cross term $\mu^\top \phi z$ is small. When $G$ is the cycle, this term will
in fact be zero; when $G$ is the path, the cross term is relevant due to the asymmetry between the
vertices closer to the endpoints or to the middle. However, this will not be a problem as long as
$\|z\|_\infty$ is not too large, which indeed holds when $p$ is not highly concentrated.

\begin{proposition}
    \label{prop:z-delta-cross-term}
    Let $\eta \in (0,1]$ and let $\phi = \Ex{\bPhi}$ be the corresponding expected join matrix.
    Let $\delta > 0$ be any positive real number. Then for any $z \in \bR^n$ satisfying
    \begin{enumerate}
        \item $\sum_i z_i = 0$; and
        \item $\|z\|_\infty \le \delta$,
    \end{enumerate}
    it is the case that
    \[
        \abs*{\mu^\top \phi z} \le \frac{2\delta}{n \eta^2} \,.
    \]
\end{proposition}
\begin{proof}
    When $\phi = \phi^\cyc$, we have that $\mu^\top \phi$ is a constant vector (this is true for any
    circulant matrix), and hence $\mu^\top \phi z = 0$ (since $\sum_i z_i = 0$). Therefore we may
    now assume that $\phi = \phi^\pth$.

    Note that, by symmetry between $z$ and $-z$ in the LHS, it suffices upper bound $\mu^\top \phi z$.
    We expand this expression as follows:
    \[
        \mu^\top \phi z
        = \sum_{i,j \in \bZ_n} \mu_i z_j \phi_{i,j}
        = \frac{1}{n} \sum_{j=0}^{n-1} z_j \left(\sum_{i=0}^{n-1} \phi_{i,j}\right) \,.
    \]
    Hence our goal is to show
    \begin{equation}
        \label{eq:z-j-sum}
        \sum_{j=0}^{n-1} z_j S_j \lequestion \frac{2\delta}{\eta^2} \,,
    \end{equation}
    where $S_j \define \left(\sum_{i=0}^{n-1} \phi_{i,j}\right)$ is the sum of the entries in the $j$-th
    column of $\phi$. Note that $(S_j)_{j=0,\dotsc,n-1}$ is a symmetric unimodal sequence (first
    increasing, then decreasing) with strict inequalities everywhere except for indices
    $\floor{(n-1)/2}$ and $\ceil{(n-1)/2}$ when $n$ is even.
    We will use a ``rearrangement and saturation'' argument
    to construct a vector $z^*$ that upper bounds the LHS of \eqref{eq:z-j-sum} (hereafter called
    the \emph{objective value}).

    Let $z'$ be a vector satisfying the conditions from the statement (hereafter called a
    \emph{feasible solution}) that maximizes the objective value. Let $\sigma$ be a permutation of
    $\{0, \dotsc, n-1\}$ that puts the sequence of column sums in non-decreasing order:
    $S_{\sigma(0)} \le \dotsm \le S_{\sigma(n-1)}$. Then we can also assume that $z'$ respects this
    order: $z'_{\sigma(0)} \le \dotsm \le z'_{\sigma(n-1)}$, since otherwise rearranging the entries
    of $z'$ would yield another feasible solution with equal or larger objective value.

    We now argue that we may assume that, among all nonzero entries of $z'$, all have absolute
    value equal to $\delta$ (which we call \emph{saturated entries}) except for at most one positive
    entry and one negative entry. Indeed, if two consecutive (under $\sigma)$ nonzero entries with
    the same sign are not saturated, \ie they satisfy
    $\abs*{z'_{\sigma(i)}}, \abs*{z'_{\sigma(i+1)}} < \delta$, then we can obtain another feasible
    solution with equal or larger objective value by ``saturating'' this pair of entries, \ie
    making $z'_{\sigma(i)}$ smaller and $z'_{\sigma(i+1)}$ larger until either of them reaches
    a value in $\{-\delta, 0, \delta\}$.

We claim that we may also assume that the multiset of values of the positive entries of $z'$ is
equal to the multiset of absolute values of the negative entries of $z'$.  Suppose $z'$ has $N^+$
entries equal to $\delta$, $N^-$ entries equal to $-\delta$, $M^+ \in \{0,1\}$ entries in the
interval $(0, \delta)$, and $M^- \in \{0,1\}$ entries in the interval $(-\delta, 0)$. If $N^+ =
N^-$, then since $\sum_j z'_j = 0$, we must also have $M^+ = M^-$ and, if this value is $1$, then
the corresponding entries of $z'$ must have the same absolute value so that they add to zero. On the
other hand, if $N^+ \ne N^-$, say $N^+ > N^-$ without loss of generality, then $\sum_i z'_i > 0$
since the sum of the saturated values of $z'$ is at least $\delta$ while the sum of the
non-saturated values must be in $(-\delta, \delta)$. This contradicts the fact that $z'$ is a
feasible solution.

Now we construct $z^*$ by saturating the remaining (zero or two) entries of $z'$:
\[
    z^*_i \define \begin{cases}
        \delta,  & \text{if $z'_i > 0$} \\
        -\delta, & \text{if $z'_i < 0$} \\
        0,       & \text{if $z'_i = 0$.}
    \end{cases}
\]
Then by the same arguments as above, $z^*$ has equal or larger objective value as $z'$.
We now upper bound this objective value by the RHS of \eqref{eq:z-j-sum}, which will
conclude the argument.

    Let $N$ be the number of positive entries of $z^*$. By construction, we have
    \[
        N = |\{i \in \bZ_n : z^*_i = \delta\}| = |\{i \in \bZ_n : z^*_i = -\delta\}| \,.
    \]
    Then our objective value is
    \begin{equation}
        \label{eq:obj-val}
        \sum_j z^*_j S_j
        = \delta \left[ -\sum_{i=0}^{N-1} S_{\sigma(i)} + \sum_{i=0}^{N-1} S_{\sigma(n-1-i)} \right] \,.
    \end{equation}

    Let $h \define (n-1)/2$.
    Since $(S_j)_{j=0,\dotsc,n-1}$ is a symmetric unimodal sequence attaining its maximum in the
    middle, we may say without loss of generality that the indices $\sigma(i)$ in the first
    summation term in the RHS of \eqref{eq:obj-val} are
    $\{0, \dotsc, \ceil{N/2}-1\} \cup \{n-1, \dotsc, n - \floor{N/2}\}$.
    As for the indices $\sigma(n-1-i)$, an exact account depends on the parity of $n$,
    but we can only make the objective value larger by simply using the indices
    $\{\ceil{h}, \dotsc, \ceil{h} + \ceil{N/2} - 1\} \cup
    \{\floor{h}, \dotsc, \floor{h} - \floor{N/2} + 1\}$. Note that when $n$ is odd, this choice
    slightly overestimates the objective value by using the maximum value $S_{h}$ twice,
    but this looser bound suffices for our purposes.

    Therefore, we may finally express and compute our upper bound on the objective value of
    any feasible vector $z$. Recall that $\phi_{i,j} = \nu^{\abs{i-j}}$.
    In the edge case when $\eta=1$ and thus $\nu=0$, we have that $\phi$ is the identity matrix
    and hence $S_j = 1$ for every $j \in \bZ_n$. Therefore we obtain
    \begin{align*}
        \sum_{j=0}^{n-1} z_j S_j = \sum_{j=0}^{n-1} z_j = 0 \,,
    \end{align*}
    which satisfies \eqref{eq:z-j-sum} and we are done. Now, suppose $0 < \nu < 1$.
    Then
    \begin{align*}
        \sum_{j=0}^{n-1} z_j S_j
        &\le \delta\left[ -\sum_{i=0}^{N-1} S_{\sigma(i)} + \sum_{i=0}^{N-1} S_{\sigma(n-1-i)} \right] \\
        &\le \delta\left[
            - \sum_{i=0}^{\ceil{N/2}-1} S_i
            - \sum_{i=0}^{\floor{N/2}-1} S_{n-1-i}
            + \sum_{i=0}^{\ceil{N/2}-1} S_{\ceil{h}+i}
            + \sum_{i=0}^{\floor{N/2}-1} S_{\floor{h}-i}
            \right]
        \le \frac{2\delta}{\eta^2} \,,
    \end{align*}
    where we defer the tedious geometric sum calculations for the last inequality to
    \cref{prop:geometric-sum-expression}.
\end{proof}

\begin{lemma}
    \label{lemma:cross-term-bound}
    Let $C > 0$, and let $p$ be a distribution over $\bZ_n$ that is not $C$-highly concentrated.
    Then as long as $\frac{8C \log n}{\eta^3 \epsilon^2} \le m \le C n \log n$,
    the following holds:
    \[
        \abs*{2m \mu^\top \phi z} \le \frac{m}{2n} \epsilon^2 \eta \,.
    \]
\end{lemma}
\begin{proof}
    Write $p = \mu + z$. Since $\|p\|_1 = \|\mu\|_1 = 1$, it follows that $\sum_i z_i = 0$,
    satisfying the first condition of \cref{prop:z-delta-cross-term}. We will show that $z$ also
    satisfies the second condition with $\delta \define \frac{C \log n}{m}$.

    Indeed, since $0 \le \frac{1}{n} + z_i \le \frac{C \log n}{m}$ (the second inequality
    by \cref{remark:non-concentrated-l-infty-confused-collector}) we get, on the one hand,
    \[
        z_i \ge -\frac{1}{n} \ge -\frac{C \log n}{m} \,,
    \]
    where we used the assumption that $m \le C n \log n$, and on the other hand,
    \[
        z_i \le \frac{C \log n}{m} - \frac{1}{n} \le \frac{C \log n}{m} \,,
    \]
    and thus $\|z\|_\infty \le \frac{C \log n}{m}$ as desired. \cref{prop:z-delta-cross-term}
    implies that
    \[
        \abs*{\mu^\top \phi z} \le \frac{2C \log n}{m n \eta^2} \,.
    \]
    Finally, it suffices to combine this inequality with our assumed lower bound
    on $m$. We obtain
    \[
        \abs*{2m \mu^\top \phi z}
        \le 2m \cdot \frac{2C \log n}{m n \eta^2}
        = \frac{8C \log n}{\eta^3 \epsilon^2} \cdot \frac{\epsilon^2 \eta}{2n}
        \le \frac{m}{2n} \epsilon^2 \eta \,. \qedhere
    \]
\end{proof}

\noindent
We combine the previous results to show the desired separation in the expected value of $\bm{Y}$:

\begin{lemma}[Separation in the expected value of $\bm{Y}$]
    \label{lemma:separation-expected-value-confused-collector}
    Let $C, c > 0$ be constants, and let
    $m = c \cdot \frac{\sqrt{n}}{\epsilon^2 \eta^{3/2}} \log^2 n$.
    Then for all sufficiently large $n$ and all $\epsilon, \eta \in (0, 1]$ satisfying
    $\eta \ge (c/C)^{2/3} \frac{\log^{2/3} n}{n^{1/3} \epsilon^{4/3}}$, the following holds.
    Suppose $p$ is a distribution over $\bZ_n$ that is not $C$-highly concentrated such that
    $\dist_\TV(p, \mu) > \epsilon$. Write $p = \mu + z$. Then the test statistic $\bm{Y}$ satisfies
    \[
        \Ex{\bm{Y}} > \Ex{\bm{Y}^{(\mu)}} + \frac{m \eta}{8} \|z\|_2^2 \,.
    \]
\end{lemma}
\begin{proof}
    By \cref{cor:separation-general-form}, we have
    \[
        \Ex{\bm{Y}} > \Ex{\bm{Y}^{(\mu)}} + \frac{m \eta}{8} \|z\|_2^2
                + \frac{m}{2n} \epsilon^2 \eta + 2m \mu^\top \phi z \,.
    \]
    Hence, we will be done if we can show that
    $2m \mu^\top \phi z \ge -\frac{m}{2n} \epsilon^2 \eta$. This will follow immediately from
    \cref{lemma:cross-term-bound} as long as we can verify the preconditions on $m$. We first
    check the lower bound:
    \[
        m \ge \frac{8C \log n}{\eta^3 \epsilon^2}
        \iff c \frac{\sqrt{n}}{\epsilon^2} \cdot \frac{\log^2 n}{\eta^{3/2}}
            \ge \frac{8C \log n}{\eta^3 \epsilon^2}
        \iff \eta \ge (8C/c)^{2/3} \frac{1}{n^{1/3} \log^{2/3} n} \,,
    \]
    which holds for all sufficiently large $n$ by our assumption on $\eta$. As for the upper bound,
    \[
        m \le C n \log n
        \iff c \frac{\sqrt{n}}{\epsilon^2} \cdot \frac{\log^2 n}{\eta^{3/2}} \le C n \log n
        \iff \eta \ge (c/C)^{2/3} \frac{\log^{2/3} n}{n^{1/3} \epsilon^{4/3}} \,,
        \]
    which holds by assumption. Hence \cref{lemma:cross-term-bound} applies and we are done.
\end{proof}

\begin{remark}
    \label{remark:eta-range-ok}
    The condition $\eta = \Omega\left(\frac{\log^{2/3} n}{n^{1/3} \epsilon^{4/3}}\right)$ in
    the statement above will hold in the range of parameters considered by the present argument.
    Concretely, when $\eta = \Omega\left(\frac{\log^{4/5} n}{n^{1/5} \epsilon^{4/5}}\right)$ and
    $\epsilon = \Omega\left(\frac{1}{n^{1/4}}\right)$, the condition holds because
    \[
        \frac{\left(\frac{\log^{4/5} n}{n^{1/5} \epsilon^{4/5}}\right)}
             {\left(\frac{\log^{2/3} n}{n^{1/3} \epsilon^{4/3}}\right)}
        = (\log n)^{2/15} n^{2/15} \epsilon^{8/15}
        \ge (\log n)^{2/15} n^{2/15} \Omega(n^{-2/15})
        = \omega(1) \,.
    \]
\end{remark}

\subsection{Concentration of the Test Statistic}

We apply the general results presented in \cref{section:common-variance} to upper bound the
variance of $\bm{Y}$.

\begin{lemma}[First component of the variance]
    \label{lemma:var-first-component-confused-collector}
    Let $C > 0$ be a constant, let $n \in \bN$ be sufficiently large and suppose $p$ is a
    probability distribution over $\bZ_n$ that is not $C$-highly concentrated.
    Suppose that $m \le \poly(n)$ and $\eta \ge \Omega(n^{-1/5})$. Then
    \[
        \Varu{\bm{H}}{\Exuc{\bm{T}}{\bm{Y}}{\bm{H}}}
        \le \frac{\|p\|_2^2}{\eta} \cdot O\left( \log^4 n \right) \,.
    \]
\end{lemma}
\begin{proof}
    By \cref{prop:general-variance-first-component}, for some constant $c > 0$ we have
    \[
        \Varu{\bm{H}}{\Exuc{\bm{T}}{\bm{Y}}{\bm{H}}}
            \le 5m^2 \zeta(\cI)\|p\|_1^4
            + c m^2 \cdot \sum_{\substack{I=\llangle i,d \rrangle \in \cI \\ 1 \le d \le n}}
            p_{i} p_{i+d-1} p[I]^2 \Ex{\bm{J}(I)}  \,.
    \]
    We start with the second component of the RHS.
    Recall that for any $I = \llangle i,d \rrangle$, $\Ex{\bm{J}(I)} \le \nu^{d-1}$ (this value
    may be zero if $\cI = \cI^\pth$ and $I$ crosses the edge between vertices $0$ and $n-1$).
    We have
    \begin{align*}
        &\sum_{\substack{I=\llangle i,d \rrangle \in \cI \\ 1 \le d \le n}} p_{i} p_{i+d-1} p[I]^2 \Ex{\bm{J}(I)}
        \le \sum_{i=0}^{n-1} \sum_{d=1}^{n} p_i p_{i+d-1} p[\llangle i,d \rrangle]^2 \nu^{d-1} \\
        &\quad \le \sum_{i=0}^{n-1} \sum_{d=1}^{n} p_i p_{i+d-1} \nu^{d-1}
            \left[ \frac{C\log^2 n}{m} \cdot \max\left\{ \eta(d-1), \frac{1}{\log n} \right\} \right]^2
            & \text{($p$ is not $C$-highly-concentrated)} \\
        &\quad \le \sum_{i=0}^{n-1} \sum_{d=1}^n p_i p_{i+d-1} \nu^{d-1}
            \left[ \frac{C^2\log^4 n}{m^2} \cdot \left( \eta^2 d^2 + \frac{1}{\log^2 n} \right) \right] \\
        &\quad = \frac{C^2 \log^2 n}{m^2} \left(
            \eta^2 \log^2(n) \sum_{d=1}^n d^2 \nu^{d-1} \sum_{i=0}^{n-1} p_i p_{i+d-1} +
            \sum_{d=1}^n \nu^{d-1} \sum_{i=0}^{n-1} p_i p_{i+d-1} \right) \\
        &\quad \le \frac{C^2 \|p\|_2^2 \log^2 n}{m^2} \left(
            \eta^2 \log^2(n) \sum_{d=1}^n d^2 \nu^{d-1} + \sum_{d=1}^n \nu^{d-1} \right)
            & \text{(by Cauchy-Schwarz)} \\
        &\quad \le \frac{C^2 \|p\|_2^2 \log^2 n}{m^2} \left(
            \eta^2 \log^2(n) \frac{1+\nu}{\eta^3} + \frac{1}{\eta} \right)
            & \text{(since $\eta = 1-\nu$)} \\
        &\quad \le \frac{3 C^2 \|p\|_2^2 \log^4 n}{m^2 \eta}
            & \text{(since $\nu < 1$)} \\
        &\quad = \frac{1}{m^2} \cdot \frac{\|p\|_2^2}{\eta} \cdot O\left( \log^4 n \right) \,,
    \end{align*}
    as desired. Then, it suffices to show that the term $5m^2\zeta(\cI)\|p\|_1^4$ is $O(1/n)$,
    since $\|p\|_2^2 \ge \|\mu\|_2^2 = 1/n$. Let $K > 0$ be a constant such that
    $m \le n^K$ for sufficiently large $n$, as per the assumption that $m \le \poly(n)$.
    Then we have $\|p\|_1^4 = 1$ and
    $m^2 \zeta(\cI) = o(m^2/n^{2K+1}) \le o(1/n)$
    by \cref{prop:zeta-is-small-confused-collector}, as needed.
\end{proof}

\begin{lemma}[Second component of the variance]
    \label{lemma:var-second-component-confused-collector}
    Let $C > 0$ be a constant, let $n \in \bN$ be sufficiently large and suppose $p$ is a
    probability distribution over $\bZ_n$ that is not $C$-highly concentrated.
    Suppose $m \le \poly(n)$. Then
    \[
        \Exu{\bm{H}}{\Varuc{\bm{T}}{\bm{Y}}{\bm{H}}}
        \le \frac{\|p\|_2^2}{\eta} \cdot O(\log^4 n) \,.
    \]
\end{lemma}
\begin{proof}
    Let $K \ge 2$ be some constant such that $m \le O(n^{K-1})$.
    Recall that, from \cref{prop:bucket-sizes-small}, the buckets
    $\bm{\Gamma} = (\bm{\Gamma}_1, \dotsc, \bm{\Gamma}_n)$ induced by $\bm{H}$ are such that
    $\abs*{\bm{\Gamma}_i} \le \frac{2K\log n}{\eta}$ for all $i$,
    except with probability at most $1/n^K$. We show that the variance is small when this condition
    holds, and that the low-probability case where the condition fails does not contribute too much
    to the expectation.

    \textbf{Case 1.} Suppose $H$ is such that its induced buckets
    $\Gamma = (\Gamma_1, \dotsc, \Gamma_b)$ satisfy
    $\abs*{\Gamma_i} \le \frac{2K\log n}{\eta}$
    for every $i \in [n]$. We wish to show that $\Varuc{\bm{T}}{\bm{Y}}{\bm{H} = H}$
    satisfies the upper bound from the statement.
    We start with the result from
    \cref{prop:general-variance-second-component}:
    \[
        \Varuc{\bm{T}}{\bm{Y}}{\bm{H} = H} = 2\|p_{|\Gamma}\|_2^2 + 4m \|p_{|\Gamma}\|_3^3 \,.
    \]
    We start by bounding the first term in the RHS. For each bucket $\Gamma_i$, we have
    \[
        \left(p_{|\Gamma}\right)_i^2
        = \left( \sum_{j \in \Gamma_i} p_j \right)^2
        = \sum_{j,k \in \Gamma_i} p_j p_k
        \le \sum_{j,k \in \Gamma_i} \frac{p_j^2 + p_k^2}{2}
        = \abs*{\Gamma_i} \sum_{j \in \Gamma_i} p_j^2
        \le O\left(\frac{\log n}{\eta}\right) \sum_{j \in \Gamma_i} p_j^2 \,.
    \]
    Hence, we obtain
    \[
        \|p_{|\Gamma}\|_2^2
        = \sum_{i=1}^b \left(p_{|\Gamma}\right)_i^2
        \le O\left(\frac{\log n}{\eta}\right) \|p\|_2^2 \,,
    \]
    as desired. Moving on to the second term, first note that
    $\|p_{|\Gamma}\|_3^3 = \sum_i p[\Gamma_i]^3
    \le \sum_i (\max_j p[\Gamma_j]) p[\Gamma_i]^2
    = \|p_{|\Gamma}\|_\infty \|p_{|\Gamma}\|_2^2$. We claim that
    $\|p_{|\Gamma}\|_\infty \le O\left(\frac{\log^3 n}{m}\right)$.
    Fix any $i \in [b]$ and consider entry
    $\left(p_{|\Gamma}\right)_i = p\left[\Gamma_i\right]$. The anticoncentration of $p$ yields
    \[
        p\left[\Gamma_i\right] \le \frac{C \log^2 n}{m} \cdot \max\left\{
            \eta\left( \abs*{\Gamma_i}-1 \right), \frac{1}{\log n} \right\} \,.
    \]
    Combining with the assumption that $\abs*{\Gamma_i} \le O\left(\frac{\log n}{\eta}\right)$,
    we get
    \[
        p\left[\Gamma_i\right] \le \frac{C \log^2 n}{m} \cdot O(\log n) \,,
    \]
    which establishes the claim. We have already shown that
    $\|p_{|\Gamma}\|_2^2 \le \frac{\|p\|_2^2}{\eta} \cdot O(\log n)$, and thus
    \[
        m \|p_{|\Gamma}\|_3^3
        \le m \cdot O\left(\frac{\log^3 n}{m} \right) \cdot \frac{\|p\|_2^2}{\eta} \cdot O(\log n)
        = \frac{\|p\|_2^2}{\eta} \cdot O(\log^4 n) \,,
    \]
    which concludes Case 1.

    \textbf{Case 2.} In the rare event that $\bm{H} = H$ is a subgraph that
    fails the small-buckets condition of Case 1, we will fall back to a looser upper bound for
    the conditional variance that holds for every $H$. We once again start with the result
    from \cref{prop:general-variance-second-component}:
    \[
        \Varuc{\bm{T}}{\bm{Y}}{\bm{H} = H}
        = 2\|p_{|\Gamma}\|_2^2 + 4m \|p_{|\Gamma}\|_3^3 \,.
    \]
    Using $\|p_{|\Gamma}\|_1 = 1$ along with the monotonicity of $\ell^p$ norms gives
    \[
        \Varuc{\bm{T}}{\bm{Y}}{\bm{H} = H}
        \le 2\|p_{|\Gamma}\|_1^2 + 4m \|p_{|\Gamma}\|_1^3
        \le 6m
        = O(n^{K-1}) \,.
    \]

    \textbf{Concluding the argument.} We now combine both cases to upper bound the expected
    variance. Using \cref{prop:bucket-sizes-small}, we have
    \begin{align*}
        \Exu{\bm{H}}{\Varuc{\bm{T}}{\bm{Y}}{\bm{H}}}
        &\le \left[ \begin{array}{l}
            \Pr{\bm{H} \text{ satisfies } \|\bm{\Gamma}\|_\infty \le \frac{2K\log n}{\eta}}
                \Exuc{\bm{H}}
                     {\Varuc{\bm{T}}{\bm{Y}}{\bm{H}}}
                     {\bm{H} \text{ satisfies } \|\bm{\Gamma}\|_\infty \le \frac{2K\log n}{\eta}} \\
            + \Pr{\bm{H} \text{ does not satisfy } \|\bm{\Gamma}\|_\infty \le \frac{2K\log n}{\eta}}
                \max_{H} \left\{ \Varuc{\bm{T}}{\bm{Y}}{\bm{H} = H} \right\}
        \end{array} \right] \\
        &\le \frac{\|p\|_2^2}{\eta} \cdot O(\log^4 n) + \frac{1}{n^K} \cdot O(n^{K-1}) \,.
    \end{align*}
    Since $\|p\|_2^2 \ge \|\mu\|_2^2 = 1/n$, the first term dominates the second, concluding the
    proof.
\end{proof}

\noindent
We conclude that $\bm{Y}$ satisfies the following concentration bound:

\begin{lemma}[Concentration of the Test Statistic]
    \label{lemma:confused-collector-concentration}
    Let $C > 0$ be a constant, let $n \in \bN$ be sufficiently large and suppose $p$ is a
    probability distribution over $\bZ_n$ that is not $C$-highly concentrated.
    Suppose $\eta \ge \Omega(n^{-1/5})$ and $m \le \poly(n)$. Then for all $t > 0$,
    \[
        \Pr{\abs*{\bm{Y} - \Ex{\bm{Y}}} \ge t}
        \le \frac{\|p\|_2^2}{\eta t^2} \cdot O(\log^4 n) \,.
    \]
\end{lemma}
\begin{proof}
    Combine
    \cref{lemma:var-first-component-confused-collector,lemma:var-second-component-confused-collector},
    via the law of total variance, and Chebyshev's inequality.
\end{proof}

\subsection{Correctness of the Tester}

Combining our separation and concentration results above, we can show that $\bm{Y}$ is concentrated
on the correct side of the tester's threshold.

\begin{lemma}
    \label{lemma:tester-second-step-confused-collector}
    Let $\alpha, L > 0$ be constants. Then there exist constants $\beta > 0$, and
    $c = c_{\alpha,\beta} > 0$ such that 
    the following holds for all sufficiently large $n$.
    Let $\epsilon, \eta \in (0, 1]$ satisfy $\eta \ge \frac{L \log^{4/5} n}{n^{1/5} \epsilon^{4/5}}$.
    Suppose $p$ is a probability distribution over $\bZ_n$ that is not $4\alpha$-highly concentrated
    with respect to $m = c \cdot \frac{\sqrt{n}}{\epsilon^2 \eta^{3/2}} \log^2 n$.
    Let $T \define \frac{m}{n^2} \sum_{i,j} \phi_{i,j} + \beta \frac{m}{n} \epsilon^2 \eta$ be the
    threshold used by the second step of \cref{alg:tester-confused-collector}.
    Then the test statistic $\bm{Y}$ satisfies the following:
    \begin{enumerate}
        \item (Completeness) If $p = \mu$, then $\bm{Y} < T$ with probability at least $99/100$;
        \item (Soundness) If $\dist_\TV(p, \mu) > \epsilon$, then $\bm{Y} > T$ with probability at
            least $99/100$.
    \end{enumerate}
\end{lemma}
\begin{proof}
    \textbf{Completeness.} Suppose $p = \mu$.
    By \cref{prop:expectation-y-uniform-confused-collector}, $\bm{Y}$ satisfies
    $\Ex{\bm{Y}} = \frac{m}{n^2} \sum_{i,j} \phi_{i,j}$.
    Hence for any fixed $\beta$ (to be chosen below),
    it suffices to show that $\bm{Y} < \Ex{\bm{Y}} + \beta \frac{m}{n} \epsilon^2 \eta$ with good
    probability. By \cref{lemma:confused-collector-concentration} and using the fact that
    $\|\mu\|_2^2 = 1/n$,
    \[
        \Pr{\bm{Y} \ge \Ex{\bm{Y}} + \beta \frac{m}{n} \epsilon^2 \eta}
        \le \Pr{\abs*{\bm{Y} - \Ex{\bm{Y}}} \ge \beta \frac{m}{n} \epsilon^2 \eta}
        \le \frac{\|p\|_2^2}{\eta \left(\beta \frac{m}{n} \epsilon^2 \eta\right)^2}
            \cdot O(\log^4 n)
        = \frac{O(n \log^4 n)}{\beta^2 m^2 \epsilon^4 \eta^3} \,,
    \]
    and we have
    \begin{equation}
        \label{eq:m-requirement-1}
        \frac{O(n \log^4 n)}{\beta^2 m^2 \epsilon^4 \eta^3} \le 1/100
        \iff m \ge \frac{1}{\beta} \cdot
                O\left( \frac{\sqrt{n}}{\epsilon^2 \eta^{3/2}} \log^2 n\right) \,,
    \end{equation}
    as desired. Thus, there exists constant
    $c^{(1)} = c^{(1)}_{\alpha,\beta}$ such that if
    $m \ge c^{(1)} \frac{\sqrt{n}}{\epsilon^2 \eta^{3/2}} \log^2 n$
    then, for all sufficiently large $n$, $\bm{Y} < T$ with probability at least $99/100$.

    \textbf{Soundness.} We proceed similarly. By
    \cref{lemma:separation-expected-value-confused-collector,prop:expectation-y-uniform-confused-collector},
    we have
    \[
        \Ex{\bm{Y}} > \Ex{\bm{Y}^{(\mu)}} + \frac{m \eta}{8} \|z\|_2^2
        = \frac{m}{n^2}\sum_{i,j} \phi_{i,j} + \frac{m \eta}{8} \|z\|_2^2 \,.
    \]
    Therefore it suffices to show that, for appropriately chosen $\beta$, we have
    \[
        \bm{Y} \gtquestion \Ex{\bm{Y}} - \frac{m \eta}{8} \|z\|_2^2 + \beta \frac{m}{n} \epsilon^2 \eta \,.
    \]
    Recall that, when $\dist_\TV(p, \mu) > \epsilon$, we have $\|z\|_1 > 2\epsilon$, which implies
    $\|z\|_2^2 > 4\epsilon^2 / n$ and hence
    \[
        \frac{m \eta}{8} \|z\|_2^2 > \frac{m}{2n} \epsilon^2 \eta \,,
    \]
    so that
    \[
        \Ex{\bm{Y}} - \frac{m \eta}{8} \|z\|_2^2 + \beta \frac{m}{n} \epsilon^2 \eta
        < \Ex{\bm{Y}} - \frac{m \eta}{8} \|z\|_2^2 + 2\beta \frac{m \eta}{8} \|z\|_2^2
        = \Ex{\bm{Y}} - (1-2\beta) \frac{m \eta}{8} \|z\|_2^2 \,.
    \]
    Thus, for $\beta \le 1/3$, we have $1-2\beta \ge \beta$ and it suffices to show that the
    following holds with probability at least $99/100$:
    \[
        \bm{Y} \gtquestion \Ex{\bm{Y}} - \beta \frac{m \eta}{8} \|z\|_2^2 \,.
    \]
    We apply \cref{lemma:confused-collector-concentration} again, along with
    $\|p\|_2^2 = \|\mu\|_2^2 + \|z\|_2^2 = \frac{1}{n} + \|z\|_2^2$ and
    $\|z\|_2^2 \ge \frac{4\epsilon^2}{n}$.
    \begin{align*}
        \Pr{\bm{Y} \le \Ex{\bm{Y}} - \beta \frac{m \eta}{8} \|z\|_2^2}
        &\le \Pr{\abs*{\bm{Y} - \Ex{\bm{Y}}} \ge \beta \frac{m \eta}{8} \|z\|_2^2}
        \le \frac{\|p\|_2^2}{\eta \left( \beta \frac{m \eta}{8} \|z\|_2^2 \right)^2}
            \cdot O(\log^4 n) \\
        &= \frac{(\frac{1}{n} + \|z\|_2^2) O(\log^4 n)}{\beta^2 m^2 \eta^3 \|z\|_2^4}
        \le \frac{\frac{1}{n} O(\log^4 n)}{\beta^2 m^2 \eta^3 (\epsilon^2 / n)^2}
            + \frac{O(\log^4 n)}{\beta^2 m^2 \eta^3 (\epsilon^2 / n)} \\
        &= \frac{O(n \log^4 n)}{\beta^2 m^2 \eta^3 \epsilon^4} \,.
    \end{align*}
    This failure probability is asymptotically the same as that obtained in the completeness case.
    Thus there exists a constant $c^{(2)} = c^{(2)}_{\alpha,\beta} > 0$ such that, for
    $m \ge c^{(2)} \frac{\sqrt{n}}{\epsilon^2 \eta^{3/2}} \log^2 n$
    and all sufficiently large $n$, $\bm{Y} > T$ with probability at least $99/100$.
    Setting $c = \max\left\{ c^{(1)}_{\alpha,\beta}, c^{(2)}_{\alpha,\beta} \right\}$
    concludes the proof.
\end{proof}

\noindent
Finally, we establish correctness by combining our results for the two steps of the tester:

\begin{theorem}[Refinement of \cref{thm:intro-confused-collector-main}]
\label{thm:confused-collector-main}
    There exist constants $\alpha > 0$, $\beta > 0$, $c = c_{\alpha,\beta} > 0$,
    and $L = L_c > 0$ such that
    the following holds for all sufficiently large $n$. Suppose $\epsilon, \eta \in (0, 1]$ satisfy
    $\eta \ge \frac{L \log^{4/5} n}{n^{1/5} \epsilon^{4/5}}$. Let $G$ be either the cycle or the
    path on vertices $V = \bZ_n$, and let $p$ be a probability distribution
    over $\bZ_n$.
    Then \cref{alg:tester-confused-collector} instantiated with constants
    $\alpha, \beta, L \text{ and } c$ has sample complexity
    $\Theta\left(\frac{\sqrt{n}}{\epsilon^2 \eta^{3/2}} \log^2 n\right)$
    and its output on $p$ satisfies
    \begin{enumerate}
        \item (Completeness) If $p$ is the uniform distribution over $\bZ_n$, then the algorithm
            accepts with probability at least $9/10$;
        \item (Soundness) If $p$ is $\epsilon$-far from the uniform distribution in TV distance,
            then the algorithm rejects with probability at least $9/10$.
    \end{enumerate}
\end{theorem}
\begin{proof}
    We start by instantiating $\alpha > 0$ large enough as per
    \cref{lemma:concentration-test-eta}. That lemma also requires that $m \le n \eta$, which we
    now verify. Fix any constant $c > 0$ and suppose
    $m = c \cdot \frac{\sqrt{n}}{\epsilon^2 \eta^{3/2}} \log^2 n$. Then
    \[
        m \le n \eta
        \iff c \cdot \frac{\sqrt{n}}{\epsilon^2} \cdot \frac{\log^2 n}{\eta^{3/2}} \le n \eta
        \iff \eta^{5/2} \ge c \cdot \frac{\log^2 n}{n^{1/2} \epsilon^2}
        \iff \eta \ge \frac{c^{2/5} \log^{4/5} n}{n^{1/5} \epsilon^{4/5}} \,.
    \]
    Therefore, for any choice of $c$, setting $L = L_c \ge c^{2/5}$ ensures that $m \le n \eta$.

    Thus, instantiate $\beta, c > 0$ as provided by
    \cref{lemma:tester-second-step-confused-collector}, and the corresponding $L_c$ as above.
    Now, we can use \cref{lemma:concentration-test-eta,lemma:tester-second-step-confused-collector}
    to establish overall correctness of the tester. (Note that the sample complexity claim follows
    from the specification of the algorithm.)

    \textbf{Completeness.} By \cref{lemma:concentration-test-eta}, the first step of the tester
    rejects only with probability at most $1/100$. Likewise, by
    \cref{lemma:tester-second-step-confused-collector}, the second step of the tester rejects
    only with probability at most $1/100.$ Hence the total rejection probability is at most
    $2/100 < 1/10$.

    \textbf{Soundness.} There are two cases depending on the concentration of $p$. First, suppose
    $p$ is $4\alpha$-highly concentrated. Then the first step of the tester rejects with
    probability at least $99/100$ by \cref{lemma:concentration-test-eta}. On the other hand, if
    $p$ is not $4\alpha$-highly concentrated, then the second step of the tester rejects with
    probability at least $99/100$ by \cref{lemma:tester-second-step-confused-collector}. Either
    way, the tester rejects with probability at least $99/100 > 9/10$.
\end{proof}

\begin{remark}
    \label{remark:confused-collector-parameter-ranges}
    In the introduction (see \cref{thm:intro-confused-collector-main}), we stated that our sample
    complexity upper bound would apply to the regime where
    $\epsilon \ge \widetilde \Theta\left( n^{-1/4} \right)$ and
    $\eta \ge \widetilde \Theta\left( n^{-1/5} \epsilon^{-4/5} \right)$. Although the condition on
    $\epsilon$ is not explicitly stated above, it is a consequence of the condition on $\eta$
    and the fact that $\eta \le 1$ in our definition of the problem:
    \[
        \frac{L \log^{4/5} n}{n^{1/5} \epsilon^{4/5}} \le \eta \le 1
        \implies \epsilon \ge \frac{L^{5/4} \log n}{n^{1/4}} \,.
    \]
    An interesting question is whether it is possible to handle an ever wider range of parameters;
    in particular, our analysis uses the inequality $m \le n \eta$, but if we allow arbitrarily
    small $\epsilon$, then necessarily $m \gg n$. Also note that it is not possible to
    handle the \emph{full} range of parameters: for sufficiently small $\epsilon$ and (say)
    $\eta = 1/10$, one may place the deviation from uniformity on two adjacent vertices,
    and with probability at least $9/10$ this deviation will be imperceptible to the tester.
\end{remark}

\section{Testing Uniformity in the Parity Trace Model}
\label{sec:general-upper-bound}

In this section we state the upper bound portion of our main \cref{thm:intro-main-informal},
stated formally here:

\begin{theorem}
\label{thm:intro-main}
Fix domain $[2n]$. Let $\Pi$ contain only the uniform distribution.  Then the sample complexity
of $(\Pi, \far^\TV_\epsilon(\Pi))$-distribution testing under the parity trace is $\widetilde
\Theta\left( \left(\frac{n}{\epsilon}\right)^{4/5} + \frac{\sqrt n}{\epsilon^2}\right)$.
\end{theorem}

Following the setup from \cref{sec:collision-based-testing-general}, we consider the task of testing
uniformity of an unknown distribution $\pi = \pi(p, q)$ in the parity trace model. Recall that we
stitch the ends of the trace into a necklace and study the resulting circular trace.
Therefore our base graph $G = (\bZ_n, E)$ is the cycle and we think of $p$ as a partial distribution
over the vertices, whereas $q$ will determine the weights of the edges:
if $e \in E$ connects vertices $i$ and $i+1$ (mod $n$), then $w(e) = 1 - e^{-m q_i}$,
and this edge is sampled into $\bm{H}$ with probability $1 - w(e) = \Pr{\Poi(mq_i) = 0}$.

The testing algorithm has two cases: when $\epsilon$ is very small, in which case we may reduce to
the standard uniformity testing algorithm (handled in \cref{section:gen-small-eps}); and when
$\epsilon$ is not too small, in which case our main analysis applies. In our main analysis, the
tester performs 3 steps:
\begin{enumerate}
\item Bias test: check whether the counts of 1- and 0-valued symbols in the trace are too
    unbalanced, in which case the distribution must be far from uniform;
\item Concentration test: check whether any run-length (of either 1- or 0-valued symbols)
    is too large. We will show that this case can also be safely rejected.
\item Collision-based test: accept or reject depending on whether the test statistic $\bm{Y}$ is
    below a certain threshold. This step will require the most technical work.
\end{enumerate}
Formally, the algorithm is \cref{alg:uniformity-tester-linear-trace}.
It is parameterized by absolute constants $\alpha, \beta, \gamma, K$, which will be
defined later, and requires that $\epsilon \ge \frac{K\log^{3} n}{n^{1/4}}$; note that this
condition, combined with the fact that the algorithm sets $m = O\left(
\left(\frac{n}{\epsilon}\right)^{4/5} \log^{7/5} n\right)$, implies $m = O\left(\frac{n}{\log
n}\right) = o(n)$. We will use this fact throughout the analysis.

\begin{algorithm}[th]
    \caption{Uniformity tester for the case when
        $\epsilon \ge \frac{K_{\alpha,\beta,\gamma} \log^{3}}{n^{1/4}}$.}
    \label{alg:uniformity-tester-linear-trace}
    \hspace*{\algorithmicindent}
    Set $m \gets \Theta_{\alpha, \beta, \gamma} \left(
        \left(\frac{n}{\epsilon}\right)^{4/5} \log^{7/5} n\right)$. \\
    \hspace*{\algorithmicindent} \textbf{Constants:}
        $\alpha, \beta, \gamma > 0$ and $K = K_{\alpha,\beta,\gamma} > 1$, to be defined later.\\
    \hspace*{\algorithmicindent}
    \textbf{Input:} For $\pi = \pi(p,q)$ on domain $\bZ_n$,
        receive $\trace(S)$ for sample $S \gets \samp(\pi, m)$ \\
    \hspace*{\algorithmicindent} \textbf{Requires:} $\epsilon \ge \frac{K\log^{3} n}{n^{1/4}}$. \\
    \begin{algorithmic}[1]
        \Procedure{UniformityTester-ParityTrace}{$\trace(S)$}
            \State Construct the circular trace from $\trace(S)$.
            \For {$b \in \{0,1\}$}
                \State Let $X_1, \dotsc, X_n$ be the ``$b$'' run-lengths defined in
                    \cref{section:parity-trace}.
                \State $N \gets \sum_i X_i$.
                \State $Y \gets \frac{1}{m} \sum_i X_i(X_i-1)$.
                \State \textbf{If} {$\frac{N}{m} \ge \frac{1}{2} + \frac{\gamma}{\sqrt{m}}$}
                    then \textbf{reject}.
                    \label{line: reject bias}
                \State \textbf{If} {$\max_i X_i \ge \alpha \log n$} then \textbf{reject}.
                    \label{line: reject concentration}
                \State \textbf{If} {$Y \ge \frac{m}{4n^2}\sum_{i,j} \phi^{(\mu)}_{i,j}
                        + \beta \frac{\epsilon^2 m^2}{n^2}$} then \textbf{reject}.
            \EndFor
            \State \textbf{Accept}.
        \EndProcedure
    \end{algorithmic}
\end{algorithm}

We will write $\mu$ for the partial distribution such that $\pi(\mu,\mu)$ is the uniform distribution,
\ie $\mu_i = 1/2n$ for each $i \in \bZ_n$. As discussed in 
\cref{section:parity-trace}, we analyze only the statistics for the 1s in the trace (\ie the case
$b=1$ in \cref{alg:uniformity-tester-linear-trace}), as the case for the 0s is symmetric.
\noindent
\cref{section:parity-easy} will show that the first two steps of
\cref{alg:uniformity-tester-linear-trace} are correct. 
\cref{section:parity-expectation} will show that $\Ex{\bm{Y}}$ is small when the input distribution
is uniform, and large when it is far from uniform.
\cref{section:parity-variance} will give a bound on the variance of $\bm Y$.
\cref{section:parity-correctness} combines these results to prove correctness of
\cref{alg:uniformity-tester-linear-trace}. Together with the algorithm for the small $\epsilon$ case
in \cref{section:gen-small-eps}, this will prove the upper bound of \cref{thm:intro-main}.

\subsection{The Small $\epsilon$ Case}
\label{section:gen-small-eps}

Fix any constant $K > 0$.  In the case $\epsilon < \frac{K \log^{3} n}{n^{1/4}}$, the tester will
simulate the standard uniformity tester (see \eg \cite{VV17,DGPP19}). The tester under the parity
trace will use a sample of size $O\left(\frac{\sqrt n}{\epsilon^2} + n \log n\right) = \widetilde
O\left(\frac{\sqrt n}{\epsilon^2}\right)$. The simulation is possible because, using $O(n \log n)$
samples, the tester either receives a sample from every domain element (therefore gaining the
ability to correctly distinguish all elements of the support), or it can safely reject. This proof
is not particularly insightful and we defer it to \cref{section:parity-trace-small-eps}.

\begin{lemma}
\label{lemma:upper-bound-linear-trace-small-epsilon}
Suppose that $\epsilon < \frac{K \log^3 n}{n^{1/4}}$. Then there is a distribution tester under the
parity trace, with sample complexity $\widetilde O\left(\frac{\sqrt n}{\epsilon^2}\right)$,
such that on input distribution $\pi = \pi(p,q)$:
\begin{enumerate}
\item If $p = q = \mu$, the algorithm will accept with probability at least $2/3$; and,
\item If $\pi$ is $\epsilon$-far from uniform, then the algorithm will reject with probability at
least $2/3$.
\end{enumerate}
\end{lemma}

\subsection{Easy Cases: Unbalanced and Highly-Concentrated Distributions}
\label{section:parity-easy}

We now proceed to the main analysis, where $\epsilon \geq \frac{K \log^3 n }{n^{1/4}}$, where $K$ is
some constant.  In this section we handle the two ``easy'' rejection cases of the tester, which
detect whether the total probability masses on the 1-valued elements and 0-valued elements are
significantly unbalanced, or whether one of the partial distributions $p$ or $q$ is highly
concentrated relative to the other.
First, we consider the case where the partial distributions are unbalanced:

\begin{proposition}
    \label{prop:easy-case-bias}
    For sufficiently large absolute constant $\gamma > 0$ and sufficiently large $n$, the tester
    satisfies the following:
    \begin{itemize}
        \item When $\pi(p, q)$ is uniform, \cref*{line: reject bias} rejects with probability
            at most $1/100$.
        \item When $\|p\|_1 \not\in \frac{1}{2} \pm \frac{2\gamma}{\sqrt{m}}$
            (equivalently, $\|q\|_1 \not\in \frac{1}{2} \pm \frac{2\gamma}{\sqrt{m}}$),
            \cref*{line: reject bias} rejects with probability at least $99/100$.
    \end{itemize}
\end{proposition}
\begin{proof}
    Fix the iteration of the tester with $b=1$, so that $\bm{N} \sim \Poi(m\|p\|_1)$ is the number
    of ``1'' symbols observed in the trace. First suppose $\pi(p,q)$ is uniform, so that in
    particular $\|p\|_1 = 1/2$. Then by \cref{fact:poisson-concentration}, the probability that
    the test in this iteration rejects is
    \[
        \Pr{\frac{\bm{N}}{m} \ge \frac{1}{2} + \frac{\gamma}{\sqrt{m}}}
        = \Pr{\bm{N} \ge m/2 + \gamma\sqrt{m}}
        \le e^{-\frac{m \gamma^2}{2(m/2 + \gamma\sqrt{m})}}
        \le e^{-\frac{\gamma^2}{2}}
        \le 1/200 \,.
    \]
    Now suppose $\|p\|_1 > \frac{1}{2} + \frac{2\gamma}{\sqrt{m}}$. Then the probability that the
    test in this iteration \emph{fails} to reject is
    \[
        \Pr{\frac{\bm{N}}{m} < \frac{1}{2} + \frac{\gamma}{\sqrt{m}}}
        \le \Pr{\bm{N} \le m\|p\|_1 - \gamma\sqrt{m}}
        \le e^{-\frac{m \gamma^2}{2(m\|p\|_1 + \gamma\sqrt{m})}}
        \le e^{-\frac{\gamma^2}{4}}
        \le 1/200 \,.
    \]
    By symmetry, the same holds for the iteration $b=0$ with respect to $\|q\|_1$. Hence this
    step correctly accepts/rejects except with probability at most $1/100$.
\end{proof}

Next, we handle the case where $p$ or $q$ is highly concentrated relative to the other, which we
define as follows.

\begin{definition}[Highly concentrated partial distributions]
Given a constant $C > 0$, positive integer $m$, and partial distribution $\pi = \pi(p,q)$, we say
that $p$ is \emph{$C$-highly concentrated}\footnote{We would adjust this definition appropriately to
test the statistic for the 0s.} \emph{relative to $q$ (with respect to $m$)} if
$\RelativeConcentrationT{p}{q}{t} \ge C \log^2 n$, where $t = \frac{1}{m\log n}$.
\end{definition}

\begin{remark}
    \label{remark:non-concentrated-l-infty-linear-trace}
    If $p$ is not $C$-highly concentrated relative to $q$, then in particular
    $\|p\|_\infty < \frac{C\log n}{m}$,
    as can be seen by taking intervals $I = \llangle i, 1 \rrangle$ for each $i \in \bZ_n$.
\end{remark}

We can now combine this definition with \cref{lemma:relative-concentration-structural} to show
that the second step of the tester behaves as intended:

\begin{proposition}
    \label{prop:easy-case-concentrated}
    For sufficiently large absolute constant $\alpha > 0$ and sufficiently large $n$, the tester
    satisfies the following:
    \begin{enumerate}
        \item When $\pi(p, q)$ is uniform, \cref*{line: reject concentration} rejects with
            probability at most $1/100$.
        \item When at least one of $p, q$ is $4\alpha$-highly concentrated relative to the other
            with respect to $m$, \cref*{line: reject concentration} rejects with probability at
            least $99/100$.
    \end{enumerate}
\end{proposition}
\begin{proof}
    \textbf{Completeness.}
    Suppose $\pi(p, q)$ is uniform and fix iteration $b=1$,
    so that the tester is looking for long runs of ``1'' symbols. Therefore, we have
    buckets $\bm{\Gamma} = (\bm{\Gamma}_1, \dotsc, \bm{\Gamma}_{\bm{b}})$
    corresponding to the connected components of $\bm{H}$, and each bucket $\bm{\Gamma}_i$
    contributes a run length $\bm{X}_i \sim \Poi\left(m p\left[ \bm{\Gamma}_i \right]\right)$.

    First, we claim that $\max_i q\left[\bm{\Gamma}_i^*\right] \le \frac{2\log n}{m}$ with high
    probability. For each $i \in \bZ_n$, let $I_i$ be the minimal circular interval
    $I_i = \llangle i, d \rrangle$ satisfying $q[I^*_i] > \frac{2\log n}{m}$.
    If $\max_i q\left[\bm{\Gamma}_i^*\right] > \frac{2\log n}{m}$, then at least one of these
    intervals was joined, \ie $\bm{J}(I_i) = 1$ for some $i \in \bZ_n$. Recall that each edge
    $e = (i, i+1)$ appears in $\bm{H}$ with probability
    $1 - w(e) = e^{-m q_i} = \Pr{\Poi(mq_i) = 0}$. Thus
    \begin{align*}
        \Pr{\bm{J}(I_i) = 1}
        &= \Pr{\forall e \in I_i^* : e \in \bm{H}}
        = \Pr{\Poi(m q[I^*_i]) = 0} \\
        &< \Pr{\Poi\left(m \cdot \frac{2\log n}{m}\right) = 0}
        = e^{-2\log n} = 1/n^2 \,.
    \end{align*}
    By the union bound,
    $\Pr{\max_i q\left[\bm{\Gamma}_i^*\right] > \frac{2\log n}{m}} < 1/n = o(1)$.

    Now, suppose $\max_i q\left[\bm{\Gamma}^*_i\right] \le \frac{2\log n}{m}$. Since $\pi$ is
    uniform and therefore $p=q$, it follows that
    $\max_i p\left[\bm{\Gamma}_i\right] \le \frac{4\log n}{m}$ (the factor of $2$ accounts for
    the fact that $\Gamma_i$ is in general one element larger than $\Gamma_i^*$, and if
    the latter is empty and $\Gamma_i$ has size 1, then the claim follows from the fact that
    $\pi$ is the uniform distribution).
    Thus, we use \cref{fact:poisson-concentration} to upper bound the probability that any fixed
    bucket $\Gamma_i$ produces a run length that \cref*{line: reject concentration} would
    reject: for sufficiently large $\alpha$,
    \begin{align*}
        \Pruc{}{\bm{X}_i \ge \alpha \log n}{\bm{\Gamma}_i = \Gamma_i}
        &= \Pr{\Poi\left( m p[\Gamma_i] \right) \ge \alpha \log n}
        \le \Pr{\Poi(4\log n) \ge \alpha \log n} \\
        &\le \Pr{\Poi(4\log n) - 4\log n \ge (\alpha - 4)\log n}
        \le e^{-\frac{(\alpha - 4)^2 \log^2 n}{2\left( (\alpha-4)\log n + 4\log n \right)}} \\
        &\le e^{-\frac{(\alpha/2)^2 \log n}{2 \alpha}}
        \le e^{-2 \log n}
        \le 1 / n^2 \,.
    \end{align*}
    Hence, the probability that this event occurs for any bucket is at most $1/n = o(1)$,
    and by symmetry the same is true for the iteration of the tester with $b=0$.
    Therefore, when $\pi(p, q)$ is uniform \cref*{line: reject concentration} rejects only
    with $o(1)$ probability.

    \textbf{Soundness.}
    Suppose without loss of generality that $p$ is $4\alpha$-highly concentrated relative to $q$.
    Using \cref{lemma:relative-concentration-structural}, let $I$ be a circular
    interval satisfying
    \begin{enumerate}
        \item $q[I^*] \le \frac{1}{m \log n}$; and
        \item $p[I] \ge 2\alpha \frac{\log n}{m}$.
    \end{enumerate}

    First, we claim that with high probability $\bm{J}(I) = 1$. Indeed we have
    \begin{align*}
        \Pr{\bm{J}(I) = 0}
        &= 1- \Pr{\forall e \in I^* : e \in \bm{H}}
        = 1 - \Pr{\Poi(mq[I^*]) = 0} \\
        &= \Pr{\Poi(mq[I^*]) > 0}
        \le \Pr{\Poi\left( \frac{1}{\log n} \right) > 0}
        = 1 - e^{-\frac{1}{\log n}}
        = o(1) \,.
    \end{align*}
    Therefore any ``1'' symbols sampled from the vertices in $I$ will belong to the
    same run, and this run will contain at least $\Poi(mp[I])$ symbols where
    $mp[I] \ge 2\alpha \log n$. Thus the probability that this run length fails to exceed the
    rejection threshold is at most
    \begin{align*}
        \Pr{\Poi(mp[I]) < \alpha \log n}
        &\le \Pr{\Poi(2\alpha \log n) \le \alpha \log n}
        = \Pr{\Poi(2\alpha \log n) \le 2\alpha \log n - \alpha \log n} \\
        &\le e^{-\frac{\alpha^2 \log^2 n}{2(\alpha \log n + 2\alpha \log n)}}
        = e^{-\frac{\alpha \log n}{6}}
        = o(1)\,.
    \end{align*}
    Hence \cref*{line: reject concentration} rejects except with $o(1)$ probability, completing
    the proof.
\end{proof}

These two steps will allow us to assume, when useful, that
1) $\|p\|_1, \|q\|_1 \in \frac{1}{2} \pm \frac{2\gamma}{\sqrt{m}}$; and 2) neither $p$ nor $q$ is
$4\alpha$-highly concentrated relative to the other with respect to $m$.
We are now ready to analyze our main test statistic $\bm{Y}$.

\subsection{Expected Value of the Test Statistic}
\label{section:parity-expectation}

Our first goal is to show that $\Ex{\bm{Y}}$ is well-separated between the case when $\pi(p,q) =
\pi(\mu,\mu)$ and when $\dist_\TV(\pi(p,q), \pi(\mu,\mu)) > \epsilon$.  Note that the latter case
implies that $\|p - \mu\|_1 > \epsilon$ or $\|q - \mu\|_1 > \epsilon$ and that our tester is
symmetric with respect to the 0- and 1-valued symbols, so it is safe to assume without loss of
generality that $\|q - \mu\|_1 > \epsilon$.

Recall our formulation of the expected value of $\bm{Y}$ from \cref{prop:expectation-y-quadratic-form}:
\begin{equation}
    \label{eq:expectation-y-quadratic-form-linear-trace}
    \Ex{\bm{Y}} = m p^\top \phi p \,.
\end{equation}
We first show that, when $\pi(p,q) = (\mu,\mu)$ this value is precisely the baseline against which
\cref{alg:uniformity-tester-linear-trace} thresholds the test statistic. Let
$\phi^{(\mu)} \define \Ex{\bPhi^{(\mu)}}$, where $\bPhi^{(\mu)} = \bPhi(\mu)$ is the random
join matrix produced by the uniform partial distribution $q = \mu$.
Let random variable $\bm{Y}^{(\mu)}$ denote the value of the test statistic when $\pi(p,q)$
is the uniform distribution. Then we have:


\begin{proposition}[Expectation of $\bm{Y}$ in the uniform case]
    \label{prop:expectation-y-uniform}
    When $\pi(p,q)$ is the uniform distribution, the statistic $\bm{Y} = \bm{Y}^{(\mu)}$ satisfies
    \[
        \Ex{\bm{Y}^{(\mu)}} = \frac{m}{4n^2} \sum_{i,j} \phi^{(\mu)}_{i,j} \,.
    \]
\end{proposition}
\begin{proof}
    Since $p = q = \mu = \vec 1 / 2n$, \eqref{eq:expectation-y-quadratic-form-linear-trace} yields
    \[
        \Ex{\bm{Y}^{(\mu})} = m \mu^\top \phi^{(\mu)} \mu
        = \frac{m}{4n^2} \sum_{i,j} \phi^{(\mu)}_{i,j} \,. \qedhere
    \]
\end{proof}

We wish to show that when $q$ is far from $\mu$, the quadratic form $p^\top \phi p$ is large
regardless of the choice of $p$ (assuming our conditions on the relative concentration and the
bias). Our strategy is to show that even a ``worst-case'' partial distribution $\overline p =
\overline p(q)$, tailored to make $\bm{Y}$ as small as possible, would still incur a large gap
compared to the uniform case; and then argue that if $p$ deviates from $\overline p$, this can only
make the testing task easier.

We call this worst-case partial distribution the \emph{uniform conjugate} of $q$, and denote it by
$\overline p$ (since it takes the role of $p$). Formally, we say a partial distribution $\overline
p$ is a $\tau$-uniform conjugate of $q$ if
\[
    \left( \sum_{i=0}^{n-1} \overline p_i \right) + \left( \sum_{i=0}^{n-1} q_i \right) = 1
    \qquad\text{ and }\qquad
    \phi \overline p = \tau \cdot \vec 1 \,.
\]
Note that the $i$-th entry of $\phi p$ is
\[
    \left( \phi p \right)_i = \sum_{j=0}^{n-1} \phi_{i,j} p_j
    = \Exu{\bPhi}{\sum_{j=0}^{n-1} \bPhi_{i,j} p_j} \,,
\]
which is the expected sum $p[\bm{\Gamma}_{\bm{\gamma}(i)}]$ of the bucket containing element $i \in
\bZ_n$.  Therefore, when $\overline p$ is a $\tau$-uniform conjugate of $q$, this expected sum is the
same for every bucket. This is the ``worst-case'' because, in expectation, the distribution over the
odd elements sampled from $\overline p$ will be uniform.

If we allow $\overline p$ to have negative entries, one could show that every $q$ has a
uniform conjugate for some value $\tau$, because $\phi$ is positive semidefinite\footnote{If all
entries of $q$ are non-zero, then as observed in \cref{claim:join-semidefinite} $\phi$ is positive
definite and therefore invertible, so $\overline p$ exists for appropriate $\tau$.  If $q$ has some
zero entries, one can first reduce the domain by eliminating such entries, solve the inverse
problem, and then arbitrarily distribute the density in each range of $\overline p$ separated by
zero $q$-density.} (\cref{claim:join-semidefinite}). But we require an explicit $\tau$ (and
non-negative entries). We will give a closed-form solution for an \emph{approximate} uniform
conjugate:

\begin{definition}[Approximate uniform conjugate]
\label{def:approximate-uniform-conjugate}
For a partial distribution $q$ over $\bZ_n$ such that $\|q\|_1 > 0$
and sample-size parameter $m$, with expected join matrix
$\phi = \phi(q)$, let $\xi(m,q) \define \frac{e^{-m\|q\|_1}}{(1-e^{-m\|q\|_1})^2}$.
We say that $\widetilde p \in \bR^{\bZ_n}_{\ge 0}$ is an \emph{approximate uniform
conjugate} of $q$ if $\|\widetilde p\|_1 = 1 - \|q\|_1$, and, for $\tau = \tau(m,q) \define
\frac{1-\|q\|_1}{\sum_{i=0}^{n-1} \tanh\left(\frac{mq_i}{2}\right)}$, it holds that
\[
  \max_{i \in \bZ_n} \abs*{(\phi \widetilde p)_i - \tau} \leq 4 n \cdot \xi(m,q) \,.
\]
\end{definition}

\noindent
Going forward, write $\xi \define \xi(m, q)$ for convenience.
We write the expected value of our test statistic in terms of an approximate uniform conjugate:

\begin{proposition}
    \label{prop:conjugate-mean}
    Suppose $\widetilde p$ is an approximate uniform conjugate of $q$, and
    write $p = \widetilde p + z$. Then
    \[
        \Ex{\bm{Y}} = m \widetilde p^\top \phi \widetilde p + m z^\top \phi z \pm 8mn^2 \xi
    \]
    and
    \[
        \widetilde p^\top \phi \widetilde p = \|p\|_1 (\tau \pm 4n\xi) \,.
    \]
    Moreover, if $z = \vec 0$, then we simply have
    \[
        \Ex{\bm{Y}} = m \widetilde p^\top \phi \widetilde p \,.
    \]
\end{proposition}
\begin{proof}
    We have
    \begin{align*}
        \Ex{\bm{Y}}
        &= m p^\top \phi p
         = m (\widetilde p + z)^\top \phi (\widetilde p + z)
         = m \left( \widetilde p^\top \phi \widetilde p + 2 z^\top \phi \widetilde p + z^\top \phi z \right) \\
        &= m \widetilde p^\top \phi \widetilde p + m z^\top \phi z + 2m \sum_{i}  z_i (\phi \widetilde p)_i
        = m \widetilde p^\top \phi \widetilde p + m z^\top \phi z + 2m \sum_{i}  z_i (\tau \pm 4n\xi) \,.
    \end{align*}
    Since $\sum_j z_j = 0$ (because $\|p\|_1 = \|\widetilde p\|_1 = 1 - \|q\|_1$)
    and $|z_i| \leq 1$ for each $i \in \bZ_n$, the term $2m \sum_i z_i (\tau \pm 4n\xi)$
    is bounded by $8mn^2 \xi$ in absolute value, which gives the first conclusion.
    Inspecting the case when $z = \vec 0$ also gives the last conclusion.
    For the second statement, observe that
    \begin{align*}
        \widetilde p^\top \phi \widetilde p
        = \widetilde p^\top (\phi \widetilde p)
        = \sum_{i=0}^{n-1} \widetilde p_i (\tau \pm 4n\xi) \,.
    \end{align*}
    Since $\sum_{i=0}^{n-1} \widetilde p_i = \|\widetilde p\|_1 = \|p\|_1 = 1 - \|q\|_1$, this value
    is bounded from above by $(\tau + 4n\xi) \|p\|_1$, and from below by $(\tau - 4n\xi) \|p\|_1$.
\end{proof}

The positive semidefiniteness of $\phi$ (shown below) already gives that $z^\top \phi z$ is
non-negative, and we will also show that the approximation error term $8mn^2\xi$ is negligible.
Therefore, our task is to show that $\widetilde p^\top \phi \widetilde p = \|p\|_1 (\tau \pm 4n\xi)$
is large, \ie to show that $\tau$ is large when $q$ is far from uniform.

\begin{claim}
    \label{claim:join-semidefinite}
    For any partial distribution $q \in \bZ_n$, the matrix $\phi = \phi(q)$ is positive semidefinite.
    If $q_i > 0$ for every $i \in \bZ_n$, then $\phi$ is positive definite.
\end{claim}
\begin{proof}
    For any vector $u \in \bR^{\bZ_n}$, we have
    \[
        u^\top \phi u = \Ex{ u^\top \bPhi u } = \Ex{ \sum_{i,j \in \bZ_n} \bPhi_{i,j} u_iu_j }
        = \Ex{ \sum_{i = 1}^n \sum_{j,j' \in \bm{\Gamma}_i} u_j u_{j'} }
        = \Ex{ \sum_{i = 1}^n \left(\sum_{j \in \bm{\Gamma}_i} u_j \right)^2 } \,.
    \]
    This is non-negative, so we conclude that $\phi$ is positive semi-definite. If $q_i > 0$
    for every $i \in \bZ_n$, then $w(e) > 0$ for every edge $e$, so with positive probability
    $\kappa>0$ the subgraph $\bm{H}$ will be an independent set, in which case
    $\bPhi = \Phi$ will have singleton buckets $\Gamma_i = (i,1)$ for each $i \in \bZ_n$, and
    \[
        u^\top \phi u \geq \kappa \cdot u^\top \Phi u
        = \kappa \cdot \sum_{i=1}^{n} \left( \sum_{j,j' \in \Gamma_i} u_ju_{j'} \right)^2
        = \kappa \cdot \sum_{i=0}^{n-1} u_i^2
        > 0 \,,
    \]
    whenever $u \neq \vec 0$. So $\phi$ is positive definite.
\end{proof}

Before continuing to make use of the nice properties of approximate uniform conjugates, we show
that such an object does exist:

\begin{lemma}
\label{lemma:approx-conjugate-existence}
For any partial distribution $q$ such that $\|q\|_1 > 0$
and sample-size parameter $m$, there exists an approximate uniform
conjugate $\widetilde p$ of $q$. If $q = \mu$, we may take $\widetilde p = \mu$.
\end{lemma}
\begin{proof}
    Recall that we identify the vertices of the cycle with the integers modulo $n$, \ie $\bZ_n$.

    Let $E$ be the event that $\bm{H} = G$, \ie that every edge was sampled into $\bm{H}$.
    Define $\varepsilon \define \Pr{E}$ and note that $\varepsilon = e^{-m\|q\|_1}$, since each
    edge $e = (i,i+1)$ is sampled with probability $1 - w(e) = e^{-mq_i}$.

    Let $u \in \bR^{\bZ_n}$, which we view as a candidate for $\widetilde p$.
    For each $i \in \bZ_n$, define three random variables:
    \begin{itemize}
        \item $\bm{N^R}_i$ is the number of vertices joined with $i$ in $\bm{H}$ in the clockwise
            direction (including $i$ itself). Formally,
            $\bm{N^R} \define \sum_{d=1}^n \bm{J}(\llangle i, d \rrangle)$.

            Then define $\bm{R}_i \define u[ \llangle i, \bm{N^R}_i \rrangle ] =
            \sum_{t=0}^{\bm{N^R}_i-1} u_{i+t}$.

        \item $\bm{N^L}_i$ is the number of vertices joined with $i$ in $\bm{H}$ in the
            counterclockwise direction (including $i$ itself), unless all edges were sampled into
            $\bm{H}$ (\ie event $E$ occurs), in which case we define $\bm{N^L}_i = 1$.
            Formally, $\bm{N^L}_i \define \ind{E}
            + (1-\ind{E}) \sum_{d=1}^n \bm{J}(\llangle i, -d \rrangle)$.

            Then define $\bm{L}_i \define u[ \llangle i, - \bm{N^L}_i \rrangle ] =
            \sum_{t=0}^{\bm{N^L}_i-1} u_{i-t}$.

        \item Define $\bm{D}_i \define u[ \bm{\Gamma}_{\bm{\gamma}(i)} ]
            = u[ \llangle i - \bm{N^L} + 1, \bm{N^L} + \bm{N^R} - 1 \rrangle]$.

    \end{itemize}

    Note that $\bm{D}_i = \bm{L}_i + \bm{R}_i - u_i$ for every $i \in \bZ_n$; when $E$ does not occur,
    this is true because $\bm{N^L}_i$ and $\bm{N^R}_i$ count the number of joined elements to the left
    and right of $i$ and, since both encounter a non-joined element somewhere, only $i$ itself is
    counted twice. On the other hand, when $E$ does occur, then this is true by construction, since
    we get $\bm{D}_i = \bm{R}_i = \vec 1^\top u$, whereas $\bm{L}_i = u_i$.

A perfect uniform conjugate $u$ would make $\Ex{\bm{D}_i}$ equal everywhere. Observe that the random
variables $\bm{R}_i$ and $\bm{L}_i$ are not mutually independent, and $\bm{N^R}_i, \bm{N^L}_i$ are
bounded, which complicates the analysis. We relax the problem by defining random variables that are
independent and asking for an \emph{approximate} uniform conjugate.

We introduce independent random variables $\bm{N^{R'}}_i$ and $\bm{N^{L'}}_i$ which are generated by
a Markov process, and define new variables $\bm{L}'_i$, $\bm{R}'_i$, and $\bm{D}'_i$ that depend on
$\bm{N^{R'}}_i$ and $\bm{N^{L'}}_i$ in the same way as before:
\begin{itemize}
\item $\bm{N^{R'}}_i$ is generated as follows. Initialize $\bm{N^{R'}}_i$ to 1. For each $t \geq 0$
in increasing order, sample $\bm{P}_t \sim \Poi(mq_{i+t})$. If $\bm{P}_t > 0$, stop; otherwise
increment $\bm{N^{R'}}_i$.

Define $\bm{R}'_i \define u[ \llangle i, \bm{N}^{\bm{R}'}_i \rrangle ] =
\sum_{t=0}^{\bm{N}^{\bm{R}'}_i - 1} u_{i+t}$.

\item $\bm{N^{L'}}_i$ is generated as follows. Initialize $\bm{N^{L'}}_i$ to 1. For each $t \geq 1$
in increasing order, sample $\bm{P}_t \sim \Poi(mq_{i-t})$. If $\bm{P}_t > 0$, stop; otherwise
increment $\bm{N^{L'}}_i$.

Define $\bm{L}'_i \define u[ \llangle i, - \bm{N}^{\bm{L}'}_i \rrangle ] =
\sum_{t=0}^{\bm{N}^{\bm{L}'}_i - 1} u_{i-t}$.

\item Define $\bm{D}'_i \define \bm{L}'_i + \bm{R}'_i - u_i$, as in the original process above.
\end{itemize}

Recall that $\llangle i, k \rrangle$ is a multiset, so that $u[ \llangle i, k \rrangle]$ can count
an element $u_t$ more than once.  Let $f \define 1 - \|q\|_1$.
Recalling \cref{def:approximate-uniform-conjugate}, we would like $u$ to satisfy three
requirements: 1) $u_i \ge 0$ for all $i \in \bZ_n$; 2) $\sum_i u_i = f$; and 3) $\Ex{\bm{D}'_i} =
\tau = \frac{1-\|q\|_1}{\sum_{i=0}^{n-1} \tanh\left(\frac{mq_i}{2}\right)}$
for all $i \in \bZ_n$. If we obtain such $u$ and show that $\Ex{\bm{D}_i} =
\Ex{\bm{D}'_i} \pm 4n\xi$, we will have found our approximate uniform conjugate $\widetilde p$.

    We give an explicit solution and then verify it. Set
    \[
        u_i \define \tau \left( \frac{1}{1 + e^{-mq_i}} + \frac{1}{1 + e^{-mq_{i-1}}} - 1 \right)
    \]
    for every $i \in \bZ_n$. It is clear that $u_i \ge 0$, satisfying the first requirement.
    The second requirement is also satisfied:
    \begin{align*}
        \sum_{i=0}^{n-1} u_i
        &= \tau \sum_{i=0}^{n-1} \left[ \frac{1}{1 + e^{-mq_i}} + \frac{1}{1 + e^{-mq_{i-1}}} - 1 \right]
        = \tau \sum_{i=0}^{n-1} \left[ \frac{2}{1 + e^{-mq_i}} - 1 \right]
        = \tau \sum_{i=0}^{n-1} \left[ \frac{1 - e^{-mq_i}}{1 + e^{-mq_i}} \right] \\
        &= \frac{f}{\sum_{i=0}^{n-1} \tanh(mq_i/2)} \cdot \sum_{i=0}^{n-1} \tanh(mq_i/2)
        = f \,.
    \end{align*}

    We now verify the third requirement. For convenience of notation, define
    $r'_i \define \Ex{\bm{R}'_i}$, $l'_i \define \Ex{\bm{L}'_i}$, and
    $d'_i \define \Ex{\bm{D}'_i}$. Let $\overline R$ be the $n \times n$ matrix given by
    \[
        \overline R_{i, i+d} = \Pr{ \bm{N^{R'}}_i > d} = e^{-m q[\llangle i,d \rrangle]}
    \]
    for all $i \in \bZ_n$ and $0 \le d \le n-1$.
    Then $R_{i, i+d}$ is the probability that the Markov process
    generating $\bm{R}'_i$ counts $u_{i+d}$ at least once.
    Note that $\varepsilon = \Pr{E} = e^{-m \|q\|_1}$ is the probability that the process
    loops back to the same element $i$ once.
    Then using the Markov property, the expectation of $\bm{R}'_i$ is
    \begin{align*}
        r'_i 
        &= \sum_{t=0}^{\infty} \sum_{d=0}^{n-1} \varepsilon^t \overline R_{i,i+d} u_{i+d}
        = \sum_{d=0}^{n-1} \overline R_{i,i+d} u_{i+d}
            + \sum_{t=1}^\infty \sum_{d=0}^{n-1} \varepsilon^t \overline R_{i,i+d} u_{i+d} \\
        &= \sum_{d=0}^{n-1} \overline R_{i,i+d} u_{i+d}
            + \varepsilon \sum_{t=0}^\infty \sum_{d=0}^{n-1} \varepsilon^t \overline R_{i,i+d} u_{i+d}
        = \sum_{d=0}^{n-1} \overline R_{i,i+d} u_{i+d} + \varepsilon r'_i \,,
    \end{align*}
    and thus
    \begin{align*}
        (1 - \varepsilon) r'_i
        &= \sum_{d=0}^{n-1} \overline R_{i,i+d} u_{i+d}
        = \sum_{d=0}^{n-1} \overline R_{i,i+d} \cdot \tau \left(
            \frac{1}{1 + e^{-mq_{i+d}}} + \frac{1}{1 + e^{-mq_{i+d-1}}} - 1 \right) \\
        &= \tau \left( \sum_{d=0}^{n-1} \frac{\overline R_{i,i+d}}{1 + e^{-mq_{i+d}}}
            - \overline R_{i,i+d} \right)
            + \tau \left( \sum_{d=0}^{n-1} \frac{\overline R_{i,i+d}}{1 + e^{-mq_{i+d-1}}} \right)
            \\
        &= -\tau \left( \sum_{d=0}^{n-1} \overline R_{i,i+d} \frac{e^{-mq_{i+d}}}{1 + e^{-mq_{i+d}}} \right)
            + \tau \left( \sum_{d=0}^{n-1} \overline R_{i,i+d}\frac{1}{1 + e^{-mq_{i+d-1}}} \right) \,.
    \end{align*}
    We now observe that, for $1 \le d \le n-1$,
    $\overline R_{i,i+d} = e^{-mq_{i+d-1}} \overline R_{i,i+d-1}$. Also note that
    $\overline R_{i,i} = 1$. Along with a change of variables in the second sum above, we obtain
    \begin{align*}
        \frac{1 - \varepsilon}{\tau} r'_i
        &= -\left( \sum_{d=0}^{n-2} \overline R_{i,i+d} \frac{e^{-mq_{i+d}}}{1 + e^{-mq_{i+d}}} \right)
            - \overline R_{i,i+n-1} \frac{e^{-mq_{i+n-1}}}{1 + e^{-mq_{i+n-1}}}
            \\
            &\qquad \quad
            + \overline R_{i,i} \frac{1}{1 + e^{-mq_{i-1}}}
            + \left( \sum_{d=1}^{n-1} e^{-mq_{i+d-1}} R_{i,i+d-1}\frac{1}{1 + e^{-mq_{i+d-1}}} \right)
            \\
        &= -\left( \sum_{d=0}^{n-2} \overline R_{i,i+d} \frac{e^{-mq_{i+d}}}{1 + e^{-mq_{i+d}}} \right)
            - \overline R_{i,i+n-1} \frac{e^{-mq_{i+n-1}}}{1 + e^{-mq_{i+n-1}}}
            \\
            &\qquad \quad
            + \frac{1}{1 + e^{-mq_{i-1}}}
            + \left( \sum_{d=0}^{n-2} \overline R_{i,i+d} \frac{e^{-mq_{i+d}}}{1 + e^{-mq_{i+d}}} \right)
            \\
        &= - \overline R_{i,i+n-1} \frac{e^{-mq_{i+n-1}}}{1 + e^{-mq_{i+n-1}}}
            + \frac{1}{1 + e^{-mq_{i-1}}} \,.
    \end{align*}
    Also note that $\varepsilon = e^{-mq_{i+n-1}} \overline R_{i,i+n-1}$, and
    $q_{i-1} = q_{i+n-1}$, and therefore
    \[
        \frac{1 - \varepsilon}{\tau} r'_i
        = -\frac{\varepsilon}{1 + e^{-mq_{i+n-1}}} + \frac{1}{1 + e^{-mq_{i-1}}}
        = (1 - \varepsilon) \frac{1}{1 + e^{-mq_{i-1}}} \,.
    \]
    Since $\varepsilon < 1$ (because $\|q\| > 0$),
    We conclude that
    \[
        r'_i = \frac{\tau}{1 + e^{-mq_{i-1}}} \,.
    \]
    An identical analysis for the symmetrical process determining $\bm{L}'_i$ yields
    \[
        l'_i = \frac{\tau}{1 + e^{-mq_i}} \,.
    \]
    We now verify the third requirement: for every $i \in \bZ_n$,
    \[
        \Ex{\bm{D}'_i}
        = d'_i = l'_i + r'_i - u_i
        = \frac{\tau}{1 + e^{-mq_i}} + \frac{\tau}{1 + e^{-mq_{i-1}}}
            - \tau \left( \frac{1}{1 + e^{-mq_i}} + \frac{1}{1 + e^{-mq_{i-1}}} - 1 \right)
        = \tau \,,
    \]
    as needed.

    It remains to show that $\Ex{\bm{D}_i}$ does not differ from $\Ex{\bm{D}'_i} = \tau$ by
    more than $4n\xi$. 
    Write $l_i \define \Ex{\bm{L}_i}$ and $r_i \define \Ex{\bm{R}_i}$. We will show that
    $\abs{r_i - r'_i}$ and $\abs{l_i - l'_i}$ are small.

    Fix some $i \in \bZ_n$. Recall that $\bm{N}^{\bm{R}'}_i$ counts how many states the Markov process
    for $\bm{R}'_i$ visited before stopping, meaning that the process sampled $\bm{P}_t = 0$ and
    advanced to the next state (vertex) exactly $\bm{N}^{\bm{R}'}_i - 1$ consecutive times before
    stopping. Let $\bm{K} \define \lfloor \frac{\bm{N}^{\bm{R}'}_i - 1}{n} \rfloor$, so that
    $\bm{K}$ is the number of times the process ``looped back'' and reached vertex $i$ again.

    Then, recalling the definition of event $E$, note that
    $\Pr{\bm{K} \ge 1} = \Pr{E} = \varepsilon = e^{-m \|q\|_1}$.
    More generally, we have $\Pr{\bm{K} \ge k} \leq e^{-k m \|q\|_1}$ for every non-negative
    integer $k$ by the Markov property.

    Now, we may bound $\abs{r_i - r'_i}$ as follows. First, note that $r_i \le r'_i$, since
    $u$ is a non-negative vector and, although the join (or transition) probabilities are the same
    for both processes, the Markov process may continue even after visiting $n$ elements.
    In fact, we have $\Pr{\bm{N^R}_i = t} = \Pr{\bm{N}^{\bm{R}'}_i = t}$ for every
    $1 \le t \le n-1$, and thus $\Pr{\bm{N^R}_i = n} \ge \Pr{\bm{N}^{\bm{R}'}_i = n}$.
    Now, it suffices to upper bound $r'_i - r_i$, which we do as follows:
    \begin{align*}
        r'_i - r_i
        &= \sum_{t=1}^{\infty} \Pr{\bm{N}^{\bm{R}'}_i = t} \sum_{d=0}^{t-1} u_{i+d}
            - \sum_{t=1}^n \Pr{\bm{N^R}_i = t} \sum_{d=0}^{t-1} u_{i+d} \\
        &= \sum_{t=1}^{n} \left( \Pr{\bm{N}^{\bm{R}'}_i=t} - \Pr{\bm{N^R}_i = t} \right)
                                \sum_{d=0}^{t-1} u_{i+d}
            + \sum_{t=n+1}^{\infty} \Pr{\bm{N}^{\bm{R}'}_i = t} \sum_{d=0}^{t-1} u_{i+d} \\
        &\le \sum_{k=1}^{\infty} \sum_{t=1}^{n} \Pr{\bm{N}^{\bm{R}'}_i = kn+t} \sum_{d=0}^{kn+t-1} u_{i+d} \\
        &\le \sum_{k=1}^\infty n \Pr{\bm{N}^{\bm{R}'}_i \ge kn+1} (k+1) f
        = n f \sum_{k=1}^\infty \Pr{\bm{K} \ge k} (k+1) \\
        &\le n f \sum_{k=1}^\infty (k+1) e^{-km\|q\|_1}
        \le n f \cdot 2 \sum_{k=1}^\infty k e^{-km\|q\|_1}
        \le 2 n f \cdot \frac{e^{-m\|q\|_1}}{(1 - e^{-m\|q\|_1})^2} \,,
    \end{align*}
which is bounded by $2n \xi(m,q)$ where $\xi$ is defined as in
\cref{def:approximate-uniform-conjugate}. A similar analysis shows that $\abs{l'_i - l_i} \le 2 n
\xi(m,q)$. Therefore $\abs{d'_i - d_i} \le 4 n \xi(m,q)$, and hence $\Ex{\bm{D}_i} = \Ex{\bm{D}'_i}
\pm 4 n \xi(m,q) = \tau \pm 4 n \xi(m,q)$.  Hence $\widetilde p = u$ is an approximate uniform
conjugate with the desired parameters.  Moreover, one can check that the solution $\widetilde p$ we
obtained yields $\mu$ when $q=\mu$.
\end{proof}

\noindent
To analyze $\tau = \frac{1-\|q\|_1}{\sum_{i=0}^{n-1} \tanh\left(\frac{mq_i}{2}\right)}$,
we will require the following bounds on $\tanh(x)$.

\begin{fact}
\label{fact:tanh-bounds}
For sufficiently small $x > 0$,
\[
  \frac{x}{2} \leq \tanh(x) \leq 2x \,.
\]
\end{fact}
\begin{proof}
This follows from the Taylor expansion $\tanh(x) = x - \frac{x^3}{3} + \frac{2x^5}{15} + O(x^7)$.
\end{proof}

\paragraph{Notation} For vector $u \in \bR^{\bZ_n}$, we denote by $u^+$ the vector given by
$u^+_i = \max(0, u_i)$ for every $i \in \bZ_n$, and by $u^-$ the vector given by
$u^-_i = -\min(0, u_i)$ for every $i \in \bZ_n$.

\begin{proposition}[Quadratic upper bound to $\tanh$ from near zero to the right]
    \label{prop: tanh quadratic upper bound}

    For all sufficiently small real numbers $r > 0$ and all $0 \le x \le \frac{1}{2 \tanh(r)}$,
    we have
    \[
        \tanh(r+x) \le \tanh(r) + (1 - \tanh^2(r))x - \tanh(r)(1 - \tanh^2(r))x^2 \,.
    \]
\end{proposition}
\begin{proof}
    Define functions $f, g : \bR_{\ge 0} \to \bR$ as the quantities on the two sides
    of the desired inequality:
    \begin{align*}
        f(x) &\define \tanh(r+x) \,, \\
        g(x) &\define \tanh(r) + (1 - \tanh^2(r))x - \tanh(r)(1 - \tanh^2(r))x^2 \,.
    \end{align*}
    Thus we wish to show that, for sufficiently small $r$, $f(x) \le g(x)$ for all
    $0 \le x \le \frac{1}{2\tanh(r)}$.

    Recall that $\tanh$ is bounded between $0$ and $1$ in its non-negative domain, with
    $\tanh(r) = r \pm O(r^3)$ as $r \to 0$ (this follows from its Taylor series) and
    $\tanh(y) \to 1$ as $y \to \infty$.
    Since $g$ is a downward-facing parabola, we start by determining the point $x^*$ at which
    $g$ attains its maximum. We can determine this point by setting the derivative $g'$ to zero:
    \[
        g'(x^*) = 0 \implies (1 - \tanh^2(r)) - 2\tanh(r)(1 - \tanh^2(r))x^* = 0
        \implies x^* = \frac{1}{2\tanh(r)} \,.
    \]

    Now, our strategy will be to define a ``breakpoint'' $c \log(1/r)$ (for a sufficiently
    small constant $c$ to be specified) and show that $g(c\log(1/r)) \to \infty$, thus dividing
    the argument in two parts: $c\log(1/r) \le x \le x^*$, where $g$ is increasing and hence the
    result will follow immediately, and $0 \le x \le c\log(1/r)$, which will require some more work.

    We first show that for every $c > 0$, $g(c\log(1/r)) \to \infty$ as $r \to 0$:
    \begin{align*}
        g(c\log(1/r)) &\ge (1 - \tanh^2(r)) c\log(1/r) (1 - \tanh(r) c\log(1/r)) \\
        &\ge \frac{1}{2} c\log(1/r) (1 - 2c r\log(1/r)) \\
        &\ge \frac{c\log(1/r)}{4} \\
        &= \omega(1) \,,
    \end{align*}
    where we have used the fact that $r \log(1/r) \to 0$ in the last inequality.

    Note that $c\log\left(\frac{1}{r}\right) \le \frac{1}{4r} \le \frac{1}{2\tanh(r)}$ for all
    sufficiently small $r$. This means that $g(x)$ is increasing on
    $\left[ c\log(1/r), x^* \right]$ and hence $g(x) > 1$ in this range.
    Since $f(x) < 1$ for all $x$, we have shown that $f(x) \le g(x)$ when $c\log(1/r) \le x \le x^*$.

    We now proceed to the range $0 \le x \le c\log(1/r)$. By the mean-value form of Taylor's
    theorem, we have that
    \[
        f(x) = f(0) + f'(0) x + \frac{1}{2} f''(\eta) x^2
    \]
    for some $0 \le \eta \le x$. Substituting the definition of $f$, we obtain:
    \[
        f(x) = \tanh(r) + (1 - \tanh^2(r)) x - \tanh(r+\eta)(1 - \tanh^2(r+\eta)) x^2 \,.
    \]
    Hence, to show that $f(x) \le g(x)$, it suffices to show that
    \[
        \tanh(y)(1 - \tanh^2(y)) \ge \tanh(r)(1 - \tanh^2(r))
    \]
    for all $r \le y \le r + x \le r + c\log(1/r)$. We will show this for the larger interval
    $r \le y \le 2c\log(1/r)$.

    Define $h : \bR_{\ge 0} \to \bR$ by $h(x) \define \tanh(x)(1 - \tanh^2(x))$, so that we wish
    to show $h(y) \ge h(r)$ for $r \le y \le 2c\log(1/r)$. The derivative $h'$ satisfies the
    following properties:
    \begin{enumerate}
        \item $h'(x) = (1 - 3\tanh^2(x))(1 - \tanh^2(x))$;
        \item $h'(0) = 1$;
        \item $h'$ has its only positive real root at $\nu \define \frac{1}{2} \cosh^{-1}(2)
            = \frac{1}{2}\log(2 + \sqrt{3})$; and
        \item $h'$ is positive on $[0, \nu)$ and negative on $(\nu, \infty)$.
    \end{enumerate}

    It follows that $h$ is increasing on $[r, \nu]$ and decreasing on $[\nu, 2c\log(1/r)]$.
    Since $h(r) \ge h(r)$ trivially, we obtain that $h(y) \ge h(r)$ for $r \le y \le \nu$.
    Therefore it suffices to show that $h(2c\log(1/r)) \ge h(r)$ as long as $r$ is sufficiently
    small. Indeed, we have
    \begin{align*}
        h(2c\log(1/r)) &\ge h(r) \\
        \iff \tanh(2c\log(1/r))(1 - \tanh^2(2c\log(1/r))) &\ge \tanh(r)(1 - \tanh^2(r)) \\
        \iff \frac{\tanh(2c\log(1/r))}{\tanh(r)} &\ge \frac{1 - \tanh^2(r)}{1 - \tanh^2(2c\log(1/r))} \\
        \impliedby \frac{1/2}{2r} &\ge \frac{1}{1 -
            \left(\frac{1 - e^{-4c\log(1/r)}}{1 + e^{-4c\log(1/r)}}\right)^2} \\
        \impliedby \frac{1}{4r} &\ge \frac{1}{1 - (1 - e^{-4c\log(1/r)})} \\
        \iff \frac{1}{4r} &\ge \frac{1}{r^{4c}} \\
        \iff r^{1 - 4c} &\le \frac{1}{4} \,,
    \end{align*}
    which holds for all sufficiently small $r$ as long as $c < 1/4$, since then
    $r^{1 - 4c} \to 0$. This concludes the proof.
\end{proof}

\begin{lemma}[Quantitative Jensen's inequality for $\tanh$ near zero]
\label{lemma:qjensen}
For all sufficiently small $r > 0$, the following holds.  Suppose $u \in \bR^{\bZ_n}$ is a vector
satisfying $u_i \in \left[0, r + \frac{1}{2\tanh(r)}\right]$ for every $i \in \bZ_n$, and whose
average is $\frac{1}{n} \sum_i u_i = r$. Then we have
\[
\frac{1}{n} \sum_{i=0}^{n-1} \tanh(u_i) \le
\tanh(r)\left( 1 - \frac{1}{n} \left(1 - O(r^2)\right) \| (u - r \cdot \vec 1)^+ \|_2^2 \right) \,.
\]
\end{lemma}
\begin{proof}
    Write $u_i = r + x_i$, so that $x_i \le \frac{1}{2\tanh(r)}$ for every $i \in \bZ_n$
    and $\sum_i x_i = \sum_i (u_i - r) = 0$.
    Since $\tanh$ is a concave function on its non-negative domain, the first-degree Taylor series
    around $r$, namely $\tanh(r + y) \approx \tanh(r) + (1 - \tanh^2(r)) y$,
    upper bounds $\tanh(u_i)$ for every $i \in \bZ_n$.

    Therefore, our strategy will be to upper bound the entries with $x_i < 0$ via the first-degree
    series, and the entries with $x_i \ge 0$ via \cref{prop: tanh quadratic upper bound}. Then,
    the first degree terms will cancel out and the second-degree terms will yield the desired bound.
    Concretely, we have:
    \begin{align*}
        \sum_{i=0}^{n-1} \tanh(u_i) &= \sum_{i=0}^{n-1} \tanh(r + x_i) \\
        &\le \sum_{i=0}^{n-1} \Big[ \tanh(r) + (1 - \tanh^2(r)) x_i \Big]
            - \sum_{i \in \bZ_n : x_i \ge 0} \tanh(r)(1 - \tanh^2(r)) x_i^2 \\
        &= n\tanh(r) + (1 - \tanh^2(r))\sum_{i=1}^n x_i
            - \tanh(r)(1 - \tanh^2(r)) \sum_{i : x_i \ge 0} x_i^2 \\
        &\le \tanh(r) \Big( n - \left(1 - O\left(r^2\right)\right) \|(u-r)^+\|_2^2 \Big) \,.
        \qedhere
    \end{align*}
\end{proof}

We may now combine the results above to show a separation in $\Ex{\bm{Y}}$ as long as $q$
is not highly concentrated relative to $p$:

\begin{lemma}[Separation in the expected value of the test statistic]
    \label{lemma:sep-expected-value}
    Let $C, \gamma > 0$ be constants, let $n \in \bN$ be sufficiently large,
    let $\epsilon \geq \frac{1}{n^{1/4}}$, and let $m = m(n,\epsilon)$ satisfy
    \[
        \left( 2^8 \gamma^2 + (2^{18} \gamma)^{2/5} \right)
            \left(\frac{n}{\epsilon}\right)^{4/5}
        \le m
        \le \frac{4n}{3C\log n} \,.
    \]
    Let $\pi = \pi(p,q)$, where $p,q$ are partial distributions
    satisfying $\|p\|_1, \|q\|_1 = \frac{1}{2} \pm \frac{2\gamma}{\sqrt m}$, $\|q-\mu\|_1 >
    \epsilon$, and such that $q$ is not $C$-highly concentrated with respect to $p$.
    Write $p = \widetilde p + z$ where $\widetilde p$ is an approximate uniform conjugate of $q$.
    Then the expected value of the test statistic $\bm{Y}$ satisfies
    \[
        \Ex{\bm{Y}} \ge \Ex{\bm{Y}^{(\mu)}}
            + \Omega\left( \frac{\epsilon^2 m^2}{n^2} \right) + m z^\top \phi z \,.
    \]
\end{lemma}
\begin{proof}
    From \cref{lemma:approx-conjugate-existence}, $\mu$ is its own approximate uniform conjugate.
    Since we will reason about both $\mu$ as its own approximate uniform conjugate and about
    $\widetilde p$ as the approximate uniform conjugate of $q$, let
    $\xi \define \max(\xi(m,\mu), \xi(m,q))$. First, using \cref{prop:conjugate-mean},
    \begin{equation}
        \label{eq:sep-unif}
        \Ex{\bm{Y}^{(\mu)}} = m \mu^\top \phi^{(\mu)} \mu
        = m \|\mu\|_1 (\tau(m,\mu) \pm 4n\xi)
        = \frac{m \|\mu\|_1^2}{\sum_{i=0}^{n-1} \tanh\left(\frac{m}{4n}\right)} \pm 4mn \|\mu\|_1 \xi
        = \frac{m}{4 n \tanh\left(\frac{m}{4n}\right)} \pm 2mn \xi \,.
    \end{equation}
    We now consider $\Ex{\bm{Y}}$.
    By the assumption that $\|p\|_1 \geq \tfrac{1}{2} - \tfrac{2\gamma}{\sqrt m}$, we obtain
    \[
        \|p\|_1^2 \geq \left(\frac{1}{2} - \frac{2\gamma}{\sqrt m}\right)^2
        \geq \frac{1}{4} - \frac{2\gamma}{\sqrt m} \,.
    \]
    Let $\tau \define \tau(m,q)$. Using \cref{prop:conjugate-mean},
    \begin{equation}
        \label{eq:sep-non-unif}
        \begin{aligned}
            \Ex{\bm{Y}}
            &= m \|p\|_1(\tau \pm 4n\xi) + mz^\top\phi z \pm 8mn^2\xi \\
            &\geq \frac{m \|p\|_1^2}{\sum_{i=0}^{n-1} \tanh\left(\frac{mq_i}{2}\right)}
                + mz^\top \phi z - 4mn\xi\|p\|_1 - 8mn^2\xi \\
            &\geq \frac{m - 8\gamma\sqrt m}{4\sum_{i=0}^{n-1} \tanh\left(\frac{mq_i}{2}\right)}
                + mz^\top \phi z - 12mn^2\xi \,.
        \end{aligned}
    \end{equation}
    We now write $q_i = \tfrac{1}{n}\|q\|_1 + x_i$ for each $i \in \bZ_n$, so
    \[
        \sum_{i=0}^{n-1} \tanh\left(\frac{mq_i}{2}\right)
        = \sum_{i=0}^{n-1} \tanh\left(\frac{m}{2n}\|q\|_1 + \frac{mx_i}{2}\right) \,,
    \]
    and $\sum_{i=0}^{n-1} x_i = 0$.  Writing $r \define \frac{m}{2n}\|q\|_1$ and
    $u_i \define \tfrac{mq_i}{2}$, we have $u = \frac{m}{2} q$ and
    \[
        0 \leq u_i = r + \frac{mx_i}{2} \,.
    \]
    Since $q$ is not $C$-highly concentrated relative to $p$, then as observed in
    \cref{remark:non-concentrated-l-infty-linear-trace} it holds that
    $\|q\|_\infty < \frac{C \log n}{m}$, so we have
    \[
        x_i < \frac{C \log n}{m} - \frac{1}{n}\|q\|_1 < \frac{C \log n}{m} \,.
    \]
    Moreover, since $\|q\|_1 \le \frac{1}{2} + \frac{2\gamma}{\sqrt{m}}$, we have that $r$ satisfies
    \[
        r \le \frac{m}{4n} + \frac{\gamma \sqrt{m}}{n}
        \le \frac{3m}{8n}
        \implies \tanh(r) \le \frac{3m}{4n} \,,
    \]
    where in the second inequality we used the fact that $\sqrt{m}/n = o(m/n)$, and in the
    last inequality we used \cref{fact:tanh-bounds} and the assumption that $m/n \le \frac{4}{3C\log n}$
    and that $n$ is sufficiently large. Then we obtain
    \[
        x_i < \frac{C \log n}{m}
        = \frac{2}{m} \cdot \frac{C \log n}{2}
        \le \frac{2}{m} \cdot \frac{2n}{3m}
        = \frac{2}{m} \cdot \frac{1}{2 \cdot 3m / (4n)}
        \le \frac{2}{m} \cdot \frac{1}{2 \tanh(r)} \,,
    \]
    where in the second inequality we used the assumption $m \le \frac{4n}{3C\log n}$.
    Thus $u, r$ satisfy the conditions
    \[
        0 \le u_i = r + \frac{m x_i}{2} \le r + \frac{1}{2\tanh(r)} \,,
        \qquad
        \frac{1}{n} \sum_{i=0}^{n-1} u_i = \frac{m}{2n} \|q\|_1 = r \,.
    \]
    Let $q' \define \tfrac{1}{n}\|q\|_1 \cdot \vec 1$, which is the partial distribution that is
    uniform with total mass equal to the total mass of $q$. Then observing that
    $r \cdot \vec 1 = m\|q\|_1 \cdot \mu = (m/2) \cdot q'$, we apply \cref{lemma:qjensen}, yielding
    \begin{align*}
        \frac{1}{n} \sum_{i=0}^{n-1} \tanh\left(\frac{mq_i}{2}\right)
        = \frac{1}{n} \sum_{i=0}^{n-1} \tanh(u_i)
        &\leq \tanh\left(\frac{m}{2n}\|q\|_1\right)
            \left(1 - \frac{1}{n} (1 - O(r^2)) \| ((m/2)q - (m/2)q')^+ \|_2^2 \right) \\
        &\le \tanh\left(\frac{m}{2n}\|q\|_1\right)
            \left(1 - \frac{m^2}{8n} \| (q - q')^+ \|_2^2 \right) \,,
    \end{align*}
    where the last inequality used the fact that $r^2 = \left(\frac{m}{2n}\|q\|_1\right)^2 = o(1)$.
    We will also use the following upper bound on $\|(q-q')^+\|_2^2$:
    \[
        \|(q-q')^+\|_2^2
        \le \|q\|_2^2
        < \left( \frac{C\log n}{m} \right)^2 \cdot \frac{1}{\left(\frac{C\log n}{m}\right)}
        = \frac{C\log n}{m} \,,
    \]
    where we used the fact that $\|q\|_\infty < \frac{C\log n}{m}$ by the
    anticoncentration assumption, so that the maximum $\ell^2$-norm is achieved by concentrating
    the partial distribution as much as possible given this constraint. We conclude that
    \begin{equation}
        \label{eq:upper-bound-q-q'-term}
        \frac{m^2}{8n} \|(q-q')^+\|_2^2
        < \frac{C m \log n}{8n}
        \le \frac{1}{6} \,,
    \end{equation}
    the latter since $m \le \frac{4n}{3C\log n}$. Thus we use the inequality
    $\frac{1}{1-x} \geq 1+x$, valid for $x < 1$, as follows:
    \begin{equation}
        \label{eq:sep-tanh}
        \begin{aligned}
            \frac{m - 8\gamma\sqrt m}{4 \sum_{i=0}^{n-1} \tanh\left(\frac{mq_i}{2}\right)}
            &\geq \frac{m - 8\gamma\sqrt m}{4 n \tanh\left(\frac{m}{2n}\|q\|_1\right)
            \left( 1 - \frac{m^2}{8n} \|(q - q')^+\|_2^2\right)} \\
            &\geq \frac{m - 8\gamma\sqrt m}{4n \tanh\left(\frac{m}{2n}\|q\|_1\right)}
                \left( 1 + \frac{m^2}{8n} \|(q - q')^+\|_2^2\right) \,.
        \end{aligned}
    \end{equation}
    From \eqref{eq:sep-unif}, \eqref{eq:sep-non-unif}, and \eqref{eq:sep-tanh}, we now have
    \begin{equation}
        \label{eq:separation-formula}
        \Ex{\bm{Y}} - \Ex{\bm{Y}^{(\mu)}}
        \ge \frac{m}{4n} \cdot F + G - H + mz^\top \phi z - 14mn^2\xi \,,
    \end{equation}
    where
    \begin{align*}
        F &=
            \frac{1}{\tanh\left(\frac{m}{2n}\|q\|_1\right)}
            - \frac{1}{\tanh\left(\frac{m}{4n}\right)} \,, \\
        G &= \frac{m}{4n \tanh\left(\frac{m}{2n}\|q\|_1\right)} \cdot \frac{m^2}{8n}
            \|(q - q')^+\|_2^2 \,, \\
        H &=
            \frac{8\gamma\sqrt m}{4n \tanh\left(\frac{m}{2n}\|q\|_1\right)}
                \left(1 + \frac{m^2}{8n} \|(q - q')^+\|_2^2 \right) \,.
    \end{align*}
    We will show that $G$ is large enough to give the desired separation
    $\Omega(\epsilon^2 m^2 / n^2)$, while $F$ and $H$ are asymptotically small enough.
    We first lower bound $F$. Using the fact that
    $\frac{m}{2n}\|q\|_1 \le \frac{m}{2n} \left(\frac{1}{2} + \frac{2\gamma}{\sqrt{m}}\right)
    = \frac{m}{4n} + \frac{\gamma\sqrt{m}}{n}$ and
    the upper bound $\tanh(r+x) \le \tanh(r) + x(1 - \tanh^2(r))$, which holds from the Taylor
    expansion of $\tanh$ when the arguments are all non-negative, we obtain
    \[
        F
        \ge \frac{1}{\tanh\left(\frac{m}{4n} + \frac{\gamma\sqrt{m}}{n}\right)}
            - \frac{1}{\tanh\left(\frac{m}{4n}\right)}
        = \frac{\tanh\left(\frac{m}{4n}\right) - \tanh\left(\frac{m}{4n} + \frac{\gamma\sqrt{m}}{n}\right)}
                {\tanh\left(\frac{m}{4n}\right)\tanh\left(\frac{m}{4n} + \frac{\gamma\sqrt{m}}{n}\right)}
        \ge -\frac{\frac{\gamma\sqrt{m}}{n} \left( 1 - \tanh^2\left(\frac{m}{4n}\right) \right)}
                {\tanh\left(\frac{m}{4n}\right)\tanh\left(\frac{m}{4n} + \frac{\gamma\sqrt{m}}{n}\right)} \,.
    \]
    For sufficiently large $n$ and therefore sufficiently small $m/n$, we have
    $\tanh(m/4n + \gamma\sqrt{m}/n) > \tanh(m/4n) \ge m/8n$ from \cref{fact:tanh-bounds}. We obtain
    \[
        \frac{m}{4n} \cdot F \ge -\frac{m}{4n} \cdot \frac{\gamma\sqrt{m} / n} {(m/8n)^2}
        = -\frac{16 \cdot \gamma}{\sqrt{m}} \,.
    \]
    We verify that this negative factor does not overwhelm the desired separation
    $\frac{\epsilon^2 m^2}{n^2}$ as follows:
    \[
        \frac{16 \gamma}{\sqrt{m}} \le \frac{\epsilon^2 m^2}{2^{12} \cdot n^2}
        \iff m^{5/2} \ge \frac{2^{16} \cdot \gamma \cdot n^2}{\epsilon^2}
        \iff m \ge (2^{16} \cdot \gamma)^{2/5} (n/\epsilon)^{4/5} \,,
    \]
    which holds by assumption.

    As for $G$, we use the bound $\tanh\left(\frac{m}{2n}\|q\|_1\right) \le \tanh(m/2n) \le m/n$
    and the Cauchy-Schwarz inequality to obtain
    \[
        G \ge \frac{m^2}{32 n} \|(q - q')^+\|_2^2
        \ge \frac{m^2}{32 n^2} \|(q - q')^+\|_1^2
        = \frac{m^2}{2^7 \cdot n^2} \|q - q'\|_1^2 \,,
    \]
    where the equality is because, since $\|q\|_1 = \|q'\|_1$, we have
    $\|(q - q')^+\|_1 = \|(q - q')^-\|_1 = \frac{1}{2} \|q - q'\|_1$.

    Therefore, our goal is to show that $\|q - q'\|_1^2$ is not much smaller than $\epsilon^2$.
    Using the triangle inequality, we have
    \[
        \epsilon < \|q - \mu\|_1 \le \|q - q'\|_1 + \|q' - \mu\|_1
        = \|q - q'\|_1 + \sum_{i=0}^{n-1} \abs*{\frac{\|q\|_1}{n} - \frac{1}{2n}}
        = \|q - q'\|_1 + \abs*{\|q\|_1 - \frac{1}{2}}
        \le \|q - q'\|_1 + \frac{2\gamma}{\sqrt{m}} \,,
    \]
    so that, using $\epsilon \le 2$ which always holds,
    \[
        G
        \ge \frac{m^2}{2^7 \cdot n^2} \left( \epsilon - \frac{2\gamma}{\sqrt{m}} \right)^2
        > \frac{m^2}{2^7 \cdot n^2} \left( \epsilon^2 - \frac{8\gamma}{\sqrt{m}} \right)
        \ge \frac{\epsilon^2 m^2}{2^8 \cdot n^2} \,,
    \]
    where the last inequality is obtained as follows, using the assumption that
    $m \ge 2^8 \cdot \gamma^2 (n/\epsilon)^{4/5}$:
    \begin{align*}
        \frac{8\gamma}{\sqrt{m}} \le \frac{\epsilon^2}{2}
        &\iff m \ge \frac{2^8 \gamma^2}{\epsilon^4} \\
        &\impliedby 2^8 \gamma^2 (n/\epsilon)^{4/5} \ge \frac{2^8 \gamma^2}{\epsilon^4}
        \iff \epsilon^{16/5} \ge \frac{1}{n^{4/5}}
        \iff \epsilon \ge \frac{1}{n^{1/4}} \,,
    \end{align*}
    which is true by assumption.

    We also show that $H$ does not overwhelm this term. For sufficiently large $n$ and therefore
    $m$, we have the inequality
    $\tanh\left(\frac{m}{2n}\|q\|_1\right)
    \ge \tanh\left(\frac{m}{2n}\left(\frac{1}{2} - \frac{2\gamma}{\sqrt{m}}\right)\right)
    \ge \tanh\left(\frac{m}{8n}\right) \ge m/16n$.
    Along with, \eqref{eq:upper-bound-q-q'-term}, we conclude
    \[
        H < \frac{8\gamma\sqrt m}{4n \cdot (m/16n)} \left(1 + \frac{1}{6} \right)
        < \frac{2^6 \cdot \gamma}{\sqrt{m}}
        \le \frac{\epsilon^2 m^2}{2^{12} \cdot n^2} \,,
    \]
    where the last inequality holds since
    \[
        \frac{2^6 \cdot \gamma}{\sqrt{m}} \le \frac{\epsilon^2 m^2}{2^{12} \cdot n^2}
        \iff m^{5/2} \ge \frac{2^{18} \gamma n^2}{\epsilon^2}
        \iff m \ge (2^{18} \gamma)^{2/5} (n/\epsilon)^{4/5} \,.
    \]
    Finally, we inspect the error term $-14mn^2\xi$. Recall that
    $\xi = \max(\xi(m,\mu), \xi(m,q))$, where
    $\xi(m,\mu) = \frac{e^{-m \|\mu\|_1}}{(1-e^{-m \|\mu\|_1})^2}
    = \frac{e^{-m/2}}{(1 - e^{-m/2})^2}$ and
    $\xi(m,q) = \frac{e^{-m \|q\|_1}}{(1 - e^{-m \|q\|_1})^2}$. Using the bound
    $\|q\|_1 \ge \frac{1}{2} - \frac{2\gamma}{\sqrt{m}} \ge \frac{1}{4}$ as $n$ and $m$ grow,
    we conclude that $\xi \le 2 e^{-m/4}$ and therefore, using the (simplified) assumptions
    $\Omega(n^{4/5}) \le m \le n$, we conclude that
    \[
        mn^2\xi \le n^3 e^{-\Omega(n^{4/5})}
        = o(n^{-5/2})
        = o\left(\frac{\epsilon^2}{n^2}\right)
        = o\left(\frac{\epsilon^2 m^2}{n^2}\right) \,,
    \]
    where we used the assumption $\epsilon \ge n^{-1/4}$ in the penultimate step.
    Returning to \eqref{eq:separation-formula}, we obtain
    \begin{align*}
        \Ex{\bm{Y}} - \Ex{\bm{Y}^{(\mu)}}
        &\ge \frac{m}{4n} \cdot F + G - H + mz^\top \phi z - 14mn^2\xi
        \ge -3 \cdot \frac{\epsilon^2 m^2}{2^{12} n^2} + \frac{\epsilon^2 m^2}{2^8 \cdot n^2}
            + m z^\top \phi z \\
        &= \Omega\left(\frac{\epsilon^2 m^2}{n^2}\right) + m z^\top \phi z \,. \qedhere
    \end{align*}
\end{proof}

\subsection{Concentration of the Test Statistic}
\label{section:parity-variance}

In this section, we start from the general results established in \cref{section:common-variance}
and conclude specific bounds for the variance of $\bm{Y}$ in the current setting.

\begin{lemma}[First component of the variance]
\label{lemma:var-first-component}
Let $C > 0 $ be a constant, let $n \in \bN$ be sufficiently large, and let $m$ satisfy
$m \le \poly(n)$.  Let $\pi = \pi(p,q)$, where $p,q$ are partial distributions
satisfying $\|p\|_1, \|q\|_1 \geq 1/4$ such that $p$ is not $C$-highly concentrated relative
to $q$.  Then
    \[
        \Varu{\bm{H}}{\Exuc{\bm{T}}{\bm{Y}}{\bm{H}}}
        \le O\left( \log^6 n \right) \cdot p^\top \phi p + O(m^2 e^{-m/8}) \,.
    \]
\end{lemma}
\begin{proof}
    By \cref{prop:general-variance-first-component} we have, for some absolute constant $c > 0$,
    \begin{equation}
        \label{eq:var-first-formula}
        \Varu{\bm{H}}{\Exuc{\bm{T}}{\bm{Y}}{\bm{H}}}
        \le 5m^2 \zeta(\cI)\|p\|_1^4
            + c m^2 \cdot \sum_{\substack{I=\llangle i,d \rrangle \in \cI \\ 1 \le d \le n}}
                p_{i} p_{i+d-1} p[I]^2 \Ex{\bm{J}(I)} \,.
    \end{equation}
    Our assumption that $p$ is not highly concentrated relative to $q$ gives the inequality
    \[
        p[\llangle i,d+1 \rrangle]^2
        \le \left[ C\log^2(n) \cdot \max\left\{ q[\llangle i,d \rrangle], \frac{1}{m\log n} \right\} \right]^2
        \le C^2\log^4(n) q[\llangle i,d \rrangle]^2 + \frac{C^2 \log^2 n}{m^2} \,,
    \]
    and therefore, with a small change of variables in $d$,
    \begin{align*}
        &\sum_{\substack{I=\llangle i,d \rrangle \in \cI \\ 1 \le d \le n}}
            p_{i} p_{i+d-1} p[I]^2 \Ex{\bm{J}[I]} \\
        &\qquad= \sum_{i=0}^{n-1} \sum_{d=0}^{n-1}
            p_i p_{i+d} \Ex{\bm{J}(\llangle i,d+1 \rrangle)} p[ \llangle i,d+1 \rrangle]^2 \\
        &\qquad \le
            C^2 \log^4(n)
                \sum_{i=0}^{n-1} \sum_{d=0}^{n-1} p_i p_{i+d} \Ex{\bm{J}(\llangle i,d+1 \rrangle)}
                    q[\llangle i,d \rrangle]^2
            + \frac{C^2\log^2 n}{m^2}
                \sum_{i=0}^{n-1} \sum_{d=0}^{n-1} p_i p_{i+d} \Ex{\bm{J}(\llangle i,d+1 \rrangle)}
            \,.
    \end{align*}
    We show that both terms above satisfy our desired asymptotic bound. For the second term, note
    that $\bm{J}(\llangle i,d+1 \rrangle) = 1 \implies \bPhi_{i,i+d} = 1$, and therefore
    $\Ex{\bm{J}(\llangle i,d+1 \rrangle)} \le \phi_{i,i+d}$. Thus
    \begin{equation}
        \label{eq:var-first-1}
        \frac{C^2\log^2 n}{m^2}
            \sum_{i=0}^{n-1} \sum_{d=0}^{n-1} p_i p_{i+d} \Ex{\bm{J}(\llangle i,d+1 \rrangle)}
        \le \frac{C^2\log^2 n}{m^2}
            \sum_{i=0}^{n-1} \sum_{d=0}^{n-1} p_i p_{i+d} \phi_{i,i+d}
        = O\left(\frac{\log^2 n}{m^2}\right) \cdot p^\top \phi p \,,
    \end{equation}
    as desired. As for the first term, recall that
    \[
        \Ex{\bm{J}(\llangle i,d+1 \rrangle)}
        = \Pr{\forall e \in \llangle i,d \rrangle : e \in \bm{H}}
        = \prod_{e \in \llangle i,d \rrangle} (1 - w(e))
        = e^{-mq[\llangle i,d \rrangle]} \,.
    \]
    Let $K \ge 1$ be a constant such that $m \le n^K$ for all sufficiently large $n$,
    as per the assumption that $m \le \poly(n)$.
    Now, letting $x \define q[\llangle i,d \rrangle]$, which is bounded between $0$ and $1$,
    we consider two cases. First, suppose $x \ge \frac{2(K+2)\log n}{m}$. Then we obtain
    \[
        x^2 e^{-mx} \le 1 \cdot e^{-2(K+2)\log n} = n^{-2K-4} \,.
    \]
    On the other hand, if $x \le \frac{2(K+2)\log n}{m}$, then
    \[
        x^2 e^{-mx} \le O\left( \frac{\log^2 n}{m^2} \right) e^{-mx} \,.
    \]
    Therefore the first term is
    \begin{align*}
        &C^2 \log^4(n)
            \sum_{i=0}^{n-1} \sum_{d=0}^{n-1} p_i p_{i+d} e^{-mq[\llangle i,d \rrangle]} q[\llangle
i,d \rrangle ]^2 \\
        &\qquad \le O\left( \log^4 n \right)
            \sum_{i=0}^{n-1} \sum_{d=0}^{n-1} p_i p_{i+d}
                \left( n^{-2K-4} + O\left( \frac{\log^2 n}{m^2} \right) e^{-mq[\llangle i,d
\rrangle]} \right) \\
        &\qquad \le O(n^{-2K-1}) +
            O\left( \frac{\log^6 n}{m^2} \right) \cdot p^\top \phi p \,,
    \end{align*}
    where again we used the inequality
    $e^{-mq[\llangle i,d \rrangle]} = \Ex{\bm{J}(\llangle i,d+1 \rrangle)} \le \phi_{i,i+d}$
    in the last step.

    To upper bound the term $O(n^{-2K-1})$, we recall that $m \le n^K$ and observe that,
    since $\phi$ is $1$ on the diagonal and $\|p\|_1 \ge 1/4$, we have
    $p^\top \phi p \ge \|p\|_2^2 \ge \Omega(1/n)$. Therefore
    \[
        n^{-2K-1} = \frac{1}{n^{2K}} \cdot \frac{1}{n} \le \frac{1}{m^2} \cdot O(p^\top \phi p) \,.
    \]
    It follows that
    \begin{equation}
        \label{eq:var-first-2}
        C^2 \log^4(n)
            \sum_{i=0}^{n-1} \sum_{d=0}^{n-1} p_i p_{i+d} e^{-mq[\llangle i,d \rrangle]}
                q[\llangle i,d \rrangle ]^2
        \le O\left( \frac{\log^6 n}{m^2} \right) \cdot p^\top \phi p \,.
    \end{equation}

    As for the error term $5m^2 \zeta(\cI)\|p\|_1^4$, we upper bound $\|p\|_1$ by $1$
    and recall that $\zeta(\cI) \le e^{-\frac{m \|q\|_1}{2}}$ by \cref{prop:zeta-cycle-bound}.
    Along with the assumption that $\|q\|_1 \ge 1/4$, we obtain
    \begin{equation}
        \label{eq:var-first-3}
        5m^2 \zeta(\cI)\|p\|_1^4 \le O(m^2 e^{-m/8}) \,,
    \end{equation}
    as needed.
    Putting together \eqref{eq:var-first-1},\eqref{eq:var-first-2} and \eqref{eq:var-first-3}
    into \eqref{eq:var-first-formula}, we conclude that
    \[
        \Varu{\bm{H}}{\Exuc{\bm{T}}{\bm{Y}}{\bm{H}}}
        \le O(\log^6 n) p^\top \phi p + O(m^2 e^{-m/8}) \,. \qedhere
    \]
\end{proof}

To make the result above useful, we need to upper bound the quadratic form $p^\top \phi p$
by some quantity comparable to the separation shown in \cref{lemma:sep-expected-value}.
Recalling the breakdown in terms of an approximate uniform conjugate, $p = \widetilde p + z$,
our first task is to upper bound $\widetilde p^\top \phi \widetilde p$.

\begin{proposition}
\label{prop:quadratic-form-upper-bound}
Let $C > 0$ be a constant, let $n \in \bN$ be sufficiently large, and suppose $m$ satisfies
$m = \omega(\log n)$, $m = o(n)$. Let $\pi = \pi(p,q)$, where $p$, $q$ are partial distributions
satisfying $\|q\|_1 \geq 1/4$ such that $q$ is not $C$-highly concentrated relative to $p$.
Let $\widetilde p$ be an approximate uniform conjugate of $q$.  Then
\[
  {\widetilde p}^\top \phi \widetilde p = O\left(\frac{\log n}{m}\right) \,.
\]
\end{proposition}
\begin{proof}
    By \cref{prop:conjugate-mean}, ${\widetilde p}^\top \phi \widetilde p =\|p\|(\tau \pm 4n\xi)$.
    Our main task is to show that $\tau = O\left(\frac{\log n}{m}\right)$, but we first check
    that $4n\xi$ is small enough.
    Indeed, from \cref{def:approximate-uniform-conjugate} and since $\|q\|_1 \ge 1/4$, we have
    \[
        n\xi = n \cdot \frac{e^{-m \|q\|_1}}{(1 - e^{-m \|q\|_1)^2}}
        \le \frac{n \cdot e^{-\Omega(m)}}{(1 - e^{-\Omega(m)})^2}
        = o(1/n)
        \le o(1/m) \,,
    \]
    the last two steps since $m = \omega(\log n)$, $m = o(n)$.
    We now study $\tau$. Recall that
\[
    \tau = \frac{\|p\|_1}{\sum_{i=0}^{n-1} \tanh\left(\frac{mq_i}{2}\right)} \,.
\]
Since $\tanh$ is concave on the non-negative domain, our goal will be to upper bound $\tau$ using
Jensen's inequality. Write $q = \mu + y$. Let $S \define \{ i \in \bZ_n : y_i \geq 0 \}$ and
$\overline S \define \{ i \in \bZ_n : y_i < 0 \} = \bZ_n \setminus S$. Note that the vector $y^+$
only takes non-zero entries in $S$, and $y^-$ only takes non-zero entries in $\overline S$.
For each $i \in S$,
we have $\frac{1}{2n} \leq q_i \leq \frac{C \log n}{m}$, the upper bound since
$\|q\|_\infty \leq \frac{C \log n}{m}$ due to the anticoncentration assumption
(\cref{remark:non-concentrated-l-infty-linear-trace}). Then we may write
\[
  q_i = \lambda_i \cdot \frac{C \log n}{m} + (1-\lambda_i) \cdot \frac{1}{2n}\,,
  \qquad\text{ where }\qquad
\lambda_i = \frac{y_i}{\frac{C \log n}{m} - \frac{1}{2n}} \in [0,1] \,.
\]
For $i \in \overline S$, we have $0 \leq q_i < \frac{1}{2n}$, so we may
write
\[
  q_i = \lambda_i \cdot \frac{1}{2n} \,,
  \qquad\text{ where }\qquad
  \lambda_i = 2nq_i \,.
\]
Applying Jensen's inequality, and using $\tanh\left(\frac{m}{4n}\right) \geq \frac{m}{8n}$ which
holds for sufficiently small $m/4n$ (\cref{fact:tanh-bounds}),
\begin{align*}
\sum_{i=0}^{n-1} \tanh\left(\frac{mq_i}{2}\right)
&= \sum_{i \in S} \tanh\left(\lambda_i \cdot \frac{C \log n}{2}
    + (1-\lambda_i) \cdot \frac{m}{4n}\right)
  + \sum_{i \in \overline S} \tanh\left(\lambda_i \cdot \frac{m}{4n} \right) \\
&\geq \sum_{i \in S} \left[\lambda_i \tanh\left(\frac{C \log n}{2}\right)
    + (1-\lambda_i) \tanh\left(\frac{m}{4n}\right)\right]
  + \sum_{i \in \overline S} \lambda_i \tanh\left(\frac{m}{4n}\right) \\
&\geq \frac{1}{2} \sum_{i \in S} \lambda_i + \sum_{i \in S} (1-\lambda_i) \cdot \frac{m}{8n}
  + \sum_{i \in \overline S} \lambda_i \cdot \frac{m}{8n} \\
&= \frac{1}{2} \sum_{i \in S} \lambda_i + \frac{m}{8n} \left(
        \sum_{i \in S} (1-\lambda_i) + \sum_{i \in \overline S} \lambda_i \right) \,,
\end{align*}
where we have used the fact that $\tanh\left(\frac{C \log n}{2}\right) \geq \frac12$ for
sufficiently large $n$.  We consider two cases. First assume that $\|y^+\|_1 \geq
\frac{1}{8}$. Then
\[
    \frac{1}{2}\sum_{i \in S} \lambda_i
    = \frac{1}{2} \cdot \frac{m}{C \log n - \frac{m}{2n}} \|y^+\|_1
    \geq \frac{m}{16 C \log n} \,,
\]
and therefore
\[
\tau = \frac{\|p\|_1}{\sum_{i=0}^{n-1} \tanh\left(\frac{mq_i}{2}\right)} \leq \frac{16C\log n}{m} \,,
\]
as desired. Next assume that $\|y^+\|_1 < \frac{1}{8}$.  Observe that for $i \in S$,
\[
  \lambda_i = \frac{y_i}{ \frac{C \log n}{m} - \frac{1}{2n} }
    = \frac{my_i}{C\log n - \frac{m}{2n}} \leq \frac{2my_i}{C \log n} \,,
\]
since $\frac{m}{2n} \leq \frac12 C\log n$. 
So
\begin{align*}
\sum_{i \in S} (1-\lambda_i) + \sum_{i \in \overline S} \lambda_i
  &\geq \sum_{i \in S} \left( 1-\frac{2my_i}{C \log n} \right)
      + 2n \sum_{i \in \overline S} q_i
  = \sum_{i \in S} \left( 1-\frac{2my_i}{C \log n} \right)
      + \sum_{i \in \overline S} ( 1 + 2ny_i ) \\
  &= |S| - \frac{2m}{C \log n} \|y^+\|_1 + |\overline S| - 2n \|y^-\|_1 \\
  &= n - \frac{2m}{C \log n} \|y^+\|_1 - 2n \|y^-\|_1 \,.
\end{align*}
Recalling the assumption that $\|q\|_1 \ge 1/4$, we now observe that
\begin{align*}
    \|q\|_1 = \sum_{i \in S} \left(\frac{1}{2n} + |y_i|\right)
        + \sum_{i \in \overline S} \left(\frac{1}{2n} - |y_i|\right)
    = \frac{1}{2} + \|y^+\|_1 - \|y^-\|_1 \geq \frac{1}{4} \,,
\end{align*}
so $\|y^-\|_1 \leq \|y^+\|_1 + \frac{1}{4}$. Then
\begin{align*}
\sum_{i \in S} (1-\lambda_i) + \sum_{i \in \overline S} \lambda_i
  &\geq n - \frac{2m}{C \log n} \|y^+\|_1 - 2n \|y^-\|_1 \\
  &\geq n - \frac{2m}{C \log n} \|y^+\|_1 - 2n \left(\|y^+\|_1 + \frac{1}{4}\right) \\
  &= n - \|y^+\|_1 \left( \frac{2m}{C \log n} + 2n \right) - \frac{n}{2} \\
  &\geq \frac{n}{2} - \frac{1}{8} \left( \frac{2m}{C \log n} + 2n \right)
  = \left(\frac14 - o(1)\right)n \,.
\end{align*}
We conclude that
\[
  \sum_{i=0}^{n-1} \tanh\left(\frac{mq_i}{2}\right)
  \geq \frac{m}{8n} \left( \sum_{i \in S} (1-\lambda_i) + \sum_{i \in
\overline S} \lambda_i \right)
  \geq \frac{m}{8n} \left(\frac{1}{4} - o(1)\right) n \geq \frac{m}{33} \,,
\]
for sufficiently large $n$, in which case $\tau \leq \frac{33}{m}$. This concludes the proof.
\end{proof}

\begin{corollary}
    \label{cor:variance-first-component}
    Under the assumptions of \cref{lemma:var-first-component} and
    \cref{prop:quadratic-form-upper-bound}, and writing $p = \widetilde p + z$,
    the first component of the variance satisfies
    \[
        \Varu{\bm{H}}{\Exuc{\bm{T}}{\bm{Y}}{\bm{H}}}
        \le O\left( \frac{\log^7 n}{m} \right) + O\left(\log^6 n\right) z^\top \phi z \,.
    \]
\end{corollary}
\begin{proof}
    Note that $m^2 e^{-m/8} = o(1/m)$ since $m = \omega(\log n) = \omega(1)$.
    Thus \cref{lemma:var-first-component}, along with the breakdown $p = \widetilde p + z$, yields
    \begin{align*}
        \Varu{\bm{H}}{\Exuc{\bm{T}}{\bm{Y}}{\bm{H}}}
        &\le O\left( \log^6 n \right) \cdot (\widetilde p + z)^\top \phi (\widetilde p + z)
            + O(m^2 e^{-m/8}) \\
        &= O\left( \log^6 n \right) \left(
            \widetilde p^\top \phi \widetilde p + z^\top \phi z + 2 z^\top \phi \widetilde p
            \right)
            + o(1/m) \,.
    \end{align*}
    By a similar argument as in the proof of \cref{prop:conjugate-mean}, we have
    $z^\top \phi \widetilde p \le 4 n^2 \xi = o(1/m)$, the last step as in
    the proof of \cref{prop:quadratic-form-upper-bound}.
    Applying \cref{prop:quadratic-form-upper-bound} to the term
    $\widetilde p^\top \phi \widetilde p$, we obtain
    \[
        \Varu{\bm{H}}{\Exuc{\bm{T}}{\bm{Y}}{\bm{H}}}
        \le O\left(\frac{\log^7 n}{m}\right) + O\left(\log^6 n\right) z^\top \phi z \,.
        \qedhere
    \]
\end{proof}

We now upper bound the second component of the variance. The key step is to show that, with
high probability, no bucket contains too much probability mass:

\begin{proposition}
    \label{prop:var-q-bin}
    Let $C, K > 0$ be constants.
    Let $\pi = \pi(p,q)$ where $p,q$ are partial distributions such that $p$ is not $C$-highly
    concentrated relative to $q$. Then the random bucketing
    $\bm{\Gamma} = (\bm{\Gamma}_1, \dotsc, \bm{\Gamma}_{\bm{b}})$ induced by $\bm{H}$ satisfies
    \[
        \Pr{ \max_j p[\bm{\Gamma}_j] \geq 2(K+1)\frac{C \log^3 n}{m} } < 2/n^K \,.
    \]
\end{proposition}
\begin{proof}
Fix any $i \in \bZ_n$, and let $\bm{\Gamma}_j$ be the (random) bucket containing $i$. We wish to
bound the probability that $p[ \bm{\Gamma}_j ] > 2(K+1) C \frac{\log^3 n}{m}$.
Let $R = \llangle i, d_R \rrangle$ be the minimal circular interval in the clockwise direction
starting at $i$ satisfying $p[R] \ge (K+1) C\frac{\log^3 n}{m}$. Likewise, let
$L = \llangle i, -d_L \rrangle$ be the minimal circular interval in the counterclockwise direction
starting at $i$ satisfying $p[L] \ge (K+1) C\frac{\log^3 n}{m}$. Observe that, if
$p[\bm{\Gamma}_j] \ge 2(K+1) C\frac{\log^3 n}{m}$, then the bucket contains at least one of these
intervals: $\cE(L) \subseteq \cE(\bm{\Gamma}_j)$ or $\cE(R) \subseteq \cE(\bm{\Gamma}_j)$,
and therefore $\bm{J}(L) = 1$ or $\bm{J}(R) = 1$.

    Since $p$ is not $C$-highly concentrated relative to $q$, we have
    $(K+1) C\frac{\log^3 n}{m} \le p[R] < C\log^2(n) \cdot \max\{ q[R^*], \frac{1}{m\log n} \}$.
    Therefore, it must be the case that
    \[
        q[R^*] > \frac{p[R]}{C\log^2 n} \ge (K+1)\frac{\log n}{m} \,.
    \]
    Now
    \[
        \Pr{ \bm{J}(R) = 1 }
        = \Pr{ \forall z \in R^* : z \in \bm{H} }
        = \prod_{z \in R^*} (1 - w(z))
        = \prod_{z \in R^*} e^{-mq_z}
        = e^{-m q[ R^* ]}
        < e^{-(K+1) \log n}
        = \frac{1}{n^{K+1}} \,.
    \]
    The same holds for $\Pr{ \bm{J}(L) = 1 }$. Then, by the union bound over $i \in \bZ_n$,
    \[
        \Pr{ \max_j p[\bm{\Gamma}_j] \geq 2(K+1)\frac{C \log^3 n}{m} } < 2/n^K \,. \qedhere
    \]
\end{proof}

\begin{lemma}[Second component of the variance]
    \label{lemma:var-second-component}
    Let $C > 0$ be a constant.
    Let $n$ be sufficiently large and suppose $m$ satisfies $m \le \poly(n)$.
    Suppose $\pi = \pi(p,q)$ where $p,q$ are partial distributions such that $p$ is not $C$-highly
    concentrated relative to $q$. Then we have
    \[
        \Exu{\bm{H}}{\Varuc{\bm{T}}{\bm{Y}}{\bm{H}}} \le O\left( \frac{\log^6 n}{m} \right) \,.
    \]
\end{lemma}
\begin{proof}
    We start from the general result from \cref{prop:general-variance-second-component}:
    for some absolute constant $c > 0$, for every $H$ in the support of $\bm{H}$ with induced
    buckets $\Gamma = (\Gamma_1, \dotsc, \Gamma_b)$,
    \[
        \Varuc{\bm{T}}{\bm{Y}}{\bm{H} = H} \le c \|p_{|\Gamma}\|_2^2 + cm \|p_{|\Gamma}\|_3^3 \,.
    \]

    Let $K > 0$ be a constant such that $m \le n^{K/2}$ for all sufficiently large $n$,
    which exists by the assumption that $m \le \poly(n)$.
    First, suppose the subgraph $H$ induces bucketing $\Gamma = (\Gamma_1, \dotsc, \Gamma_b)$ satisfying
    $\max_j p[\Gamma_j] \le 2(K+1)C\frac{\log^3 n}{m}$. Since $\|p_{|\Gamma}\|_1 = \|p\|_1 \le 1$,
    we can upper bound the values that $\|p_{|\Gamma}\|_2^2$ and $\|p_{|\Gamma}\|_3^3$ can take
    by distributing $1$ total weight in a maximally concentrated way, \ie meeting the per-bucket
    upper bound we have just assumed. Therefore, we obtain
    \[
        \|p_{|\Gamma}\|_2^2
        \le \left(2(K+1)C\frac{\log^3 n}{m}\right)^2 \cdot \frac{1}{\left(2(K+1)C\frac{\log^3 n}{m}\right)}
        = O\left(\frac{\log^3 n}{m}\right)
    \]
    and
    \[
        \|p_{|\Gamma}\|_3^3
        \le \left(2(K+1)C\frac{\log^3 n}{m}\right)^3 \cdot \frac{1}{\left(2(K+1)C\frac{\log^3 n}{m}\right)}
        = O\left(\frac{\log^6 n}{m^2}\right) \,.
    \]
    Therefore, in this case, we have
    \[
        \Varuc{\bm{T}}{\bm{Y}}{\bm{H} = H} \le O\left(\frac{\log^6 n}{m}\right) \,.
    \]
    On the other hand, since $\|p_{|\Gamma}\|_1 \le 1$, every $H$ satisfies the simpler bound
    \[
        \Varuc{\bm{T}}{\bm{Y}}{\bm{H} = H}
        \le c \|p_{|\Gamma}\|_2^2 + cm \|p_{|\Gamma}\|_3^3
        = O(m) \,.
    \]
    Using \cref{prop:var-q-bin}, we write
    \begin{align*}
        &\Exu{\bm{H}}{\Varuc{\bm{T}}{\bm{Y}}{\bm{H}}} \\
        &\qquad \leq \Pr{ \max_j p[\bm{\Gamma}_j] \leq \frac{2(K+1)C\log^3 n}{m}} \cdot
            \Exuc{\bm{H}}{\Varuc{\bm{T}}{\bm{Y}}{\bm{H}}}
                {\max_j p[\bm{\Gamma}_j] \leq \frac{2(K+1)C \log^3 n}{m}} \\
        &\qquad\quad + \Pr{ \max_j p[\bm{\Gamma}_j] > \frac{2(K+1)C \log n}{m}} \cdot
            \Exuc{\bm{H}}{\Varuc{\bm{T}}{\bm{Y}}{\bm{H}}}
                {\max_j p[\bm{\Gamma}_j] > \frac{2(K+1)C\log^3 n}{m}} \\
        &\qquad \leq O\left(\frac{\log^6 n}{m}\right) + O(m) \cdot \frac{2}{n^K} \,,
    \end{align*}
    and since $m \le n^{K/2}$, we have $O(m) \cdot \frac{2}{n^K} \le O(1/n^{K/2}) \le O(1/m)$,
    as needed.
\end{proof}

We can now use the law of total variance to combine these results into a concentration bound for
the test statistic:

\begin{lemma}[Concentration of the Test Statistic]
    \label{lemma:concentration-parity-trace}
    Let $C > 0$ be a constant and $n \in \bN$ be sufficiently large. Suppose $m = m(n,\epsilon)$
    satisfies $m \le \poly(n)$. Let $\pi = \pi(p,q)$, where $p,q$ are partial distributions
    satisfying $\|p\|_1, \|q\|_1 \ge 1/4$ such that $p$ is not $C$-highly concentrated
    relative to $q$. Then for all $t > 0$,
    \[
        \Pr{\abs*{\bm{Y}-\Ex{Y}} \ge t}
        \le \frac{\frac{1}{m} + p^\top \phi p}{t^2} \cdot O(\log^6 n) \,.
    \]
    Moreover, suppose $q$ is not $C$-highly concentrated relative to $p$ and $m$ satisfies
    $m = \omega(\log n)$, $m = o(n)$. Then writing $p = \widetilde p + z$ where $\widetilde p$
    is an approximate uniform conjugate of $q$, we also have
    \[
        \Pr{\abs*{\bm{Y}-\Ex{Y}} \ge t}
        \le \frac{\frac{1}{m} + z^\top \phi z}{t^2} \cdot O(\log^7 n) \,.
    \]
\end{lemma}
\begin{proof}
    By the law of total variance,
    \[
        \Var{\bm{Y}} = \Varu{\bm{H}}{ \Exuc{\bm{T}}{ \bm{Y} }{ \bm{H} } }
        + \Exu{\bm{H}}{ \Varuc{\bm{T}}{ \bm{Y} }{ \bm{H} } } \,.
    \]
    The first term is bounded by
    \[
        \Varu{\bm{H}}{\Exuc{\bm{T}}{\bm{Y}}{\bm{H}}}
        \le O\left( \log^6 n \right) \cdot p^\top \phi p + O(m^2 e^{-m/8})
    \]
    by \cref{lemma:var-first-component}, and the second term is bounded by
    \[
        \Exu{\bm{H}}{\Varuc{\bm{T}}{\bm{Y}}{\bm{H}}} \le O\left( \frac{\log^6 n}{m} \right)
    \]
    by \cref{lemma:var-second-component}. Moreover, for any constant $c > 0$,
    the function $m \mapsto m^3 e^{-cm}$ has a global maximum of
    $\frac{3^3}{e^3 c^3}$, and therefore
    $m^2 e^{-m/8} = O(1/m)$. The first statement follows from Chebyshev's inequality.

    Making also the second set of assumptions, \cref{cor:variance-first-component} implies that
    \[
        \Varu{\bm{H}}{\Exuc{\bm{T}}{\bm{Y}}{\bm{H}}}
        \le O\left( \frac{1}{m} + z^\top \phi z \right) \cdot O(\log^7 n) \,,
    \]
    so the second statement follows again from Chebyshev's inequality.
\end{proof}

\subsection{Correctness of the Tester for Large \texorpdfstring{$\epsilon$}{Epsilon}}
\label{section:parity-correctness}

We can use our separation and concentration results above to show that $\bm{Y}$ is concentrated
on the correct side of the tester's threshold. Combining this with the easy cases of biased and
highly concentrated distributions will yield the correctness result.

\begin{lemma}
    \label{lemma:test-statistic-concentration}

    Let $\alpha, \gamma > 0 $ be constants. There exist constants
    $\beta = \beta_{\alpha,\gamma} > 0$ and $K = K_{\alpha,\beta,\gamma} > 1$
    such that the following holds for all sufficiently large $n$.
    Suppose $\epsilon \ge \frac{K \log^{3} n}{n^{1/4}}$.
    Let $\pi = \pi(p,q)$, where $p,q$ are partial
    distributions satisfying $\|p\|_1, \|q\|_1 = \frac{1}{2} \pm \frac{2\gamma}{\sqrt m}$
    and suppose that $p, q$ are not $4\alpha$-highly concentrated relative to the other.

    Let $T \define \frac{m}{4n^2}\sum_{i,j} \phi^{(\mu)}_{i,j}
    + \beta \frac{\epsilon^2 m^2}{n^2}$ be the threshold used by
    \cref{alg:uniformity-tester-linear-trace}.
    Let $\bm{Y}^{(0)}$ and $\bm{Y}^{(1)}$ be random variables denoting the value of the test
    statistic $\bm{Y}$ in the iterations $b=0$ and $b=1$ of the algorithm, respectively.
    Then when
    $m = \Theta_{\alpha,\beta,\gamma}\left( \left(\frac{n}{\epsilon}\right)^{4/5} \log^{7/5} n \right)$,
    the following statements hold:
    \begin{enumerate}
        \item (Completeness) If $\pi(p, q)$ is the uniform distribution over $[2n]$, then
            $\max\{ \bm{Y}^{(0)}, \bm{Y}^{(1)} \} < T$ with probability at least $99/100$; and
        \item (Soundness) If $\dist_\TV(\pi(p, q), \pi(\mu, \mu)) > \epsilon$,
            then $\max\{ \bm{Y}^{(0)}, \bm{Y}^{(1)} \} > T$ with probability at least $99/100$.
    \end{enumerate}
\end{lemma}
\begin{proof}
    Note that we can simply write $T = \Ex{\bm{Y}^{(\mu)}} + \beta \frac{\epsilon^2 m^2}{n^2}$
    by \cref{prop:expectation-y-uniform}.

    \textbf{Completeness.}
    In this case, $\Ex{\bm{Y}^{(0)}} = \Ex{\bm{Y}^{(1)}} = \Ex{\bm{Y}^{(\mu)}}$.
    Moreover, in this case we can write $p = \widetilde p + z$ for $\widetilde p = \mu$ and
    $z = \vec 0$ since $\mu$ is its own uniform conjugate by \cref{lemma:approx-conjugate-existence}.
    Hence \cref{lemma:concentration-parity-trace} gives
    \[
        \Pr{\bm{Y}^{(1)} \ge T}
        \le \Pr{\abs*{\bm{Y}^{(1)} - \Ex{\bm{Y}^{(1)}}} \ge \beta \frac{\epsilon^2 m^2}{n^2}}
        \le \frac{\frac{1}{m} + z^\top \phi z}{\left( \beta \frac{\epsilon^2 m^2}{n^2} \right)^2}
            \cdot O(\log^7 n)
        = \frac{n^4}{\beta^2 \epsilon^4 m^5} \cdot O(\log^7 n) \,.
    \]
    Thus for any constant $\beta$ (to be chosen below), this probability is at most (say) $1/200$
    when $m = \Omega((n/\epsilon)^{4/5} \log^{7/5} n)$.
    By symmetry, the same is true for $\bm{Y}^{(0)}$, and hence the probability that
    $\max\{ \bm{Y}^{(0)}, \bm{Y}^{(1)} \} < T$ fails to hold is at most $1/100$, as desired.

    \textbf{Soundness.}
    Without loss of generality, it suffices to consider the case when $\|q-\mu\|_1 > \epsilon$
    and show that $\bm{Y}^{(1)} > T$ with probability at least $99/100$.

    Since $\epsilon \ge \frac{K \log^{3} n}{n^{1/4}}$ and
    $m = \Theta_{\alpha,\beta,\gamma}( (n/\epsilon)^{4/5} \log^{7/5} n )$,
    for any value of $\beta$
    we can ensure that $\frac{m}{n / \log n}$ is smaller than any constant by making
    $K = K_{\alpha,\beta,\gamma}$ sufficiently large. Indeed, for some constant
    $A = A_{\alpha,\beta,\gamma} > 0$ and sufficiently large $n$, we have
    \[
        m \le A \left(\frac{n}{\epsilon}\right)^{4/5} \log^{7/5} n
        \le \frac{A n \log^{7/5} n}{K^{4/5} \log^{12/5} n}
        = \frac{A}{K^{4/5}} \cdot \frac{n}{\log n} \,,
    \]
    which can be made sufficiently small by making $K$ sufficiently large.
    Therefore the conditions of \cref{lemma:sep-expected-value} are satisfied and we obtain
    \[
        \Ex{\bm{Y}^{(1)}} \ge \Ex{\bm{Y}^{(\mu)}}
            + \Omega\left( \frac{\epsilon^2 m^2}{n^2} \right) + m z^\top \phi z \,.
    \]
    For concreteness, let $L = L_{\alpha,\gamma} > 0$ be a constant such that,
    for sufficiently large $n$, we have
    \[
        \Ex{\bm{Y}^{(1)}} \ge \Ex{\bm{Y}^{(\mu)}} + L\frac{\epsilon^2 m^2}{n^2} + m z^\top \phi z \,.
    \]
    Then as long as $\beta < L/2$, \cref{lemma:concentration-parity-trace} yields
    \begin{align*}
        \Pr{\bm{Y}^{(1)} \le T}
        &\le \Pr{\bm{Y}^{(1)} - \Ex{\bm{Y}^{(1)}} \le \beta \frac{\epsilon^2 m^2}{n^2}
            - L\frac{\epsilon^2 m^2}{n^2} - m z^\top \phi z} \\
        &\le \Pr{\abs*{\bm{Y}^{(1)} - \Ex{\bm{Y}^{(1)}}} \ge
            (L-\beta) \frac{\epsilon^2 m^2}{n^2} + m z^\top \phi z} \\
        &\le \Pr{\abs*{\bm{Y}^{(1)} - \Ex{\bm{Y}^{(1)}}} \ge
            \frac{L}{2} \cdot \frac{\epsilon^2 m^2}{n^2} + m z^\top \phi z} \\
        &\le \frac{\frac{1}{m} + z^\top \phi z}
                  {\left( \frac{L}{2} \cdot \frac{\epsilon^2 m^2}{n^2} + m z^\top \phi z \right)^2}
                  \cdot O(\log^7 n) \\
        &\le \frac{\frac{1}{m} + z^\top \phi z}
                  {\frac{L^2}{4} \cdot \frac{\epsilon^4 m^4}{n^4} + m^2 \left(z^\top \phi z\right)^2}
                  \cdot O(\log^7 n) \,,
    \end{align*}
    the last step since $z^\top \phi z \ge 0$ due to the positive semidefiniteness of $\phi$.
    We now consider two cases. First, suppose $z^\top \phi z \le 1/m$. Then
    \[
        \Pr{\bm{Y}^{(1)} \le T}
        \le \frac{2/m}{\frac{L^2}{4} \cdot \frac{\epsilon^4 m^4}{n^4}} \cdot O(\log^7 n)
        = \frac{n^4}{L^2 \epsilon^4 m^5} \cdot O(\log^7 n) \,,
    \]
    which is again at most $1/200$. On the other hand, suppose $z^\top \phi z \ge 1/m$. Then
    \[
        \Pr{\bm{Y}^{(1)} \le T}
        \le \frac{1/m}{\frac{L^2}{4} \cdot \frac{\epsilon^4 m^4}{n^4}} \cdot O(\log^7 n)
            + \frac{z^\top \phi z}{m^2 (z^\top \phi z)^2} \cdot O(\log^7 n)
        \le \frac{n^4}{L^2 \epsilon^4 m^5} \cdot O(\log^7 n) + \frac{1}{m} \cdot O(\log^7 n) \,.
    \]
    We have already seen that the first term is at most $1/200$, and the second term is clearly
    $o(1)$. Hence $\Pr{\bm{Y}^{(1)} \le T} \le 1/100$, concluding the proof.
\end{proof}

We may now combine the previous results to conclude the correctness of the tester:

\begin{theorem}
    \label{thm:upper-bound-linear-trace-main}
    There exist constants $\alpha, \beta, \gamma > 0$ and $K > 1$ such that the following holds
    for all sufficiently large $n$.
    Suppose $\epsilon \ge \frac{K \log^{3} n}{n^{1/4}}$.
    Let $\pi = \pi(p,q)$, where $p,q$ are partial distributions.

    Then \cref{alg:uniformity-tester-linear-trace} instantiated with constants
    $\alpha, \beta, \text{and } \gamma$
    has sample complexity $\Theta\left( \left(\frac{n}{\epsilon}\right)^{4/5} \log^{7/5} n\right)$
    and satisfies the following:
    \begin{enumerate}
        \item (Completeness) If $\pi(p, q)$ is the uniform distribution over $[2n]$, the algorithm
            accepts with probability at least $9/10$; and
        \item (Soundness) If $\dist_\TV(\pi(p, q), \pi(\mu, \mu)) > \epsilon$, the algorithm
            rejects with probability at least $9/10$.
    \end{enumerate}
\end{theorem}
\begin{proof}
    We first instantiate sufficiently large $\alpha, \gamma > 0$, sufficiently small $\beta > 0$ and
    sufficiently large $K > 1$ (in this order) to satisfy the conditions of
    \cref{prop:easy-case-bias,prop:easy-case-concentrated,lemma:test-statistic-concentration}.
    The sample complexity follows from the definition of the algorithm; we now show that it
    correctly accepts/rejects.

    \textbf{Completeness.}
    By \cref{prop:easy-case-bias,prop:easy-case-concentrated,lemma:test-statistic-concentration},
    the algorithm rejects with probability at most
    $1/100 + 1/100 + 1/100 < 1/10$.

    \textbf{Soundness.}
    We consider three cases. First, suppose $\|p\|_1 \not\in \frac{1}{2} \pm \frac{2\gamma}{\sqrt{m}}$.
    Then by \cref{prop:easy-case-bias}, the algorithm rejects with probability at least $99/100$.

    Second, suppose $p$ or $q$ is $4\alpha$-highly concentrated relative to the other. Then
    by \cref{prop:easy-case-concentrated}, the algorithm rejects with probability at least $99/100$.

    Finally, suppose $\|p\|_1, \|q\|_1 = \frac{1}{2} \pm \frac{2\gamma}{\sqrt{m}}$ and neither $p$
    nor $q$ is $4\alpha$-highly concentrated. Then by \cref{lemma:test-statistic-concentration},
    the algorithm rejects with probability at least $99/100$, as desired.
\end{proof}

\noindent
Combining \cref{thm:upper-bound-linear-trace-main,lemma:upper-bound-linear-trace-small-epsilon}
establishes the upper bound portion of \cref{thm:intro-main}.

\section{Lower Bound for Testing Uniformity in the Parity Trace Model}
\label{section:lower-bound}

\paragraph{Notation} In this section, let $\mu$ denote the partial distribution for domain $[2n]$
with total mass $1/2$ uniformly distributed over its support, so that $\pi(\mu,\mu)$ is the uniform
distribution over $[2n]$.

We wish to prove the following result:

\begin{theorem}[Lower bound portion of \cref{thm:intro-main}]
    \label{thm: lower bound}

    Let $\Pi_1$ contain only the uniform distribution over $[2n]$, and let $\Pi_2$ be the set of
    distributions over $[2n]$ that are $\epsilon$-far from uniform in total variation distance. Then
    $(\Pi_1, \Pi_2, 51/100)$-testing under the parity trace requires sample complexity at least
    $\widetilde \Omega\left(
        \left( \frac{n}{\epsilon} \right)^{4/5} +\frac{\sqrt{n}}{\epsilon^2} \right)$,
    where the $\widetilde \Omega$ notation only hides polylogarithmic factors in $n$.
    Furthermore, this bound holds even if the input distribution
    $\pi$ is guaranteed to have $1/2$ mass uniformly distributed over the zero-valued (\ie even)
    coordinates.
\end{theorem}

We divide the analysis into two parts: a reduction from the standard uniformity testing model,
which establishes an $\Omega(\sqrt{n}/\epsilon^2)$ lower bound, and a more sophisticated argument
that applies when $\epsilon \ge n^{-1/4}$; fortunately, this is precisely the regime where
$(n/\epsilon)^{4/5} \ge \sqrt{n}/\epsilon^2$. First, the easier bound:

\begin{proposition}
    \label{prop:lb-reduction-from-standard}
    Let $\Pi_1$ contain only the uniform distribution over $[2n]$, and let $\Pi_2$ be the set of
    distributions over $[2n]$ that are $\epsilon$-far from uniform in total variation distance. Then
    $(\Pi_1, \Pi_2, 51/100)$-testing under the parity trace requires sample complexity at least
    $\Omega(\sqrt{n}/\epsilon^2)$.
    Furthermore, this bound holds even if the input distribution
    $\pi$ is guaranteed to have $1/2$ mass uniformly distributed over the zero-valued (\ie even)
    coordinates.
\end{proposition}
\begin{proof}
We reduce from testing uniformity of a distribution over $[n]$, for which there is a lower bound of
$\Omega(\sqrt n / \epsilon^2)$ \cite{Pan08}.  For input distribution $\pi$ over $[n]$, let $\pi'$ be
the distribution on $[2n]$ defined by setting $\pi'(2i-1) = \pi(i)/2$ for each $i \in [n]$ and
$\pi(2i) = \tfrac{1}{2n}$ for $i \in [n]$, so that $\pi'$ is uniform over the even elements.
Observe that we may simulate a sample from $\pi'$ by sampling $\bm x \sim \pi$ and taking
$2\bm x - 1$ with probability $1/2$, and otherwise taking a uniformly random even element of $[2n]$.
Then the following hold:
\begin{enumerate}
\item If $\pi$ is uniform over $[n]$ then $\pi'$ is uniform over $[2n]$; and
\item If $\pi$ is $\epsilon$-far from uniform then $\pi'$ is $\epsilon/2$-far from uniform (with
respect to TV distance).
\end{enumerate}
Therefore the tester for uniformity may simulate the parity trace tester with parameter
$\epsilon/2$.
\end{proof}

We now give our main technical argument to show the
$\widetilde \Omega\left((n/\epsilon)^{4/5}\right)$ bound for the case $\epsilon \ge n^{-1/4}$.

\subsection{Outline of the Argument}

Our approach, inspired by \cite{DK16}, is to construct distributions over YES and NO inputs
such that, when $\bm{Z}$ is a random variable indicating the YES/NO case and $\bm{\cT}$ is the input
to the algorithm (a parity trace drawn from a YES or NO distribution),
the mutual information $I(\bm{Z} : \bm{\cT})$
is small, so that no algorithm can predict $\bm{Z}$ from $\bm{\cT}$ with good probability.
Concretely, we follow \cite{DK16} and use the following simple consequence of Fano's inequality:

\begin{fact}[Fano's inequality]
    \label{fact:fanos}
    Suppose $\bm{Z}$ is a uniform random bit, $\bm{\cT}$ is a random variable, and there exists
    a function $f$ such that $f(\bm{\cT}) = \bm{Z}$ with probability at least $51\%$.
    Then $I(\bm{Z} : \bm{\cT}) \ge 2 \cdot 10^{-4}$.
\end{fact}

Therefore, our goal is to construct ``distributions over distributions''
(hereby called \emph{distributions})
$\cD_0$ (YES case) and $\cD_1$ (NO case), which are supported on distributions $\pi$ over $[2n]$,
satisfying the following: let $m = m(n, \epsilon)$ be the sample complexity of the tester, and
assume the
Poissonized setting (which will be convenient later). Then we want to satisfy the following:
\begin{enumerate}
    \item $\cD_0$ is supported on a single element $\pi(\mu, \mu)$,
        the uniform distribution over $[2n]$;
    \item Every $\pi$ in the support of $\cD_1$ satisfies
        $\dist_\TV(\pi, \pi(\mu,\mu)) \ge \Omega(\epsilon)$; and
    \item Let $\bm{Z} \sim \Ber(1/2)$, and $\bm{\pi} \sim \cD_{\bm{Z}}$.
        Let $\bm{\cT}$ be distributed as follows: draw $\bm{S} \sim \samp(\bm{\pi}, \Poi(m))$
        and let $\bm{\cT} = \trace(\bm{S})$.
        Then when $m = o\left( \left(\frac{n}{\epsilon}\right)^{4/5} \frac{1}{\log^4 n} \right)$,
        we have $I(\bm{Z} : \bm{\cT}) = o(1)$.
\end{enumerate}

We now outline the main ingredients of our proof, and then present the full argument.
For simplicity, we will assume that $n$ is even.

\paragraph{YES and NO distributions.}
Recall that a probability distribution $\pi = \pi(p,q)$ over $[2n]$ consists of partial
distributions $p$ over the 1-valued elements (odd indices) and $q$ over the 0-valued elements
(even indices). We will partition the domain $[2n]$ into $n/2$ consecutive
length-4 intervals, called \emph{dominoes}, such that the $i$-th domino determines the entries
$(p_j, q_j, p_{j+1}, q_{j+1})$, where $j = 2i-1$, and contributes to the trace a string (called
a \emph{subtrace}) distributed as
\[
    1^{\bm{A}_j} 0^{\bm{B}_j} 1^{\bm{A}_{j+1}} 0^{\bm{B}_{j+1}} \,,
\]
where $\bm{A}_k \sim \Poi(m p_k), \bm{B}_k \sim \Poi(m q_k)$ independently for each $k \in [n]$.

We will always set $q = \mu$, \ie the partial distribution over the 0-valued elements is
uniform with total mass $1/2$. In the YES distribution $\cD_0$, $p = \mu$ as well.
In the NO distribution $\cD_1$, we will set either
$(p_j, p_{j+1}) = \left(\frac{1+\epsilon}{2n}, \frac{1-\epsilon}{2n}\right)$ or
$(p_j, p_{j+1}) = \left(\frac{1-\epsilon}{2n}, \frac{1+\epsilon}{2n}\right)$,
with equal probability and independently for each domino. Hence each domino is ``balanced''
and the subtraces produced by different dominoes are independent conditional on $\bm{Z}$.
Moreover, we will show that sampling at most 2 symbols from a domino reveals no information about
$\bm{Z}$, \ie only 3-way or larger collisions are informative.

\paragraph{Partial fingerprints.}
Since each domino is uninformative if at most 2 symbols are drawn from it, we will study the
distributional properties of those dominoes from which a larger number of symbols was
sampled---this is where information about $\bm{Z}$ may be revealed to the algorithm.
Drawing inspiration from standard distribution testing theory, we will study the \emph{partial
fingerprint} over the dominoes, which essentially measures how many information-revealing symbols
were sampled.

Roughly speaking, we will show that the probability of a partial fingerprint decreases exponentially
in the number of information-revealing sample elements (namely, those coming from dominoes
from which 3 or more symbols were drawn), which places an
upper bound on how much the algorithm can learn from these elements. We remark that the lower
bound argument of~\cite{DKN15a} for testing closeness of structured distributions uses the similar
idea of constructing a gadget from which up to two samples are distributed identically under YES
and NO conditions.

\paragraph{Partition of the domain and chain rule of mutual information.}
Given the observations above, one might hope to conclude the argument by 1) upper bounding the
mutual information between $\bm{Z}$ and the subtrace from each domino; and 2) adding up, by
the chain rule of mutual information, this quantity over all the dominoes.
(If random variables $\bm{T}_1, \dotsc, \bm{T}_k$ are independent conditional on $\bm{Z}$,
the chain rule of mutual information implies that
$I(\bm{T}_1, \dotsc, \bm{T}_k : \bm{Z}) \le \sum_{i=1}^k I(\bm{T}_i : \bm{Z})$.)
Unfortunately,
this strategy does not give a good bound; intuitively, it assumes that the
algorithm ``knows'' too much---namely the boundaries of all the dominoes in the trace it sees,
which, in reality, should be very difficult to predict.\footnote{Another interesting attempt is to
condition the analysis on the identities of the 0-valued symbols seen in the trace, and
then consider the distribution of the 1-valued symbols inside each range delimited by the zeroes.
This also seems to fail for a similar reason: by the birthday paradox, when one
draws $n^{4/5}$ samples from $[2n]$, many of the intervals delimited by the 0-valued symbols
will be very small, which also amounts to ``revealing'' too much information.}

As it turns out, one solution is to consider $\Theta(m)$ contiguous ranges, each consisting
of $\Theta(n/m)$ dominoes. Since we sample $\Poi(m)$ symbols in total, the expected number of
symbols sampled from each such range is $\Theta(1)$, which makes the analysis tractable, and adding
up the contribution from each of these ranges to the mutual information gives the desired bound.

\subsection{Construction of YES and NO Distributions}
We now formally define dominoes, subtraces, and the YES and NO distributions.

\begin{definition}[Dominoes]
    For any integer $i \in [n/2]$, let $j = 2i-1$ and $j' = 4i-3$.
    We call the range $\{j',j'+1,j'+2,j'+3\}$ of the domain $[2n]$, along with the probability
    masses of $p$ and $q$ associated with these positions (namely $p_{j}, q_{j}, p_{j+1}, q_{j+1}$)
    the \emph{$i$-th domino}.

    In particular, we categorize dominoes as one of three kinds according to the probability
    masses of its $p$ entries (which will be chosen differently under the YES and NO distributions):
    \begin{enumerate}
        \item \emph{Unbiased}: when $p_{j} = p_{j+1} = \frac{1}{2n}$.
        \item \emph{Left $\epsilon$-biased}: when $p_{j} = \frac{1}{2n}(1 + \epsilon)$ and
            $p_{j+1} = \frac{1}{2n}(1 - \epsilon)$.
        \item \emph{Right $\epsilon$-biased}: when $p_{j} = \frac{1}{2n}(1 - \epsilon)$ and
            $p_{j+1} = \frac{1}{2n}(1 + \epsilon)$.
    \end{enumerate}
\end{definition}

\begin{definition}[Subtraces]
    \label{def:subtrace}

    Given a probability distribution $\pi(p,q)$ over $[2n]$, and for each $i \in [n/2]$, we say
    that the \emph{subtrace produced by the $i$-th domino} is the random binary string
    \[
        \bm{t}_i \define 1^{\bm{A}_j} 0^{\bm{B}_j} 1^{\bm{A}_{j+1}} 0^{\bm{B}_{j+1}} \,,
    \]
    where $j = 2i-1$ and $\bm{A}_k \sim \Poi(m p_k), \bm{B}_k \sim \Poi(m q_k)$ independently.

    Given a contiguous range of $r$ dominoes indexed by $\{i, i+1, \dotsc, i+r-1\}$, the subtrace
    produced by this range of dominoes is
    \[
        \bm{T}_{i,r} \define \bm{t}_i \circ \dotsc \circ \bm{t}_{i+r-1} \,,
    \]
    where $\circ$ stands for concatenation.
\end{definition}

\begin{observation}
    \label{observation:domino-subtrace-length}
    Recall that any domino has $q_{j} = q_{j+1} = 1/2n$, \ie the partial distribution over the
    0-valued elements is uniform with total mass $1/2$. Therefore each domino satisfies
    \[
        p_{j} + q_{j} + p_{j+1} + q_{j+1} = 2/n \,,
    \]
    and therefore the length of the subtrace produced by each domino is independently distributed
    as $\Poi(2m/n)$ regardless of the value of $\bm{Z}$.
\end{observation}

Using the definitions above, we can see that the full trace $\cT$ is distributed as
\[
    \bm{\cT} = \bm{t}_1 \circ \dotsm \circ \bm{t}_{n/2} \,.
\]
Alternatively, if we partition the set of all dominoes into contiguous ranges
$\{i_1, \dotsc, i_1 + r_1 - 1\}, \dotsc, \{i_k, \dotsc, i_k + r_k - 1\}$, then
\[
    \bm{\cT} = \bm{T}_{i_1,r_1} \circ \dotsm \circ \bm{T}_{i_k,r_k} \,.
\]

\noindent
We now define the YES and NO distributions.

\begin{definition}[YES and NO distributions]
    \label{def:yes-no-distributions}
    Let $\epsilon > 0$.
    The \emph{{\normalfont YES} distribution} $\cD_0$ is a distribution supported on a single element
    $\pi(\mu,\mu)$, the uniform distribution over $[2n]$. The \emph{{\normalfont NO} distribution}
    $\cD_1$ is a distribution supported on distributions over $[2n]$ drawn as follows: for each
    $i \in [n/2]$, make the $i$-th domino left $\epsilon$-biased or right $\epsilon$-biased, with
    equal probability independently for each domino.

    When we are thinking of the distribution $\pi(p,q)$ as a random variable drawn from these
    distributions, we will accordingly write $\bm{\pi} = \pi(\bm{p}, \bm{q})$.
\end{definition}

\begin{observation}
    Every $\pi(p,q)$ in the support of $\cD_1$ satisfies
    $\dist_\TV(\pi(p,q), \pi(\mu,\mu)) = \epsilon/4$.
\end{observation}

Therefore, we seek to show the following result:

\begin{claim}
    \label{claim: YES vs NO}
    Let $\bm{Z} \sim \Ber(1/2)$, and let $\bm{\cT}$ be a parity trace of size $\Poi(m)$ sampled from
    distribution $\bm{\pi}$ over $[2n]$, where $\bm{\pi} \sim \cD_{\bm{Z}}$. Then if
    $m = o\left( \left(\frac{n}{\epsilon}\right)^{4/5} \frac{1}{\log^4 n} \right)$, it follows
    that $I(\bm{Z} : \bm{\cT}) = o(1)$.
\end{claim}

\subsection{Partial Fingerprints and their Probabilities}

In the standard model of distribution testing, the \emph{fingerprint} of a sample is a complete
description of the relevant information for testing symmetric properties of discrete
distributions~\cite{Bat01}. The fingerprint is the ``histogram of the histogram'': for each positive
integer $k$, the number of elements that occurred exactly $k$ times in the sample.

In our construction, we would like to analyze the fingerprint over the \emph{dominoes}, as follows.
For each $i$, let $d_i$ be the number of trace symbols produced from the $i$-th domino.
Then $d$ is our histogram and the corresponding fingerprint counts, for each positive integer $k$,
how many trace symbols came from dominoes satisfying $d_i=k$.

Studying the fingerprint over the dominoes is useful because, as our analysis will show,
each domino is uninformative about $\bm{Z}$ when at most 2 symbols are sampled it, and
when 3 or more symbols are sampled, the amount of information revealed grows according to the number
of symbols. This phenomenon suggests that we consider a \emph{partial fingerprint}, which is
obtained from the fingerprint by collapsing the counts corresponding to all integers $k \ge 3$
into a single category ``$3^+$''. We give the following equivalent formulation, which is more
convenient for our analysis:

\begin{definition}[$(h,k,s)$-collisions]
    \label{def:collisions}
    Consider the process of throwing $b$ identical balls into $r$ bins, each ball at a bin
    selected independently uniformly at random. A vector $d = (d_1, \dotsc, d_r) \in \bZ_{\ge 0}^r$
    such that $d_i$ is the number of balls in the $i$-th bin, for each $i \in [r]$, is called
    an \emph{outcome} of this process.

    We say that outcome $d$ is an \emph{$(h,k,s)$-collision} if, among the $r$ bins,
    exactly $h$ of them contain exactly two balls,
    exactly $k$ of them contain at least three balls and, moreover, the total number of balls in
    those $k$ bins is $k+s$ (in other words, $s$ is the number of ``surplus'' balls in the bins
    with at least 3 balls). We define $\cC_{r,b}(h,k,s)$, the
    \emph{set of $(h,k,s)$-collision outcomes}, as
    \[
        \cC_{r,b}(h, k, s) \define \left\{
            \begin{aligned}
                &(d_1, \dotsc, d_r) \in \bZ_{\ge 0}^r : \\
                &\qquad
                \sum_{i=1}^r d_i = b,
                \sum_{i=1}^r \ind{d_i = 2} = h,
                \sum_{i=1}^r \ind{d_i \ge 3} = k,
                \sum_{i=1}^r \ind{d_i \ge 3} \cdot (d_i-1) = s
            \end{aligned}
            \right\} \,.
    \]
    Note that $\cC_{r,b}(h, k, s) = \emptyset$ whenever $s < 2k$,
    and similarly when $b < 2h$ or $b < k+s$.
\end{definition}

Note that the random vector $\bm{d}$ expressing the outcome of the random process described above
is distributed as $\bm{d} \sim \Multinomial(b, r, (1/r, \dotsc, 1/r))$.
We show that, for appropriate range of values, $(h,k,s)$-collisions are exponentially
unlikely in $h$ and $s$:

\begin{proposition}
    \label{prop: collisions unlikely}
    Let $\bm{d} = (\bm{d}_1, \dotsc, \bm{d}_r) \sim \Multinomial(b, r, (1/r, \dotsc, 1/r))$.
    Then for each $h \ge 0$ and $k, s > 0$, we have
    \[
        \Pr{\bm{d} \in C_{r,b}(h, k, s)} \le
        \left( \frac{(bh)^2}{r} \right)^h
        \left( \frac{(bk)^{3/2}}{r} \right)^s
        \,,
    \]
    where $0^0$ is interpreted as $1$.
\end{proposition}
\begin{proof}
    We can upper bound this probability by ranging over which bins will contain exactly two
    balls, if any---call these ``$2$-collisions''---and three or more balls---call
    these ``$3^+$-collisions''---, as well as which balls fall into those bins,
    and then roughly upper bounding the combinatorial quantities determining each.
    Let notation $\binom{[n]}{m}$ denote the set of subsets of $[n]$ of size $m$. We have
    \begin{align*}
        &\Pr{\bm{d} \in C_{r,b}(h, k, s)} \\
        &= \sum_{I_2 \in \binom{[r]}{h}} \sum_{I_3 \in \binom{[r] \setminus I_2}{k}}
            \sum_{J_2 \in \binom{[b]}{2h}} \sum_{J_3 \in \binom{[b] \setminus J_2}{k+s}}
            \left[
                \begin{array}{l}
                    \Pr{\text{balls $J_2$ form $2$-collisions on bins $I_2$}} \\
                    \cdot \Pr{\text{balls $J_3$ form $3^+$-collisions on bins $I_3$}} \\
                    \cdot \Pr{\text{balls $[b] \setminus (J_2 \cup J_3)$ fall on bins
                        $[r] \setminus (I_2 \cup I_3)$ without collisions}}
                \end{array}
                \right]
                    \\
        &\le \sum_{I_2 \in \binom{[r]}{h}} \sum_{I_3 \in \binom{[r] \setminus I_2}{k}}
            \sum_{J_2 \in \binom{[b]}{2h}} \sum_{J_3 \in \binom{[b] \setminus J_2}{k+s}}
                \Pr{\text{balls $J_2$ fall within bins $I_2$}}
                \Pr{\text{balls $J_3$ fall within bins $I_3$}} \\
        &\le \binom{r}{h} \binom{r}{k} \binom{b}{2h} \binom{b}{k+s}
                \left(\frac{h}{r}\right)^{2h} \left(\frac{k}{r}\right)^{k+s} \\
        &\le r^{h+k} b^{2h+k+s} h^{2h} k^{k+s} r^{-2h-k-s} \\
        &\le (bh)^{2h} r^{-h} \cdot (bk)^{\frac{3}{2}s} r^{-s}
            \hspace{25em} \text{(Since $k \le s/2$)} \\
        &= \left( \frac{(bh)^2}{r} \right)^h \left( \frac{(bk)^{3/2}}{r} \right)^s \,,
    \end{align*}
    where we used $k \le s/2$ which holds unless the probability is zero, in which case the
    conclusion follows trivially.
\end{proof}

We will also need the following simple ``birthday problem'' bound:

\begin{proposition}
    \label{prop: collision free}
    Let $\bm{d} = (\bm{d}_1, \dotsc, \bm{d}_r) \sim \Multinomial(b, r, (1/r, \dotsc, 1/r))$.
    Then the probability of seeing no collisions satisfies
    \[
        \Pr{\bm{d} \in \cC_{r,b}(0,0,0)} \ge 1 - \frac{b^2}{r} \,.
    \]
\end{proposition}
\begin{proof}
    This probability is
    \[
        \Pr{\bm{d} \in \cC_{r,b}(0,0,0)}
        = \frac{r \cdot (r - 1) \dotsm (r - b + 1)}{r^b}
        \ge \left( \frac{r - b}{r} \right)^b
        = \left( 1 - \frac{b}{r} \right)^b
        \ge 1 - \frac{b^2}{r} \,.
    \]
\end{proof}

\subsection{YES and NO Dominoes Behave Similarly}

We now show that each subtrace that is not too long must have similar probabilities of being
produced by a domino under the YES and NO distributions.

We first need the following simple bound, which informally encapsulates the property that the
``information'' revealed by a domino decays as $O(\epsilon^2)$ even though its relative probability
masses are $\Theta(\epsilon)$-biased.

\begin{proposition}
    \label{prop: epsilon squared bound}
    For all non-negative integers $x$ and $y$, and all $0 < \epsilon < 1$,
    \[
        \frac{1}{2}(1+\epsilon)^x(1-\epsilon)^y + \frac{1}{2}(1-\epsilon)^x(1+\epsilon)^y
        = 1 \pm \epsilon^2 \cdot 2^{x+y} \,.
    \]
\end{proposition}
\begin{proof}
    By the binomial theorem, we have
    \begin{align*}
        &\frac{1}{2}(1+\epsilon)^x(1-\epsilon)^y + \frac{1}{2}(1-\epsilon)^x(1+\epsilon)^y \\
        &\qquad = \frac{1}{2}
            \left(\sum_{i=0}^x \binom{x}{i} \epsilon^i\right)
            \left(\sum_{j=0}^y \binom{y}{j} (-1)^j \epsilon^j\right)
         +
            \frac{1}{2}
            \left(\sum_{i=0}^x \binom{x}{i} (-1)^i \epsilon^i\right)
            \left(\sum_{j=0}^y \binom{y}{j} \epsilon^j\right) \\
        &\qquad = \sum_{i=0}^x \sum_{j=0}^y \binom{x}{i} \binom{y}{j} \epsilon^{i+j}
            \left( \frac{(-1)^i + (-1)^j}{2} \right) \\
        &\qquad = 1 \pm \epsilon^2
            \left(\sum_{i=0}^x \binom{x}{i}\right) \left(\sum_{j=0}^y \binom{y}{j}\right) \\
        &\qquad = 1 \pm \epsilon^2 \cdot 2^{x+y} \,,
    \end{align*}
    where the third (in)equality holds because when $i=j=0$ the entire inner expression is equal to
    $1$, when $i+j=1$ it is zero since $i$ and $j$ have different parities, and otherwise we have
    $\epsilon^{i+j} \le \epsilon^2$.
\end{proof}

\begin{lemma}
    \label{lemma:domino-subtrace}
    Let $\epsilon \in (0,1)$.
    Let $\bm{t} = \bm{t}_i$ denote the random variable corresponding to the subtrace produced by
    a domino. Then for any binary string $t$, we have the following two cases:
    \begin{enumerate}
        \item If $t$ contains at least one ``0'' symbol and two ``1'' symbols, then
            \[
                \Pruc{}{\bm{t}=t}{|\bm{t}|=|t|, \bm{Z}=1} =
                \left( \Pruc{}{\bm{t}=t}{|\bm{t}|=|t|, \bm{Z}=0} \right)
                \left( 1 \pm \epsilon^2 \cdot 2^{|t|-1} \right) \,.
            \]
        \item Otherwise,
            \[
                \Pruc{}{\bm{t}=t}{|\bm{t}|=|t|, \bm{Z}=1}
                = \Pruc{}{\bm{t}=t}{|\bm{t}|=|t|, \bm{Z}=0} \,.
            \]
    \end{enumerate}
    Note that the probabilities are taken over the choice of distribution
    $\bm{\pi} \sim \cD_{\bm{Z}}$ and random vectors $\bm{A}, \bm{B}$ corresponding to the
    (Poissonized) trace from $\bm{\pi}$.
\end{lemma}
\begin{proof}
    Note that if $t$ is not in the regular language $1^*0^*1^*0^*$, all probabilities above are
    zero (since such trace cannot be produced by a domino) and the claim holds.
    Suppose $t$ has form $1^*0^*1^*0^*$.

    Without loss of generality, say $\bm{t}$ is the trace produced by the first domino, so that
    $\bm{t} = 1^{\bm{A}_1} 0^{\bm{B}_1} 1^{\bm{A}_2} 0^{\bm{B}_2}$ for
    $\bm{A}_1 \sim \Poi(m \bm{p}_1), \bm{A}_2 \sim \Poi(m \bm{p}_2),
    \bm{B}_1 \sim \Poi(m \bm{q}_1), \bm{B}_2 \sim \Poi(m \bm{q}_2)$, where $\bm{p},\bm{q}$ are
    the partial distributions of $\bm{\pi}$ and therefore $\bm{q}_1 = \bm{q}_2 = 1/2n$.
    Note that $|\bm{t}| = \bm{A}_1 + \bm{B}_1 + \bm{A}_2 + \bm{B}_2$.

    By standard arguments, once we condition on $|\bm{t}|= |t|$, the random variables
    $\bm{A}_i$ and $\bm{B}_i$ are distributed according to a multinomial distribution given by
    $|t|$ trials and $4$ bins with probabilities corresponding to the relative weights of the
    probability masses on each position:
    \begin{gather*}
        \Pruc{}{\bm{A}_1=a_1, \bm{B}_1=b_1, \bm{A}_2=a_2, \bm{B}_2=b_2}
            {\bm{A}_1 + \bm{B}_1 + \bm{A}_2 + \bm{B}_2 = |t|} \\
        = \Pr{(\bm{X}_1, \bm{Y}_1, \bm{X}_2, \bm{Y}_2) = (a_1, b_1, a_2, b_2)} \,, \\
        (\bm{X}_1, \bm{Y}_1, \bm{X}_2, \bm{Y}_2) \sim \Multinomial\left(
        |t|, \left(\frac{n}{2} \bm{p}_1, \frac{1}{4}, \frac{n}{2} \bm{p}_2, \frac{1}{4} \right)
        \right)
        \,.
    \end{gather*}
    For convenience,
    let $\bm{p}'_1 \define \frac{n}{2} \bm{p}_1$ and $\bm{p}'_2 \define \frac{n}{2} \bm{p}_2$.
    Note that, in the YES case ($\bm{Z}=0$), we have $\bm{p}'_1 = \bm{p}'_2 = \frac{1}{4}$,
    and in the NO case ($\bm{Z}=1$), we have one of the following with equal probability:
    \begin{enumerate}
        \item $\bm{p}'_1 = \frac{1}{4}(1 + \epsilon)$ and $\bm{p}'_2 = \frac{1}{4}(1 - \epsilon)$;
        \item $\bm{p}'_1 = \frac{1}{4}(1 - \epsilon)$ and $\bm{p}'_2 = \frac{1}{4}(1 + \epsilon)$.
    \end{enumerate}
    We now prove the claim. We start with the second case, which is simpler.
    First, suppose $t$ contains no ``0'' symbols. Then since the total weight of the 1-valued
    positions is $\bm{p}'_1 + \bm{p}'_2 = 1/2$ regardless of the value of $\bm{Z}$, we have
    \[
        \Pruc{}{\bm{t}=t}{|\bm{t}|=|t|,\bm{Z}=1} = \Pruc{}{\bm{t}=t}{|\bm{t}|=|t|,\bm{Z}=0}
        = \left(\frac{1}{2}\right)^{|t|} \,.
    \]
    On the other hand, suppose $t$ contains at most one ``1'' symbol. If it contains no ``1'' symbols,
    the same logic applies, so we can assume that $t$ contains exactly one ``1'' symbol.
    We may write the probability of $t$ as the sum of the probabilities of all
    $(a_1, b_1, a_2, b_2)$ that produce $t$ as a binary string, \ie
    $t = 1^{a_1} 0^{b_1} 1^{a_2} 0^{b_2}$. Let $\cS$ denote the set of such tuples that produce $t$.
    Using the multinomial formulation, this yields, for any possible values
    $p'_1, p'_2$ of $\bm{p}'_1, \bm{p}'_2$,
    \[
        \Pruc{}{\bm{t}=t}{|\bm{t}|=|t|, \bm{p}'_1 = p'_1, \bm{p}'_2 = p'_2}
        = \sum_{(a_1, b_1, a_2, b_2) \in \cS}
            \frac{|t|!}{a_1! b_1! a_2! b_2!}
            \left(p'_1\right)^{a_1}
            \left(\frac{1}{4}\right)^{b_1}
            \left(p'_2\right)^{a_2}
            \left(\frac{1}{4}\right)^{b_2} \,.
    \]
    Since $t$ contains exactly one ``1'' symbol---say $t = 0^x 1 0^y$ for some $x, y \ge 0$---,
    we have that $1^{a_1} 0^{b_1} 1^{a_2} 0^{b_2} = t$ if and only if
    \begin{enumerate}
        \item $a_1 = 0$, $a_2 = 1$, $b_1 = x$, and $b_2 = y$; or, mutually exclusively,
        \item $a_1 = 1$, $a_2 = 0$, $b_1 + b_2 = y$, and $x = 0$.
    \end{enumerate}
    Thus we may write the probability of $t = 0^x 1 0^y$ as
    \[
        \Pruc{}{\bm{t}=t}{|\bm{t}|=|t|, \bm{p}'_1 = p'_1, \bm{p}'_2 = p'_2}
        = \frac{(x+y+1)!}{x!y!} \left(\frac{1}{4}\right)^{x+y} p'_2
            + \ind{x=0} \frac{(y+1)!}{y!} \left(\frac{1}{2}\right)^y p'_1 \,.
    \]
    We verify that, when $\bm{Z}=1$, the cases where $\bm{p}'_1$ and $\bm{p}'_2$ are positively
    and negatively biased cancel out and we obtain the same probability as when $\bm{Z}=0$:
    \begin{align*}
        &\Pruc{}{\bm{t}=t}{|\bm{t}|=|t|, \bm{Z}=1} \\
        &\qquad =
            \frac{1}{2} \Pruc{}{\bm{t}=t}{|\bm{t}|=|t|,
                \bm{p}'_1 = \frac{1}{4}(1-\epsilon), \bm{p}'_2 = \frac{1}{4}(1+\epsilon)} \\
            &\qquad \qquad +
                \frac{1}{2} \Pruc{}{\bm{t}=t}{|\bm{t}|=|t|,
                    \bm{p}'_1 = \frac{1}{4}(1+\epsilon), \bm{p}'_2 = \frac{1}{4}(1-\epsilon)} \\
        &\qquad =
            \frac{1}{2} \Bigg[
                \frac{(x+y+1)!}{x!y!} \left(\frac{1}{4}\right)^{x+y} \frac{1}{4}(1+\epsilon)
                + \ind{x=0} \frac{(y+1)!}{y!} \left(\frac{1}{2}\right)^y \frac{1}{4}(1-\epsilon)
                \Bigg] \\
            &\qquad \qquad +
            \frac{1}{2} \Bigg[
                \frac{(x+y+1)!}{x!y!} \left(\frac{1}{4}\right)^{x+y} \frac{1}{4}(1-\epsilon)
                    + \ind{x=0} \frac{(y+1)!}{y!} \left(\frac{1}{2}\right)^y \frac{1}{4}(1+\epsilon)
            \Bigg] \\
        &\qquad =
            \frac{(x+y+1)!}{x!y!} \left(\frac{1}{4}\right)^{x+y} \left(\frac{1}{4}\right)
                + \ind{x=0} \frac{(y+1)!}{y!} \left(\frac{1}{2}\right)^y \left(\frac{1}{4}\right) \\
        &\qquad = \Pruc{}{\bm{t}=t}{|\bm{t}|=|t|, \bm{Z}=0} \,,
    \end{align*}
    completing the proof of the second case of the claim.

    Let us return to the first case. Suppose $t$ contains at least one ``0'' symbol and two ``1''
    symbols; say $t = 1^x 0^z 1^y 0^w$ with $x+y \ge 2$ and $z+w \ge 1$.
    We start with the general multinomial formulation again: let $\cS$ be the set of tuples
    $(a_1, b_1, a_2, b_2)$ satisfying $t = 1^{a_1} 0^{b_1} 1^{a_2} 0^{b_2}$. We have
    \[
        \Pruc{}{\bm{t}=t}{|\bm{t}|=|t|, \bm{p}'_1 = p'_1, \bm{p}'_2 = p'_2}
        = \sum_{(a_1, b_1, a_2, b_2) \in \cS}
            \frac{|t|!}{a_1! b_1! a_2! b_2!}
            \left(\frac{1}{4}\right)^{b_1}
            \left(\frac{1}{4}\right)^{b_2}
            \cdot
            \left(p'_1\right)^{a_1}
            \left(p'_2\right)^{a_2} \,.
    \]
    Define $F_{a_1,b_1,a_2,b_2} \define 
            \frac{|t|!}{a_1! b_1! a_2! b_2!}
            \left(\frac{1}{4}\right)^{a_1}
            \left(\frac{1}{4}\right)^{b_1}
            \left(\frac{1}{4}\right)^{a_2}
            \left(\frac{1}{4}\right)^{b_2}
            $,
    so that
    \[
        \Pruc{}{\bm{t}=t}{|\bm{t}|=|t|, \bm{Z}=0}
        = \sum_{(a_1, b_1, a_2, b_2) \in \cS} F_{a_1,b_1,a_2,b_2}
    \]
    and
    \begin{align*}
        \Pruc{}{\bm{t}=t}{|\bm{t}|=|t|, \bm{Z}=1}
        &=
            &\frac{1}{2}
                \sum_{(a_1, b_1, a_2, b_2) \in \cS}
                \frac{|t|!}{a_1! b_1! a_2! b_2!}
                \left(\frac{1}{4}\right)^{b_1}
                \left(\frac{1}{4}\right)^{b_2}
                \cdot
                \left(\frac{1}{4}(1 + \epsilon)\right)^{a_1}
                \left(\frac{1}{4}(1 - \epsilon)\right)^{a_2} \\
            &\quad + &\frac{1}{2}
                \sum_{(a_1, b_1, a_2, b_2) \in \cS}
                \frac{|t|!}{a_1! b_1! a_2! b_2!}
                \left(\frac{1}{4}\right)^{b_1}
                \left(\frac{1}{4}\right)^{b_2}
                \cdot
                \left(\frac{1}{4}(1 - \epsilon)\right)^{a_1}
                \left(\frac{1}{4}(1 + \epsilon)\right)^{a_2} \\
        &=
            &\sum_{(a_1,b_1,a_2,b_2) \in \cS}
            F_{a_1,b_1,a_2,b_2} \left(
                \frac{1}{2} (1+\epsilon)^{a_1} (1-\epsilon)^{a_2} +
                \frac{1}{2} (1-\epsilon)^{a_1} (1+\epsilon)^{a_2}
                \right)
            \,.
    \end{align*}
    Thus it suffices to show that for every $(a_1, b_1, a_2, b_2) \in \cS$,
    \[
        \frac{1}{2} (1+\epsilon)^{a_1} (1-\epsilon)^{a_2}
            + \frac{1}{2} (1-\epsilon)^{a_1} (1+\epsilon)^{a_2}
        \eqquestion 1 \pm \epsilon^2 \cdot 2^{|t|-1} \,,
    \]
    and since $a_1 + a_2 \le |t|-1$ (because $t$ contains at least one ``0'' symbol), this follows
    from \cref{prop: epsilon squared bound}, completing the proof.
\end{proof}

\subsection{Information Bound}

Recall that we wish to upper bound the mutual information $I(\bm{Z} : \bm{\cT})$, which we will do,
using the chain rule of mutual information, by summing over the quantities
$I(\bm{Z} : \bm{T}_{i_j,r_j})$ where each $\{i_j, i_j+1, \dotsc, i_j+r_j-1\}$ is a contiguous range
of dominoes (forming a partition) and $\bm{T}_{i_j,r_j}$ is the subtrace produced by such a range.
For simplicity, let $\bm{T} = \bm{T}_{i_j,r_j}$ denote one such variable.
Let $P_0$ and $P_1$ be the conditional distributions of $\bm{T}$ under each value of $\bm{Z}$:
for each binary string $T$ and $z \in \zo$, $P_z(T) \define \Pruc{}{\bm{T}=T}{\bm{Z}=z}$.

The following fact states that, if the pointwise ratios between $P_1$ and $P_0$ are close to $1$,
then the mutual information $I(\bm{Z} : \bm{T})$ is small.
Since the full argument will require a refined version
that also handles low-probability outcomes (for which the ratio bound may fail), we state this
fact without proof for intuition only. The formulation is inspired by \cite{DK16}.

\begin{fact}
    \label{fact: discrepancy to information}
    Let $P_0$ and $P_1$ be discrete probability distributions over some domain $\cX$.
    Let $\xi > 0$ and suppose that, for every $T \in \cX$, it holds that
    \[
        \frac{P_1(T)}{P_0(T)} = 1 \pm \xi \,.
    \]
    Then we have
    \[
        \chi^2(P_1 \| P_0) \le \xi^2 \,,
    \]
    where
    $\chi^2(P_1 \| P_0) = \Exu{\bm{T} \sim P_0}{\left(\frac{P_1(\bm{T})}{P_0(\bm{T})} - 1\right)^2}$
    is the Pearson $\chi^2$-divergence.
    Moreover, if $\bm{Z}$ is a uniform random bit and $\bm{T}$ is distributed according to
    $P_{\bm{Z}}$, then
    \[
        I(\bm{Z} : \bm{T}) \le \frac{1}{2} \chi^2(P_1 \| P_0) \le \frac{1}{2} \xi^2 \,.
    \]
\end{fact}

Therefore, our task is to upper bound $\abs*{\frac{P_1(T)}{P_0(T)} - 1}$.
The following result accomplishes this for any range of $r$ dominoes and string $T$ that is not
too long compared to $r$. Later, we will see that strings $T$ that are too long are so unlikely
that they have little effect on the mutual information.

\begin{lemma}
    \label{lemma: range probs}
    There exists a universal constant $c > 0$ such that the following holds.
    Let $\epsilon \in (0,1)$, and let $r \ge 1$ be an integer.
    Let $\bm{T}$ denote the subtrace produced by a range of $r$ consecutive dominoes, and let
    $P_z$ be the probability distribution of $\bm{T}$ conditional on $\bm{Z}=z$ as above.
    Then for any binary string $T$ satisfying $|T|^4 \le r / 100$, we have
    \[
        \frac{P_1(T)}{P_0(T)} = 1 \pm c \cdot \epsilon^2 \cdot \frac{|T|^6}{r^2} \,.
    \]
\end{lemma}
\begin{proof}
    Let us denote by $g_z(t)$ the probability, as in
    \cref{lemma:domino-subtrace}, that any given domino produces subtrace $t$ conditional on
    $\bm{Z}=z$ and the length of the subtrace:
    $g_z(t) \define \Pruc{}{\bm{t}=t}{|\bm{t}|=|t|, \bm{Z}=z}$ where $\bm{t}$ is the random variable
    corresponding the subtrace from the domino under consideration. Recall that, by definition of
    the dominoes, the probabilities $g_z(t)$ are the same for every domino.

    For each $i \in [r]$, let $\bm{D}_i$ be the random variable corresponding to the length of the
    subtrace produced by the $i$-th domino in the range. As noted in
    \cref{observation:domino-subtrace-length}, the additive property of Poisson random variables
    and the construction of dominoes implies that $\bm{D}_i \sim \Poi(2m/n)$ for all $i$
    independently. For convenience, let $\lambda \define 2m/n$.

    For $T$ to be the trace produced by the range under consideration, each domino in this range
    must produce a subtrace
    in such a way that 1) the total length of all subtraces is $|T|$; and 2) the subtrace from each
    domino is equal to the appropriate substring of $T$. Toward this end, let $\cM$ denote the set
    of vectors of subtrace lengths that add up to $|T|$:
    \[
        \cM \define \left\{ d \in \bZ_{\ge 0}^r : \sum_{i=1}^r d_i = |T| \right\} \,.
    \]
    Recalling \cref{def:collisions}, we may write $\cM$ as the disjoint union
    \[
        \cM = \biguplus_{h, k, s \ge 0} \cC_{r,|T|}(h,k,s) \,.
    \]
    We will use the following notation to refer to substrings of $T$. For indices
    $1 \le a, b \le |T|$, let $T[a..b]$ denote the substring of $T$ between indices $a$ and $b$
    (inclusive) when $a \le b$, and set $T[a..b] \define \varnothing$ when $a>b$.
    For a histogram $d \in \cM$ and for each $i \in [r]$, set
    \[
        T(d, i) \define T\left[ \left(1+\sum_{j=1}^{i-1} d_j\right)
            \, .. \, \left(1+\sum_{j=1}^{i-1} d_j \right) + d_i - 1 \right] \,,
    \]
    Then for all $d \in \cM$, $T$ is equal to the concatenation $T(d,1) \circ \dotsm \circ T(d,r)$.

    We now have, for each $z \in \{0,1\}$,
    \begin{align*}
        P_z(T)
        &= \sum_{d \in \cM} \prod_{i=1}^r \Pr{\bm{D}_i = d_i} g_z(T(d, i))
        = \sum_{d \in \cM} \prod_{i=1}^r
            \frac{e^{-\lambda} \lambda^{d_i}}{d_i!}
            g_z(T(d, i)) \\
        &= \frac{e^{-r\lambda} (r\lambda)^{|T|}}{|T|!} \sum_{d \in \cM} \left[
            \left( \frac{|T|!}{d_1! \dotsm d_r!} \left(\frac{1}{r}\right)^{|T|} \right)
            \left( \prod_{i=1}^r g_z(T(d, i)) \right)
            \right] \,. \\
    \end{align*}
    Notice that the first factor inside the summation is a multinomial probability: letting
    $\bm{d} = (\bm{d}_1, \dotsc, \bm{d}_r) \sim \Multinomial(|T|, r, (1/r, \dotsc, 1/r))$,
    the first factor is precisely $\Pr{\bm{d} = d}$. This is the ``balls and bins'' process
    introduced in \cref{def:collisions}.\footnote{We have essentially ``factored out'' the
    Poissonization for this part of the analysis.}

    As for the second factor, for each $d$ and $i$ define
    \[
        \delta(d, i) \define \frac{g_1(T(d, i))}{g_0(T(d, i))} - 1 \,,
    \]
    so that $g_1(T(d, i)) = (1 + \delta(d, i)) g_0(T(d, i))$.
    By slightly loosening \cref{lemma:domino-subtrace} for simplicity, we may bound each
    $\delta(d, i)$ as follows:
    \begin{align*}
        d_i \le 2 &\implies \delta(d, i) = 0 \,, \text{ and} \\
        d_i \ge 3 &\implies \abs*{\delta(d, i)} \le \epsilon^2 \cdot 2^{d_i-1} \,.
    \end{align*}
    We then obtain
    \begin{align*}
        P_z(T)
        &= \frac{e^{-r\lambda} (r\lambda)^{|T|}}{|T|!} \sum_{d \in \cM}
            \Pr{\bm{d} = d} \prod_{i=1}^r g_0(T(d, i)) (1 + z\delta(d, i)) \\
        &= \frac{e^{-r\lambda} (r\lambda)^{|T|}}{|T|!}
            \sum_{h, k, s \ge 0} \sum_{d \in \cC_{r, |T|}(h, k, s)}
                \Pr{\bm{d}=d}
                \left( \prod_{i=1}^r g_0(T(d, i)) \right)
                \left( \prod_{i=1}^r \left( 1 \pm z \ind{d_i \ge 3} \epsilon^2 2^{d_i-1} \right) \right)
                \,.
    \end{align*}
    For any $d \in \cC_{r, |T|}(h, k, s)$, the term
    $\prod_{i=1}^r \left( 1 \pm z \ind{d_i \ge 3} \epsilon^2 2^{d_i-1} \right)$ is a product in
    which all but $k$ terms are simply $1$, since only $k$ entries $d_i$ may satisfy
    $d_i \ge 3$ by definition of $(h,k,s)$-collision.
    Therefore, upon expanding this product, we obtain $2^k$ terms; one of them is $1$, and
    $2^k-1$ of them each contain at least one $\epsilon^2$ factor and a $2^x$ factor for some
    $x \le \sum_{i=1}^r \ind{d_i \ge 3} (d_i-1)$.
    Thus, using the identity $\sum_{i=1}^r \ind{d_i \ge 3}(d_i - 1) = s$ from the definition of
    $(h,k,s)$-collision, we obtain
    \[
        \prod_{i=1}^r \left( 1 \pm z \ind{d_i \ge 3} \epsilon^2 2^{d_i-1} \right)
        = 1 \pm z \epsilon^2 (2^k - 1) 2^s \,.
    \]
    Therefore we can write
    \begin{align*}
        P_z(T)
        &= \frac{e^{-r\lambda} (r\lambda)^{|T|}}{|T|!}
            \sum_{h, k, s \ge 0}
                \left( 1 \pm z \epsilon^2 (2^k-1) 2^s \right)
                \sum_{d \in \cC_{r, |T|}(h, k, s)}
                    \Pr{\bm{d}=d}
                    \left( \prod_{i=1}^r g_0(T(d, i)) \right) \,.
    \end{align*}
    Recall that we want to show that
    $P_1(T) / P_0(T) = 1 \pm \epsilon^2 \cdot O\left(|T|^6 / r^2\right)$. Substituting the
    formulation above, we obtain
    \begin{align*}
        \frac{P_1(T)}{P_0(T)}
        &=
        \frac{
            \sum_{h, k, s \ge 0}
                \left( 1 \pm \epsilon^2 (2^k-1) 2^s \right)
                \sum_{d \in \cC_{r, |T|}(h, k, s)}
                    \Pr{\bm{d}=d}
                    \left( \prod_{i=1}^r g_0(T(d, i)) \right)
                }{
            \sum_{h, k, s \ge 0}
                \sum_{d \in \cC_{r, |T|}(h, k, s)}
                    \Pr{\bm{d}=d}
                    \left( \prod_{i=1}^r g_0(T(d, i)) \right)
                } \\
        &=
        1 \pm
        \sum_{h, k, s \ge 0} \left[ \epsilon^2 (2^k-1) 2^s \cdot \frac{
                \sum_{d \in \cC_{r, |T|}(h, k, s)}
                    \Pr{\bm{d}=d}
                    \left( \prod_{i=1}^r g_0(T(d, i)) \right)
                }{
            \sum_{d \in \cM}
                    \Pr{\bm{d}=d}
                    \left( \prod_{i=1}^r g_0(T(d, i)) \right)
                } \right] \,.
    \end{align*}
    Therefore, our goal is to show the following:
    \begin{equation}
        \label{eq:lb-desired-ineq}
        \sum_{h \ge 0, s \ge 2k \ge 2} \epsilon^2 (2^k-1) 2^s \cdot
            \frac{\sum_{d \in \cC_{r,|T|}(h,k,s)} \Pr{\bm{d}=d} \prod_{i=1}^r g_0(T(d, i))}{
                \sum_{d \in \cM} \Pr{\bm{d}=d} \prod_{i=1}^r g_0(T(d, i))}
        \lequestion c \cdot \epsilon^2 \cdot \frac{|T|^6}{r^2} \,,
    \end{equation}
    where we used the fact that $\epsilon^2 (2^k-1) 2^s = 0$ when $k=0$ to limit the range
    of $k$ in the summation to $k \ge 1$, and then used the fact that $s \ge 2k$ for any
    nonempty $\cC_{r,|T|}(h,k,s)$ to limit the range of $s$.

    First, note that for any single-character binary string $t$ (\ie strings ``0'' and ``1''),
    we have $g_0(t) = 1/2$. We may hence lower bound the denominator of \eqref{eq:lb-desired-ineq}
    by counting only those $d \in \cM$ that have no collisions at all
    (\ie $d \in \cC_{r,|T|}(0,0,0)$), whose total probability is lower bounded by
    \cref{prop: collision free}:
    \[
        \sum_{d \in \cM} \Pr{\bm{d}=d} \prod_{i=1}^r g_0(T(d, i))
        \ge \sum_{d \in \cC_{r,|T|}(0,0,0)} \Pr{\bm{d}=d} \left(\frac{1}{2}\right)^{|T|}
        \ge \left(\frac{1}{2}\right)^{|T|} \left( 1 - \frac{|T|^2}{r} \right)
        > \left(\frac{1}{2}\right)^{|T|+1} \,,
    \]
    where we used the assumption that $|T|^4 \le r/100$ in the last inequality.

    We proceed similarly to upper bound the numerator of \eqref{eq:lb-desired-ineq}. For any
    $d \in \cC_{r,|T|}(h,k,s)$, the terms in $\prod_{i=1}^r g_0(T(d,i))$ satisfying $d_i=1$
    are again equal to $1/2$, while all other terms are trivially at most $1$.
    Moreover, by definition of $(h,k,s)$-collisions we have
    $\sum_{i=1}^r \ind{d_i=1} = |T| - (2h + k + s)$.
    Hence, for any $h \ge 0$ and $s \ge 2k \ge 2$, we have
    \[
        \prod_{i=1}^r g_0(T(d, i))
        \le \left(\frac{1}{2}\right)^{|T|-(2h+k+s)}
        \le \left(\frac{1}{2}\right)^{|T|} \cdot 2^{2h + \frac{3}{2}s} \,,
    \]
    and therefore, using \cref{prop: collisions unlikely},
    \begin{align*}
        \sum_{d \in \cC_{r,|T|}(h,k,s)} \Pr{\bm{d}=d} \prod_{i=1}^r g_0(T(d, i))
        &\le \left(\frac{1}{2}\right)^{|T|} \cdot 2^{2h + \frac{3}{2}s}
            \cdot \sum_{d \in \cC_{r,|T|}(h,k,s)} \Pr{\bm{d}=d} \\
        &\le \left(\frac{1}{2}\right)^{|T|}
            \cdot \left( \frac{(2h|T|)^2}{r} \right)^h
            \cdot \left( \frac{(2k|T|)^{3/2}}{r} \right)^s \,.
    \end{align*}
    Combining the results above, along with the observation that $h, k, s \le |T|$ for any nonzero
    terms in the numerator of \eqref{eq:lb-desired-ineq}, and using the notation $\lesssim$ to
    absorb constant factors, we obtain
    \begin{align*}
        &\sum_{h \ge 0, s \ge 2k \ge 2} (2^k-1) 2^s \cdot
            \frac{\sum_{d \in \cC_{r,|T|}(h,k,s)} \Pr{\bm{d}=d} \prod_{i=1}^r g_0(T(d, i))}{
                \sum_{d \in \cM} \Pr{\bm{d}=d} \prod_{i=1}^r g_0(T(d, i))} \\
        &\qquad < \sum_{\substack{|T| \ge h \ge 0,\\ |T| \ge s \ge 2k \ge 2}}
        2^{\frac{3}{2} s} \cdot
        \frac{
            \left(\frac{1}{2}\right)^{|T|}
            \cdot \left( \frac{(2h|T|)^2}{r} \right)^h
            \cdot \left( \frac{(2k|T|)^{3/2}}{r} \right)^s }{
            \left( 1/2 \right)^{|T|+1}} \\
        &\qquad = 2 \cdot \sum_{\substack{|T| \ge h \ge 0,\\ |T| \ge s \ge 2k \ge 2}}
            \left( \frac{(2h|T|)^2}{r} \right)^h
            \cdot \left( \frac{(4k|T|)^{3/2}}{r} \right)^s \\
        &\qquad \le 2 \cdot \sum_{\substack{|T| \ge h \ge 0,\\ |T| \ge s \ge 2k \ge 2}}
            \left( \frac{4 |T|^4}{r} \right)^h
            \left( \frac{8 |T|^3}{r} \right)^s \\
        &\qquad < 2 \cdot
            \sum_{h \ge 0} \left( \frac{4 |T|^4}{r} \right)^h \left[
            \sum_{s \ge 2} \sum_{1 \le k \le s/2} \cdot \left( \frac{8 |T|^3}{r} \right)^s \right] \\
        &\qquad \le 2 \cdot
            \sum_{h \ge 0} \left( \frac{4 |T|^4}{r} \right)^h \left[
            \sum_{s \ge 2} \frac{s}{2} \cdot \left( \frac{8 |T|^3}{r} \right)^s \right] \\
        &\qquad \lesssim
            \left( \frac{|T|^3}{r} \right)^2
            \sum_{h \ge 0} \left( \frac{4 |T|^4}{r} \right)^h \\
        &\qquad \lesssim
            \frac{|T|^6}{r^2} \,,
    \end{align*}
    where we used the assumption that $|T|^4 \le r/100$ to establish the convergence of the two
    geometric series,\footnote{Namely, we used the formulas
    $\sum_{i \ge 0} x^i = \frac{1}{1-x}$ and $\sum_{i \ge 2} i x^i = \frac{(2-x)x^2}{(1-x)^2}$
    for $|x| < 1$.} thus concluding the proof.
\end{proof}

We now use this result to upper bound the mutual information between $\bm{Z}$ and the subtrace
produced by a range of $\Theta(n/m)$ consecutive dominoes.

\begin{lemma}
    \label{lemma:range-info}
    Suppose $n^{-1/4} < \epsilon < 1$, and let
    $\frac{n}{m} \le r \le 2\frac{n}{m}$ be a positive integer.
    Suppose $\bm{Z} \sim \Ber(1/2)$ and
    let $\bm{T}$ denote the subtrace generated by a range of $r$ consecutive dominoes,
    according to distribution $\cD_{\bm{Z}}$.
    Suppose $m = o\left(\left(\frac{n}{\epsilon}\right)^{4/5} \frac{1}{\log^4 n}\right)$.
    Then as $n \to \infty$ we have
    \[
        I(\bm{Z} : \bm{T}) = O\left(\frac{\epsilon^4 \log^{12} n}{r^4}\right) \,.
    \]
\end{lemma}
\begin{proof}
    As before, let $P_z$ denote the conditional probabilities of $\bm{T}$ given $\bm{Z}=z$.
    Let $Q \define (P_0 + P_1)/2$ denote the (marginal) distribution of $\bm{T}$.

    Our strategy will be to decompose the set of possible subtraces $T$ according to whether
    $|T| \lesssim \log n$ (the typical case) or $|T| \gtrsim \log n$. In the former case,
    \cref{lemma: range probs} will give that the ratio $P_1(T) / P_0(T)$ is close to $1$, while
    in the latter case, we will use Poisson concentration bounds to argue that such long traces
    cannot contribute too much to the mutual information.

    Concretely, we start by upper bounding $I(\bm{Z} : \bm{T})$ by the sum a $\chi^2$-type
    expression for $|T| \le 20\log n$, and tail probabilities for $|T| > 20\log n$:
    \begin{align*}
        I(\bm{Z} : \bm{T})
        &= H(\bm{T}) - H(\bm{T}|\bm{Z}) \\
        &= -\sum_{T} Q(T) \log(Q(T))
            + \sum_{z \in \{0,1\}} \Pr{Z=z} \sum_{T} P_z(T) \log(P_z(T)) \\
        &= -\sum_{T} \frac{P_0(T) + P_1(T)}{2} \log(Q(T))
            + \sum_{T} \left[ \frac{P_0(T)}{2} \log(P_0(T)) + \frac{P_1(T)}{2} \log(P_1(T)) \right] \\
        &= \frac{1}{2} \sum_{T} \left[
            P_0(T) \log\left( \frac{P_0(T)}{Q(T)} \right) + P_1(T) \log\left( \frac{P_1(T)}{Q(T)} \right)
            \right] \\
        &\le \frac{1}{2} \sum_{T} \left[
            P_0(T) \left( \frac{P_0(T)}{Q(T)} - 1 \right) +
            P_1(T) \left( \frac{P_1(T)}{Q(T)} - 1 \right)
            \right] \\
        &= \frac{1}{2} \sum_{T} \left[
            P_0(T) \left( \frac{P_0(T) - P_1(T)}{P_0(T) + P_1(T)} \right) +
            P_1(T) \left( \frac{P_1(T) - P_0(T)}{P_0(T) + P_1(T)} \right)
            \right] \\
        &= \frac{1}{2} \sum_{T} \left[
            \frac{\left(P_1(T) - P_0(T)\right)^2}{P_0(T) + P_1(T)}
            \right] \\
        &< \frac{1}{2} \sum_{T : |T| \le 20\log n} \frac{\left(P_1(T)-P_0(T)\right)^2}{P_0(T)}
            + \frac{1}{2} \sum_{T : |T| > 20\log n} \left( P_0(T) + P_1(T) \right) \\
        &= \frac{1}{2} \sum_{T : |T| \le 20\log n} P_0(T) \left( \frac{P_1(T)}{P_0(T)} - 1 \right)^2 \\
            &\qquad + \frac{1}{2}\Pruc{}{|\bm{T}| > 20\log n}{\bm{Z}=0}
                    + \frac{1}{2}\Pruc{}{|\bm{T}| > 20\log n}{\bm{Z}=1} \,.
    \end{align*}

    We start with the first term in the last expression above. We want to show that, when
    $|T| \le 20 \log n$, we have $|T|^4 = o(r)$, which is sufficient for satisfying the condition of
    \cref{lemma: range probs}. Recalling the assumptions $\frac{n}{m} \le r \le 2\frac{n}{m}$,
    $m = o\left( \left(\frac{n}{\epsilon}\right)^{4/5} \frac{1}{\log^4 n} \right)$ and
    $\epsilon \ge \frac{1}{n^{1/4}}$, we have
    \begin{align*}
        \frac{|T|^4 / 20^4}{r} &\le \frac{\log^4 n}{r}
        \le \frac{m \log^4 n}{n}
        \ll \frac{\left(\frac{n}{\epsilon}\right)^{4/5} \frac{1}{\log^4 n} \log^4 n}{n}
        \le 1 \,,
    \end{align*}
    and hence the condition $|T|^4 \le r/100$ holds for sufficiently large $n$.
    Therefore \cref{lemma: range probs} yields
    \[
        \frac{1}{2} \sum_{T : |T| \le 20\log n} P_0(T) \left( \frac{P_1(T)}{P_0(T)} - 1 \right)^2
        \le \frac{1}{2} \sum_{T : |T| \le 20\log n} P_0(T)
            \left( c \cdot \epsilon^2 \cdot \frac{|T|^6}{r^2} \right)^2
        = O\left( \frac{\epsilon^4 \log^{12} n}{r^4} \right) \,.
    \]

    We now deal with the second component. Recall (see \cref{observation:domino-subtrace-length})
    that $|\bm{T}|$ is distributed according to a
    Poisson distribution completely determined by the number of dominoes in the range:
    \[
        |\bm{T}| \sim \Poi\left( r \cdot \frac{2m}{n} \right) \,,
    \]
    independently of $\bm{Z}$. Let $\lambda \define r \cdot \frac{2m}{n}$ and note that
    $2 \le \lambda \le 4 \le \log n$. \cref{fact:poisson-concentration} gives, for $z \in \{0,1\}$,
    \[
        \Pruc{}{|\bm{T}| > 20\log n}{\bm{Z}=z}
        \le \Pr{|\bm{T}| - \lambda > 19\log n}
        \le e^{-\frac{(19\log n)^2}{2(\lambda + 19\log n)}}
        \le e^{-\frac{361\log n}{40}}
        \le \frac{1}{n^9} \,.
    \]
    Finally, we have
    \begin{align*}
        \frac{1}{n^9} \le \frac{\epsilon^4}{r^4}
        \impliedby \frac{1}{n^9} \le \frac{1 / n}{16 n^4 / m^4}
        \iff \frac{16}{n^4} \le m^4
        \iff \frac{2}{n} \le m \,,
    \end{align*}
    which holds trivially, completing the proof.
\end{proof}

Since subtraces produced by disjoint ranges are conditionally independent given $\bm{Z}$, applying
the chain rule along with the data processing inequality concludes the proof.

\begin{lemma}[Refinement of \cref{claim: YES vs NO}]
    \label{lemma:refined-lower-bound}
    Suppose $n^{-1/4} < \epsilon < 1$.
    Let $\bm{Z} \sim \Ber(1/2)$, $\bm{\pi} \sim \cD_{\bm{\pi}}$, and let $\bm{\cT}$ be a trace of
    size $\Poi(m)$ drawn from $\pi$. Then if
    $m = o\left( \left(\frac{n}{\epsilon}\right)^{4/5} \frac{1}{\log^4 n} \right)$, we have
    \[
        I(\bm{Z} : \bm{\cT}) = O\left( \frac{\epsilon^4 m^5}{n^4} \log^{12}n \right) = o(1) \,,
    \]
    and hence any algorithm that succeeds in distinguishing the YES and NO cases with probability
    at least $51\%$ (over $\bm{Z}$ and $\bm{\cT}$) requires sample complexity at least
    \[
        \Omega\left( \left(\frac{n}{\epsilon}\right)^{4/5} \cdot \frac{1}{\log^4 n} \right) \,.
    \]
\end{lemma}
\begin{proof}
    Fix an arbitrary partition of the domain $[2n]$ into consecutive ranges $R_1, \dotsc, R_{m'}$
    such that 1) each $R_i$ is a contiguous range with multiple of $4$ length, and hence
    consists of $|R_i|/4$ consecutive dominoes (recall that we assume even $n$ for simplicity);
    and 2) each $R_i$ satisfies $\frac{n}{m} \le |R_i|/4 \le 2\frac{n}{m}$.
    It follows that $m' = \Theta(m)$ and, letting $r_i \define |R_i|/4$ for each $i \in [m']$, each
    $r_i$ satisfies the requirements of \cref{lemma:range-info}.

    Let $\bm{T}_i$ be the subtrace generated by range $R_i$, so that the final trace is obtained by
    concatenation of all subtraces:
    \[
        \bm{\cT} = \bm{T}_1 \circ \bm{T}_2 \dotsm \circ \bm{T}_{m'} \,.
    \]
    The data processing inequality yields
    \[
        I(\bm{Z} : \bm{\cT}) \le I(\bm{Z} : \bm{T}_1, \dotsc, \bm{T}_{m'}) \,.
    \]

    Note that, conditional on $\bm{Z}$, the entries of $\bm{\pi}$ in different dominoes are
    mutually independent as per the process described in \cref{def:yes-no-distributions}.
    Then, recalling the distribution of subtraces described in \cref{def:subtrace}, it follows that
    the $\bm{T}_i$ are mutually independent conditional on $\bm{Z}$.
    Thus the chain rule of mutual information and \cref{lemma:range-info} give
    \[
        I(\bm{Z} : \bm{\cT}) \le \sum_{i=1}^{m'} I(\bm{Z} : \bm{T}_i)
        = \Theta(m) \cdot O\left( \frac{\epsilon^4 m^4 \log^{12} n}{n^4} \right)
        = O\left( m^5 \cdot \frac{\epsilon^4}{n^4} \log^{12} n \right)
        = o(1) \,,
    \]
    as desired. Finally, applying \cref{fact:fanos} establishes the second conclusion.
\end{proof}

Putting together \cref{prop:lb-reduction-from-standard} and \cref{lemma:refined-lower-bound}
establishes \cref{thm: lower bound}:

\begin{corollary}[Refinement of \cref{thm: lower bound}]
    Let $\Pi_1$ contain only the uniform distribution over $[2n]$, and let $\Pi_2$ be the set of
    distributions over $[2n]$ that are $\epsilon$-far from uniform in total variation distance. Then
    $(\Pi_1, \Pi_2, 51/100)$-testing under the parity trace requires sample complexity at least
    $\Omega\left(
        \left( \frac{n}{\epsilon} \right)^{4/5} \frac{1}{\log^4 n}
        + \frac{\sqrt{n}}{\epsilon^2} \right)$
    samples.
    Furthermore, this bound holds even if the input distribution
    $\pi$ is guaranteed to have $1/2$ mass uniformly distributed over the zero-valued (\ie even)
    coordinates.
\end{corollary}
\begin{proof}
    The lower bound of $\Omega(\sqrt{n}/\epsilon^2)$ holds by \cref{prop:lb-reduction-from-standard}.
    Moreover, we have
    \[
        \left(\frac{n}{\epsilon}\right)^{4/5} \frac{1}{\log^4 n} \ge \frac{\sqrt{n}}{\epsilon^2}
        \iff \epsilon \ge \frac{\log^{10/3} n}{n^{1/4}} \,,
    \]
    in which case \cref{lemma:refined-lower-bound} establishes the bound.
\end{proof}

\section{Distribution-Free Sample-Based Property Testing}
\label{section:sample-based-testing}

In this section, we relate distribution testing under the parity trace to distribution-free
sample-based property testing. The main ideas of this section are:
\begin{enumerate}[itemsep=0pt]
\item We define \emph{labeled-distribution testing} as a generalized reformulation of
distribution-free sample-based property testing that makes the connection to distribution testing
more explicit.
\item There is a natural type of labeled distribution properties, which we call \emph{density
properties}, that includes some property testing problems, and some more challenging versions of
standard distribution testing problems. We use the edit distance and Ramsey theory to show that
testing these properties is equivalent to testing distributions under the parity trace.
\item Using this equivalence, we get new tight positive results for distribution-free sample-based
testing (in the more general labeled-distribution definition) by applying \cref{thm:intro-main-informal}.
\item There is a \emph{testing-by-learning} reduction for labeled-distribution testing, similar to
the standard testing-by-learning reduction of \cite{GGR98}, that allows non-constructive upper
bounds on distribution testing under the parity trace. This will be used in
\cref{section:trace-testing} to get upper bounds for some testing problems in the trace
reconstruction model.
\end{enumerate}
The section is organized as follows:
\begin{description}[itemsep=0pt]
\item[\cref{section:labeled-distributions}:] The definition of labeled distributions.
\item[\cref{section:edit-distance}:] The definition of edit distance, which is closely related to
labeled distributions and will be necessary for all of our applications in the remainder of the
paper.
\item[\cref{section:labeled-to-standard}:] The definition of labeled distribution testing, and how
it generalizes the conventional distribution testing and distribution-free sample-based property
testing models.
\item[\cref{section:testing-by-learning}:] A testing-by-learning reduction for labeled-distribution
testing.
\item[\cref{section:density-properties}:] The definition of \emph{density properties}, and the
equivalence of testing density properties to distribution testing under the parity trace.
\item[\cref{section:uniform-k-alternating-theorem}:] The proof of
our main result on labeled distribution testing, \cref{thm:intro-main-testing-informal}, which is an
application of \cref{thm:intro-main-informal}.
\item[\cref{section:testing-uniform-nopromise}:] An upper bound on testing uniform distributions
against unrestricted distributions under the parity trace
(\cref{thm:testing-uniform-k-alternating-no-promise}).
\item[\cref{section:testing-support-size}:] The equivalence between testing support size $k$ under
the parity trace, and testing $k$-alternating functions in the distribution-free sample-based model
(\cref{thm:testing-support-k}), and an alternate proof of the lower bounds of \cite{BFH21} for
testing halfspaces, among others.
\end{description}

\subsection{Labeled Distributions}
\label{section:labeled-distributions}

We shall now define labeled distributions and edit distance, which are closely related.

\begin{definition}[Labeled Distribution]
\label{def:labeled-distributions}
On any fixed domain $\cX$, a \emph{labeled distribution} is a pair $(f, \cD)$ of a function $f : \cX
\to \zo$ and a probability distribution $\cD$ over $\cX$. We write $\cD_f$ for the probability
distribution over $\cX \times \zo$, where the density of any $(x, b) \in \cX \times \zo$ is defined
as
\[
  \cD_f(x,b) \define \begin{cases}
    \cD(x) &\text{ if } b = f(x) \\
    0      &\text{ otherwise.}
  \end{cases}
\]
In other words, a sample from $\cD_f$ is obtained by sampling $x \sim \cD$ and taking $(x, f(x))$.
\end{definition}

We study the case $\cX = \bZ$. For a labeled distribution $(f, \cD)$ over
$\bZ$, it may be the case that $f$ ``alternates'' an infinite number of times. We restrict our
attention to the ``proper'' labeled distributions, where $f$ has a finite number of alternations
``on the left''\!\!, defined as follows.

\begin{definition}[Proper Labeled Distributions]
A labeled distribution $(f,\cD)$ is \emph{1-proper} if there exists $t \in \bZ$ such that $f(x) = 1$
for all $x < t$. It is \emph{0-proper} if, instead, $f(x) = 0$ for all $x < t$.  If $(f,\cD)$
is either 0- or 1-proper, we call it \emph{proper}.
\end{definition}

\begin{remark}
When studying labeled distribution testing, it suffices to consider proper labeled distributions.
This is because, for every labeled distribution $\cD_f$ and every $\delta > 0$, there exists a
\emph{proper} labeled distribution $\cD_g$ such that $\dist_\TV(\cD_f, \cD_g) < \delta$. So every
distribution is indistinguishable (to any algorithm with bounded sample size) from some proper
distribution.
\end{remark}

\begin{definition}[Alternation Sequence]
For any proper labeled distribution $(f,\cD)$, the \emph{alternation sequence} is the unique
sequence $a_1 < a_2 < a_3 < \dotsm$ such that $f$ is constant on the intervals $(-\infty, a_1]$,
$(a_1, a_2]$, $\dotsc$, and $f(a_{i-1}) \neq f(a_i)$. Note that if $(f,\cD)$ is 1-proper, then $f$
takes value $\parity(i)$ on the interval $(a_{i-1}, a_i]$, and value 1 on $(-\infty, a_1]$. If
$(f,\cD)$ is 0-proper, it takes the opposite values. Note that this sequence always exists when
$(f,\cD)$ is proper, and it may be an infinite sequence.
\end{definition}

\begin{definition}[Density Sequence]
For any proper labeled distribution $(f, \cD)$ with alternation sequence $a_1 < a_2 < a_3 <
\dotsm$, we define the \emph{density sequence} $\pi_{f,\cD} : \bN \to \bR_{\geq 0}$ as follows. If
$(f,\cD)$ is 1-proper, we define
\[
  \pi_{f,\cD}(i) \define \begin{cases}
    \cD(-\infty, a_1] &\text{ if } i = 1 \\
    \cD(a_{i-1}, a_i] &\text{ if } i > 1 \,.
  \end{cases}
\]
If $(f,\cD)$ is 0-proper, we define
\[
  \pi_{f,\cD}(i) \define \begin{cases}
    0 &\text{ if } i = 1 \\
    \cD(-\infty, a_1] &\text{ if } i = 2 \\
    \cD(a_{i-2}, a_{i-1}] &\text{ if } i > 2 \,.
  \end{cases}
\]
Note that $\pi_{f,\cD}$ is a probability distribution, since $\sum_{i=1}^\infty \pi_{f,\cD}(i) =
\sum_{x \in \bZ} \cD(x) = 1$. For any set $\Xi$ of proper labeled distributions, we write
\[
  \Pi(\Xi) \define \{ \pi_{f,\cD} : (f,\cD) \in \Xi \}
\]
for the set of density sequences (probability distributions) associated with the proper labeled
distributions in $\Xi$.
\end{definition}

The following simple formula for TV distance for labeled distributions is often useful.
\begin{proposition}
\label{prop:tv-distance-labeled}
Let $(f,\cD)$ and $(g,\cE)$ be labeled distributions over $\bZ$. Then
\[
  \dist_\TV(\cD_f, \cE_g)
  = \frac{1}{2} \sum_{i \in \bZ} \ind{f(i) \neq g(i)}( \cD(i) + \cE(i) )
                  +  \ind{f(i) = g(i)} | \cD(i) - \cE(i) | \,.
\]
\end{proposition}
\begin{proof}
By definition,
\begin{align*}
  \dist_\TV(\cD_f, \cE_g) 
  &= \frac{1}{2} \sum_{(i, b) \in \bZ \times \zo} | \cD_f(i,b) - \cE_g(i,b) | \\
  &= \frac{1}{2}
    \left(\sum_{\substack{(i, b) \in \bZ \times \zo \\ f(i) = g(i) = b}} | \cD(i) - \cE(i) |
    + \sum_{\substack{(i,b) \in \bZ \times \zo \\ f(i) = b, g(i) \neq b }} \cD(i) 
    + \sum_{\substack{(i,b) \in \bZ \times \zo \\ f(i) \neq b, g(i) =  b }} \cE(i) \right) \\
  &= \frac{1}{2} \sum_{i \in \bZ} \ind{f(i) \neq g(i)}( \cD(i) + \cE(i) )
                  +  \ind{f(i) = g(i)} | \cD(i) - \cE(i) | \,. \qedhere
\end{align*}
\end{proof}

\subsection{Edit Distance}
\label{section:edit-distance}

We define two notions of edit distance: one for labeled distributions on domain $\bZ$, and one for
distributions over $\bN$.

\begin{definition}[Edit Distance for Labeled Distributions]
\label{def:labeled-distribution-edit-distance}
For any two proper labeled distributions $(f, \cD)$ and $(g, \cE)$ on domain $\bZ$, define
\[
  \dist_\edit( (f, \cD), (g, \cE) ) \define \inf \dist_\TV( \cD'_{f'}, \cE'_{g'} ) \,,
\]
where the infimum is taken over all proper labeled distributions $(f', \cD')$ and $(g', \cE')$ that
have the same density sequences as the original distributions, \ie that satisfy $\pi_{f', \cD'} =
\pi_{f,\cD}$ and $\pi_{g', \cE'} = \pi_{g, \cE}$.
\end{definition}

Next, we will define the edit distance for distributions over $\bN$. Recall that the TV distance is
not the natural distance metric for distribution testing under the parity trace, because
distributions may have maximum TV distance 1 while being indistinguishable under the parity trace.
Edit distance replaces the TV distance as the natural (pseudo-)metric for the parity trace.
We begin by introducing the notion of a \emph{fractional string}.

\begin{definition}[Fractional String]
A \emph{fractional string} is a finite sequence $\sigma_1^{p_1} \sigma_2^{p_2} \dots \sigma_n^{p_n}$
where each \emph{fractional character} $\sigma_i^{p_i}$ consists of a symbol $\sigma_i \in \zo$ and
a value $p_i \in \bR_{\geq 0}$.
\end{definition}

We now define the edit distance for fractional strings, which is an analog of the standard edit
distance for strings.

\begin{definition}[Edit Distance for Fractional Strings]
\label{def:fractional-string-edit-distance}
Let $a = a_1^{p_1} a_2^{p_2} \dotsm a_n^{p_n}$ be a fractional string. We define the following
permitted edit operations on $a$, with associated cost:
\begin{description}
\item[Insert:] For $i \in [n+1]$ and $b \in \zo$, $\mathsf{ins}_{i,b}(a)$ is the fractional string
obtained by inserting the fractional character $b^0$ immediately before $a_i^{p_i}$. The cost of
this operation is 0.
\item[Delete:] For $i \in [n]$ such that $p_i = 0$, $\mathsf{del}_{i}(a)$ is the fractional string
obtained by deleting the fractional character $a_i^{p_i} = a_i^0$. The cost of this operation is 0.
\item[Rearrange:] For $i \in [n-1]$ such that $a_i = a_{i+1}$, and $-p_i \leq \delta \leq p_{i+1}$,
$\mathsf{rearr}_{i,\delta}(a)$ is the fractional string obtained by replacing
$a_i^{p_i}a_{i+1}^{p_{i+1}}$ with $a_i^{p_i+\delta} a_{i+1}^{p_{i+1}-\delta}$. The cost of this
operation is 0.
\item[Adjust:] For $i \in [n]$ and $\delta \geq -p_i$, $\mathsf{adj}_{i,\delta}(a)$ is the
fractional string obtained by replacing $p_i$ with $p_i + \delta$, so that the fractional character
$a_i^{p_i}$ becomes $a_i^{p_i+\delta}$. The cost of this operation is $|\delta|/2$.
\end{description}
For a fractional string $a$, we say that a sequence of operations $O_1, \dotsc, O_k$ is permitted if
for each $i \in [k]$, $O_i$ is a permitted operation on the fractional string $(O_{i-1} \circ
O_{i-2} \circ \dotsm \circ O_1)(a)$.

For two fractional strings $a$ and $b$, we define the \emph{edit distance}
$\dist_{\mathsf{fr-edit}}(a,b)$ as the minimum $c$ such that there exists a sequence of permitted
operations $O_1, \dots, O_k$ such that $O_k \circ O_{k-1} \circ \dotsm \circ O_1 (a) = b$ and the
sum of costs of operations $O_i$ is $c$.
\end{definition}

Let $\pi : \bN \to \bR_{\geq 0}$ be any finitely-supported probability distribution. We define the
fractional string $\mathsf{str}(\pi)$ as follows. Since $\pi$ is finitely-supported, there is some
$k \in \bN$ such that $\pi(i) = 0$ for all $i > k$. Then we define
\[
\mathsf{str}(\pi) \define 1^{\pi(1)} 0^{\pi(2)} 1^{\pi(3)} 0^{\pi(4)} \dotsm (\parity(k))^{\pi(k)} \,.
\]

Finally, we may define the edit distance for distributions.

\begin{definition}[Edit Distance for Distributions]
\label{def:distribution-edit-distance}
For two finitely-supported distributions $\pi, \pi'$, we define
\[
  \dist_\edit(\pi, \pi') \define \dist_{\mathsf{fr-edit}}(\mathsf{str}(\pi), \mathsf{str}(\pi')) \,.
\]
\end{definition}
The following alternate characterization of edit distance is helpful. We defer the proof to
\cref{section:edit-distance-definition-equivalence}.

\begin{restatable}{lemma}{lemmaeditdistancealternate}
\label{lemma:edit-distance-alternate}
Let $\pi$ and $\pi'$ be finitely-supported distributions over $\bN$. Then
\[
  \dist_\edit(\pi, \pi') = \inf \dist_\TV(\cD_f, \cE_g) \,,
\]
where the infimum is taken over labeled distributions $(f,\cD)$ and $(g,\cE)$ such that $\pi =
\pi_{f,\cD}$ and $\pi' = \pi_{g,\cE}$.
\end{restatable}

From this lemma, we can see that the edit distance for distributions and labeled distributions are
essentially equivalent: for two proper labeled distributions $(f, \cD)$ and $(g,\cE)$, the lemma
implies
\[
  \dist_\edit((f,\cD), (g,\cE)) = \dist_\edit(\pi_{f,\cD}, \pi_{g,\cE}) \,.
\]

It is easy to see that the following inequality holds in general:
\begin{equation}
\label{eq:edit-to-tv-inequality}
  \dist_\edit(\pi, \pi') \leq \dist_\TV(\pi, \pi') \,.
\end{equation}
This can be verified by taking $(f,\cD)$ where $f(i) = \parity(i)$ for $i \in \bN$ and $f(i) = 1$
for $i \leq 0$, and $\cD(i) = \pi(i)$ for $i \in \bN$ and $\cD(i) = 0$ for $i \leq 0$. Define
$(g,\cE)$ similarly for $\pi'$. This satisfies $\pi_{f,\cD} = \pi$ and $\pi_{g,\cE} = \pi'$, and
$\dist_\TV(\cD_f, \cE_g) = \dist_\TV(\pi, \pi')$.

\subsection{Labeled Distribution Testing}
\label{section:labeled-to-standard}

We now introduce \emph{labeled distribution testing}. For labeled distributions $(f,\cD)$ and
$(g,\cE)$ over a fixed domain $\cX$, we abuse notation and write
\[
  \dist_\TV( (f,\cD), (g, \cE)) \define \dist_\TV( \cD_f, \cE_g ) \,,
\]
so that, for a property $\Xi$ of labeled distributions, we have
\[
\far^\TV_\epsilon(\Xi) \define \{ (f, \cD) : \forall (g,\cE) \in \Xi, \dist_\TV(\cD_f, \cE_g) >
\epsilon \} \,.
\]

\begin{definition}[Labeled Distribution Testing]
\label{def:labeled-distribution-testing}
Let $\Xi_1$ and $\Xi_2$ be properties of labeled distributions over a fixed domain $\cX$.
A $(\Xi_1, \Xi_2, \alpha)$-labeled distribution tester, with sample complexity $m$, is an algorithm
$A$ that satisfies the following, for every labeled distribution $(f,\cD)$ over $\cX$:
\begin{enumerate}
\item If $(f,\cD) \in \Xi_1$, then $\Pru{S_f \sim \samp(\cD_f, m)}{ A(S_f) \text{ accepts }} \geq
\alpha$; and
\item If $(f,\cD) \in \Xi_2$, then $\Pru{S_f \sim \samp(\cD_f, m)}{ A(S_f) \text{ rejects }} \geq
\alpha$.
\end{enumerate}
\noindent
The canonical form of this problem has $\Xi_2 \define \far^\TV_\epsilon(\Xi_1)$ for some $\epsilon >
0$.
\end{definition}

We prove here that one can obtain the standard distribution testing and distribution-free
sample-based property testing models from our more general labeled distribution testing model.  To
obtain, from the labeled distribution testing model, the standard distribution testing model, where
the goal is to test a property $\Pi$ of distributions over domain $\cX$, it suffices to choose the
property $\Xi = \Lambda \times \Pi$, where $\Lambda$ contains only the constant 0 function over
domain $\cX$.

It is slightly less obvious how to obtain distribution-free sample-based property testing from the
labeled distribution testing model.  Distribution-free sample-based property testing is defined as
follows.

\begin{definition}[Distribution-Free Sample-Based Property Testing]
Let $\Lambda$ be a property of functions $\cX \to \zo$ for some fixed domain $\cX$, and let
$\epsilon > 0$. A $(\Lambda, \epsilon, \alpha)$-distribution-free sample-based property tester, with
sample complexity $m$, is an algorithm $A$ that satisfies the following, for every function $f : \cX
\to \zo$ and distribution $\cD$ over $\cX$:
\begin{enumerate}
\item If $f \in \Lambda$, then $\Pru{S_f \sim \samp(\cD_f, m)}{ A(S_f) \text{ accepts }} \geq
\alpha$; and
\item If $\Pru{x \sim \cD}{ f(x) \neq g(x) } > \epsilon$ for all $g \in \Lambda$, then $\Pru{S_f
\sim \samp(\cD_f, m)}{ A(S_f) \text{ rejects }} \geq \alpha$.
\end{enumerate}
\end{definition}
This problem cannot be expressed neatly as the problem of distinguishing between properties
$\Lambda_1$ and $\Lambda_2$ of functions $\cX \to \zo$, because the set of functions that should be
rejected depends on the distribution $\cD$. But we can express it as distinguishing two properties
$\Xi_1$, $\Xi_2$ of labeled distributions, as shown in the next two propositions.

\begin{proposition}
\label{prop:app-equivalence-tv}
Fix a domain $\cX$ and let $(f,\cD)$, $(g,\cE)$ be labeled distributions. Then
\[
  \Pru{x \sim \cD}{f(x) \neq g(x)} \leq 2 \cdot \dist_\TV(\cD_f, \cE_G) \,.
\]
\end{proposition}
\begin{proof}
For any event $X \subseteq \cX \times \zo$, write $\cD_f(X) \define \Pru{x \sim \cD}{ (x, f(x)) \in
X }$.  Write $\epsilon \define \dist_\TV( \cD_f , \cE_g )$, so that for any event $X \subseteq \cX
\times \zo$, we have $\left| \cD_f(X) - \cE_g(X) \right| \leq \epsilon$.

Let $\Delta \define \{ x \in \cX : f(x) \neq g(x) \}$ and define the event $E \define \{ (x,b) : x
\in \Delta, b = g(x) \}$. Then $\cD_f(E) = 0$ by definition, and $\left| \cD_f(E) - \cE_g(E)
\right| \leq \epsilon$, so $\cE_g(E) \leq \epsilon$. Then also $\cE(\Delta) = \cE_g(E) \leq
\epsilon$.

Define the event $F \define \{ (x,b) : x \in \Delta, b \in \zo \}$, which satisfies $\cD_f(F) =
\cD(\Delta)$ and $\cE_g(F) = \cE(\Delta)$. Then $\left| \cD(\Delta) - \cE(\Delta) \right| \leq
\epsilon$, so $\cD(\Delta) \leq \cE(\Delta) + \epsilon \leq 2\epsilon$.
We conclude that
\[
  \Pru{x \sim \cD}{ f(x) \neq g(x) } = \cD(\Delta) \leq 2\epsilon \,. \qedhere
\]
\end{proof}

\begin{proposition}
\label{prop:tester-equivalence}
Fix any domain $\cX$. Let $\Lambda$ be any property of functions $\cX \to \zo$, and let $\Pi$ be the
set of all distributions over $\cX$. Then, for $\Xi = \Lambda \times \Pi$,
\begin{enumerate}
\item If there is a $(\Xi, \far^\TV_{\epsilon/2}(\Xi),
\alpha)$-labeled distribution tester with sample complexity $m$, then there is a
$(\Lambda, \epsilon, \alpha)$-distribution-free sample-based tester with sample
complexity $m$.
\item If there is a $(\Lambda, \epsilon, \alpha)$-distribution-free sample-based tester with sample
complexity $m$, then there is a $(\Xi, \far^\TV_\epsilon(\Xi))$-labeled distribution tester with
sample complexity $m$.
\end{enumerate}
\end{proposition}
\begin{proof}
For the first conclusion, the input to the distribution-free sample-based property tester is a
function $f : \cX \to \zo$ and a distribution $\cD$ over $\cX$. The algorithm will take a labeled
sample $S_f$ where $S \sim \samp(\cD, m)$, and simulate the labeled distribution tester on $S_f$.
Suppose that $f \in \Lambda$. Then $\cD_f \in \Lambda \times \Pi = \Xi$, so the labeled distribution
tester will accept with probability at least $\alpha$.

Now suppose that $f$ is $\epsilon$-far from $\Lambda$ with respect to $\cD$. Then for all $g : \cX
\to \zo$ and all distributions $\cE$ over $\cX$, we have $\dist_\TV(\cD_f, \cE_g) \geq \frac{1}{2}
\Pru{x \sim \cD}{ f(x) \neq g(x) } > \epsilon/2$, due to \cref{prop:app-equivalence-tv}. Therefore
$(f,\cD) \in \far^\TV_{\epsilon/2}(\Xi)$, so the tester will reject with probability at least
$\alpha$.

For the second conclusion, we obtain a labeled distribution tester for $\Xi$ by taking, on input
$(f,\cD)$, the labeled sample $S_f \sim \samp(\cD_f, m)$, and running the distribution-free
sample-based property tester on $S_f$. 

If $(f,\cD) \in \Xi$ then $f \in \Lambda$, so the property tester will accept with probability at
least $\alpha$.

If $\dist_\TV(\cD_f, \cE_g) > \epsilon$ for all $(g,\cE) \in \Xi$, then in particular
$\dist_\TV(\cD_f, \cD_g) > \epsilon$ for all $g \in \Lambda$. Then there is an event $E \subseteq
\cX \times \zo$ such that $\left|\cD_f(E) - \cD_g(E)\right| > \epsilon$.
We may assume without loss of generality that $f(x) \neq g(x)$ for each $(x,b) \in E$ (since we may
remove any set of pairs $(x,b)$ where $f(x) = g(x)$ without changing this difference). Then we can
partition $E = E_f \cup E_g$ where $E_f \define \{ (x,b) \in E : b = f(x) \}$ and $E_g \define \{
(x,b) \in E : b = g(x) \}$. Then $\cD_f(E) = \cD_f(E_f)$ and $\cD_g(E) = \cD_g(E_g)$.

It must be the case that either $\cD_f(E_f) > \epsilon$ or $\cD_g(E_g) > \epsilon$. If $\cD_f(E_f) >
\epsilon$, then we can choose $\Delta \define \{ x \in \cX : (x, f(x)) \in E_f \}$, so
\[
  \Pru{x \sim \cD}{ f(x) \neq g(x) } \geq \cD(\Delta) = \cD_f(E_f) > \epsilon \,.
\]
If $\cD_g(E_g) > \epsilon$, a similar conclusion holds. Since this holds for all $g \in \Lambda$, we
see that the property tester should reject with probability at least $\alpha$.
\end{proof}

\subsection{Testing-by-Learning}
\label{section:testing-by-learning}

\cite{GGR98} observed that a proper learning algorithm for a hypothesis class $\Pi$ can be used as a
property tester, by including an extra ``verification'' step. It is convenient for us to adapt the
same technique to a different type of learning algorithm that works for classes of labeled
distributions, where the ``hypothesis class'' is not just a set of functions $\cX \to \zo$, but a
joint set of function-distribution pairs.

For fixed domain $\cX$ and property $\Pi$ of distributions over $\cX$, write
$\close^\TV_\epsilon(\Pi)$ for the set of distributions $\pi$ over $\cX$ satisfying
$\dist_\TV(\pi, \Pi) \le \epsilon$.

\begin{definition}[Labeled Distribution Learning]
Let $\Xi$ be a property of labeled distributions on some fixed domain $\cX$. A \emph{labeled
distribution learning algorithm} for $\Xi$, with success probability $\alpha$, error $\epsilon > 0$,
and sample complexity $m$, is an algorithm $A$ that, on any input $(f, \cD)$, receives a labeled
sample $S_f \sim \samp(\cD_f, m)$ and outputs a function $g : \cX \to \zo$, and \emph{succeeds} with
probability at least $\alpha$ over $S_f$ and the randomness of the algorithm. The \emph{success}
event is defined as follows:
\begin{description}
\item[Success:] If $(f, \cD) \in \Xi$, then $\dist_\TV(\cD_f, \cD_g) < \epsilon$.
\end{description}
We call the algorithm \emph{proper} if the success event also has the following additional
conditions:
\begin{description}
\item[Success$^*$:]\;\\
If $(f, \cD) \notin \Xi$, then there exists $\cE$ such that $\cE_g \in \Xi$.

If $(f, \cD) \in \Xi$, then there exists $\cE$ such that $\cE_g \in \Xi$ and
$\dist_\TV(\cD, \cE) < \epsilon$.
\end{description}
\end{definition}

\begin{definition}[Learner-Verifier Pair]
\label{def:learner-verifier-pair}
Let $\Xi$ be a property of labeled distributions on some fixed domain $\cX$,
let $\epsilon, \delta > 0$ and let $A$ be a \emph{proper} labeled distribution learning
algorithm for $\Xi$ with success probability $1-\delta/3$, error $\epsilon/4$, and sample
complexity $m_A$. Let $\cG_A \subseteq \zo^{\cX}$ be the range of $A$. For every
$g \in \cG_A$, let $\Pi_g$ be the property
\[
\Pi_g \define \{ \cE : (g, \cE) \in \Xi \} \,.
\]
Suppose $B = \{B_g\}_{g \in \cG_A}$ is a family of algorithms such that, for every $g \in \cG_A$,
algorithm $B_g$ is a
$(\close^\TV_{\epsilon/4}(\Pi_g), \far^\TV_{\epsilon/2}(\Pi_g), 1-\delta/3)$-distribution tester
with sample complexity $m_B$. We call $(A, B)$ a \emph{learner-verifier pair} for $\Xi$
with success probability $1-\delta$, error $\epsilon$, and sample complexity $m_A + m_B$.
\end{definition}

\begin{proposition}
\label{prop:testing-by-learning}
Let $\Xi$ be a property of labeled distributions such that there exists a learner-verifier pair
$(A,B)$ with success probability $1-\delta$, error $\epsilon > 0$, and sample complexity $m = m_A +
m_B$. Then there is a $(\Xi, \far^\TV_\epsilon(\Xi))$-labeled distribution tester with sample
complexity $m + O(1/\epsilon)$.
\end{proposition}
\begin{proof}
On input $(f, \cD)$, the tester performs the following.
\begin{enumerate}
\item Use $m_A$ samples to run the proper learner $A$, and obtain an output function $g$. 
\item Use $O(1/\epsilon)$ samples from $\cD$ to compute an estimate $z$ of $\Pru{x \sim \cD}{f(x)
\neq g(x)}$, and reject if this is greater than $\tfrac{3}{8} \epsilon$.
\item Use $m_B$ samples to run the distribution tester $B_g$.
\end{enumerate}
Suppose that $(f, \cD) \in \Xi$. Suppose that the algorithms $A$ succeeds, which occurs with
probability at least $1-\delta/3$. Then there exists $\cE$ such that $\cE_g \in \Xi$ and
\[
  \dist_\TV(\cD_f, \cD_g) \leq \epsilon/4 \,,\qquad\text{ and }\qquad \dist_\TV(\cD,\cE) \leq
\epsilon/4 \,,
\]
by the conditions on the algorithm $A$, which is a proper learner. By the multiplicative Chernoff
bound, we have $z < \tfrac{3}{8} \epsilon$ with probability at least $1-\delta/3$, after using
$O(1/\epsilon)$ samples, so the second step passes. Finally, assume the algorithm $B_g$ succeeds,
which occurs with probability $1-\delta/3$. Since $\cE_g \in \Xi$, and $\dist_\TV(\cD,\cE) \leq
\epsilon/4$, we have 
\[
  \dist_\TV(\cD, \Pi_g) \leq \dist_\TV(\cD, \cE) \leq \epsilon/4 \,.
\]
So $\cD \in \close^\TV_{\epsilon/4}(\Pi_g)$, and the algorithm $B_g$ will accept. The probability of any of
these steps failing is at most $\delta$, by the union bound.

Now suppose that $(f, \cD) \in \far^{\TV}_\epsilon(\Xi)$. Suppose for contradiction that 
\[
  \dist_\TV(\cD, \Pi_g) < \epsilon/2
      \qquad \text{ and } \qquad
  \dist_\TV(\cD_f, \cD_g) < \epsilon/2 \,.
\]
Then there exists $\cE \in \Pi_g$ such that $\cE_g \in \Xi$ and $\dist_\TV(\cD, \cE) < \epsilon/2$.
Then
\[
  \dist_\TV(\cD_f, \cE_g) \leq \dist_\TV(\cD_f, \cD_g) + \dist_\TV(\cD_g, \cE_g)
                          =    \dist_\TV(\cD_f, \cD_g) + \dist_\TV(\cD, \cE)
                          < \epsilon \,,
\]
which is a contradiction. So it must be that either $\dist_\TV(\cD, \Pi_g) \geq \epsilon/2$, in
which case the third step rejects with probability at least $1-\delta/3$, or that $\dist_\TV(\cD_f,
\cD_g) \geq \epsilon/2$, in which case the second step rejects with probability at least
$1-\delta/3$, again using the multiplicative Chernoff bound.
\end{proof}

\begin{remark}
This formalization captures the standard testing-by-learning reduction, when $\Xi$ is a property of
labeled distributions obtained by choosing a property $\Lambda_1$ of functions $\cX \to \zo$,
setting $\Lambda_2$ to be the set of all distributions over $\cX$, and letting $\Xi = \Lambda_1
\times \Lambda_2$. In this case, the learner $A$ is the standard PAC learning algorithm (see the
proof of \cref{thm:uniform-k-alternating-learning-bound}), and the
verifier $B$ simply accepts everything.
\end{remark}

\subsection{Density Properties and Distribution Testing Under the Parity Trace}
\label{subsection:labeled-testing-to-parity-testing}
\label{section:density-properties}
\newcommand{\LD}{\textsf{LabeledDist}}
\newcommand{\PT}{\textsf{ParityTrace}}

Labeled-distribution testing is a more general reformulation of distribution-free sample-based
property testing, which allows a richer class of properties to be defined. For the remainder of
\cref{section:sample-based-testing}, we are interested in a certain family of
labeled-distribution properties that we call \emph{density properties}.

\begin{definition}[Density Property]
\label{def:density-property}
A property $\Xi$ of proper labeled distributions is a \emph{density property} if there exists a set
$\Pi$ of probability distributions over $\bN$ such that $\Pi = \Pi(\Xi)$; \ie for any proper labeled
distribution $(f,\cD)$, it holds that $(f,\cD) \in \Xi$ if and only if $\pi_{f,\cD} \in \Pi$. 
\end{definition}

In this subsection, our goal is to establish the relationships between testing density properties
and distribution testing under the parity trace, which are illustrated in
\cref{fig:relative-strength}. Here we prove the $\to$ relations; examples showing the $\not \to$
relations are discussed in \cref{section:relative-strength}.

\begin{figure}[h]
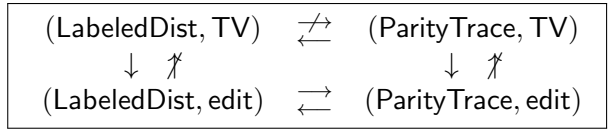

    \centering
\fbox{
    \begin{tabular}{ccc}
        $(\LD, \TV)$   & $\substack{\centernot\longrightarrow \\ \longleftarrow}$ & $(\PT, \TV)$ \\
        $\downarrow \quad  \centernot\uparrow$ &      & $\downarrow \quad \centernot\uparrow$ \\
        $(\LD, \edit)$ & $\substack{\longrightarrow \\ \longleftarrow}$           & $(\PT, \edit)$ \\
    \end{tabular}}
    \caption{Summary of the relative strengths of testing models and distance metrics. An arrow
$(\cM_1, d_1) \to (\cM_2, d_2)$ means that a tester in model $\cM_1$ with respect to distance $d_1$
implies a tester in model $\cM_2$ with respect to distance $d_2$, while $\not \to$ means that the
implication does not hold in general.}
    \label{fig:relative-strength}
\end{figure}

\begin{lemma}
\label{lemma:tv-density-property}
Let $\Xi$ be any density property, and let $\epsilon > 0$. Then $\Pi(\far^\TV_\epsilon(\Xi))
\subseteq \far^\TV_\epsilon(\Pi(\Xi))$.
\end{lemma}
\begin{proof}
Suppose that $\pi \in \Pi( \far^\TV_\epsilon(\Xi) )$, so that $\pi = \pi_{f,\cD}$ for some $\cD_f
\in \far^\TV_\epsilon(\Xi)$. Suppose for contradiction that $\dist_\TV(\pi, \pi') \leq \epsilon$ for
some $\pi' = \pi_{g,\cE} \in \Pi(\Xi)$. Let $a_1 < a_2 < a_3 < \dotsm$ be the alternation sequence
of $f$, and define $a_0 = -\infty$. Then we define the labeled distribution $\cF_f$ as follows.
Below, we assume without loss of generality that $\cD_f$ is 1-proper; if it is 0-proper, we require
to adjust some of the indices by 1.

For each $i \in \bN$, let $\delta_i \define \pi'(i) - \sum_{x \in (a_{i-1}, a_i]} \cD(x)$.
If $\delta_i \geq 0$, for $x \in (a_{i-1}, a_i]$, we may obtain $\cF$ on interval $(a_{i-1}, a_i]$
by setting $\cF(z) = \cD(z)+\delta$ for an arbitrarily chosen element $z \in (a_{i-1}, a_i]$, and
then choosing $\cF(x) = \cD(x)$ for the remaining $x \neq z$ in this interval.
 
If $\delta_i < 0$, we may obtain $\cF$ on interval $(a_{i-1}, a_i]$ by subtracting a total of
$|\delta_i|$ from the densities $\cD(x)$ inside the interval; this is possible, since we must have
$\sum_{a_{i-1} < x \leq a_i} \cD(x) \geq |\delta_i|$.

It is easy to verify that this construction satisfies $\pi_{f,\cF} = \pi' = \pi_{g,\cE}$. Since
$\Xi$ is a density property, it holds that $\cF_f \in \Xi$.  We can see that
\begin{align*}
  \dist_\TV(\cD_f, \cF_f)
    &= \frac{1}{2} \sum_x |\cD(x) - \cF(x)| 
    = \frac{1}{2} \sum_i \left( \sum_{a_{i-1} < x \leq a_i} |\cD(x) - \cF(x)| \right)  \\
    &= \frac{1}{2} \sum_i |\delta_i| 
    = \dist_\TV(\pi, \pi') \leq \epsilon \,.
\end{align*}
But this contradicts the assumption that $\cD_f \in \far^\TV_\epsilon(\Xi)$.  So
$\Pi(\far^\TV_\epsilon(\Xi)) \subseteq \far^\TV_\epsilon(\Pi(\Xi))$.
\end{proof}

Density properties are closed under Boolean operations; in particular, we require the next fact,
which follows by definition.

\begin{fact}
\label{fact:density-intersection}
Suppose $\Xi_1$ and $\Xi_2$ are density properties. Then $\Xi_1 \cap \Xi_2$ is a density property.
\end{fact}

We will show that labeled distribution testing, and distribution testing under the parity trace,
are essentially equivalent for density properties. The first step is to show that labeled
distribution testers for density properties can always be transformed into a restricted type of
tester that ignores the absolute position of the sample points, and keeps only their labels and
their order. This proof is inspired by one in \cite{DKN15a}. We require some notation.

For any multiset $S \subseteq \bZ$ of size $m$, write $S = \{ x_1, x_2, \dotsm, x_m \}$ where we put
$x_1 \leq x_2 \leq \dotsm \leq x_m$. Then for any $f : \bZ \to \zo$, we will write the ordered
sequence of points in $S$ labeled by $f$ as
\[
  S_f \define \left( (x_1, f(x_1)), (x_2, f(x_2)), \dotsc, (x_m, f(x_m)) \right) \,.
\]
Then we define
\[
  \trace^*(S_f) \define ( f(x_1), f(x_2), \dotsc, f(x_k) ) \,.
\]

\begin{fact}
\label{fact:trace-distribution-equiv}
Let $(f, \cD)$ be any proper labeled distribution. For any $m$, let $T \sim \samp(\pi_{f,\cD}, m)$
and $S \sim \samp(\cD,m)$. Then the random variables $\trace^*(S_f)$ and $\trace(T)$ are identically
distributed.
\end{fact}

\begin{lemma}
\label{lemma:density-property-ramsey}
Let $\Xi_1$ and $\Xi_2$ be any density properties, let $\alpha \in (0,1)$, and suppose there is
an algorithm $A$ and a number $m$ such that: 
\begin{enumerate}
\item If $\cD_f \in \Xi_1$
  then $\Pru{S \sim \samp(\cD,m)}{ A(S_f) \text{ accepts }} \geq \alpha$; and,
\item If $\cD_f \in \Xi_2$
  then $\Pru{S \sim \samp(\cD,m)}{ A(S_f) \text{ rejects }} \geq \alpha$.
\end{enumerate}
Then for any $\delta > 0$, there is an algorithm $A'$ satisfying
\begin{enumerate}
\item If $\cD_f \in \Xi_1$
  then $\Pru{S \sim \samp(\cD,m)}{ A'(\trace^*(S_f)) \text{ accepts }} \geq \alpha - \delta$; and,
\item If $\cD_f \in \Xi_2$
  then $\Pru{S \sim \samp(\cD,m)}{ A'(\trace^*(S_f)) \text{ rejects }} \geq \alpha - \delta$.
\end{enumerate}
\end{lemma}

To prove this lemma, we require the infinite Ramsey theorem.  For any set $X$ and $n \in \bN$, let
${X \choose n}$ denote the set of size $n$ subsets of $X$.

\begin{theorem}[Infinite Ramsey Theorem~\cite{Ram09}]
\label{thm:ramsey}
Fix any $c,n \in \bN$ and let $X$ be any countably infinite set. For any coloring of ${ X \choose
n}$ by $c$ colors, there exists an infinite set $Y \subseteq X$ such that sets in ${ Y \choose n }$
have the same color.
\end{theorem}

We may now prove our \cref{lemma:density-property-ramsey}.

\begin{proof}[Proof of \cref{lemma:density-property-ramsey}] Since $\Xi_1, \Xi_2$ are density
properties, they have associated sets of density sequences $\Pi(\Xi_1)$ and $\Pi(\Xi_2)$.

On input $(f,\cD)$, the algorithm $A$ receives $S_f$, where $S \sim \samp(\cD,m)$. For each multiset
$S \subseteq \bN$ of size $m$, the algorithm's decision on $S_f$ can be written as $A(S_f) =
A_S(\trace^*(S_f))$, where $A_S : \zo^m \to \zo$. There are at most $b = 2^{2^m}$ possible decision
functions. We identify each possible decision function $\Delta : \zo^m \to \zo$ with an element of
$[b]$.

For each subset $S \subset \bN$ of size $m$, we color $S$ with the function $A_S$, which we have
identified with an element of $[b]$. By \cref{thm:ramsey}, there exists an infinite set $N \subset
\bN$ such that all $S \subset \bN$ of size $m$ have the same color. Then there exists a decision
function $\Delta : \zo^m \to \zo$ such that, for each $S \in {N \choose m}$, $A_S = \Delta$.

We now define the algorithm $A'$ as follows. On input $\trace^*(S_f) = (f(x_1), f(x_2), \dotsc,
f(x_m))$, $A'$ will simply output $\Delta(f(x_1), f(x_2), \dotsc, f(x_m))$. It remains to show that
this algorithm will be correct.

Fix any input $(f,\cD)$; without loss of generality, we assume $(f,\cD) \in \Xi_1$, since the
analogous argument will hold for $\Xi_2$.  Let $a_1 < a_2 < a_3 < \dotsm$ be the alternation
sequence for $f$.  We will also assume that $(f,\cD)$ is 1-proper, since a similar argument will
hold when $(f,\cD)$ is 0-proper (the difference being that we would have $\pi_{f,\cD}(1) = 0$ and
$\pi_{f,\cD}(i+1) = \cD(a_{i-1}, a_i]$ instead of $\pi_{f,\cD}(i) = \cD(a_{i-1}, a_i]$).

Choose $C = \lceil m^2 / \delta \rceil$.  Since $N$ is infinite, we may choose a one-to-one mapping
$\gamma : \bN \to N \cup \{\infty\}$, satisfying
\[
  1 < \gamma(a_1) < \gamma(a_2) < \dotsm < \gamma(a_k) < \dotsm \,,
\]
and define $\gamma(a_0) = 1$ for ease of notation. We may choose $\gamma$ to satisfy $|N \cap
(\gamma(a_{i-1}), \gamma(a_i)] | \geq C$ for each $i \geq 1$. Define a distribution $\cD'$ by
assigning value $\cD'(x) = \frac{1}{C} \cdot \cD( a_{i-1}, a_i ]$ for the first $C$ elements $x \in
N \cap (\gamma(a_{i-1}), \gamma(a_i)]$, and define the function $f' : \bZ \to \zo$ as the unique
function with alternation sequence $\gamma(a_1) < \gamma(a_2) < \dotsm$.  Observe that, for each $i
\in \bN$,
\[
  \pi_{f',\cD'}(i) = \cD(\gamma(a_{i-1}), \gamma(a_i)] = C \cdot \frac{1}{C} \cD(a_{i-1},a_i]
                  = \pi_{f,\cD}(i) \,,
\]
so $(f',\cD') \in \Xi_1$ because $\Xi_1$ is a density property. So
\[
  \alpha \leq \Pru{S' \sim \samp(\cD',m)}{ A_{S'}( \trace^*(S'_{f'}) ) \text{ accepts } } \,.
\]
We have $\supp(\cD') \subseteq N$, so $S' \subseteq N$ with probability 1. So, if $S'$ is a set of
size $m$ (\ie each element of the multiset $S$ occurs with multiplicity 1), then $A_{S'} = \Delta$.
Let $F(S')$ be the event that $S'$ is a set of size $m$. Since each element of $N$ has density at
most $1/C$, the union bound gives
\[
  \Pru{S'}{ \neg F(S') } < \frac{m^2}{C} \leq \delta \,.
\]
Then
\begin{align*}
  \alpha \leq \Pru{S'}{ A_{S'}( \trace^*(S'_{f'}) ) \text{ accepts } }
         &\leq \Pru{S'}{ \Delta( \trace^*(S'_{f'}) ) \text{ accepts } } 
                + \Pru{S'}{ \neg F(S') } \\
         &< \Pru{S'}{ \Delta( \trace^*(S'_{f'}) ) \text{ accepts } } 
                + \delta \,.
\end{align*}
Now observe that, for $S \sim \samp(\cD,m)$, the variables
$\trace^*(S'_{f'})$ and $\trace^*(S_f)$ are identically distributed. So we have
\[
  \alpha-\delta < \Pru{S'}{ \Delta( \trace^*(S'_{f'}) ) \text{ accepts } } 
  = \Pru{S}{ \Delta( \trace^*(S_{f}) ) \text{ accepts } } \,.
\]
This concludes the proof, since we may repeat the analogous argument for $\Xi_2$.
\end{proof}

We may now establish the general equivalence between labeled distribution testing and distribution
testing under the parity trace. (Note that the second part of the lemma below does not require that
$\Xi_1$ and $\Xi_2$ are density properties.)

\begin{lemma}
\label{lemma:labeled-to-linear-trace-reductions}
Let $\Xi_1$ and $\Xi_2$ be properties of labeled distributions. Then:
\begin{enumerate}
\item Suppose that $\Xi_1$ and $\Xi_2$ are density properties.
If there is a $(\Xi_1, \Xi_2)$-labeled distribution tester with sample complexity $m$ and
success probability $\alpha$, then for any $\delta > 0$ there is a $(\Pi(\Xi_1),
\Pi(\Xi_2))$-distribution tester under the parity trace, with sample complexity $m$ and success
probability $\alpha-\delta$.
\item If there is a $(\Pi(\Xi_1), \Pi(\Xi_2))$-distribution tester under the parity
trace with sample complexity $m$ and success probability $\alpha$, then there is a $(\Xi_1,
\Xi_2)$-labeled distribution tester with sample complexity $m$ and success probability $\alpha$.
\end{enumerate}
\end{lemma}
\begin{proof}
Suppose there is a $(\Xi_1, \Xi_2, \alpha)$-labeled distribution tester $A$, with sample
complexity $m$. By \cref{lemma:density-property-ramsey}, for any $\delta > 0$, there is a tester
$A'$ such that
\begin{align*}
  \cD_f \in \Xi_1
    &\implies \Pru{S \sim \samp(\cD,m)}{ A'(\trace^*(S_f)) \text{ accepts }} > \alpha-\delta \\
  \cD_f \in \Xi_2
    &\implies \Pru{S \sim \samp(\cD,m)}{ A'(\trace^*(S_f)) \text{ rejects }} > \alpha-\delta \,.
\end{align*}
Suppose $\pi \in \Pi(\Xi_1)$. Then there exists $\cD_f \in \Xi_1$ such that $\pi =
\pi_{f,\cD}$. Using the fact that $\trace(T)$ and $\trace^*(S_f)$ are identically distributed when
$T \sim \samp(\pi,m)$ and $S \sim \samp(\cD,m)$ (\cref{fact:trace-distribution-equiv}), we have
\[
    \Pru{T \sim \samp(\pi,m)}{ A'(\trace(T)) \text{ accepts }} =
    \Pru{S \sim \samp(\cD,m)}{ A'(\trace^*(S_f)) \text{ accepts }} > \alpha-\delta \,.
\]
The analogous argument holds when $\pi \in \Pi( \Xi_2 )$, which concludes the first part of the
proof.

Now suppose there is a $(\Pi(\Xi_1), \Pi(\Xi_2))$-distribution tester under the parity
trace, with sample complexity $m$ and success probability $\alpha$, so there is an algorithm $A$
such that
\begin{align*}
  \pi \in \Pi(\Xi_1)
    &\implies \Pru{T \sim \samp(\pi,m)}{ A(T) \text{ accepts }} > \alpha \\
  \pi \in \Pi(\Xi_2)
    &\implies \Pru{T \sim \samp(\pi,m)}{ A(T) \text{ rejects }} > \alpha \,.
\end{align*}
Suppose that $\cD_f \in \Xi_1$. Then $\pi_{f,\cD} \in \Pi(\Xi_1)$. Again using the fact that
$\trace^*(S_f)$ and $\trace(T)$ are identically distributed when $S \sim \samp(\cD,m)$ and $T \sim
\samp(\pi_{f,\cD},m)$ (\cref{fact:trace-distribution-equiv}), we have
\[
    \Pru{S \sim \samp(\cD,m)}{ A(\trace^*(S_f)) \text{ accepts }} =
    \Pru{T \sim \samp(\pi,m)}{ A(\trace(T)) \text{ accepts }} > \alpha \,.
\]
The analogous argument holds when $\cD_f \in \Xi_2$, which concludes the proof.
\end{proof}

To introduce the distance metrics into the equivalence, we require:

\begin{proposition}
\label{prop:density-property-edit}
Let $\Xi$ be any density property and let $\epsilon > 0$. Then $\far^\edit_\epsilon(\Xi)$ is a
density property, and
\[
  \far^\edit_\epsilon(\Pi(\Xi)) = \Pi(\far^\edit_\epsilon(\Xi)) \,.
\]
\end{proposition}
\begin{proof}
It is evident that $\far^\edit_\epsilon(\Xi)$ is a density property, because $\dist_\edit((f,\cD),
(g,\cE)) = \dist_\edit(\pi_{f,\cD}, \pi_{g,\cE})$ for any labeled distributions $\cD_f$ and
$\cE_g$, so that $\far^\edit_\epsilon(\Xi)$ is defined entirely by the density sequences.

We first prove $\far^\edit_\epsilon(\Pi(\Xi)) \subseteq \Pi(\far^\edit_\epsilon(\Xi))$.
Let $\pi \in \far^\edit_\epsilon(\Pi(\Xi))$, so that $\dist_\edit(\pi, \pi') > \epsilon$ for all
$\pi' \in \Pi(\Xi)$. Suppose for contradiction that $\pi \notin \Pi(\far^\edit_\epsilon(\Xi))$. Let
$(f,\cD)$ be any labeled distribution such that $\pi_{f,\cD} = \pi$. Then $\dist_\edit((f,\cD), \Xi)
\leq \epsilon$, so there exists $(g,\cE) \in \Xi$, and $\cD'_{f'}$ with $\pi_{f,\cD} =
\pi_{f',\cD'}$, such that $\dist_\TV(\cD'_{f'}, \cE_g) \leq \epsilon$. But then
\[
\dist_\edit(\pi, \Pi(\Xi)) \leq \dist_\edit(\pi, \pi_{g,\cE}) = \dist_\edit(\pi_{f',\cD'},
\pi_{g,\cE}) \leq \dist_\TV(\cD'_{f'}, \cE_g) \leq \epsilon \,,
\]
which is a contradiction. This establishes $\far^\edit_\epsilon(\Pi(\Xi)) \subseteq
\Pi(\far^\edit_\epsilon(\Xi))$.

Next, we prove $\Pi(\far^\edit_\epsilon(\Xi)) \subseteq \far^\edit_\epsilon(\Pi(\Xi))$. Let $\pi \in
\Pi(\far^\edit_\epsilon(\Xi))$, so that $\pi = \pi_{f,\cD}$ for some $(f,\cD)$ that satisfies
$\dist_\edit((f,\cD), \Xi) > \epsilon$. Suppose for contradiction that $\pi \notin
\far^\edit_\epsilon(\Pi(\Xi))$, so that there exists $\pi' \in \Pi(\Xi)$ such that
$\dist_\edit(\pi,\pi') \leq \epsilon$. Then $\pi' = \pi_{g,\cE}$ for some $(g,\cE) \in \Xi$, so
\[
  \dist_\edit((f,\cD), \Xi)
    \leq \dist_\edit((f,\cD), (g,\cE))
    = \dist_\edit(\pi, \pi') \leq \epsilon \,,
\]
which is a contradiction. This concludes the proof.
\end{proof}

We may now establish the arrow $(\LD, \edit) \to (\PT, \edit)$ from \cref{fig:relative-strength}.

\begin{lemma}
    \label{lemma:labeled-edit-tester-to-parity-tester}
    Let $\Xi$ be any density property and suppose there is a $(\Xi, \far^\edit_\epsilon(\Xi),
    \alpha)$-labeled distribution tester with sample complexity $m$. Then for any $\delta > 0$,
    there is a $(\Pi(\Xi), \far^\edit_\epsilon(\Pi(\Xi)), \alpha-\delta)$-distribution tester under
    the parity trace, with sample complexity $m$.
\end{lemma}
\begin{proof}
    By \cref{prop:density-property-edit}, $\far^\edit_\epsilon(\Xi)$ is a density property.
    Therefore \cref{lemma:labeled-to-linear-trace-reductions} yields a
    $(\Pi(\Xi), \Pi(\far^\edit_\epsilon(\Xi)), \alpha-\delta)$-distribution tester under the parity
    trace, for any $\delta>0$, with sample complexity $m$. By \cref{prop:density-property-edit}, we
    obtain a $(\Pi(\Xi), \far_\epsilon^\edit(\Pi(\Xi)), \alpha-\delta)$-distribution tester under
    the parity trace.
\end{proof}

The following simple fact establishes $(\LD, \TV) \to (\LD, \edit)$ from
\cref{fig:relative-strength}.

\begin{fact}
\label{fact:edit-tv-subset}
Let $\Xi$ be any property of (proper) labeled distributions, and let $\epsilon > 0$. Then
$\far^\edit_\epsilon(\Xi) \subseteq \far^\TV_\epsilon(\Xi)$.
\end{fact}
\begin{proof}
This follows from the inequality $\dist_\edit((f,\cD), (g,\cE)) \leq \dist_\TV(\cD_f, \cE_g)$ for
any two proper labeled distributions $\cD_f$ and $\cE_g$.
\end{proof}

We now state a convenient lemma for later use.

\begin{restatable}{lemma}{lemmatestertoparity}
\label{lemma:labeled-tester-to-parity-tester}
Let $\Xi$ be any density property and suppose there is a $(\Xi, \far^\TV_\epsilon(\Xi),
\alpha)$-labeled distribution tester with sample complexity $m$. Then for any $\delta > 0$,
there is a $(\Pi(\Xi), \far^\edit_\epsilon(\Pi(\Xi)), \alpha-\delta)$-distribution tester under the
parity trace, with sample complexity $m$.
\end{restatable}

The arrow $(\PT, \edit) \rightarrow (\LD, \edit)$ is proved as follows. Suppose we have
a $(\Pi, \far^\edit_{\epsilon}(\Pi), \alpha)$-distribution tester under the parity trace.
Let $\Xi = \Xi(\Pi)$ be the corresponding density property, so that $\Pi = \Pi(\Xi)$.
By \cref{prop:density-property-edit},
$\far^\edit_\epsilon(\Pi(\Xi)) = \Pi(\far^\edit_\epsilon(\Xi))$, and thus by
\cref{lemma:labeled-to-linear-trace-reductions} we have a
$(\Xi, \far^\edit_\epsilon(\Xi), \alpha)$-labeled distribution tester.

The arrow $(\PT, \TV) \rightarrow (\LD, \TV)$ is similar. Suppose we have
a $(\Pi, \far^\TV_\epsilon(\Pi), \alpha)$-distribution tester under the parity trace.
Let $\Xi = \Xi(\Pi)$ be the corresponding density property, so that $\Pi = \Pi(\Xi)$.
By \cref{lemma:tv-density-property},
$\Pi(\far^\TV_\epsilon(\Xi)) \subseteq \far^\TV_\epsilon(\Pi(\Xi))$, so we get a
$(\Pi(\Xi), \Pi(\far^\TV_\epsilon(\Xi)), \alpha)$-distribution tester under the parity trace.
By \cref{lemma:labeled-to-linear-trace-reductions} we have a
$(\Xi, \far^\TV_\epsilon(\Xi), \alpha)$-labeled distribution tester.

\subsection{Testing Uniformly \texorpdfstring{$k$}{k}-Alternating Functions}
\label{section:uniform-k-alternating-theorem}
\label{section:edit-tv-equivalence}

We now prove our main result for labeled distribution testing, restated below for convenience,
which is an application of our main \cref{thm:intro-main-informal}. First, we observe that the edit and TV
distances coincide when one of the distributions is uniform; we defer the proof to
\cref{section:edit-distance-uniform-distribution}.

\begin{restatable}{lemma}{lemmaedittotvuniform}
\label{lemma:edit-to-tv-uniform}
    There exists an absolute constant $c > 0$ such that the following holds.
    Let $\pi$ be the distribution over $\bN$ that is uniformly supported on $[k]$, and
    $\pi'$ be another probability distribution over $\bN$ supported within $[k]$.
    Then $\dist_\edit(\pi, \pi') \ge c \cdot \dist_\TV(\pi, \pi')$.
\end{restatable}

\begin{remark}
\label{remark:intro-main-edit}
The statement of our main \cref{thm:intro-main-informal} leaves open the possibility of a uniformity tester
under the parity trace, with respect to the \emph{edit distance}, that beats the lower bound of that
theorem. This is because an edit distance tester is weaker than a TV distance tester, due to 
inequality \eqref{eq:edit-to-tv-inequality}. The above lemma shows that we may strengthen the
lower bound in \cref{thm:intro-main-informal} to hold for testers in the edit distance as well.
\end{remark}

\begin{restatable}{theorem}{thmintromaintesting}
\label{thm:intro-main-testing}
Let $\Xi_1$ be the uniformly $2k$-alternating labeled distributions, and let $\Xi_2$ be the
$2k$-alternating labeled distributions that are $\epsilon$-far in total variation distance from
$\Xi_1$. Then the sample complexity of $(\Xi_1, \Xi_2, 2/3)$-labeled distribution testing is
$\widetilde \Theta\left( \frac{k^{4/5}}{\epsilon^{4/5}} + \frac{\sqrt k}{\epsilon^2}\right)$.
\end{restatable}

\begin{proof}[Proof of upper bound]
Let $K$ be the property of $2k$-alternating labeled distributions, which is a density property,
with $\Pi(K)$ being the property of density sequences supported on $[2k]$. Then $\Xi_2 =
\far^\TV_\epsilon(\Xi_1) \cap K$ and $\Pi(\Xi_2) = \Pi(\far^\TV_\epsilon(\Xi_1) \cap K) =
\Pi(\far^\TV_\epsilon(\Xi_1)) \cap \Pi(K)$, while $\Pi(\Xi_1)$ contains only the uniform
distribution $\mu$ supported on $[2k]$. By \cref{lemma:tv-density-property},
$\Pi(\far^\TV_\epsilon(\Xi_1)) \subseteq \far^\TV_\epsilon(\Pi(\Xi_1))$.  Therefore, a $(\Pi(\Xi_1),
\far^\TV_\epsilon(\Pi(\Xi_1)) \cap \Pi(K), \alpha)$-distribution tester under the parity trace, with
sample complexity $m$, is also a $(\Pi(\Xi_1), \Pi(\far^\TV_\epsilon(\Xi_1) \cap K),
\alpha)$-distribution tester under the parity trace, with sample complexity $m$.  The conclusion now
follows from \cref{lemma:labeled-to-linear-trace-reductions}.
\end{proof}

\begin{proof}[Proof of lower bound]
We begin with a specialized variant of the argument from
\cref{lemma:labeled-edit-tester-to-parity-tester}.
Suppose there is a $(\Xi_1, \Xi_2, 2/3)$-labeled distribution tester with sample complexity $m$,
and recall that $\Xi_2 = \far^\TV_\epsilon(\Xi_1) \cap K$. By \cref{fact:edit-tv-subset}, we have
$\far^\edit(\Xi_1) \subseteq \far^\TV(\Xi_1)$, so this is also a $(\Xi_1, \far^\edit_\epsilon(\Xi_1)
\cap K, 2/3)$-labeled distribution tester. By \cref{prop:density-property-edit},
$\far^\edit_\epsilon(\Xi_1)$ is a density property, so by \cref{fact:density-intersection},
$\far^\edit_\epsilon(\Xi_1) \cap K$ is a density property. From
\cref{lemma:labeled-to-linear-trace-reductions}, we then obtain a $(\Pi(\Xi_1),
\Pi(\far^\edit_\epsilon(\Xi_1) \cap K), 2/3-\delta)$-distribution tester under the parity trace, for
any $\delta > 0$, with sample complexity $m$. Observe that $\Pi(\far^\edit_\epsilon(\Xi_1) \cap K) =
\Pi(\far^\edit_\epsilon(\Xi_1)) \cap \Pi(K)$. By \cref{prop:density-property-edit}, we have a
$(\Pi(\Xi_1), \far^\edit_\epsilon(\Pi(\Xi_1)) \cap \Pi(K), 2/3-\delta)$-distribution tester under
the parity trace.

Note that $\Pi(\Xi_1) = \{ \mu \}$, where $\mu$ is the uniform distribution supported on $[2k]$.
Let $0 < c < 1$ be the constant from \cref{lemma:edit-to-tv-uniform}, and consider any distribution
$\pi \in \far^\TV_{\epsilon/c}(\Pi(\Xi_1))$. Then by \cref{lemma:edit-to-tv-uniform}, we have
\[
  \dist_\edit(\pi, \mu) \geq c \cdot \dist_\TV(\pi, \mu) > c \cdot (\epsilon/c) = \epsilon \,,
\]
so $\pi \in \far^\edit_\epsilon(\Pi(\Xi_1))$. Then $\far^\TV_{\epsilon/c}(\Pi(X_1)) \subseteq
\far^\edit_\epsilon(\Pi(\Xi_1))$.

Therefore, our tester is also a $(\Pi(\Xi_1), \far^\TV_{\epsilon/c}(\Pi(\Xi_1)) \cap \Pi(K),
2/3-\delta)$-distribution tester under the parity trace, with sample complexity $m$. By \cref{thm:
lower bound}, we must have the desired lower bound of
\[
  m = \widetilde \Omega\left( \left(\frac{n}{\epsilon}\right)^{4/5}
                              + \frac{\sqrt n}{\epsilon^2} \right) \,. \qedhere
\]
\end{proof}

\subsection{Promise-Free Testing $k$-Alternating and Uniformly \texorpdfstring{$k$}{k}-Alternating
Functions}
\label{section:testing-uniform-nopromise}

\cref{thm:intro-main-informal} proves a tight bound on testing whether a distribution supported on $[k]$ is
uniform, under the parity trace. We use testing-by-learning to prove a bound on the harder problem
of testing whether a distribution is uniform on $[k]$, \emph{without} the promise that the input is
supported on $[k]$.

\begin{restatable}{theorem}{thmtestinguniformnopromise}
\label{thm:testing-uniform-k-alternating-no-promise}
Fix domain $\bN$.  Let $\Pi$ contain only the uniform distribution over $[k]$. There is a $(\Pi,
\far^\edit_\epsilon(\Pi), 2/3)$-distribution tester under the parity trace, with sample complexity
$O\left(\frac{k}{\epsilon} + \frac{k}{\epsilon^2 \log k}\right)$.
\end{restatable}


This will follow from the next lemma, using fact that $\Pi = \Pi(\Xi)$, where $\Xi$ is the property
of uniformly $k$-alternating labeled distributions, together with
\cref{lemma:labeled-tester-to-parity-tester}.

\begin{lemma}
\label{thm:uniform-k-alternating-learning-bound}
Let $\Xi$ be the uniformly $k$-alternating functions. Then there is a $(\Xi,
\far^\TV_\epsilon(\Xi),2/3)$-labeled distribution tester with sample complexity
$O\left(\frac{k}{\epsilon} + \frac{k}{\epsilon^2 \log k} \right)$.
\end{lemma}
\begin{proof}
Let $c \in (0,1)$ be the universal constant in \cref{lemma:edit-to-tv-uniform}.  We will construct a
learner-verifier pair. Let $A$ be the standard PAC learning algorithm for the class of
$k$-alternating functions, with error $c\epsilon/4$. This algorithm, using a sample of size $m_A =
O(k/\epsilon)$, outputs a $k$-alternating function $g : \bZ \to \zo$, such that with probability at
least $8/9$,
\[
  \dist_\TV(\cD_f, \cD_g) = \Pru{x \sim \cD}{ f(x) \neq g(x) } < c\epsilon/4 \,,
\]
where the equality is due to \cref{fact:tv-to-ham}.  It is clear that there exists a distribution
$\cE$ such that $\cE_g \in \Xi$.

Suppose that $(f, \cD) \in \Xi$.  We must show that there exists a distribution $\cE$ such that
$\cE_g \in \Xi$ and $\dist_\TV(\cD, \cE) \leq \epsilon/4$. Let $\mu$ be the uniform distribution
over $[k+1]$. Using \cref{prop:labeled-distribution-construction}, we obtain $\cE$ such that
\[
  \dist_\TV(\cD_g, \cE_g) = \dist_\TV(\pi_{g,\cD}, \pi_{g,\cE}) = \dist_\TV(\pi_{g,\cD}, \mu) \,.
\]
Using $\pi_{g,\cE} = \mu = \pi_{f,\cD}$ and \cref{lemma:edit-to-tv-uniform}, we have
\[
  \dist_\TV(\pi_{g,\cD}, \pi_{g,\cE})
  = \dist_\TV(\pi_{g,\cD}, \pi_{f,\cD})
  \leq \frac{1}{c} \cdot \dist_\edit(\pi_{g,\cD}, \pi_{f,\cD})
  \leq \frac{1}{c} \cdot \dist_\TV(\cD_g , \cD_f) \,,
\]
where the last inequality holds by definition. Then
\[
  \dist_\TV(\cD, \cE) = \dist_\TV(\cD_g, \cE_g)
    \leq \frac{1}{c} \cdot \dist_\TV(\cD_g, \cD_f) \leq \epsilon/4 \,.
\]
So algorithm $A$ satisfies the conditions for the learner-verifier pair. It remains to construct the
verifier $B$.  Let $g : \bZ \to \zo$ be any possible output of the learner, which must be a
$k$-alternating function. Then $\Pi_g$ is the set of all distributions $\cE$ such that $\pi_{g,\cE}
= \mu$, where $\mu$ is the uniform distribution over $[k+1]$. Define the algorithm $B_g$ as follows:
\begin{enumerate}
\item Sample $S \sim \samp(\cD,m_B)$ and construct the multiset $S'$ by taking each $x \in S$ and
including the number $i \in [k+1]$ in $S'$, where $i$ is the unique interval $(a_{i-1}, a_i]$ that
contains $x$. Then $S'$ is distributed as $\samp(\pi_{g,\cD}, m_B)$.
\item Use an $(\epsilon/4, \epsilon/2)$-tolerant uniformity tester on sample $S'$ to test if
$\dist_\TV(\pi_{g,\cD}, \mu) < \epsilon/4$ or $\dist_\TV(\pi_{g,\cD}, \mu) > \epsilon/2$.
This step requires $m_B = O\left(\frac{k}{\epsilon^2 \log k}\right)$ samples \cite{VV17a}.
\end{enumerate}
Suppose that $\cE \in \close^\TV_{\epsilon/4}(\Pi_g)$, so there exists $\cF$ such that
$\pi_{g,\cF} = \mu$ and $\dist_\TV(\cE_g, \cF_g) < \epsilon/4$. Then
\[
  \dist_\TV(\pi_{g,\cE}, \mu)
  = \dist_\TV(\pi_{g,\cE}, \pi_{g,\cF})
  \leq \dist_\TV(\cE_g, \cF_g) < \epsilon/4 \,,
\]
so the tolerant uniformity tester will accept.

Now suppose that $\cE \in \far^\TV_{\epsilon/2}(\Pi_g)$. For contradiction, suppose that
\[
  \dist_\TV(\pi_{g,\cE}, \mu) \leq \epsilon/2 \,.
\]
Using \cref{prop:labeled-distribution-construction}, we obtain $\cF$ such that $\pi_{g,\cF} = \mu$
(so $\cF \in \Pi_g$), and
\[
  \dist_\TV(\cE_g, \cF_g) = \dist_\TV(\pi_{g,\cE}, \pi_{g,\cF}) = \dist_\TV(\pi_{g,\cE}, \mu) \leq
    \epsilon/2 \,.
\]
This contradicts $\cE \in \far^\TV_{\epsilon/2}(\Pi_g)$, so it must be the case that
$\dist_\TV(\pi_{g,\cE}, \mu) > \epsilon/2$. Then the tolerant uniformity tester will correctly
reject.
\end{proof}

We also note that the no-promise problem of testing the $k$-alternating labeled distributions
inherits an upper bound of $O(k/\epsilon)$ from the equivalence to distribution-free sample-based
testing:

\begin{lemma}
    \label{lemma:labeled-distribution-testing-k-alternating}
    Let $\Xi$ be the set of labeled distributions $(f, \cD)$ such that $f$ is a $k$-alternating
    function and $\cD$ is any distribution over $\bZ$.
    Let $m(k, \epsilon)$ be the optimal sample complexity of a distribution-free sample-based
    tester for $k$-alternating functions.
    Then the optimal sample complexity of
    $(\Xi, \far^\TV_\epsilon(\Xi))$-labeled distribution testing is $\Theta(m(k, \epsilon))$.
    In particular, there is such a tester with sample complexity $O(k/\epsilon)$.
\end{lemma}
\begin{proof}
    The first part of the statement follows from \cref{prop:tester-equivalence}.
    The second part follows from standard PAC learning theory, since the class of $k$-alternating
    functions has VC dimension $k+1$, along with the testing by learning reduction \cite{GGR98}.
\end{proof}

As with the previous result, this lemma implies a bound for testing the \emph{support size}
of distributions under the parity trace; this is the starting point for our next discussion, on the
connections between distribution testing under the parity trace and distribution-free sample based
testing.

\subsection{Distribution-Free Sample-Based Property Testing}
\label{section:testing-support-size}

We now prove that testing support size $k$ under the parity trace is equivalent to testing
$k$-alternating functions in the standard distribution-free sample-based model (whose optimal sample
complexity is unknown \cite{BFH21}). This has the interesting consequence, in
\cref{lemma:support-size-to-halfspace}, that an improved lower bound on testing support size under the
parity trace could give a better lower bound for testing halfspaces in the distribution-free
sample-based model.

We require the following proposition about edit distance, which is proved in
\cref{section:edit-distance-support-size}.

\begin{restatable}{proposition}{propedittvdistancestosupportsizek}
\label{lemma:edit-tv-distances-to-support-size-k}
Let $\Xi$ be the property of proper labeled distributions $(g,\cE)$ where $\pi_{g,\cE}$ has support
size at most $k$. Then for any proper labeled distribution $(f,\cD)$, $\dist_\TV((f,\cD), \Xi) \leq
\dist_\edit((f,\cD), \Xi)$.
\end{restatable}

\begin{restatable}{theorem}{thmtestingsupportk}
\label{thm:testing-support-k}
Let $\Pi$ be the class of distributions on domain $\bN$ with support size at most $k$. Let
$m_1(k,\epsilon)$ be the optimal sample complexity of a $(\Pi, \far^\edit_\epsilon(\Pi),
2/3)$-distribution tester under the parity trace, and let $m_2(k,\epsilon)$ be the optimal sample
complexity of a distribution-free sample-based tester for $k$-alternating functions. Then
$m_1(k,\epsilon) = \Theta(m_2(k,\epsilon))$.
\end{restatable}

\begin{proof}[Proof of first direction]
We wish to construct a $(\Pi, \far^\edit_\epsilon(\Pi), 2/3)$-distribution tester under the parity
trace, with sample complexity $O(m_2(k,\epsilon))$.
Let $\Xi$ be the set of labeled distributions $(f,\cD)$ such that $f$ is a $(k-1)$-alternating
function and $\cD$ is any distribution over $\bZ$.
Let $\Xi'$ be the set of labeled distributions
$(f,\cD)$ on domain $\bZ$ such that $\pi_{f,\cD}$ has support size at most $k$, so that $\Pi(\Xi') =
\Pi$. \cref{lemma:labeled-distribution-testing-k-alternating} gives a
$(\Xi, \far^\TV_\epsilon(\Xi), 3/4)$-labeled distribution tester with sample complexity
$O(m_2(k, \epsilon))$.
We will construct a $(\Xi', \far^\TV_\epsilon(\Xi'), 3/4)$-labeled distribution tester with
sample complexity $O(m_2(k,\epsilon))$, from which the conclusion will hold by
\cref{lemma:labeled-tester-to-parity-tester}.

Observe that $\Xi \subset \Xi'$ since for any $(k-1)$ alternating function $f$ and any distribution
$\cD$, $\pi_{f,\cD}$ has support size at most $k$.  We show that for any labeled distribution $(f,
\cD) \in \Xi'$, there exists a labeled distribution $(g, \cD) \in \Xi$ such that $\dist_\TV(\cD_f,
\cD_g) = 0$, so that $\cD_f = \cD_g$. Let $a_1 < a_2 < \dotsm$ be the alternation sequence for $f$,
and use the convention $a_0 = -\infty$. Since $\pi_{f,\cD}$ has support size at most $k$, there are
at most $k$ intervals $(a_{i-1}, a_i]$ such that $\cD(a_{i-1}, a_i] > 0$. Construct $g$ by assigning
$g(x) = 1- f(x)$ for all $x$ belonging to any interval $(a_{i-1}, a_i]$ that satisfies $\cD(a_{i-1},
a_i] = 0$. By \cref{fact:edit-zero-mass}, we have $\dist_\TV(\cD_f, \cD_g) = 0$. The resulting
function has at most $k-1$ alternation points, so $(g,\cD) \in \Xi$.

Suppose $(f, \cD) \in \far^\TV_\epsilon(\Xi)$ and suppose for contradiction that there is $(g,\cE)
\in \Xi'$ such that $\dist_\TV(\cD_f, \cE_g) \leq \epsilon$. Then there is $(g', \cE') \in \Xi$ such
that $\cE'_{g'} = \cE_g$, so $\dist_\TV(\cD_f, \cE'_{g'}) = \dist_\TV(\cD_f, \cE_g) \leq \epsilon$,
and $(f, \cD) \notin \far^\TV_\epsilon(\Xi)$, a contradiction. So $(f, \cD) \in
\far^\TV_\epsilon(\Xi')$. Then $\far^\TV_\epsilon(\Xi) = \far^\TV_\epsilon(\Xi')$.

Then any $(\Xi, \far^\TV_\epsilon(\Xi), 3/4)$-labeled distribution tester is also a $(\Xi',
\far^\TV_\epsilon(\Xi'), 3/4)$-labeled distribution tester, since samples from elements of $\Xi$ are
indistinguishable from samples from elements of $\Xi'$.
\end{proof}

\begin{proof}[Proof of second direction]
We wish to construct a $(\Xi, \far^\TV_\epsilon(\Xi), 2/3)$-labeled distribution tester; then the
conclusion will follow from \cref{prop:tester-equivalence}.

As shown in the upper bound argument, this is equivalent to a $(\Xi', \far^\TV_\epsilon(\Xi'),
2/3)$-labeled distribution tester, where $\Xi'$ is the class of labeled distributions $(f,\cD)$
where $\pi_{f,\cD}$ has support size at most $k$. Note that $\Pi = \Pi(\Xi')$.

By \cref{lemma:edit-tv-distances-to-support-size-k}, $\far^\TV_\epsilon(\Xi') =
\far^\edit_\epsilon(\Xi')$, so this is equivalent to a $(\Xi', \far^\edit_\epsilon(\Xi'),
2/3)$-labeled distribution tester. Since $\Xi'$ and $\far^\edit_\epsilon(\Xi')$ are density
properties, it suffices to obtain a $(\Pi(\Xi'), \Pi(\far^\edit_\epsilon(\Xi')), 2/3)$-distribution
tester under the parity trace, due to \cref{lemma:labeled-to-linear-trace-reductions}.  Finally,
apply \cref{prop:density-property-edit}.
\end{proof}

The above theorem relates the sample complexity of testing $k$-alternating functions to the
complexity of testing support size under the parity trace, with respect to the edit distance.
From here, we will reproduce the result of \cite{BFH21}, that testing $k$-alternating functions
requires $\Omega(k / \log k)$ samples, which will follow from the $\Omega(n / \log n)$ lower bound
for estimating support size, due to \cite{VV11,WY19}.
We use the following formulation of the result of \cite{VV11}:

\begin{theorem}[\cite{VV11}]
    For any sufficiently small constant $\delta > 0$, there exists a pair of distributions
    $\pi^+, \pi^-$ whose non-zero densities are at least $\frac{1}{n}$,
    such that $\pi^+$ has support size at least $(1-\delta)n$, $\pi^-$ has support size at most
    $(1+\delta) \frac{n}{2}$, and distinguishing between them requires
    $\Omega\left(\frac{n}{\log n}\right)$ samples.
\end{theorem}

Their result also applies to estimating the \emph{entropy} of distributions, in which
case they obtain an $\Omega(n/\epsilon \log n)$ lower bound by constructing distributions
$\pi^+_\epsilon, \pi^-_\epsilon$ that with probability $\epsilon$ draw from $\pi^+,\pi^-$
respectively, and otherwise draw an element $\bot$; this shrinks the entropy gap to an $\epsilon$
fraction of the original gap, and distinguishing between $\pi^+_\epsilon$ and $\pi^-_\epsilon$
requires an $\frac{1}{\epsilon}$ fraction more samples.
While this argument does not apply to the support size estimation problem, which requires that
densities be lower bounded by $1/n$, it does apply to testing support size against TV
distance:

\begin{corollary}
    \label{thm:supp-size-testing-lower-bound}
    Let $\Pi$ be the set of distributions over $\bN$ with support size at most $n$, and let
    $\epsilon > 0$. Then any $(\Pi, \far^\TV_\epsilon(\Pi))$-distribution tester requires sample
    size at least $\Omega\left(\frac{n}{\epsilon \log n}\right)$.
\end{corollary}

To apply this lower bound, we reduce from testing with respect to TV distance, to testing with
respect to the edit distance. We require the following lemma, whose proof we defer to
\cref{section:edit-distance-distribution-support-size}.

\begin{restatable}{lemma}{lemmadistfreereductionedit}
\label{lemma:distr-free-reduction-edit}
    Let $k \in \bN$.
    Let $\Pi_k$ be the set of distributions over $\bN$ supported on at most $k$ elements, and
    let $\Pi_{2k}$ be the set of distributions over $\bN$ supported on at most $2k$ elements.
    Let $\pi$ be a finitely-supported probability distribution over $\bN$, and let $\pi'$ be the
    probability distribution over $\bN$ given by $\pi'(2i-1) = \pi'(2i) \define \pi(i)/2$ for each
    $i \in \supp(\pi)$.
    Then $\dist_\edit(\pi', \Pi_{2k}) \ge \frac{1}{4} \cdot \dist_\TV(\pi, \Pi_k)$.
\end{restatable}

We may now establish the simple reduction.

\begin{lemma}
\label{lemma:support-size-lower-bound}
Let $\Pi_{2k}$ be the set of distributions on domain $\bN$ with support size at most $2k$, and let
$\epsilon > 0$. Then the sample complexity of a $(\Pi_{2k},
\far^\edit_\epsilon(\Pi_{2k}))$-distribution tester is at least
$\Omega\left(\frac{k}{\epsilon \log k}\right)$.
\end{lemma}
\begin{proof}
Let $\Pi_k$ be the set of distributions on domain $\bN$ with support size at most $k$. We will
reduce $(\Pi_k, \far^\TV_\epsilon(\Pi_k))$-distribution testing to $(\Pi_{2k},
\far^\edit_{\epsilon/4}(\Pi_k))$-distribution testing, from which the conclusion follows, due to
\cref{thm:supp-size-testing-lower-bound}.

On input distribution $\pi$, the algorithm proceeds as follows. We define the distribution $\pi'$
where for each $i \in \bN$, $\pi'(2i) \define \pi(i)/2$ and $\pi'(2i-1) \define \pi(i)/2$. The
algorithm may simulate a sample from $\pi'$ by sampling $\bm i \sim \pi$ and then taking $2\bm i$ or
$2 \bm i - 1$ with equal probability. The algorithm then simulates the $(\Pi_{2k},
\far^\edit_{\epsilon/4}(\Pi_{2k}))$-distribution tester on input $\pi'$. If $\pi \in \Pi_k$, then it
is clear that $\pi' \in \Pi_{2k}$, so the algorithm will correctly accept (with probability at least
$2/3$). If $\pi \in \far^\TV_\epsilon(\Pi_k)$, then by \cref{lemma:distr-free-reduction-edit}, $\pi'
\in \far^\edit_{\epsilon/4}(\Pi_{2k})$, so the algorithm will correctly reject (with probability at
least $2/3$).
\end{proof}

We are now prepared to recover a number of the results of \cite{BFH21} for testing properties with
domain $\bR^n$, including halfspaces, intersections of halfspaces, and decision
trees. The idea is to reduce from testing $k$-alternating functions to testing the property in
question, by taking the one-dimensional space and embedding it into $\bR^n$ in an appropriate way.
This technique was also used in \cite{ES20,BFH21}. We provide a formal proof for halfspaces, and
refer to \cite{BFH21} for the details on intersections of halfspaces and decision trees.

\begin{lemma}
\label{lemma:support-size-to-halfspace}
Let $\Pi$ be the property of distributions on $\bN$ with support size at most $k$. For any $d$ and
$\epsilon > 0$, let $h(d,\epsilon)$ be the sample complexity of testing \emph{halfspaces} on domain
$\bR^d$, in the distribution-free sample-based model.Then there is a $(\Pi,
\far^\edit_\epsilon(\Pi), 2/3)$-distribution tester under the parity trace, with sample complexity
$O(h(k+1,\epsilon))$.
\end{lemma}
\begin{proof}
This follows from \cref{thm:testing-support-k} and the following reduction from testing
$k$-alternating functions to testing halfspaces.

On input $f : \bZ \to \zo$ and distribution $\cD$ over $\bZ$, consider the one-to-one function
$\psi(x) \define (1, x, x^2, \dotsc, x^k)$ and the distribution $\psi\cD$ defined as the
distribution over $\psi(x)$ where $x \sim \cD$. The tester will simulate the halfspace tester on
samples $(\psi(x), f(x))$ where $(x, f(x)) \sim \cD_f$. 

Note that a function $f : \bZ \to \zo$ is $k$-alternating if and only if there exists a degree $k$
polynomial $p : \bZ \to \bR$ such that $f(x) = \sign(p(x))$, where we define $\sign(z) = 0$ of $z
\leq 0$ and $\sign(z) = 1$ if $z > 0$.  Then $f$ is $k$-alternating if and only if there exists $w =
(w_0, w_1, \dotsc, w_k) \in \bR^{k+1}$ such that $f(x) = \sign( \inn{w, \psi(x)} )$. So $f$ is
$k$-alternating if and only if there exists a halfspace $g : \bR^{k+1} \to \zo$ such that $f(x) =
g(\phi(x))$ on all $x$.

Write $\psi(\cD_f)$ for the distribution of $(\psi(x), f(x))$ when $x \sim \cD$.
So for any $k$-alternating function $f$, there exists a halfspace $g$ such that $\psi(\cD_f) =
(\psi\cD)_g$. On the other hand, for any halfspace $g$, there exists a $k$-alternating function $f$
such that $\psi(\cD_f) = (\psi\cD)_g$.

Then, for any $k$-alternating function $f$ and distribution $\cD$, samples from $\psi(\cD_f)$ are
indistinguishable from samples from $(\psi\cD)_g$, where $g$ is a halfspace, so a halfspace tester
will accept. On the other hand, for any function $f$ that is $\epsilon$-far from $k$-alternating
under distribution $\cD$, consider an arbitrary function $f' : \bR^{k+1} \to \zo$ such that
$f'(\psi(x)) = f(x)$ on all $x \in \bZ$, so $(\psi\cD)_{f'} = \psi(\cD_f)$, so halfspace tester
will perform identically on the simulated samples $(\psi(x), f(x))$ as on the samples $(z, f'(z))
\sim (\psi\cD)_{f'}$, so the tester performs as if it was given input $f'$ and $\psi\cD$.

If there exists a halfspace $g$ such that $\Pru{x \sim \cD}{ f'(\psi(x)) = g(\psi(x)) } \leq
\epsilon$, then $\Pru{x \sim \cD}{ f(x) = g(\psi(x)) } \leq \epsilon$, where $g(\psi( \cdot ))$ is
$k$-alternating, which is a contradiction. So it must be that $f'$ is $\epsilon$-far from being a
halfspace with respect to $\psi\cD$, so the halfspace tester rejects $f'$.
\end{proof}

With this reduction, together with the lower bound provided by 
\cref{lemma:support-size-lower-bound}, we recover the following bounds. Note that \cite{BFH21}
only stated their bounds for constant $\epsilon$, but the amplification argument above could also
be applied directly to their results.

\begin{corollary}[See \cite{BFH21}]
\label{cor:bfh-recovery}
The following lower bounds hold for the sample complexity of testing in the distribution-free
sample-based model:
\begin{enumerate}[topsep=0pt,itemsep=0pt]
\item $k$-Alternating functions over $\bR$: $\Omega\left(\frac{k}{\epsilon \log k}\right)$;
\item Halfspaces over $\bR^n$: $\Omega\left(\frac{n}{\epsilon \log n}\right)$;
\item Intersections of $k$ halfspaces over $\bR^n$: $\Omega\left(\frac{nk}{\epsilon \log(nk)}\right)$;
\item Size $k$ decision trees over $\bR^n$: $\Omega\left(\frac{k}{\epsilon \log k}\right)$.
\end{enumerate}
\end{corollary}

\section{Property Testing in the Trace Reconstruction Model}
\label{section:trace-testing}

We begin by formally defining property testing in the trace reconstruction model. We then discuss
the connection between the (relative) edit distance on strings and the edit distance on
distributions that we introduced, which will be a crucial component of our results.

For a string $x \in \zo^N$ and retention rate $\rho \in (0,1)$, $\del(x,\rho)$ is the distribution
of substrings of $x$ obtained by deleting each character of $x$ independently with probability
$1-\rho$.  A sample $\bm{T} \sim \del(x,\rho)$ is called a \emph{trace} from $x$ with deletion rate
$1-\rho$.

\begin{definition}[Trace Testing]
Let $\Psi_1$ and $\Psi_2$ be properties of strings in $\zo^N$, and let $\alpha, \rho \in (0,1)$,
which we call the \emph{success probability} and \emph{retention rate}, respectively. A
\emph{$(\Psi_1, \Psi_2, \rho, \alpha)$-trace testing algorithm using $m$ traces} is an algorithm $A$
such that, for $m$ independent traces $\bm{T}_1(x), \dotsc, \bm{T}_m(x)$ obtained from $x$ with
deletion rate $1-\rho$,
\begin{enumerate}
\item If $x \in \Psi_1$ then
    $\Pr{ A(\bm{T}_1(x), \dotsc, \bm{T}_m(x)) \text{ accepts } } \geq \alpha$; and,
\item If $x \in \Psi_2$ then
    $\Pr{ A(\bm{T}_1(x), \dotsc, \bm{T}_m(x)) \text{ rejects } } \geq \alpha$.
\end{enumerate}
\end{definition}

Many of our results refer to $n$-block strings and uniform $n$-block strings, which we now define.

\begin{definition}
    Fix $N \in \bN$.
    We say $x \in \{0,1\}^N$ is an \emph{$n$-block string} if $x$ consists of at most $n$ blocks,
    where a block is an all-1s string or all-0s string. For integer $n$ that divides $N$, the
    \emph{1-uniform $n$-block string} is
    $1^{N/n} 0^{N/n} 1^{N/n} \dotsm \parity(n)^{N/n}$ and the \emph{0-uniform $n$-block string}
    is $0^{N/n} 1^{N/n} 0^{N/n} \dotsm (1-\parity(n))^{N/n}$. We say that $x$ is
    a \emph{uniform $n$-block string} if it is the 1-uniform or 0-uniform $n$-block string.
\end{definition}

\newcommand{\stringedit}{\mathsf{string-edit}}

\begin{definition}[Relative Edit Distance]
    Writing $\dist_\stringedit : \zo^* \times \zo^* \to \bZ$ for the edit distance on
    strings, we define the \emph{relative} string edit distance
    on strings $x \in \zo^N$ and $y \in \zo^M$ as
    \[
        \dist_\reledit(x,y) = \frac{2}{N+M} \dist_\stringedit(x,y) \,.
    \]
\end{definition}

We define a correspondence between strings and probability distributions, which allow us to
relate property testing for trace reconstruction to distribution testing under the parity trace.

\begin{definition}[String to Distribution Correspondence]
\label{def:string-to-distribution}
For any fixed $N \in \bN$ and probability distribution $\pi$ over $\bN$, whose densities are integer
multiples of $1/N$, we define the string $\psi(\pi) \in \zo^N$ as
\[
  \psi(\pi) \define 1^{\pi(1) \cdot N} 0^{\pi(2) \cdot N} 1^{\pi(3) \cdot N} 0^{\pi(3) \cdot N}
    \dotsm \,,
\]
where $b^k$ denotes the character $b$ repeated $k$ times.  This map is not one-to-one. But, for
strings $x \in \zo^N$, we define a probability distribution $\psi^{-1}(x)$ as follows. Define the
function $f_x : \bN \to \zo$ as $f_x(i) = x_i$ for each $i \in [N]$ (and 1 elsewhere), and let $\cD$
be the uniform distribution over $[N]$. Then
\[
  \psi^{-1}(x) \define \pi_{f_x,\cD} \,.
\]
One may verify that $\psi(\psi^{-1}(x)) = x$ for any string $x \in \zo^N$.
To each property $\Psi$ of strings in $\zo^N$, we associate the property of probability
distributions $\Pi = \Pi(\Psi) \define \{ \psi^{-1}(x) : x \in \Psi \}$, with $\psi^{-1}$ as
defined in \cref{def:string-to-distribution}.
For any such $\Psi$, let $\far^\reledit_\epsilon(\Psi)$ denote the set of strings $x \in \zo^N$
such that $\dist_\reledit(x, \Psi) > \epsilon$, where
\[
    \dist_\reledit(x, \Psi) \define \min_{y \in \Psi} \dist_\reledit(x, y) \,.
\]
\end{definition}

\begin{observation}
    \label{obs:n-block-support}
    If $x$ is an $n$-block string, then $\psi^{-1}(x)$ is supported on at most $n$ elements.
    If $x$ is the 1-uniform $n$-block string, then $\psi^{-1}(x)$ is the uniform distribution over
    $\{1,2,\dotsc,n\}$.
    If $x$ is the 0-uniform $n$-block string, then $\psi^{-1}(x)$ is the uniform distribution over
    $\{2,3,\dotsc,n+1\}$.
\end{observation}

Recall that we have defined the edit distance on distributions as the natural metric for
distribution testing under the parity trace. The next lemma shows that the edit distance for
distributions is essentially equivalent to the relative edit distance on strings, under the
string-to-distribution correspondence. This will allow us to obtain equivalences between
distribution testing under the parity trace, and property testing for trace reconstruction.
We defer the proof to \cref{section:equivalence-edit-distance-strings}.

\begin{restatable}{lemma}{lemmarelativeeditdistance}
\label{lemma:relative-edit-distance}
Fix any $N$ and let $\pi, \pi'$ be probability distributions over $\bN$ whose densities are integer
multiples of $1/N$. Then
$\frac{1}{2} \cdot \dist_\reledit(\psi(\pi), \psi(\pi'))
\leq \dist_\edit(\pi, \pi')
\leq \dist_\mathsf{rel-edit}(\psi(\pi), \psi(\pi'))$.
\end{restatable}

\subsection{Single-Trace Upper Bounds}

We seek to obtain algorithms for testing properties of strings, with respect to the relative
edit distance, by reducing to testing properties of distributions under the parity
trace, with respect to the edit distance on distributions. We will make use of the following simple
technique, which turns a trace from a string (\ie produced by a deletion channel) into
(the parity trace of) a Poissonized sample from the associated probability distribution
(\ie the result of sampling with replacement).

\newcommand{\ZTPoi}{\mathsf{Poi}_{>0}}

\begin{proposition}
    \label{prop:poissonize-trace}
    Fix $\rho \in (0,1)$. There exists an algorithm \textsc{Poissonize} which consumes a binary
    string and produces another binary string satisfying the following. Let $N \in \bN$ and
    $x \in \zo^N$, and suppose the input is a random trace from $x$ with deletion
    rate $1-\rho$. Then the output is distributed as $\trace(\bm{S})$, where
    $\bm{S} \sim \samp(\psi^{-1}(x), \bm{m})$ and
    $\bm{m} \sim \Poi\left(N \log\left(\frac{1}{1-\rho}\right)\right)$.
\end{proposition}
\begin{proof}
    The idea is to treat each symbol in the input as indicating the event that a corresponding
    Poisson random variable was non-zero, and then up-sample the symbol to the appropriate
    conditional distribution to obtain a Poissonized sample.

    Let $\lambda \define \log\left(\frac{1}{1-\rho}\right)$ and let $\ZTPoi(\lambda)$ denote the
    distribution of a $\Poi(\lambda)$ random variable conditional on being nonzero.

    The algorithm proceeds as follows: on input string $s$, for each symbol $s_j$ from left to
    right, independently sample $\bm{z}_j \sim \ZTPoi(\lambda)$ and append $\bm{z}_j$ copies of
    $s_j$ to the output.

    For each $i \in [N]$, let $\bm{X}_i \sim \Ber(\rho)$ independently.
    Then the input is distributed as
    \[
        x_1^{\bm{X}_1} x_2^{\bm{X}_2} \dotsm x_{N-1}^{\bm{X}_{N-1}} x_N^{\bm{X}_N} \,.
    \]
    Let $\pi \define \psi^{-1}(x)$, say it is supported on $[n]$. For each $i \in [n]$, let
    $\bm{Y}_i \sim \Poi(N \lambda \pi(i))$ independently.
    Then the target output distribution is identical to that of
    \[
        1^{\bm{Y}_1} 0^{\bm{Y}_2} \dotsm \parity(n-1)^{\bm{Y}_{n-1}} \parity(n)^{\bm{Y}_n} \,.
    \]
    By additivity of the Poisson distribution and definition of $\psi^{-1}$, this distribution
    is identical to
    \[
        x_1^{\bm{Z}_1} x_2^{\bm{Z}_2} \dotsm x_{N-1}^{\bm{Z}_{N-1}} x_N^{\bm{Z}_N} \,.
    \]
    where for each $i \in [N]$, $\bm{Z}_i \sim \Poi(\lambda)$ independently.

    By considering the random process that produces the trace from $x$ along with the random process
of the algorithm, we may identify each symbol in the output of the algorithm with the location $i
\in [N]$ corresponding to the appearance of $x_i$ in the trace.  For each $i \in [N]$, let
$\bm{K}_i$ be the random variable denoting how many times $x_i$ was appended to the output. Then
$\bm{K}_i$ is distributed according to the following random process: if $\bm{X}_i = 0$ then
$\bm{K}_i \gets 0$, otherwise $\bm{K}_i \gets \ZTPoi(\lambda)$. Note that the $\bm{K}_i$ are
mutually independent, and the output of the algorithm is
    \[
        x_1^{\bm{K}_1} x_2^{\bm{K}_2} \dotsm x_{N-1}^{\bm{K}_{N-1}} x_N^{\bm{K}_N} \,.
    \]
    Therefore we will be done if, for each $i \in [N]$, $\bm{Z}_i$ and
    $\bm{K}_i$ are distributed identically, which we now check. We have
    $\Pr{\bm{K}_i = 0} = \Pr{\bm{X}_i = 0} = 1 - \rho$ and
    $\Pr{\bm{Z}_i = 0} = e^{-\lambda} = 1 - \rho$, and for each $k \ge 1$,
    \begin{align*}
        \Pr{\bm{Z}_i = k}
        &= \Pr{\bm{Z}_i > 0} \Pruc{}{\bm{Z}_i = k}{\bm{Z}_i > 0}
        = \rho \cdot \Pr{\ZTPoi(\lambda) = k} \\
        &= \Pr{\bm{X}_i = 1} \cdot \Pruc{}{\bm{K}_i = k}{\bm{X}_i = 1}
        = \Pr{\bm{K}_i = k \text{ and } \bm{X}_i=1}
        = \Pr{\bm{K}_i = k} \,. \qedhere
    \end{align*}
\end{proof}

\begin{remark}
    Although we assume that $\rho$ is explicitly known to obtain \cref{prop:poissonize-trace},
    this assumption is not crucial: if we only knew a lower bound $\rho_0$ on $\rho$, we could
    obtain essentially equivalent results by sub-sampling $\Bin(N, \rho')$ elements from
    the trace with $\rho' = \rho_0/C$ for some large constant $C$. Then, except with negligible
    probability of failure, the sample would be distributed as a trace with known deletion rate
    $1-\rho'$.
\end{remark}

Equipped with this result, we obtain testers in the trace reconstruction model from testers
in the parity trace model via a black-box reduction.

\begin{lemma}
    \label{lemma:trace-tester-from-parity-trace-tester}
    Let $N, m \in \bN$.
    Let $\Psi_1, \Psi_2$ be properties of strings in $\zo^N$, and let $\alpha > 0$.
    If there is a Poissonized $(\Pi(\Psi_1), \Pi(\Psi_2), \alpha)$-distribution tester under the
    parity trace with sample complexity $m$, then there is a
    $(\Psi_1, \Psi_2, \rho, \alpha)$-trace tester using one trace, for $\rho = 1 - e^{-m/N}$.
\end{lemma}
\begin{proof}
    Let $A$ be a Poissonized $(\Pi(\Psi_1), \Pi(\Psi_2), \alpha)$-distribution
    tester under the parity trace with sample complexity $m$. Our trace tester $B$ works as follows:
    \begin{enumerate}
        \item Receive a trace $x$ with deletion rate $1-\rho$.
        \item Let $y \gets \textsc{Poissonize}(x)$.
        \item Return $A(y)$.
    \end{enumerate}
    Let $\bm{x}$ and $\bm{y}$ be random variables denoting the inputs to $B$ and $A$, respectively.
    By \cref{prop:poissonize-trace}, $\bm{y}$ is distributed as $\trace(\bm{S})$ where
    $\bm{S} \sim \samp(\psi^{-1}(x), \bm{m})$ and $\bm{m} \sim \Poi(N\lambda)$, where
    $\lambda = \log\left(\frac{1}{1-\rho}\right) = m/N$.
    In other words, $\bm{y}$ is distributed as the parity trace of a sample from $\psi^{-1}(x)$ of
    size $\Poi(m)$. Moreover, by definition of $\Pi(\Psi_1)$ we have that if
    $x \in \Psi_1$ then $\psi^{-1}(x) \in \Pi(\Psi_1)$, and the same for $\Psi_2$. Therefore
    the correctness of $B$ follows from the correctness of $A$.
\end{proof}

We now conclude each of our single-trace upper bounds from
\cref{section:intro-trace-reconstruction}, using the following immediate consequence of
the equivalence of edit distances between strings and distributions.

\begin{proposition}
    \label{prop:far-reledit-far-edit}
    Let $N \in \bN$ and $\epsilon > 0$, and let $\Psi$ be a property of strings in $\zo^N$.
    Then $\Pi(\far^\reledit_\epsilon(\Psi)) \subseteq \far^\edit_{\epsilon/2}(\Pi(\Psi))$.
\end{proposition}
\begin{proof}
    Let $\pi \in \Pi(\far^\reledit_\epsilon(\Psi))$. By definition of
    $\Pi(\far^\reledit_\epsilon(\Psi))$, we have $\pi = \psi^{-1}(x)$ for some
    $x \in \far^\reledit_\epsilon(\Psi)$. Thus $x \in \zo^N$ and for each $y \in \Psi$,
    $\dist_\reledit(x,y) > \epsilon$, and by \cref{lemma:relative-edit-distance}
    $\dist_\edit(\psi^{-1}(x), \psi^{-1}(y)) > \epsilon/2$.

    By definition of $\Pi(\Psi)$, for each $\pi' \in \Pi(\Psi)$ we have
    $\pi' = \psi^{-1}(y')$ for some $y' \in \Psi$. But then
    $\dist_\edit(\pi, \pi') = \dist_\edit(\psi^{-1}(x), \psi^{-1}(y')) > \epsilon/2$. Therefore
    $\pi \in \far^\edit_{\epsilon/2}(\Pi(\Psi))$.
\end{proof}

Our result for testing $n$-block strings will require the following equivalence between the
relative edit distance of strings to the property of $n$-block strings, and the edit distance of
appropriate probability distributions to the property of distributions supported on at most $n$
elements. We defer the proof to \cref{section:string-edit-distance-support-size}.

\begin{restatable}{proposition}{propdisttonstringsdistributions}
    \label{prop:dist-to-n-strings-distributions}
    Let $N,n \in \bN$. Let $\Psi$ be the set of $n$-block strings in $\zo^N$, and
    let $\Pi$ be the set of probability distributions over $\bN$ with support size at most
    $n$. Then for every distribution $\pi$ over $\bN$ whose densities are integer multiples of $1/N$
    and for $x = \psi(\pi)$,
    \[
        \dist_\edit(\pi, \Pi) \le \dist_\reledit(x, \Psi) \le 2 \cdot \dist_\edit(\pi, \Pi) \,.
    \]
\end{restatable}

\noindent
This implies:

\begin{proposition}
    \label{prop:pi-far-reledit-to-far-edit}
    Let $N \in \bN$. Let $\Psi$ be a property of strings in $\zo^N$ and let $\Pi$ be a property
    of probability distributions over $\bN$.
    Suppose that for every distribution $\pi$ over $\bN$ whose densities are integer multiples
    of $1/N$ and for $x = \psi(\pi)$, it holds
    that $\dist_\reledit(x, \Psi) \le 2 \cdot \dist_\edit(\pi, \Pi)$. Then
    \[
        \Pi(\far^\reledit_\epsilon(\Psi)) \subseteq \far^\edit_{\epsilon/4}(\Pi) \,.
    \]
\end{proposition}
\begin{proof}

    Let $I$ be the set of distributions over $\bN$ whose densities are integer multiples of $1/N$.
    We claim that $\far^\edit_{\epsilon/2}(\Pi(\Psi)) \cap I \subseteq \far^\edit_{\epsilon/4}(\Pi)$.

    Fix any $\pi \in \far^\edit_{\epsilon/2}(\Pi(\Psi)) \cap I$. We have that
    $\dist_\edit(\pi, \psi^{-1}(y)) > \epsilon/2$ for every $y \in \Psi$, and since all densities
    in $\pi$ are integer multiples of $1/N$ by $\pi \in I$, letting $x \define \psi(\pi)$ we
    conclude by \cref{lemma:relative-edit-distance} that $\dist_\reledit(x, y) > \epsilon/2$.
    Therefore $\dist_\reledit(x, \Psi) > \epsilon/2$, and using
    the hypothesis, we conclude that
    $\dist_\edit(\pi, \Pi) > \epsilon/4$ and thus $\pi \in \far^\edit_{\epsilon/4}(\Pi)$,
    establishing the first claim.

    By \cref{prop:far-reledit-far-edit},
    $\Pi(\far^\reledit_{\epsilon}(\Psi)) \subseteq \far^\edit_{\epsilon/2}(\Pi(\Psi))$,
    and therefore
    $\Pi(\far^\reledit_{\epsilon}(\Psi)) \cap I
    \subseteq \far^\edit_{\epsilon/2}(\Pi(\Psi)) \cap I
    \subseteq \far^\edit_{\epsilon/4}(\Pi)$. Finally, note that
    $\Pi(\far^\reledit_{\epsilon}(\Psi)) \cap I = \Pi(\far^\reledit_{\epsilon}(\Psi))$
    because every member of the latter has the form $\pi = \psi^{-1}(x)$ for some $x \in \zo^N$,
    so we are done.
\end{proof}

The following result establishes the upper bound portion of
\cref{thm:intro-trace-testing-support-n-informal}.

\begin{restatable}{theorem}{thmtracetestingsupportk}
    \label{thm:trace-testing-support-n}
    Let $N,n \in \bN$ and $\epsilon > 0$, and let $\Psi$ be the set of $n$-block strings in $\zo^N$.
    There is a $(\Psi, \far^\reledit_\epsilon(\Psi), \rho, 2/3)$-trace tester using one trace
    with expected trace size $\rho N = O(n/\epsilon)$.
\end{restatable}
\begin{proof}
    Let $\Pi$ be the class of distributions over $\bN$ with support size at most $n$.
    By \cref{thm:testing-support-k}, there is a
    $(\Pi, \far^\edit_{\epsilon/4}(\Pi), 2/3)$-distribution tester under the parity
    trace, which we may assume is Poissonized by \cref{prop:poissonization}, with sample complexity
    $O(n/\epsilon)$.

    Since for every $x \in \Psi$ we have that $\psi^{-1}(x)$ has support size at most $n$
    (\cref{obs:n-block-support}), it
    follows that $\Pi(\Psi) \subseteq \Pi$.
    \cref{prop:pi-far-reledit-to-far-edit} together with
    \cref{prop:dist-to-n-strings-distributions} gives that
    $\Pi(\far^\reledit_\epsilon(\Psi)) \subseteq \far^\edit_{\epsilon/4}(\Pi)$.

    Therefore we obtain a
    $(\Pi(\Psi), \Pi(\far^\reledit_\epsilon(\Psi)), 2/3)$-distribution tester under the parity
    trace with sample complexity $m = O(n/\epsilon)$.
    Then \cref{lemma:trace-tester-from-parity-trace-tester} yields a
    $(\Psi, \far^\reledit_\epsilon(\Psi), \rho, 2/3)$-trace tester using one trace for
    $\rho = 1 - e^{-m/N} \le m/N$, \ie expected trace size $\rho N = O(n/\epsilon)$.
\end{proof}

Our results for trace testing uniform $n$-block strings are simpler to obtain from distribution
testing under the parity trace, because now the corresponding property $\Pi$ of probability
distributions contains only the distributions corresponding to the uniform $n$-block strings.

\paragraph{Notation.} For fixed $N \in \bN$ and $n$ that divides $N$, let
$u^{(1)}$ denote the 1-uniform $n$-block string and let $u^{(0)}$ denote the 0-uniform $n$-block
string.

\begin{proposition}
    \label{prop:dist-to-uniform-n-strings-distributions}
    Let $N,n \in \bN$ be such that $n$ divides $N$.
    Let $\Psi$ contain only the 1-uniform $n$-block string $u^{(1)}$, and let
    $\Pi$ contain only the uniform distribution over $[n]$.
    For every distribution $\pi$ over $\bN$ whose densities are integer multiples of $1/N$ and
    for $x = \psi(\pi)$,
    \[
        \dist_\edit(\pi, \Pi) \le \dist_\reledit(x, \Psi) \le 2 \cdot \dist_\edit(\pi, \Pi) \,.
    \]
\end{proposition}
\begin{proof}
    Let $\pi^*$ be the uniform distribution over $[n]$, so that $\pi^* = \psi^{-1}(u^{(1)})$
    (\cref{obs:n-block-support}).
    The two inequalities are immediate consequences of \cref{lemma:relative-edit-distance}.
    First, we have
    $\dist_\edit(\pi, \Pi) = \dist_\edit(\pi, \pi^*)
    \le \dist_\reledit(x, u^{(1)}) = \dist_\reledit(x, \Psi)$.
    Similarly,
    $\dist_\edit(\pi, \Pi) = \dist_\edit(\pi, \pi^*)
    \ge \frac{1}{2} \dist_\edit(x, u^{(1)}) = \frac{1}{2} \dist_\edit(x, \Psi)$.
\end{proof}

We first show a tester for the 1-uniform $n$-block strings, and then generalize it to both types
of uniform strings to obtain \cref{thm:trace-testing-uniform-k-no-promise-informal}, restated here:

\begin{restatable}{theorem}{thmtracetestinguniformknopromise}
\label{thm:trace-testing-uniform-k-no-promise}
Let $N, n \in \bN$ be such that $n$ divides $N$, and let $\epsilon > 0$.
Let $\Psi$ contain only the uniform $n$-block strings in $\zo^N$. There is a
$(\Psi, \far^\reledit_\epsilon(\Psi), \rho, 2/3)$-trace tester using one trace with expected
trace size $\rho N = O\left( \frac{n}{\epsilon} + \frac{n}{\epsilon^2 \log n} \right)$.
\end{restatable}

\begin{lemma}
\label{lemma:trace-testing-1-uniform-k-no-promise}
Let $N,n \in \bN$ be such that $n$ divides $N$, and let $\epsilon > 0$.  Let $\Psi$ contain only the
1-uniform $n$-block string $u^{(1)} \in \zo^N$.  There is a $(\Psi, \far^\reledit_\epsilon(\Psi), \rho,
2/3)$-trace tester using one trace with expected trace size $\rho N = O\left(\frac{n}{\epsilon} +
\frac{n}{\epsilon^2 \log n}\right)$.
\end{lemma}
\begin{proof}
    Let $\Pi$ contain only the uniform distribution $\pi = \psi^{-1}(u^{(1)})$ over $[n]$.
    By \cref{thm:testing-uniform-k-alternating-no-promise}, there is a
    $(\Pi, \far^\edit_{\epsilon/4}(\Pi), 2/3)$-distribution tester under the parity trace, which we
    may assume is Poissonized by \cref{prop:poissonization}, with sample complexity
    $O\left(\frac{n}{\epsilon} + \frac{n}{\epsilon^2 \log n}\right)$.

    Note that $\Pi = \{\pi\} = \{\psi^{-1}(u^{(1)})\} = \Pi(\Psi)$.
    Moreover, \cref{prop:pi-far-reledit-to-far-edit} together with
    \cref{prop:dist-to-uniform-n-strings-distributions} gives that
    $\Pi(\far^\reledit_\epsilon(\Psi)) \subseteq \far^\edit_{\epsilon/4}(\Pi)$, so we obtain a
    $(\Pi(\Psi), \Pi(\far^\reledit_\epsilon(\Psi), 2/3)$-distribution tester under the parity
    trace with sample complexity $m=O\left(\frac{n}{\epsilon} + \frac{n}{\epsilon^2 \log n}\right)$.
    Then \cref{lemma:trace-tester-from-parity-trace-tester} yields a
    $(\Psi, \far^\reledit_\epsilon(\Psi), \rho, 2/3)$-trace tester using one trace for
    $\rho = 1 - e^{-m/N}$, \ie expected trace size
    $\rho N = O\left(\frac{n}{\epsilon} + \frac{n}{\epsilon^2 \log n}\right)$.
\end{proof}

\begin{proof}[Proof of \cref{thm:trace-testing-uniform-k-no-promise}]
    Use a version of the tester for the 1-uniform $n$-block string from
    \cref{lemma:trace-testing-1-uniform-k-no-promise} with success probability $5/6$, repeat it with
    all symbols negated, and accept if either execution accepts.
    If the input $x = u^{(1)}$, the first execution accepts with probability at least $5/6$, and
    if $x = u^{(0)}$, the second execution accepts with probability at least $5/6$.
    If $\dist_\reledit(x, \Psi) > \epsilon$, then $x$ is far from both $u^{(1)}$ and $u^{(0)}$,
    so each execution only accepts with probability at most $1/6$, and by the union bound the
    probability that $x$ is accepted is at most $1/3$.
\end{proof}

Toward establishing the upper bound portion of \cref{thm:intro-trace-testing-uniform}, we introduce
the following definition. We say that $x \in \{0,1\}^N$ is a \emph{type-1 $n$-block string} if
$x = 1^{t_1} 0^{t_2} 1^{t_3} \dotsm \parity(n)^{t_n}$ for some choice of non-negative integers
$t_1, \dotsc, t_n$. We say that $x$ is a \emph{type-0 $n$-block string} if
$x = 0^{t_1} 1^{t_2} 0^{t_3} \dotsm (1-\parity(n))^{t_n}$ for some choice of non-negative
integers $t_1, \dotsc, t_n$.

\begin{remark}
    \label{remark:type-1-0-strings}
    A string $x$ may be \emph{both} a type-1 $n$-block string and a type-0 $n$-block string.
    Moreover, $x$ is an $n$-block string if and only if it is a type-1 $n$-block string or
    a type-0 $n$-block string.

    If $x$ is a type-1 $n$-block string, then $\psi^{-1}(x)$ is supported within $[n]$, and if
    $x$ is a type-0 $n$-block string, then $\psi^{-1}(x)$ is supported within $\{2,3,\dotsc,n+1\}$.
\end{remark}

\begin{lemma}
    \label{lemma:trace-testing-1-uniform-k}
    Let $N, n \in \bN$ be such that $n$ is even and divides $N$, and let $\epsilon > 0$.
    Let $\Psi_1$ contain only the 1-uniform $n$-block string $u^{(1)} \in \zo^N$,
    and let $\Psi_2$ contain all type-1 $n$-block strings in $\zo^N$ that are $\epsilon$-far from
    $\Psi_1$ in (relative) edit distance.
    There is a $(\Psi_1, \Psi_2, \rho, 2/3)$-trace tester using one trace of expected size $\rho N =
    \widetilde O((n/\epsilon)^{4/5} + \sqrt{n}/\epsilon^2)$.
\end{lemma}
\begin{proof}
    Let $\Pi_1$ contain only the uniform distribution $\pi^* = \psi^{-1}(u^{(1)})$ over $[n]$,
    and let $\Pi_2$ contain the distributions over $[n]$ that are $\epsilon/2$-far from uniform in
    edit distance. By \cref{thm:intro-main} together with \cref{fact:edit-tv-subset},
    there is a $(\Pi_1, \Pi_2, 2/3)$-distribution tester
    under the parity trace, which we may assume is Poissonized by \cref{prop:poissonization},
    with sample complexity $\widetilde O((n/\epsilon)^{4/5} + \sqrt{n}/\epsilon^2)$.

    We claim that $\Pi(\Psi_1) \subseteq \Pi_1$ and $\Pi(\Psi_2) \subseteq \Pi_2$. First, let
    $\pi \in \Pi(\Psi_1)$, so that necessarily $\pi = \psi^{-1}(u^{(1)})$. Then indeed
    $\pi$ is uniform over $[n]$, so $\pi \in \Pi_1$. Now, suppose $\pi \in \Pi(\Psi_2)$, so that
    $\pi = \psi^{-1}(x)$ for some $x \in \zo^N$ such that $x$ is a type-1 $n$-block string and
    $\dist_\reledit(x,u^{(1)}) > \epsilon$. It follows that $\pi$ is supported within $[n]$
    by \cref{remark:type-1-0-strings} and, by \cref{lemma:relative-edit-distance},
    $\dist_\edit(\pi, \pi^*) \ge \frac{1}{2} \dist_\reledit(\psi(\pi), \psi(\pi^*))
    = \frac{1}{2} \dist_\reledit(x, u^{(1)}) > \epsilon/2$, so $\pi \in \Pi_2$.

    Therefore we obtain a $(\Pi(\Psi_1), \Pi(\Psi_2), 2/3)$-distribution tester under the parity
    trace with sample complexity $m = \widetilde O((n/\epsilon)^{4/5} + \sqrt{n}/\epsilon^2)$.
    Then \cref{lemma:trace-tester-from-parity-trace-tester} yields a
    $(\Psi_1, \Psi_2, \rho, 2/3)$-trace tester using one trace for $\rho = 1 - e^{-m/N}$,
    \ie expected trace size
    $\rho N = \widetilde O((n/\epsilon)^{4/5} + \sqrt{n}/\epsilon^2)$.
\end{proof}

Now, we obtain the single-trace upper bound portion of \cref{thm:intro-trace-testing-uniform}:

\begin{theorem}
    \label{thm:trace-testing-uniform-k}
    Let $N, n \in \bN$ be such that $n$ is even and divides $N$, and let $\epsilon > 0$.
    Let $\Psi_1$ contain only the uniform $n$-block strings $u^{(1)}, u^{(0)} \in \zo^N$,
    and let $\Psi_2$ contain all $n$-block strings in $\zo^N$ that are $\epsilon$-far from $\Psi_1$
    in (relative) edit distance.
    There is a $(\Psi_1, \Psi_2, \rho, 2/3)$-trace tester using one trace of expected size $\rho N =
    \widetilde O((n/\epsilon)^{4/5} + \sqrt{n}/\epsilon^2)$.
\end{theorem}
\begin{proof}
    The key observation is that the algorithm $A$ obtained in \cref{lemma:trace-testing-1-uniform-k}
    is invariant to negation of all the symbols in the input: it is a combination of the
    distribution tester from \cref{thm:intro-main}, which only depends on run lengths, and the
    \textsc{Poissonize} algorithm from \cref{prop:poissonize-trace}, which transforms the input
    in the same way regardless of the values of the symbols. Formally, for $x \in \{0,1\}^N$ and
    letting $\overline x$ denote the string obtained by negating every symbol in $x$, the
    outputs $A(x)$ and $A(\overline x)$ are identically distributed.

    Therefore, we claim that $A$ is also a $(\Psi_1, \Psi_2, \rho, 2/3)$-trace tester. Indeed if
    the input $x \in \Psi_1$, then either $x = u^{(1)}$ and $A$ accepts with probability at least
    $2/3$ by \cref{lemma:trace-testing-1-uniform-k}, or $x = u^{(0)}$ and therefore
    $\overline x = u^{(1)}$, so again $A$ accepts. On the other hand, if $x \in \Psi_2$ then
    $\dist_\reledit(x, u^{(1)}) > \epsilon$ and
    $\dist_\reledit(\overline x, u^{(1)}) = \dist_\reledit(x, u^{(0)}) > \epsilon$. Moreover,
    either $x$ or $\overline x$ is a type-1 $n$-block string, so $A$ rejects with probability
    at least $2/3$.
\end{proof}

We remark that the probability of success $2/3$ in the results above could be replaced with any
higher constant without affecting the asymptotic bounds, by using the distribution tester under the
parity trace with correspondingly better constant probability of success. Alternatively,
multiple independent traces may be used to amplify the result into the high probability regime.

\subsection{Multiple-Trace Upper Bound}

The ability to make inferences from multiple independent traces is a central component of the
trace reconstruction model. Accordingly, we would like to test the class of uniform $n$-block
strings with smaller traces than afforded by our single-trace results, at the cost of taking more
traces. The main idea is to \emph{reduce} to the single-trace case by concatenating the $k$ traces
together, and thinking of the result as one trace from the input string copied $k$ times.

For any strings $x, y \in \zo^*$ and integer $k > 0$, denote by $x \circ y$ the concatenation
of $x$ and $y$, and by $x^{\circ k}$ the concatenation $x \circ \dotsm \circ x$ with $k$
terms in total.

\begin{proposition}
    \label{prop:concat-edit-dist}
    There exists a universal constant $c > 0$ such the following holds.
    Let $N,n,k \in \bN$ be such that $n$ is even and divides $N$.
    Let $u = u^{(1)} \in \zo^N$ be the 1-uniform $n$-block string and let $x \in \zo^N$ be a type-1
    $n$-block string. Then
    \[
        \dist_\reledit(u^{\circ k}, x^{\circ k}) \ge c \cdot \dist_\reledit(u, x) \,.
    \]
\end{proposition}
\begin{proof}
    Let $\pi_u \define \psi^{-1}(u)$ and $\pi_x \define \psi^{-1}(x)$, so that $\pi_u$ and $\pi_x$
    are supported within $[n]$ and $\psi(\pi_u) = u, \psi(\pi_x) = x$.
    Define $\pi_u^k$ as the following distribution on $\bN$: for all $t \in \bZ_{\ge 0}$
    and $i \in [n]$,
    \[
        \pi_u^k(tn + i) \define \begin{cases}
            \pi_u(i)/k, & \text{if $t \le k-1$} \\
            0,          & \text{if $t \ge k$.}
        \end{cases}
    \]
    Define $\pi_x^k$ analogously. Then
    $\dist_\TV(\pi_u^k, \pi_x^k) = \dist_\TV(\pi_u, \pi_x)$, since the entries of the former
    are aligned in each of the $k$ rescaled copies of the latter.
    Also, the entries of $\pi_u^k, \pi_x^k$ are integer multiples of $1/kN$, and their
    corresponding strings over $\zo^{kN}$ satisfy $\psi(\pi_u^k) = u^{\circ k}$ and
    $\psi(\pi_x^k) = x^{\circ k}$.

    By \cref{lemma:relative-edit-distance}, we have
    $\dist_\reledit(u^{\circ k}, x^{\circ k}) = \dist_\reledit(\psi(\pi_u^k), \psi(\pi_x^k))
    \ge \dist_\edit(\pi_u^k, \pi_x^k)$.
    Since $n$ is even, $u^{\circ k}$ is the 1-uniform $kn$-block string, and therefore
    $\pi_u^k$ is uniformly distributed on $[kn]$. It is also clear that
    $\pi_x^k$ is supported within $[kn]$. It follows from \cref{lemma:edit-to-tv-uniform} that
    $\dist_\edit(\pi_u^k, \pi_x^k) \ge c' \cdot \dist_\TV(\pi_u^k, \pi_x^k)
    = \dist_\TV(\pi_u, \pi_x)$ for some universal constant $c' > 0$.
    It is easy to see from the definition of edit distance that
    $\dist_\TV(\pi_u, \pi_x) \ge \dist_\edit(\pi_u, \pi_x)$. Finally, applying
    \cref{lemma:relative-edit-distance} again yields
    $\dist_\edit(\pi_u, \pi_x) \ge \frac{1}{2} \dist_\reledit(\psi(\pi_u), \psi(\pi_x))$.
    Recalling that $\psi(\pi_u) = u$ and $\psi(\psi_x) = x$, this concludes the proof.
\end{proof}

We first use the result above to show a multiple-trace upper bound for testing the 1-uniform strings
among the type-1 $n$-block strings, and then generalize this result to both types of (uniform)
strings to obtain the upper bound portion of \cref{thm:intro-trace-testing-uniform}.

\begin{lemma}
    \label{lemma:multiple-trace-testing-1-uniform}
    Let $N, n, k \in \bN$ be such that $n$ is even and divides $N$, and let $\epsilon > 0$.
    Let $\Psi_1$ contain only the 1-uniform $n$-block string $u^{(1)} \in \zo^N$,
    and let $\Psi_2$ be the set of type-1 $n$-block strings that are $\epsilon$-far from $\Psi_1$
    in (relative) edit distance. Then there is a $(\Psi_1, \Psi_2, \rho, 2/3)$-trace tester using
    $k$ traces of expected size $\rho N = \widetilde
    O\left( \frac{n^{4/5}}{k^{1/5} \epsilon^{4/5}} + \frac{\sqrt{n}}{\sqrt{k} \epsilon^2} \right)$.
\end{lemma}
\begin{proof}
    Let $u \define u^{(1)}$ for convenience of notation.
    Let $\Psi'_1$ contain only the 1-uniform $kn$-block string $u^{\circ k} \in \zo^{kN}$, and let
    $\Psi'_2$ be the set of type-1 $kn$-block strings in $\zo^{kN}$ that are $c\epsilon$-far from
    $\Psi'_1$ in relative edit distance, where $c$ is the constant from \cref{prop:concat-edit-dist}.

    By \cref{lemma:trace-testing-1-uniform-k} we obtain algorithm $A$, a
    $(\Psi'_1, \Psi'_2, \rho', 2/3)$-trace tester using one trace of expected size
    $\rho' (kN) = \widetilde O((kn)^{4/5}/\epsilon^{4/5} + \sqrt{kn}/\epsilon^2)$.
    Our algorithm $B$, which will be a $(\Psi_1, \Psi_2, \rho, 2/3)$-trace tester using $k$
    traces, works as follows:
    \begin{enumerate}
        \item Obtain $k$ independent traces $x_1, \dotsc, x_k$ of expected size
            $\rho N = \widetilde O\left(
            \frac{n^{4/5}}{k^{1/5} \epsilon^{4/5}} + \frac{\sqrt{n}}{\sqrt{k} \epsilon^2} \right)$.
        \item Return the output of $A$ on $x_1 \circ \dotsm \circ x_k$.
    \end{enumerate}

    Let $x \in \zo^N$ be the unknown input to $B$. If $x \in \Psi_1$, then $x = u$ and hence
    $x^{\circ k} \in \Psi'_1$.  On the other hand, if $x \in \Psi_2$, then
    $\dist_\reledit(x, u) > \epsilon$ and we use \cref{prop:concat-edit-dist} to conclude that
    $\dist_\reledit(x^{\circ k}, u^{\circ k}) > c\epsilon$, so
    $x^{\circ k} \in \Psi'_2$. Moreover, the input $\bm{x}_1 \circ \dotsm \circ \bm{x}_k$ to $A$
    is distributed as a trace from $x^{\circ k}$ of expected size $k (\rho N)$, \ie deletion rate
    $\rho$. Therefore $A$ will produce the correct output (and hence so will $B$) with probability
    at least $2/3$ as long as we satisfy
    \[
        k \rho N \ge \widetilde O\left(
            \frac{(kn)^{4/5}}{\epsilon^{4/5}} + \frac{\sqrt{kn}}{\epsilon^2} \right) \,,
    \]
    which holds when
    \[
        \rho N \ge \widetilde O\left(
            \frac{n^{4/5}}{k^{1/5} \epsilon^{4/5}} + \frac{\sqrt{n}}{\sqrt{k} \epsilon^2}
        \right) \,. \qedhere
    \]
\end{proof}

We now obtain the upper bound portion of \cref{thm:intro-trace-testing-uniform}.

\begin{theorem}
Let $N, n, k \in \bN$ be such that $n$ is even and divides $N$, and let $\epsilon > 0$.  Let
$\Psi_1$ contain only the uniform $n$-block strings in $\zo^N$, and let $\Psi_2$ be the set of
$n$-block strings that are $\epsilon$-far from $\Psi_1$ in (relative) edit distance.  Then there is
a $(\Psi_1, \Psi_2, \rho, 2/3)$-trace tester using $k$ traces of expected size $\rho N = \widetilde
O\left( \frac{n^{4/5}}{k^{1/5} \epsilon^{4/5}} + \frac{\sqrt{n}}{\sqrt{k} \epsilon^2} \right)$.
\end{theorem}
\begin{proof}
    The argument is identical to the proof of \cref{thm:trace-testing-uniform-k}. Letting $A$
    be the algorithm from \cref{lemma:multiple-trace-testing-1-uniform}, we observe that
    the output $A(x)$ is distributed identically to $A(\overline x)$. Therefore $A$ is also a
    $(\Psi_1, \Psi_2, \rho, 2/3)$-trace tester using $k$ traces: if $x \in \Psi_1$, then either
    $x = u^{(1)}$ or $\overline x = u^{(1)}$, so $A$ accepts,
    and if $x \in \Psi_2$, then both $x$ and $\overline x$ are far from $u^{(1)}$ and moreover
    either $x$ or $\overline x$ is a type-1 $n$-block string, so $A$ rejects.
\end{proof}

\subsection{Lower Bounds}

We wish to reduce from distribution testing under the parity trace to testing properties of
strings in the trace reconstruction model. We define a ``trace splitting'' procedure,
which takes a parity trace from distribution $\pi$ and produces $k$ strings that look like
independent traces from string $\psi(\pi)$.

\begin{proposition}[Poissonized trace splitting]
    \label{prop:trace-splitting}
    There exists a randomized algorithm \textsc{Split} that satisfies the following.
    Let $N, k \in \bN$ and let $\rho > 0$ satisfy
    $\rho < \frac{1}{20 \cdot \sqrt{kN}}$. Let $\pi$ be any probability distribution over $\bN$
    whose densities are integer multiples of $1/N$. Then on inputs $N, k \text{ and } \rho$,
    \textsc{Split} draws a parity trace of size
    $\Poi\left(k \cdot \frac{\rho}{1-\rho} \cdot N\right)$ from $\pi$ and outputs
    a sequence of $k$ binary strings satisfying the following. Let $\bm{x}_1, \dotsc, \bm{x}_k$
    be the random variables denoting the output of \textsc{Split} (over the randomness of the
    parity trace and internal randomness of the algorithm), and let
    $\bm{y}_1, \dotsc, \bm{y}_k$ be such that each $\bm{y}_i$ is an independent trace from $\psi(\pi)$
    with expected size $\rho N$. Then
    \[
        \dist_\TV((\bm{x}_1, \dotsc, \bm{x}_k), (\bm{y}_1, \dotsc, \bm{y}_k)) < 1/100 \,.
    \]
\end{proposition}
\begin{proof}
    Let $\lambda \define \rho/(1-\rho)$. The algorithm \textsc{Split} proceeds as follows:
    \begin{enumerate}
        \item Draw a parity trace $T$ of size $\Poi(k \lambda N)$ from $\pi$;
        \item Initialize empty strings $x_1, \dotsc, x_k$;
        \item For each symbol $b$ in $T$ from left to right, append $b$ to $x_i$ where
            $i$ is drawn uniformly at random from $[k]$;
        \item Return the strings $(x_1, \dotsc, x_k)$.
    \end{enumerate}

    For each $i \in [k]$ and $j \in [N]$, let $\bm{X}_{i,j} \sim \Poi(\lambda)$ independently.
    Define random variables $\bm{x}'_1, \dotsc, \bm{x}'_k$ as follows: for each
    $i \in [k]$,
    \[
        \bm{x}'_i
        = \psi(\pi)_1^{\bm{X}_{i,1}} \psi(\pi)_2^{\bm{X}_{i,2}} \dotsm \psi(\pi)_N^{\bm{X}_{i,N}} \,.
    \]
    We claim that $(\bm{x}_1, \dotsc, \bm{x}_k)$ is distributed identically to
    $(\bm{x}'_1, \dotsc, \bm{x}'_k)$. Indeed, first recall that the parity trace $\bm{T}$ from
    Step 1 is distributed as follows:
    \[
        \bm{T} = \psi(\pi)_1^{\bm{A}_1} \psi(\pi)_2^{\bm{A}_2} \dotsm \psi(\pi)_N^{\bm{A}_N} \,,
    \]
    where $\bm{A}_j \sim \Poi(k \lambda)$ independently for each $j \in [N]$.
    Then, Step 3 is equivalent to splitting each $\bm{A}_j$ into random variables
    $(\bm{A}_{1,j}, \dotsc, \bm{A}_{k,j}) \sim \Multinomial(\bm{A}_j, (1/k, \dotsc, 1/k))$,
    and producing each $\bm{x}_i$ by concatenation:
    \[
        \bm{x}_i
        = \psi(\pi)_1^{\bm{A}_{i,1}} \psi(\pi)_2^{\bm{A}_{i,2}} \dotsm \psi(\pi)_N^{\bm{A}_{i,N}} \,.
    \]
    It follows from standard arguments that $\bm{A}_{1,j}, \dotsc, \bm{A}_{k,j}$ are i.i.d.
    random variables distributed as $\bm{A}_{i,j} \sim \Poi(\lambda)$ for each $i \in [k]$.
    Therefore the $\bm{A}_{i,j}$ are distributed identically to the $\bm{X}_{i,j}$, and so
    $(\bm{x}_1, \dotsc, \bm{x}_k)$ is distributed identically to $(\bm{x}'_1, \dotsc, \bm{x'}_k)$.

    Now, for each $i \in [k]$ and $j \in [N]$, let $\bm{Y}_{i,j} \sim \Ber(\rho)$ independently.
    By definition of trace, we have
    \[
        \bm{y}_i
        = \psi(\pi)_1^{\bm{Y}_{i,1}} \psi(\pi)_2^{\bm{Y}_{i,2}} \dotsm \psi(\pi)_N^{\bm{Y}_{i,N}} \,.
    \]
    Therefore, we will be done if we can show that the $\bm{X}_{i,j}$ are sufficiently similar to
    the $\bm{Y}_{i,j}$.

    Concretely, fix some $i \in [k]$ and $j \in [N]$, and let $\bm{X} = \bm{X}_{i,j}$ and
    $\bm{Y} = \bm{Y}_{i,j}$ for convenience.
    We claim that $\dist_\TV(\bm{X}, \bm{Y}) \le \lambda^2$.
    Indeed, first, note that the distribution of $\bm{X}$ conditional on $\bm{X} \le 1$ is identical
    to that of $\bm{Y}$:
    \[
        \Pruc{}{\bm{X} = 1}{\bm{X} \le 1} = \frac{\Pr{\bm{X}=1}}{\Pr{\bm{X}=0} + \Pr{\bm{X}=1}}
        = \frac{\lambda e^{-\lambda}}{e^{-\lambda} + \lambda e^{-\lambda}}
        = \frac{\lambda}{1+\lambda}
        = \rho
        = \Pr{\bm{Y} = 1} \,,
    \]
    and therefore
    \begin{align*}
        \dist_\TV(\bm{X}, \bm{Y})
        &= \Pr{\bm{X} \ge 2}
        = 1 - \Pr{\bm{X}=0} - \Pr{\bm{X}=1}
        = 1 - e^{-\lambda} - \lambda e^{-\lambda}
        = 1 - e^{-\lambda} (1 + \lambda) \\
        &\le 1 - (1 - \lambda)(1 + \lambda)
        = \lambda^2 \,.
    \end{align*}

    Finally, since the $\bm{X}_{i,j}$ and $\bm{Y}_{i,j}$ are all mutually independent, we obtain
    \[
        \dist_\TV((\bm{X}_{i,j})_{i \in [k], j \in [N]}, (\bm{Y}_{i,j})_{i \in [k], j \in [N]})
        \le k N \cdot \lambda^2
        < k N (2\rho)^2
        < 1/100 \,,
    \]
    and thus $\dist_\TV((\bm{x}_1, \dotsc, \bm{x}_k), (\bm{y}_1, \dotsc, \bm{y}_k)) < 1/100$
    by the data processing inequality.
\end{proof}

We use this procedure to give our general lower bound for trace testing:

\begin{lemma}
    \label{lemma:trace-testing-general-lower-bound}
    Let $\alpha > 0$.
    Let $N, k \in \bN$, and let $\Pi_1, \Pi_2$ be properties of probability distributions over $\bN$
    whose densities are integer multiples of $1/N$, and such that
    $(\Pi_1, \Pi_2, \alpha)$-testing under the parity trace requires sample complexity at
    least $m$. Then any $(\Psi(\Pi_1), \Psi(\Pi_2), \rho, \alpha+1/100)$-trace tester using $k$
    traces of expected size $\rho N$ must satisfy
    $k \rho N = \Omega\left( \min\left( m, \sqrt{kN} \right) \right)$.
\end{lemma}
\begin{proof}
    Suppose $B$ is a $(\Psi(\Pi_1), \Psi(\Pi_2), \alpha+1/100)$-trace tester using $k$ traces
    of expected size $\rho N$, and suppose $k \rho N < \sqrt{kN}/20$. Then
    $\rho < \frac{1}{20\sqrt{kN}}$ and our goal is to show that $k \rho N = \Omega(m)$. We do
    so by constructing an algorithm $A$ in the parity trace model and showing that $A$ is a
    $(\Pi_1, \Pi_2, \alpha)$-tester under the parity trace with sample complexity $O(k \rho N)$.
    The algorithm works as follows:
    \begin{enumerate}
        \item Take a parity trace of size $\Poi\left( k \cdot \frac{\rho}{1-\rho} \cdot N \right)$.
        \item Apply \textsc{Split} to obtain $k$ binary strings $x_1, \dotsc, x_k$.
        \item Return the output of $B$ on inputs $x_1, \dotsc, x_k$.
    \end{enumerate}
    Note that $A$ has sample complexity $O(k \rho N)$. Let $\pi$ be the input distribution;
    recall that if $\pi \in \Pi_1$ then $\psi(\pi) \in \Psi(\Pi_1)$, and if $\pi \in \Pi_2$ then
    $\psi(\pi) \in \Psi(\Pi_2)$.
    Let $\bm{x}_1, \dotsc, \bm{x}_k$ be the inputs fed to $B$, and let
    $\bm{y}_1, \dotsc, \bm{y}_k$ be mutually independent traces from $\psi(\pi)$, each with
    expected size $\rho N$. By \cref{prop:trace-splitting},
    $\dist_\TV((\bm{x}_1, \dotsc, \bm{x}_k), (\bm{y}_1, \dotsc, \bm{y}_k)) < 1/100$, and we know
    that $B$ would succeed (\ie accept if $\psi(\pi) \in \Psi(\Pi_1)$, reject if
    $\psi(\pi) \in \Psi(\Pi_2)$) with probability at least $\alpha+1/100$ if it were given inputs
    $\bm{y}_1, \dotsc, \bm{y}_k$. Therefore $A$ succeeds with probability at least $\alpha$,
    and $k \rho N = \Omega(m)$.
\end{proof}

We now obtain the lower bounds stated in
\cref{thm:intro-trace-testing-uniform,thm:intro-trace-testing-support-n-informal}.

\begin{theorem}
    \label{thm:multiple-trace-lower-bound}
    There exists a universal constant $C > 0$ such that the following holds.
    Let $n, N \in \bN$ and $\rho, \epsilon > 0$ be such that $n$ is even and divides $N$,
    $\epsilon$ is smaller than some sufficiently small universal constant, and
    $N \ge C \cdot \max\left\{ (n/\epsilon)^{8/5}, n/\epsilon^4 \right\}$.
    Let $\Psi_1$ contain only the uniform $n$-block strings in $\zo^N$, and let $\Psi_2$ be the set
    of all $n$-block strings that are $\epsilon$-far from $\Psi_1$ in (relative) edit distance.
    Then any $(\Psi_1, \Psi_2, \rho, 2/3)$-trace tester using $k$ traces of expected size $\rho N$
    must satisfy $k \rho N = \widetilde \Omega((n/\epsilon)^{4/5} + \sqrt{n}/\epsilon^2)$.
\end{theorem}
\begin{proof}
    Let $\Pi_1$ contain only the uniform distribution
    $\pi^* = \psi^{-1}(u^{(1)})$ on $[n]$. Let $\epsilon^* \define 8\epsilon$ for convenience.
    Using the value $\epsilon' \in (\epsilon^*, 2\epsilon^*]$ defined below, let
    $\Pi_2$ be the set of distributions over $[n]$ that 1) are $(\epsilon'/c)$-far from
    uniform in total variation distance (where $c$ is the constant from
    \cref{lemma:edit-to-tv-uniform}); and 2) have all densities in the set
    $\left\{ \frac{1}{n}, \frac{1-4\epsilon'/c}{n}, \frac{1+4\epsilon'/c}{n} \right\}$.
    Note that, by \cref{lemma:edit-to-tv-uniform}, every distribution in $\Pi_2$ is
    $\epsilon'$-far from uniform in \emph{edit} distance.

    We define $\epsilon'$ as the smallest $\epsilon' > \epsilon^*$ such that $\frac{4\epsilon'/c}{n}$
    is an integer multiple of $1/N$, and claim that $\epsilon' \le 2\epsilon^*$. Indeed, we have
    \[
        \frac{4\epsilon^*/c}{n} \ge \frac{1}{N} \iff N \ge \frac{cn}{4\epsilon^*} \,,
    \]
    which holds by assumption for sufficiently large $C$, so there exists an integer multiple of
    $1/N$ between $\frac{4\epsilon^*/c}{n}$ and $\frac{8\epsilon^*/c}{n}$.

    Since the proof of \cref{thm: lower bound} only uses distributions of the form of $\Pi_2$,
    it follows that $(\Pi_1, \Pi_2, 51/100)$-distribution testing under the parity trace requires
    $\widetilde\Omega\left( (n/\epsilon)^{4/5} + \sqrt{n}/\epsilon^2 \right)$ samples.
    Therefore, noticing that by assumption we have
    $\sqrt{kN} = \Omega\left( (n/\epsilon)^{4/5} + \sqrt{n}/\epsilon^2 \right)$,
    \cref{lemma:trace-testing-general-lower-bound} gives that any
    $(\Psi(\Pi_1), \Psi(\Pi_2), 52/100)$-trace tester using $k$ traces of expected size $\rho N$
    must satisfy $k \rho N = \widetilde \Omega((n/\epsilon)^{4/5} + \sqrt{n}/\epsilon^2)$.
    The result will follow if we show that $\Psi(\Pi_1) \subseteq \Psi_1$ and
    $\Psi(\Pi_2) \subseteq \Psi_2$.

    First, suppose $x \in \Psi(\Pi_1)$. Then $x = \psi(\pi^*) = \psi(\psi^{-1}(u^{(1)})) = u^{(1)}$
    and hence $x \in \Psi_1$.

    Second, suppose $x \in \Psi(\Pi_2)$, so $x = \psi(\pi)$ for some $\pi \in \Pi_2$
    We claim that $x \in \Psi_2$. It is clear that $x$ is an $n$-block string, so it remains to
    show that $\dist_\reledit(x, \Psi_1) > \epsilon$.
    Since $\pi \in \Pi_2$, we have $\dist_\TV(\pi, \pi^*) \ge \epsilon'/c$ and thus,
    as was observed, $\dist_\edit(\pi, \pi^*) \ge \epsilon' > \epsilon^*$.
    Then by \cref{lemma:relative-edit-distance},
    $\dist_\reledit(x, u^{(1)}) = \dist_\reledit(\psi(\pi), \psi(\pi^*))
    \ge \dist_\edit(\pi, \pi^*) > \epsilon^* > \epsilon$. We also need to show that
    $\dist_\reledit(x, u^{(0)}) > \epsilon$. We consider two cases.

    First, suppose $\epsilon^* \ge 4/n$. By the triangle inequality,
    $\dist_\reledit(\psi(\pi), u^{(1)})
    \le \dist_\reledit(\psi(\pi), u^{(0)}) + \dist_\reledit(u^{(0)}, u^{(1)})
    = \dist_\reledit(\psi(\pi), u^{(0)}) + 2/n$.
    Thus $\dist_\reledit(\psi(\pi), u^{(0)}) > \epsilon^* - 2/n > \epsilon$.

    Second, suppose $\epsilon^* < 4/n$, so $1/2n > \epsilon^*/8$. The first block of $\psi(\pi)$
    has length at least $N \cdot \frac{1 - 4\epsilon'/c}{n}$ by construction, and for sufficiently
    small $\epsilon$, this is at least
    $N \cdot \frac{1}{2n}$. Moreover this first block of $\psi(\pi)$ is a block of 1s, whereas
    the first block of $u^{(0)}$ is a block of 0s of length $N \cdot \frac{1}{n}$. Therefore
    $\dist_\reledit(\psi(\pi), u^{(0)}) \ge 1/2n > \epsilon^*/8 = \epsilon$, as desired.
    Therefore $x \in \Psi_2$, completing the proof.
\end{proof}

\begin{theorem}
    \label{thm:trace-testing-support-n-lower-bound}
    The following holds for all sufficiently small constant $\epsilon > 0$.
    There exists a function $N(n) = \Theta(n^2)$ such that, for all $n, k \in \bN$ and
    $N \ge N(n)$, the following is true. Let $\Psi$ be the set of $n$-block strings in $\zo^N$.
    Then any $(\Psi, \far^\reledit_\epsilon(\Psi), \rho, 2/3)$-trace tester using $k$ traces of
    expected size $\rho N$ must satisfy $k \rho N = \Omega(n / \log n)$.
\end{theorem}
\begin{proof}
    Let $\Pi_1$ be the set of probability distributions over $\bN$ with support size
    at most $n$.  By \cref{thm:testing-support-k} and the lower bound on testing $n$-alternating
    functions from~\cite{BFH21}, any
    $(\Pi_1,\allowbreak \far^\edit_{2\epsilon}(\Pi_1),\allowbreak 51/100)$-distribution
    tester under the parity trace must have sample complexity $\Omega(n / \log n)$. In fact, a stronger
    statement holds: for some sufficiently large universal constant $C > 0$, let $\Pi_2$ be the
    restriction of $\far^\edit_{2\epsilon}(\Pi_1)$ to those distributions with support size at most
    $Cn$. Then any $(\Pi_1, \Pi_2, 51/100)$-distribution tester under the parity trace must have sample
    complexity $\Omega(n / \log n)$. This is because the lower bound on testing $n$-alternating
    functions from~\cite{BFH21} is proved via a reduction from the \emph{support size distinction}
    problem~\cite{VV11,WY19}, and the hard examples for that problem have support size linear in $n$.

    To apply \cref{lemma:trace-testing-general-lower-bound}, we need the all densities to be integer
    multiples of $1/N$. Let $\Pi'_1$ be a property obtained by taking each distribution
    $\pi \in \Pi_1$ and rounding all the densities of $\pi$ to a multiple of $1/N$ in such a way
    that we obtain another probability distribution $\pi'$.
    Let $\Pi'_2$ be a property obtained from $\Pi_2$ in the same way. Then every
    $\pi \in \Pi_1$ satisfies $\dist_\TV(\pi, \Pi'_1) \leq \tfrac{1}{N} \cdot O(n)$ and every
    $\pi \in \Pi_2$ satisfies $\dist_\TV(\pi, \Pi'_2) \leq \tfrac{1}{N} \cdot  O(n)$.

    We claim that any $(\Pi'_1, \Pi'_2, 60/100)$-distribution tester under the parity trace must
    have sample complexity $\Omega(n / \log n)$. Suppose $A$ is a
    $(\Pi'_1, \Pi'_2, 60/100)$-distribution tester under the parity trace with sample complexity
    $m \le n / \log n$. Then $A$ is a $(\Pi_1, \Pi_2, 55/100)$-distribution tester under the parity
    trace, as we now prove. For any input $\pi \in \Pi_1$, there exists $\pi' \in \Pi'_1$ such that
    $\dist_\TV(\pi, \pi') = O(n/N)$. Then the random variables
    $\bm{S} \sim \samp(\pi, m)$ and $\bm{S'} \sim \samp(\pi', m)$ satisfy
    $\dist_\TV(\bm{S}, \bm{S'}) \le m \cdot \dist_\TV(\pi, \pi') = O(\frac{n^2}{N \log n}) = o(1)$.
    Since $A$ accepts $\pi'$ with probability at least $60/100$,
    it accepts $\pi$ with probability at least $60/100 - o(1) \ge 55/100$.
    The same argument holds for $\Pi_2$ and
    $\Pi'_2$, and therefore $A$ is a $(\Pi_1, \Pi_2, 55/100)$-distribution tester under the
    parity trace. Hence $m = \Omega(n / \log n)$, proving the claim.

    Now, noticing that by assumption we have $\sqrt{Nk} = \Omega(n)$,
    \cref{lemma:trace-testing-general-lower-bound} gives that any
    $(\Psi(\Pi'_1),\allowbreak \Psi(\Pi'_2),\allowbreak 61/100)$-trace
    tester using $k$ traces of expected size
    $\rho N$ must satisfy $k \rho N = \Omega(n / \log n)$. The result will follow if we show
    that $\Psi(\Pi'_1) \subseteq \Psi$ and $\Psi(\Pi'_2) \subseteq \far^\reledit_\epsilon(\Psi)$.

    First, let $x \in \Psi(\Pi'_1)$, so $x = \psi(\pi'_1)$ for some $\pi'_1 \in \Pi'_1$.
    Since the process to obtain $\Pi'_1$ from $\Pi_1$ does not add any new elements to the support
    of the distributions, every $\pi'_1 \in \Pi'_1$ has support size at most $n$, and thus
    $x$ is an $n$-block string. Hence $\Psi(\Pi'_1) \subseteq \Psi$.

    Second, let $x \in \Psi(\Pi'_2)$, so $x = \psi(\pi'_2)$ for some $\pi'_2 \in \Pi'_2$.
    We need to show that $\dist_\reledit(x, \Psi) > \epsilon$. First, we claim that
    $\dist_\edit(\pi'_2, \Pi_1) > \epsilon$, \ie $\pi'_2$ is $\epsilon$-far in edit distance
    from any distribution (not necessarily rounded) with support size at most $n$.
    Suppose for a contradiction there exists $\pi_1 \in \Pi_1$ such that
    $\dist_\edit(\pi_1, \pi'_2) \le \epsilon$. Let $\pi_2 \in \Pi_2$ be some distribution
    satisfying $\dist_\TV(\pi_2, \pi'_2) = O(n/N)$, which exists by construction of $\Pi'_2$.
    Then $\dist_\edit(\pi_2, \pi'_2) = O(n/N)$ and, since $\epsilon$ is a constant, we
    may assume that $\dist_\edit(\pi_2, \pi'_2) < \epsilon/2$. Then
    \[
        \dist_\edit(\pi_1, \pi_2) \le \dist_\edit(\pi_1, \pi'_2) + \dist_\edit(\pi'_2, \pi_2)
        < \epsilon + \epsilon/2 < 2\epsilon \,,
    \]
    contradicting the definition of $\Pi_2$. Therefore $\dist_\edit(\pi'_2, \Pi_1) > \epsilon$.
    Now, let $y \in \Psi$. We claim that $\dist_\reledit(x, y) > \epsilon$.
    Since $y$ is an $n$-block string, we have $\psi^{-1}(y) \in \Pi_1$, and hence
    $\dist_\edit(\pi'_2, \psi^{-1}(y)) > \epsilon$. Then \cref{lemma:relative-edit-distance} gives
    that $\dist_\reledit(x, y) = \dist_\reledit(\psi(\pi'_2), \psi(\psi^{-1}(y)))
    \ge \dist_\edit(\pi'_2, \psi^{-1}(y)) > \epsilon$. Hence $\dist_\reledit(x, \Psi) > \epsilon$
    and therefore $\Psi(\Pi'_2) \subseteq \far^\reledit_\epsilon(\Psi)$, concluding the proof.
\end{proof}

\iftoggle{anonymous}{}{%
\begin{samepage}
\newpage
\begin{center}
\Large \bf Acknowledgments
\end{center}

We thank Eric Blais for helpful discussions and comments on the presentation of this article,
and Maryam Aliakbarpour for references on testing with imperfect information.
We thank anonymous reviewers for their comments and references to related work.
\end{samepage}
}

\bibliographystyle{alpha}
\bibliography{references.bib}

\appendix
\addtocontents{toc}{\protect\setcounter{tocdepth}{1}}

\section{Poissonization and Boosting}
\label{section:poissonization}

\subsubsection{Poissonization}

It is standard (see \eg \cite{VV11,VV17,WY19}) to analyze distribution testing algorithms in the
``Poissonized'' setting, where, instead of taking $m$ independent samples, the algorithm first
samples $\bm{m} \sim \Poi(m)$ and then takes $\bm{m}$ independent samples. We slightly abuse
notation and simply say that the tester takes $\Poi(m)$ samples.  The advantage of this technique is
that the number of times each domain element appears in the sample becomes independent.  Taking
$\Poi(m)$ independent samples from distribution $\pi$ over domain $\cX$ is equivalent to taking
$\Poi(m \pi(x))$ samples independently from each $x \in \cX$.  Since $\Poi(m)$ is tightly
concentrated around $m$, one can convert back and forth between the Poissonized and non-Poissonized
model while preserving upper and lower sample complexity bounds.  We briefly state the conversions
relevant to us, and refer the reader to \eg \cite[Appendix C]{Can22} and references therein for
details:

\begin{proposition}
    \label{prop:poissonization}
    We say that an algorithm is a \emph{Poissonized} $(\Pi_1, \Pi_2, \alpha)$-distribution tester
    under the parity trace with sample complexity $m$ if it satisfies the same conditions as
    \cref{def:testing-parity-trace}, except that it draws (the parity trace of) a sample
    of size $\Poi(m)$ instead of $m$.
    Then for all $\delta > 0$, the following hold:
    \begin{enumerate}
        \item If there is a (standard) $(\Pi_1, \Pi_2, 1-\delta/2)$-distribution tester under the
            parity trace with sample complexity $m$, then there is a Poissonized
            $(\Pi_1, \Pi_2, 1-\delta)$-distribution tester under the parity trace with sample
            complexity $\max\left\{2m, 12\log(4/\delta)\right\}$.
        \item If there is a Poissonized $(\Pi_1, \Pi_2, 1-\delta/2)$-distribution tester under the
            parity trace with sample complexity $m$, then there is a (standard)
            $(\Pi_1, \Pi_2, 1-\delta)$-distribution tester under the parity trace with sample
            complexity $\max\left\{\frac{3}{2} m, 18\log(4/\delta)\right\}$.
    \end{enumerate}
    We may similarly convert the confused collector model between the standard and Poissonized
    versions by adapting \cref{def:testing-confused-collector}, and note that the analogous results
    hold for that model.
\end{proposition}

\subsubsection{Boosting Success Probabilities}

In standard distribution testing, one can usually boost the probability of success of
an algorithm to any desired level by amplification: repeat the algorithm many times and take a
majority vote.

In the parity trace model, the algorithm receives only the trace of a single sample, so it cannot
simply repeat the test multiple times. Therefore we require a different technique for boosting the
success probability. By taking a larger original sample, the tester can perform ``sample splitting''
to produce a number of independent traces, which it can then test independently, as we describe
below.

Recall that, in the (Poissonized) parity trace model, when the tester draws $\Poi(m)$ samples from
$\pi$, it receives a trace
\[
    \bm{\cT} = 1^{\bm{A}_1} 0^{\bm{B}_1} \dotsc 1^{\bm{A}_n} 0^{\bm{B}_n} \,,
\]
where each $\bm{A}_i \sim \Poi(m p_i)$ and each $\bm{B}_i \sim \Poi(m q_i)$ are mutually
independent.

\begin{fact}
    Let $\lambda > 0$ and $k \in \bN$. Define random variables $\bm{X}_1, \dotsc, \bm{X}_k$ via the
    following probabilistic process:
    \begin{enumerate}
        \item Draw $X \gets \Poi(k \lambda)$;
        \item Draw $X_1, \dotsc, X_k \gets \Multinomial(X, (1/k, \dotsc, 1/k))$.
    \end{enumerate}
    Then $\bm{X}_1, \dotsc, \bm{X}_k$ are mutually independent random variables and
    $\bm{X}_i \sim \Poi(\lambda)$ for each $i \in [k]$.
\end{fact}

Therefore, we may simulate $k$ parity traces of size $\Poi(m)$ by
\begin{enumerate}
\item Drawing a parity trace
$\bm{\cT}$ of size $\Poi(mk)$; and
\item Assigning each symbol in $\bm{\cT}$, from left to right, to
$\bm{\cT}_{\bm{j}}$ where $\bm{j}$ is chosen from $[k]$ independently uniformly at random.
\end{enumerate}

Then the fact above implies that the $j^{th}$ trace is distributed as
\[
    \bm{\cT}_j = 1^{\bm{A}_{j,1}} 0^{\bm{B}_{j,1}} \dotsc 1^{\bm{A}_{j,n}} 0^{\bm{B}_{j,n}}
\]
where $\bm{A}_{j,i} \sim \Poi(m p_i)$ and $\bm{B}_{j,i} \sim \Poi(m q_i)$ independently for all
$i$ and $j$, as desired.

As a consequence, the probability of success of a tester under the parity trace may be boosted
to any level $1-\delta$ by incurring a multiplicative factor of $\Theta(\log(1/\delta))$
in the sample complexity.

\begin{remark}
Probability boosting is not possible in the confused collector model, in the conventional sense,
because the algorithm does not have control over its resolution parameter $\eta$.
\end{remark}

\section{Missing Proofs from
\texorpdfstring{\cref{sec:upper-bound-confused-collector}}{Section~\ref{sec:upper-bound-confused-collector}}}

\begin{proposition}
    \label{prop:geometric-sum-expression}
    Let $\eta \in (0,1)$. Let $S_i$ denote the sum of the entries in the $i$-th column of
    $\phi^\pth$ for each $i \in \bZ_n$, and let $h \define (n-1)/2$.
    Then for any non-negative integer $N \le n/2$,
    \[
        - \sum_{i=0}^{\ceil{N/2}-1} S_i
            - \sum_{i=0}^{\floor{N/2}-1} S_{n-1-i}
            + \sum_{i=0}^{\ceil{N/2}-1} S_{\ceil{h}+i}
            + \sum_{i=0}^{\floor{N/2}-1} S_{\floor{h}-i}
        <
        \frac{2}{\eta^2} \,.
    \]
\end{proposition}
\begin{proof}
    Recall that $\phi^\pth_{i,j} = \nu^{\abs{i-j}}$ where $\nu = 1-\eta$.
    We express the column sums explicitly and reduce the geometric sums that emerge:
    \begin{align*}
        &\left[
            - \sum_{i=0}^{\ceil{N/2}-1} S_i
            - \sum_{i=0}^{\floor{N/2}-1} S_{n-1-i}
            + \sum_{i=0}^{\ceil{N/2}-1} S_{\ceil{h}+i}
            + \sum_{i=0}^{\floor{N/2}-1} S_{\floor{h}-i}
            \right] \\
        &\quad= \left[
            \begin{array}{l}
                - \sum_{i=0}^{\ceil{N/2}-1} \sum_{j=0}^{n-1} \nu^{\abs{i-j}}
                - \sum_{i=0}^{\floor{N/2}-1} \sum_{j=0}^{n-1} \nu^{\abs{n-1-i-j}} \\
                + \sum_{i=0}^{\ceil{N/2}-1} \sum_{j=0}^{n-1} \nu^{\abs{\ceil{h}+i-j}}
                + \sum_{i=0}^{\floor{N/2}-1} \sum_{j=0}^{n-1} \nu^{\abs{\floor{h}-i-j}}
            \end{array}
            \right] \\
        &\quad= \left[
            \begin{array}{l}
                - \sum_{i=0}^{\ceil{N/2}-1} \left(
                    \sum_{j=0}^i \nu^{i-j}
                    + \sum_{j=i+1}^{n-1} \nu^{j-i} \right)
                \\
                - \sum_{i=0}^{\floor{N/2}-1} \left(
                    \sum_{j=0}^{n-1-i} \nu^{n-1-i-j}
                    + \sum_{j=n-i}^{n-1} \nu^{j+i+1-n} \right)
                \\
                + \sum_{i=0}^{\ceil{N/2}-1} \left(
                    \sum_{j=0}^{\ceil{h}+i} \nu^{\ceil{h}+i-j}
                    + \sum_{j=\ceil{h}+i+1}^{n-1} \nu^{j-i-\ceil{h}} \right)
                \\
                + \sum_{i=0}^{\floor{N/2}-1} \left(
                    \sum_{j=0}^{\floor{h}-i} \nu^{\floor{h}-i-j}
                    + \sum_{j=\floor{h}-i+1}^{n-1} \nu^{j+i-\floor{h}} \right)
            \end{array}
            \right] \\
        &\quad= \left[
            \begin{array}{l}
                - \sum_{i=0}^{\ceil{N/2}-1} \left(
                    \frac{\nu^{i} - \nu^{-1}}{1 - \nu^{-1}}
                    + \frac{\nu - \nu^{n-i}}{1 - \nu} \right)
                - \sum_{i=0}^{\floor{N/2}-1} \left(
                    \frac{\nu^{n-1-i} - \nu^{-1}}{1 - \nu^{-1}}
                    + \frac{\nu - \nu^{i+1}}{1 - \nu} \right)
                \\
                + \sum_{i=0}^{\ceil{N/2}-1} \left(
                    \frac{\nu^{\ceil{h}+i} - \nu^{-1}}{1 - \nu^{-1}}
                    + \frac{\nu - \nu^{n-i-\ceil{h}}}{1 - \nu} \right)
                + \sum_{i=0}^{\floor{N/2}-1} \left(
                    \frac{\nu^{\floor{h}-i} - \nu^{-1}}{1 - \nu^{-1}}
                    + \frac{\nu - \nu^{n+i-\floor{h}}}{1 - \nu} \right)
            \end{array}
            \right] \\
        &\quad= \frac{1}{\eta} \left[
            \begin{array}{l}
                - \sum_{i=0}^{\ceil{N/2}-1} \left(
                    1 - \nu^{i+1}
                    + \nu - \nu^{n-i} \right)
                - \sum_{i=0}^{\floor{N/2}-1} \left(
                    1 - \nu^{n-i}
                    + \nu - \nu^{i+1} \right)
                \\
                + \sum_{i=0}^{\ceil{N/2}-1} \left(
                    1 - \nu^{\ceil{h}+1+i}
                    + \nu - \nu^{n-i-\ceil{h}} \right)
                + \sum_{i=0}^{\floor{N/2}-1} \left(
                    1 - \nu^{\floor{h}+1-i}
                    + \nu - \nu^{n+i-\floor{h}} \right)
            \end{array}
            \right] \\
        &\quad= \frac{1}{\eta} \left[
            \begin{array}{l}
                \sum_{i=0}^{\ceil{N/2}-1} \nu^{i+1}
                + \sum_{i=0}^{\ceil{N/2}-1} \nu^{n-i}
                + \sum_{i=0}^{\floor{N/2}-1} \nu^{n-i}
                + \sum_{i=0}^{\floor{N/2}-1} \nu^{i+1}
                \\
                - \sum_{i=0}^{\ceil{N/2}-1} \nu^{\ceil{h}+1+i}
                - \sum_{i=0}^{\ceil{N/2}-1} \nu^{n-i-\ceil{h}}
                - \sum_{i=0}^{\floor{N/2}-1} \nu^{\floor{h}+1-i}
                - \sum_{i=0}^{\floor{N/2}-1} \nu^{n+i-\floor{h}}
            \end{array}
            \right] \\
        &\quad= \frac{1}{\eta} \left[
            \begin{array}{l}
                \frac{\nu - \nu^{\ceil{N/2}+1}}{1-\nu}
                + \frac{\nu^{n} - \nu^{n-\ceil{N/2}}}{1 - \nu^{-1}}
                + \frac{\nu^{n} - \nu^{n-\floor{N/2}}}{1 - \nu^{-1}}
                + \frac{\nu - \nu^{\floor{N/2}+1}}{1 - \nu}
                \\
                - \frac{\nu^{\ceil{h}+1} - \nu^{\ceil{h}+1+\ceil{N/2}}}{1 - \nu}
                - \frac{\nu^{n-\ceil{h}} - \nu^{n-\ceil{h}-\ceil{N/2}}}{1 - \nu^{-1}}
                - \frac{\nu^{\floor{h}+1} - \nu^{\floor{h}+1-\floor{N/2}}}{1 - \nu^{-1}}
                - \frac{\nu^{n-\floor{h}} - \nu^{n+\floor{N/2}-\floor{h}}}{1 - \nu}
            \end{array}
            \right] \\
        &\quad= \frac{1}{\eta^2} \left[
            \begin{array}{l}
                \nu - \nu^{\ceil{N/2}+1}
                + \nu^{n+1-\ceil{N/2}} - \nu^{n+1}
                \\
                + \nu^{n+1-\floor{N/2}} - \nu^{n+1}
                + \nu - \nu^{\floor{N/2}+1}
                \\
                - \nu^{\ceil{h}+1} + \nu^{\ceil{h}+1+\ceil{N/2}}
                - \nu^{n+1-\ceil{h}-\ceil{N/2}} + \nu^{n+1-\ceil{h}}
                \\
                - \nu^{\floor{h}+2-\floor{N/2}} + \nu^{\floor{h}+2}
                - \nu^{n-\floor{h}} + \nu^{n+\floor{N/2}-\floor{h}}
            \end{array}
            \right] \\
        &\quad= \frac{1}{\eta^2} \left[
            \begin{array}{l}
                (\nu + \nu)
                + (\nu^{n+1-\ceil{N/2}} - \nu^{\ceil{N/2}+1})
                + (\nu^{n+1-\floor{N/2}} - \nu^{\floor{N/2}+1})
                - (\nu^{n+1} + \nu^{n+1})
                \\
                + (\nu^{\ceil{h}+1+\ceil{N/2}} - \nu^{\ceil{h}+1})
                + (\nu^{n+1-\ceil{h}} - \nu^{n+1-\ceil{h}-\ceil{N/2}})
                \\
                + (\nu^{\floor{h}+2} - \nu^{\floor{h}+2-\floor{N/2}})
                + (\nu^{n+\floor{N/2}-\floor{h}} - \nu^{n-\floor{h}})
            \end{array}
            \right] \\
        &\quad< \frac{2}{\eta^2} \,,
    \end{align*}
    where we used the facts that $0 < \nu < 1$ and, in the last step, that $2\ceil{N/2} \le n$
    (which holds because $2N \le n$ by assumption).
\end{proof}

\section{Missing Proofs from
\texorpdfstring{\cref{sec:general-upper-bound}}{Section~\ref{sec:general-upper-bound}}}

\subsection{Testing Uniformity: the Small $\epsilon$ Case}
\label{section:parity-trace-small-eps}

\newcommand{\psample}{\mathcal{S}}

Here we prove \cref{lemma:upper-bound-linear-trace-small-epsilon}.
The standard testing algorithms for uniformity make their decision based upon only the
\emph{histogram} of the samples, which is the tuple $(X_1, \dotsc, X_{2n})$ where $X_i$ is the
number of times element $i \in [2n]$ appears in the sample. For a sample $S$, we will write $H(S)$
for the histogram. The well-known uniformity testing result can be stated as follows:

\begin{theorem}[\cite{VV17,DGPP19}]
\label{thm:standard-histogram-tester}
There is a constant $C > 0$ and an algorithm \textsc{UniformityHistogramTester} (abbreviated as
\textsc{UHT}) such that, for any $n \in \bN$ and $\epsilon > 0$ and $m \geq C \cdot \frac{ \sqrt n
}{ \epsilon^2 }$:
\begin{enumerate}
\item If $\pi = \pi(\mu,\mu)$, $\Pr{ \textsc{UHT}( H(S) ) \text{ accepts } } \geq 3/4$; and,
\item If $\pi = \pi(p,q)$ is $\epsilon$-far from uniform, $\Pr{ \textsc{UHT}( H(S) ) \text{ rejects
}} \geq 3/4$.
\end{enumerate}
\end{theorem}

\begin{claim}
    \label{claim:small-eps-big-sample}
    Let $\pi = \pi(\mu,\mu)$ be the uniform distribution over $[2n]$, and let $m \geq 2 \cdot n \log(100
    n)$. Let $\bm{S} \sim \samp(\pi, m)$. Then
    \[
        \Pr{ \exists i \in [2n] : i \notin \bm{S} } < 1/50 \,.
    \]
\end{claim}
\begin{proof}
    By the union bound, using the fact that for any $i \in [2n]$,
    $ \Pr{ i \notin \bm{S} } = \left(1 - \frac{1}{2n}\right)^m < e^{-\frac{m}{2n}} = e^{-\log(100 n)}
    = \frac{1}{100 n} $.
\end{proof}

\begin{algorithm}[H]
    \caption{Uniformity tester for the case when $\epsilon < \frac{K \log^{3}}{n^{1/4}}$.}
    \label{alg:general-small-eps}

    \hspace*{\algorithmicindent}
    Set $m \gets \Theta\left(\frac{\sqrt n}{\epsilon^2} \log^7 n \right)$. \\
    \hspace*{\algorithmicindent} \textbf{Constants:}
        $K = K_{\alpha,\beta,\gamma} > 1$ as in \cref{alg:uniformity-tester-linear-trace}.\\
    \hspace*{\algorithmicindent}
    \textbf{Input:} For $\pi = \pi(p,q)$ on domain $[n]$,
        receive $\trace(S)$ for sample $S \gets \samp(\pi, m)$ \\
    \hspace*{\algorithmicindent} \textbf{Requires:} $\epsilon < \frac{K\log^{3} n}{n^{1/4}}$. \\

    \begin{algorithmic}[1]
        \Procedure{UniformityTesterSmall}{$p, q, n, \epsilon$}
            \State Let $X_1, \dotsc, X_n$ be the run-lengths of 1s in the trace, as defined in
\cref{section:parity-trace}.
            \State Let $X'_1, \dotsc, X'_n$ be the run-lengths of 0s in the trace, as in
\cref{section:parity-trace}.
            \If {$\exists i \in [n]$ such that $X_i = 0$ or $X'_i = 0$}
              \State \textbf{Reject}
            \Else
              \State \textbf{Output}
                    $\textsc{UniformityHistogramTester}(X_1, X'_1, \dotsc, X_{n}, X'_n)$
            \EndIf
        \EndProcedure
    \end{algorithmic}
\end{algorithm}

We will need the following fact about the total variation distance.

\begin{fact}
    \label{fact:tv-distance-conditional}
    Let $\cD$ be a probability distribution and $E$ an event in the same probability space.
    Denote by $\cD_{|E}$ the probability distribution of a random variable distributed by $\cD$
    conditional on $E$. Then
    \[
        \dist_\TV(\cD, \cD_{|E}) \le \Pr{\neg E} \,.
    \]
\end{fact}

The following proves \cref{lemma:upper-bound-linear-trace-small-epsilon}.

\begin{lemma}
    Suppose that $\epsilon < \frac{K \log^{3} n}{n^{1/4}}$. Then \cref{alg:general-small-eps} satisfies
    the following:
    \begin{enumerate}
        \item If $\pi = \pi(\mu,\mu)$, the algorithm will accept with probability at least $2/3$.
        \item If $\pi = \pi(p,q)$ is $\epsilon$-far from uniform, then the algorithm will reject with
            probability at least $2/3$.
    \end{enumerate}
\end{lemma}
\begin{proof}
By \cref{thm:standard-histogram-tester}, we know that for appropriate choice of constant $C >
0$, if $m \geq C \cdot \frac{\sqrt n}{\epsilon^2}$, then the \textsc{UHT} algorithm will be correct
with probability at least $3/4$. In our case, $m$ satisfies this condition.

    Let $\bm{S} \sim \samp(\pi, m)$ and $\bm{T} = \trace(\bm{S})$. Define
    \[
(\bm{Z}_1, \dotsc,
\bm{Z}_{2n}) \define (\bm{X}_1, \bm{X}'_1, \bm{X}_2, \bm{X}'_2, \dotsc, \bm{X}_n, \bm{X}'_n) \,,
    \]
    so that $\bm{Z}$ is the vector of run-lengths in $\bm{T}$.

    Write $A$ for the event that $\bm{Z}_i > 0$ for all $i \in [2n]$.  Observe that, if event $A$
    occurs, then $\bm{Z}$ is the histogram $H(\bm{S})$.
    Suppose that $\pi = \pi(\mu,\mu)$.
    We first argue that \textsc{UHT} has small probability of rejection even if its input comes
    from a sample conditioned on event $A$. Let $\bm{S}'$ be the random variable distributed
    as the sample $\bm{S}$ conditional on $A$ occurring.
    Then, by \cref{fact:tv-distance-conditional} and \cref{claim:small-eps-big-sample},
    \[
        \dist_\TV(\bm{S}, \bm{S}') \le \Pr{\neg A} < \frac{1}{50} \,.
    \]
    Therefore $\Pr{\textsc{UHT}(H(\bm{S})) \ne \textsc{UHT}(H(\bm{S}'))} < 1/50$ and
    \begin{align*}
        \Pruc{}{\textsc{UHT}(\bm{Z}) \text{ rejects}}{A}
        &= \Pruc{}{\textsc{UHT}(H(\bm{S})) \text{ rejects}}{A} \\
        &= \Pr{\textsc{UHT}(H(\bm{S}')) \text{ rejects}} \\
        &< \Pr{\textsc{UHT}(H(\bm{S})) \text{ rejects}} + \frac{1}{50}
        < \frac{1}{4} + \frac{1}{50} \,.
    \end{align*}
    Then the probability that \cref{alg:general-small-eps} rejects is
    \begin{align*}
        &\Pr{\textsc{UniformityTesterSmall}(\bm{Z}) \text{ rejects}} \\
        &\qquad = \Pr{ A }\Pruc{}{ \textsc{UHT}(\bm{Z}) \text{ rejects } }{ A } + \Pr{ \neg A }
        < \frac{1}{4} + \frac{2}{50}
        < \frac{1}{3} \,.
    \end{align*}
    Now suppose that $\pi = \pi(p,q)$ is $\epsilon$-far from uniform.
    Since $\bm{Z} = H(\bm{S})$ when $A$ occurs, we have
    \begin{align*}
        \Pruc{}{\textsc{UniformityTesterSmall}(\bm{Z}) \text{ rejects}}{A}
        &= \Pruc{}{\textsc{UHT}(\bm{Z}) \text{ rejects}}{A} \\
        &= \Pruc{}{\textsc{UHT}(H(\bm{S})) \text{ rejects}}{A}\,.
    \end{align*}
    Moreover, since \cref{alg:general-small-eps} always rejects when $A$ does not occur, we have
    \[
        1 = \Pruc{}{\textsc{UniformityTesterSmall}(\bm{Z}) \text{ rejects}}{\neg A}
        \ge \Pruc{}{\textsc{UHT}(H(\bm{S})) \text{ rejects}}{\neg A} \,.
    \]
    Hence the probability that \cref{alg:general-small-eps} rejects is
    \begin{align*}
        &\Pr{\textsc{UniformityTesterSmall}(\bm{Z}) \text{ rejects}} \\
        &\qquad = \Pr{A} \Pruc{}{\textsc{UniformityTesterSmall}(\bm{Z}) \text{ rejects}}{A} \\
            &\qquad \qquad + \Pr{\neg A}
                \Pruc{}{\textsc{UniformityTesterSmall}(\bm{Z}) \text{ rejects}}{\neg A} \\
        &\qquad \ge \Pr{A} \Pruc{}{\textsc{UHT}(H(\bm{S})) \text{ rejects}}{A} +
            \Pr{\neg A} \Pruc{}{\textsc{UHT}(H(\bm{S})) \text{ rejects}}{\neg A} \\
        &\qquad = \Pr{\textsc{UHT}(H(\bm{S})) \text{ rejects}}
        > 3/4 > 2/3 \,,
    \end{align*}
which concludes the proof.
\end{proof}

\section{Edit Distance Proofs}
\label{section:appendix-edit-distance-proofs}

\subsection{Facts About Edit Distance and Labeled Distributions}
\begin{fact}
\label{fact:tv-to-ham}
Let $\cD_f$ and $\cD_g$ be labeled distributions over any domain $\cX$. Then
\[
  \dist_\TV(\cD_f, \cD_g) = \Pru{x \sim \cD}{f(x) \neq g(x)} \,.
\]
\end{fact}
\begin{proof}
Assume, for simplicity of notation, that $\cX$ is countable. Using
\cref{prop:tv-distance-labeled},
\[
  \dist_\TV(\cD_f, \cD_g)
  = \frac{1}{2} \sum_x \ind{f(x) \neq g(x)} (2 \cD(x))
  = \sum_x \ind{f(x) \neq g(x)} \cdot \cD(x)
  = \Pru{x \sim \cD}{f(x) \neq g(x)} \,. \qedhere
\]
\end{proof}

\begin{proposition}
\label{prop:labeled-distribution-construction}
For any proper labeled distribution $(f,\cD)$ on domain $\bZ$ and any distribution $\pi$ on $\bN$,
there exists a distribution $\cE$ on $\bZ$ such that $\pi = \pi_{f,\cE}$ and
\[
\dist_\TV(\cD, \cE) = \dist_\TV(\pi_{f,\cD}, \pi) \,.
\]
\end{proposition}
\begin{proof}
Let $a_1 < a_2 < \dotsm $ be the alternation sequence for $f$, and use the convention $a_0 =
-\infty$; if the sequence is of finite length $t$, also define $a_{t+1} = \infty$.
For each interval $I = (a_{i-1}, a_i]$, we define $\cE$ on the points $x \in I$ as follows.
\begin{itemize}
\item If $\cD(I) \leq \pi(i)$, choose an arbitrary point $x^* \in I$. For $x \in I \setminus \{
x^* \}$, let $\cE(x) \gets \cD(x)$. Then let $\cE(x^*) \gets \pi(i) - \sum_{x \in I \setminus \{ x^*
\} } \cD(x)$. Observe that $\cE(I) = \pi(i)$, as desired, and
\begin{align*}
  \sum_{x \in I} |\cD(x) - \cE(x)|
      &= \cE(x^*) - \cD(x^*) + \left(\sum_{x \in I \setminus \{x^*\}} \cE(x) \right)
                            - \left(\sum_{x \in I \setminus \{x^*\}} \cD(x) \right) \\
      &= |\cE(I) - \cD(I)| \,.
\end{align*}
\item If $\cD(I) > \pi(i)$, let $\cE(x) \gets \cD(x) - \delta_x$ for an arbitrary choice of
values $\delta_x$ satisfying $0 \leq \delta_x \leq \cD(x)$ and $\sum_{x \in I} \delta_x = \cD(I) -
\pi(i)$; it is easy to verify that such a choice exists. Observe that, as desired,
\[
  \cE(I) = \cD(I) - \sum_{x \in I} \delta_x = \pi(i) \,,
\]
and
\[
  \sum_{x \in I} |\cE(x) - \cD(x)| = \sum_{x \in I} \delta_x = \left|\cD(I) - \pi(i)\right|
    = |\cD(I) - \cE(I)| \,.
\]
\end{itemize}
We now have $\pi_{f,\cE} = \pi$, and
\begin{align*}
  \dist_\TV(\pi_{f,\cD}, \pi)
  &= \dist_\TV(\pi_{f,\cD}, \pi_{f,\cE})
  = \frac{1}{2} \sum_{i=1}^{k+1} |\pi_{f,\cD}(i) - \pi_{f,\cE}(i)|
  = \frac{1}{2} \sum_{i=1}^{k+1} |\cD(a_{i-1}, a_i] -  \cE(a_{i-1}, a_i] | \\
  & = \frac{1}{2} \sum_x |\cD(x) - \cE(x)|
  = \dist_\TV(\cD_f, \cE_f) \,. \qedhere
\end{align*}
\end{proof}

\begin{fact}
\label{fact:edit-match-mass}
Let $a_1, \dotsc, a_m \geq 0$ and $b \geq 0$. Then there exist $b_1, \dotsc, b_m$ such that $\sum_i
b_i = b$ and $\sum_i |a_i - b_i| = \left| b - \sum_i a_i \right|$.
\end{fact}
\begin{proof}
First assume $\sum_i a_i \leq b$. Then assign $b_i = a_i$ for $i < m$ and $b_m = b - \sum_{i < m}
a_i$. Then $\sum_i b_i = b$ and $\sum_i |a_i - b_i| = b_m - a_m = b - \left(\sum_{i < m} a_i \right)
- a_m = b - \sum_i a_i$, as desired.

Now assume $\sum_i a_i > b$.  Let $j$ be the smallest number such that $\sum_{i \leq j} a_i > b$.
Assign $b_i = a_i$ for $i < j$, $b_j = b - \sum_{i < j} a_i$, and $b_i = 0$ for $i > j$. Then
$\sum_i b_i = b$ and
\[
  \sum_i |a_i - b_i| = (a_j - b_j) + \sum_{i > j} a_i  = a_j - b + \sum_{i < j} a_i + \sum_{i > j}
a_i = \left(\sum_i a_i\right) - b \,.\qedhere
\]
\end{proof}

\begin{fact}
\label{fact:edit-zero-mass}
Let $(f, \cD)$ and $(g, \cE)$ be any two proper labeled distributions, and let $g'$ be any function
such that $g'(x) = g(x)$ when $\min(\cD(x), \cE(x)) > 0$. Then
\[
  \dist_\TV(\cD_f, \cE_{g'}) = \dist_\TV(\cD_f, \cE_g) \,.
\]
\end{fact}
\begin{proof}
This follows from \cref{prop:tv-distance-labeled}.
\end{proof}

\begin{fact}
\label{fact:edit-zero-interval}
Let $(f, \cD)$ and $(g, \cE)$ be any two proper labeled distributions. Then there exist $(f', \cD')$
and $(g', \cE')$ which satisfy $\pi_{f', \cD'} = \pi_{f,\cD}$, $\pi_{g',\cE'} = \pi_{g,\cE}$, and
$\dist_\TV(\cD'_{f'}, \cE'_{g'}) \leq \dist_\TV(\cD_f, \cE_g)$, which satisfy the following
conditions:
\begin{enumerate}
\item If $I$ is any interval such that $f'$ and $g'$ are both constant on $I$, and $f'(x) \neq
g'(x)$ on all $x \in I$, then either $\cD'(I) = 0$ or $\cE'(I) = 0$.
\item $f'$ and $g'$ have no alternation points in common.
\end{enumerate}
\end{fact}
\begin{proof}
Let $a_1 < a_2 < \dotsm$ and $b_1 < b_2 < \dotsm$ be the alternation sequences for $f$ and $g$.  We
may assume without loss of generality that $f$ and $g$ do not have any alternation points in common.
This is because if $a_i = b_j$, then we increment all values $b_{j'} \geq b_j$ and $a_{i'} > a_i$ by
1, shift all densities $\cD(x)$ and $\cE(x)$ to the right by one position for $x > a_i$ and $y >
b_j$, and redefine $\cD(a_i+1) = \cE(b_j+1)=0$.

Any interval $I$ such that $f$ and $g$ are both constant on $I$ and $f(x) \neq g(x)$ on all $x \in
I$, must satisfy $I \subset I^*$ where $I^* \define (a_{i-1}, a_i] \cap (b_{j-1}, b_j]$ for some
alternation points $a_i, b_j$. Since $f$ and $g$ do not have any alternation points in common, then
either there exists $z \in (a_{i-1}, a_i]$ such that $f(x) = g(x)$, or there exists $z \in (b_{j-1},
b_j]$ such that $f(z) = g(z)$. In the first case, define $\cD'$ the same as $\cD$ except on $I^*
\cup \{z\}$, and define $\cD'(x) = 0$ for $x \in I^*$ and $\cD'(z) = \cD(z) + \cD(I^*)$. Since
$\{z\} \cup I^* \subset (b_{j-1}, b_j]$, we have $\cD'(b_{j-1}, b_j] = \cD(b_{j-1}, b_j]$, so
$\pi_{f,\cD'} = \pi_{f,\cD}$. Observe that
\[
  |\cD'(z) - \cE(z)| + \sum_{x \in I^*} (\cD'(x) + \cE(x))
  = |\cD(z) + \cD(I^*) - \cE(z)| + \cE(I^*)
  \leq |\cD(z) - \cE(z)| + \cD(I^*) + \cE(I^*) \,,
\]
so $\dist_\TV(\cD'_f, \cE_g) \leq \dist_\TV(\cD_f, \cE_g)$ by \cref{prop:tv-distance-labeled}. In
the second case, where $z \in (b_{j-1}, b-j]$, we perform the analogous adjustment on $\cE$ to get
$\cE'$.
\end{proof}

\begin{fact}
\label{fact:edit-condense}
Let $(f, \cD)$ and $(g, \cE)$ be any two proper labeled distributions. Let $a_1 < a_2 < \dotsm$ be
the alternation sequence of $f$ and let $b_1 < b_2 < \dotsm$ be the alternation sequence of $g$.
Then there exist distributions $\cD'$ and $\cE'$ that satisfy the following conditions:
\begin{enumerate}
\item $\pi_{f,\cD} = \pi_{f,\cD'}$ and $\pi_{g,\cE'} = \pi_{g,\cE}$;
\item $\cD'$ and $\cE'$ are supported on the set $C = \{ a_i \} \cup \{ b_j \}$;
\item $\dist_\TV(\cD'_f, \cE'_g) \leq \dist_\TV(\cD_f, \cE_g)$.
\end{enumerate}
\end{fact}
\begin{proof}
Define $\cD'$ and $\cE'$ as follows. Write $C = \{ c_1, c_2, \dotsc \}$ where $c_1 \leq c_2 \leq c_3
\leq \dotsm$. For each interval $(c_{i-1}, c_i] \neq \emptyset$, define $\cE'(c_i) =
\cE(c_{i-1},c_i]$ and $\cD'(c_i) = \cD(c_{i-1}, c_i]$. It is easy to verify the required properties.
\end{proof}

\subsection{Equivalence of Edit Distance Definitions}
\label{section:edit-distance-definition-equivalence}

We must prove the following lemma from \cref{section:edit-distance}.

\lemmaeditdistancealternate*

We will use the following facts, which are easy to verify, by swapping consecutive pairs
of permitted operations (and adjusting the indices appropriately).
\begin{fact}
\label{fact:edit-standard-form}
Let $a$ be any fractional string and let $O_1, \dotsc, O_k$ be any sequence of permitted operations
on $a$. Then there exists a sequence $O'_1, \dotsc, O'_k$ of permitted operations on $a$ such that
\[
(O_k \circ O_{k-1} \circ \dotsm \circ O_1)(a) = (O'_k \circ O'_{k-1} \circ \dotsm \circ O'_1)(a)
\]
and, for some $0 \leq i \leq j \leq k+1$, it holds that $O_\ell$ is an \emph{Insert} or
\emph{Rearrange} operation for all $\ell \leq i$; $O_\ell$ is an \emph{Adjust} operation for all $i
< \ell < j$; and $O_\ell$ is a \emph{Delete} or \emph{Rearrange} operation for all $\ell \geq j$.
\end{fact}

\begin{fact}
\label{fact:edit-zero-standard-form}
Let $a$ be any fractional string and let $O_1, \dotsc, O_k$ be any sequence of permitted operations
that are each \emph{Insert} or \emph{Rearrange} operations. Then there is a sequence $O'_1, \dotsc,
O'_k$ of permitted operations such that $(O'_k \circ \dotsm \circ O'_1)(a) = (O_k \circ \dotsm \circ
O_1)(a)$ and such that the following holds.  There is some $s \leq k$ such that for all $j > s$,
each $O_j$ is an \emph{Insert} operation, and for all $j \leq s$, each $O_j$ is either a
\emph{Rearrange} operation, or an \emph{Insert} operation of the form $\mathsf{ins}_{i,b}$ where
either $b = a_{i-1}$ or $b = a_i$.
\end{fact}

\begin{proof}[Proof of \cref{lemma:edit-distance-alternate}]
For a labeled distribution $(f,\cD)$ where $\cD$ is finitely-supported, there is $k \in \bN$ such
that $\cD(i) = 0$ for all $i \in \bZ$ with $|i| > k$, and we define
\[
  \mathsf{str}(\cD_f) \define (f(-k))^{\cD(-k)} (f(1-k))^{\cD(1-k)} \dotsm (f(k))^{\cD(k)} \,.
\]
\textbf{Upper bound.}
Let $(f,\cD)$ and $(g,\cE)$ be any two labeled distributions such that $\cD$ and $\cE$ are finitely
supported and $\pi_{f,\cD} = \pi$, $\pi_{g,\cE} = \pi'$. (Since $\pi, \pi'$ are finitely supported,
such labeled distributions always exist.) We will prove that $\dist_\edit(\pi,\pi') \leq \|\cD_f -
\cE_g\|_\TV$ in two steps. First, we show that $\dist_{\mathsf{fr-edit}}(\mathsf{str}(\pi),
\mathsf{str}(\cD_f)) = \dist_{\mathsf{fr-edit}}(\mathsf{str}(\pi'), \mathsf{str}(\cE_g)) = 0$.
Second, we show that $\dist_{\mathsf{fr-edit}}(\mathsf{str}(\cD_f), \mathsf{str}(\cE_g)) \leq
\dist_\TV(\cD_f, \cE_g)$. From here, the conclusion holds by the triangle inequality.

\emph{Step 1.}
Let $a_1, \dotsc, a_t$ be the alternation points of $f$, where we may assume that $f(x) = 1$ for all
$x \leq a_1$, and we may assume that there is a finite number of alternation points because
$\pi_{f,\cD}=\pi$ and $\pi$ is finitely-supported. Write $a_0 = -\infty$ and $a_{t+1} = \infty$, so
that $\cD(a_{i-1}, a_i] = \pi(i)$ for all $i \in [t+1]$. For each $i \in [t+1]$, observe that $f(x)
= \parity(i)$ for all $x \in (a_{i-1}, a_i]$. We replace each character $(\parity(i))^{\pi(i)}$ in
$\mathsf{str}(\pi)$ with the fractional string 
\begin{align*}
&(\parity(i))^{\cD(a_{i-1}+1)} 
(\parity(i))^{\cD(a_{i-1}+2)} \dotsm (\parity(i))^{\cD(a_i)} \\
=&(f(a_{i-1}+1))^{\cD(a_{i-1}+1)} (f(a_{i-1}+2))^{\cD(a_{i-1}+2)} \dotsm
  (f(a_i))^{\cD(a_i)} \,,
\end{align*}
using a finite sequence of Insert and Rearrange operations. Repeating this for each $i \in [t+1]$,
we arrive at the fractional string $\mathsf{str}(\cD_f)$, using only operations of cost 0. Repeating
the same argument for $\pi'$ and $\cE_g$, we get the similar conclusion, completing the first step
of the proof.

\emph{Step 2.}
There are $k, k'$ such that
\begin{align*}
  a \define \str(\cD_f) &= (f(-k))^{\cD(-k)} (f(1-k))^{\cD(1-k)} \dotsc (f(k))^{\cD(k)} \\
  b' \define \str(\cE_g) &= (g(-k'))^{\cE(-k')} (g(1-k'))^{\cE(1-k')} \dotsc (g(k'))^{\cE(k')} \,.
\end{align*}
Without loss of generality, we may assume $k' \leq k$ and define
\[
  b \define (g(-k))^{\cE(-k)} (g(1-k))^{\cE(1-k)} \dotsc (g(k))^{\cE(k)} \,.
\]
It is easy to see that $b$ can be obtained from $b' = \str(\cE_g)$ using only insertions, since
$\cE(x) = 0$ for $|x| > k'$. From \cref{prop:tv-distance-labeled} we have
\[
  \dist_\TV(\cD_f, \cE_g)
  = \frac12 \sum_{-k \leq i \leq k} \left(\ind{ f(i) = g(i) } |\cE(i) - \cD(i)|
                                + \ind{ f(i) \neq g(i) } (\cE(i) + \cD(i)) \right) \,.
\]
For each $-k \leq i \leq k$, we edit $a$ as follows:
\begin{itemize}
\item If $f(i) = g(i)$, use one \emph{Adjust} operation to replace the fractional character
$(f(i))^{\cD(i)}$ with $(g(i))^{\cE(i)}$, with cost $\frac 12 |\cD(i) - \cE(i)|$.
\item If $f(i) \neq g(i)$, use one \emph{Adjust} operation to replace the fractional character
$(f(i))^{\cD(i)}$ with $(f(i))^0$ with cost $\cD(i)/2$, followed by a \emph{Delete} operation and
\emph{Insert} operation to replace $(f(i))^0$ with $(g(i))^0$; and finally an \emph{Adjust}
operation to replace $(g(i))^0$ with $(g(i))^{\cE(i)}$ with cost $\cE(i)/2$. The total cost is
$\frac 12 (\cD(i) + \cE(i))$.
\end{itemize}
The resulting string is $b$ and has been obtained with cost
\[
  \frac12 \sum_{-k \leq i \leq k} \left(\ind{ f(i) = g(i) } |\cE(i) - \cD(i)|
                                + \ind{ f(i) \neq g(i) } (\cE(i) + \cD(i)) \right)
  = \dist_\TV(\cD_f, \cE_g) \,.
\]
Combined with the triangle inequality and Step 1, we have now proved that
\[
  \dist_\edit(\pi,\pi') \leq \inf \dist_\TV(\cD_f, \cE_g) \,,
\]
where the infimum is taken over all labeled distributions $(f,\cD)$ and $(g,\cE)$ that have $\cD$
and $\cE$ being \emph{finitely-supported}, and where $\pi_{f,\cD} = \pi$ and $\pi_{g,\cE} = \pi'$.
To complete the proof, we must allow labeled distributions not to be finitely supported. This is
achieved by observing that for any labeled distribution $(f,\cD)$ and any $\epsilon > 0$, we can
find a finitely-supported $\cD'$ such that $\pi_{f,\cD'} = \pi_{f,\cD}$ and $\dist_\TV(\cD_f,
\cD'_f) < \epsilon$.

\textbf{Lower bound.}
Consider any sequence of permitted edit operations $O_1, \dotsc, O_k$ such that $\str(\pi') = (O_k
\circ O_{k-1} \circ \dotsm \circ O_1)(\str(\pi))$, where due to \cref{fact:edit-standard-form} we
assume that $O_1, \dotsc, O_s$ are \emph{Insert} and \emph{Rearrange} operations, and $O_t, \dotsc,
O_k$ are \emph{Delete} and \emph{Rearrange} operations, for some $s < t$. Write $a = (O_s \circ
\dotsm \circ O_1)( \str(\pi) )$ and $b = (O_{t-1} \circ \dotsm \circ O_1)(\str(\pi))$. We
may then assume without loss of generality that the sequence is of the form described in
\cref{fact:edit-zero-standard-form}, where we write $s'$ for the index described there.
For each $j \in [s]$, write $a^{(j)} = (O_j \circ \dotsm \circ O_1)( \str(\pi) )$.

We will define a sequence $(f^{(0)}, \cD^{(0)}), (f^{(1)}, \cD^{(1)}), \dotsc, (f^{(s)}, \cD^{(s)})$
of labeled distributions inductively, in such a way that $\pi_{f^{(j)}, \cD^{(j)}} = \pi$ for each
$j$, and for each $a^{(j)} = (O_j \circ \dotsm \circ O_1)(\str(\pi))$ with $a^{(j)} =
(a^{(j)}_1)^{p^{(j)}_1} \dotsm (a^{(j)}_n)^{p^{(j)}_n}$ we will also have $\cD^{(j)}(i) =
p^{(j)}_i$, and $f^{(j)}(i) = a^{(j)}_i$ unless $p^{(j)}_i = 0$.

Define $f^{(0)}(i) = \parity(i)$ for all $i \in \bN$ and $f^{(0)}(i) = 1$ for $i \leq 1$.  Define
$\cD^{(0)}(i) = \pi(i)$ for $i \in \bN$ and $\cD^{(0)}(i) = 0$ otherwise. It holds by definition
that $\pi_{f^{(0)}, \cD^{(0)}} = \pi$.

For each $j \in [s']$, where $a^{(j)} = (a^{(j)}_1)^{p^{(j)}_1} \dotsm (a^{(j)}_n)^{p^{(j)}_n}$, we
define $(f^{(j)}, \cD^{(j)})$ as simply $f^{(j)}(i) = a^{(j)}_i$ and $\cD^{(j)}(i) = p^{(j)}_i$.
Consider the operation $O_j$. If $O_j$ is a \emph{Rearrange} operation then $f^{(j)} = f^{(j-1)}$
since none of the symbols change. If $O_j$ is a \emph{Insert} operation then it inserts a symbol
that is equal to the one before or after it. In either case, the number of alternation points of
$f^{(j)}$ is the same as the number of alternation points of $f^{(j-1)}$, and the mass of
$\cD^{(j)}$ and $\cD^{(j-1)}$ between the $i^{th}$ and $(i+1)^{th}$ respective alternation points
does not change. So $\pi_{f^{(j)}, \cD^{(j)}} = \pi_{f^{(j-1)}, \cD^{(j-1)}} = \pi$.

For the remaining operations $O_j$ with $s' < j \leq s$, we know that $O_j$ is an \emph{Insert}
operation. When inserting a new fractional character immediately before the $i^{th}$ fractional
character, we change $f^{(j-1)}$ to $f^{(j)}$ and $\cD^{(j)}$ to $\cD^{(j-1)}$ by shifting all
values $f^{(j-1)}(i')$ and $\cD^{(j-1)}(i')$ for $i' \geq i$ to the right by one place. Then we
define $\cD^{(j)}(i) = 0$ and set $f^{(j)}(i) = f^{(j-1)}(i-1)$, which does not increase the number
of alternation points. We once again have $\pi_{f^{(j)}, \cD^{(j)}} = \pi_{f^{(j-1)}, \cD^{(j-1)}} =
\pi$. In this case, we may have $f^{(j)}(i) \neq a^{(j)}_i$, but we have $p^{(j)}_i = 0$ since this
was an \emph{Insert} operation.

We now have a labeled distribution $(f, \cD) \define (f^{(s)}, \cD^{(s)})$ such that $\pi_{f,\cD} =
\pi$ and for all $i \in \bN$ it holds that $f(i) = a^{(s)}_i = a_i$ unless $p^{(s)}_i = p_i = 0$,
and $\cD(i) = p^{(s)}_i = p_i$ for all $i \in \bN$ (which further implies $\cD(i) = 0$ for $i \notin
\bN$). Note that the fractional string $b = (O_{t-1} \circ \dotsm \circ O_1)(\str(\pi))$ may be
obtained from $\str(\pi')$ only by \emph{Insert} and \emph{Rearrange} operations, and so by applying
the same argument we get $(g,\cE)$ such that $\pi_{g,\cE} = \pi'$ and for all $i \in \bN$, $g(i) =
b_i$ unless $q_i = 0$, and $\cE(i) = p_i$.

Now we must have $a_i = b_i$ for all $i$, since $b$ is obtained from $a$ using only \emph{Adjust}
operations. The cost of these \emph{Adjust} operations must be at least $\frac 12 \sum_i |p_i -
q_i|$. On the other hand, we have
\begin{align*}
  \dist_\TV(\cD_f, \cE_g)
  &= \frac 12 \sum_i \left( \ind{f(i) = g(i)} |\cD(i) - \cE(i)| + \ind{f(i) \neq g(i)} (\cD(i) + \cE(i))
\right) \\
  &= \frac 12 \sum_i \left( \ind{f(i) = g(i)} |p_i - q_i| + \ind{f(i) \neq g(i)} (p_i + q_i) \right) \\
  &= \frac 12 \sum_{i : p_i = 0 \text{ or } q_i = 0} \left( \ind{f(i) = g(i)} |p_i - q_i| + \ind{f(i) \neq g(i)} |p_i - q_i| \right) \\
  &\qquad +  \frac 12 \sum_{i : p_i > 0 \text{ and } q_i > 0} |p_i - q_i| \\
  &= \frac 12 \sum_i |p_i - q_i| \,.
\end{align*}
Therefore
\[
  \inf \dist_\TV(\cD_g, \cE_f) \leq \dist_\edit(\pi, \pi') \,,
\]
as desired.
\end{proof}

\subsection{Equivalence of Edit Distances for Strings and Distributions}
\label{section:equivalence-edit-distance-strings}

\newcommand{\ext}{\mathsf{ext}}

Write $\dist_\ham(x,y)$ for the Hamming distance between two strings $x,y$ with the same length.

\begin{definition}
    \label{def:string-extension}
    For a string $x \in \zo^*$, write $\ext(x)$ for the set of all strings $z \in
    \{0,1,\bot\}^*$ where the unique (not necessarily contiguous) subsequence $\widetilde z$ of $z$
    containing the non-$\bot$ characters is equal to $x$.
\end{definition}

\begin{fact}
    \label{fact:string-extension-hamming}
    Given strings $u \in \zo^N$ and $v \in \zo^M$, it holds that
    \[
        \dist_\stringedit(u, v) = \min_{x,y} \dist_\ham(x,y) \,,
    \]
    where the minimum is over all strings $x \in \ext(u)$ and $y \in \ext(v)$ of equal length.
\end{fact}

\lemmarelativeeditdistance*
\begin{proof} We proceed by establishing two claims.

\begin{claim}
$\dist_\edit(\pi, \pi') \leq \dist_\mathsf{rel-edit}(\psi(\pi), \psi(\pi'))$.
\end{claim}
\begin{proof}[Proof of claim]
Let $x \in \ext(\psi(\pi))$ and $y \in \ext(\psi(\pi'))$ be strings attaining
\[
    \dist_\reledit(\psi(\pi), \psi(\pi')) = \frac{1}{N} \dist_\ham(x,y) \,,
\]
and let $M$ be their
length. Note that $\psi(\pi)$ is an $n$-block string, for some $n \leq N$, and $\psi(\pi')$ is an
$n'$-block string for some $n' \leq N$. Then there exists a sequence $0=a_0 \leq a_1 < a_2 < \dotsm
< a_{n-1} \leq a_n$  such that for each $j \in [n]$ and each $i \in (a_{j-1}, a_j]$, it holds that
$x_i \in \{\bot, \parity(j)\}$.  Similarly, there exists a sequence $0=b_0 \leq b_1 < b_2 < \dotsm <
b_{n'-1} \leq b_{n'}$ such that for each $j \in [n']$ and each $i \in (b_{j-1}, b_j]$, it holds that
$y_i \in \{\bot, \parity(j)\}$.

We may then define $f : \bN \to \zo$ as the function with alternation sequence $(a_j)$, and $g : \bN
\to \zo$ as the function with alternation sequence $(b_j)$. Observe that, for each $i \in [N]$, we
have $f(i) = x_i$ when $x_i \neq \bot$, and $g(i) = y_i$ when $y_i \neq \bot$.

Now, define the probability distribution $\cD$ to have density $1/N$ on each $i \in [N]$ with $x_i
\neq \bot$, and define the probability distribution $\cE$ to have density $1/N$ on each $i \in [N]$
with $y_i \neq \bot$. It follows that $\pi_{f,\cD} = \pi$ and $\pi_{g,\cE} = \pi'$. Using
\cref{prop:tv-distance-labeled}:
\begin{align*}
  \dist_\edit(\pi, \pi')
  &\leq \dist_\TV(\cD_f, \cE_g) \\
  &= \frac{1}{2} \sum_{i=1}^N \left( \ind{f(i) \neq g(i)}(\cD(i)+\cE(i))
                            + \ind{f(i) = g(i)} |\cD(i)-\cE(i)| \right) \\
  &= \frac{1}{2} \sum_{i=1}^N \ind{f(i) \neq g(i)}(\cD(i)+\cE(i)) \,.
\end{align*}
If $f(i) \neq g(i)$ then, either:
\begin{enumerate}
\item $x_i = \bot$ and $y_i \neq \bot$, or $x_i \neq \bot$ and $y_i = \bot$, in which case
$\ind{f(i) \neq g(i)}(\cD(i) + \cE(i)) = \ind{x_i \neq y_i} \cdot \frac{1}{N}$; or
\item $x_i \neq \bot$ and $y_i \neq \bot$, in which case
$\ind{f(i) \neq g(i)}(\cD(i) + \cE(i)) = \ind{x_i \neq y_i} \cdot \frac{2}{N}$; or
\item $x_i = y_i = \bot$, in which case $\ind{f(i) \neq g(i)}(\cD(i) + \cE(i)) = 0$.
\end{enumerate}
Then 
\begin{align*}
  \dist_\edit(\pi, \pi')
  \leq \frac{1}{2} \sum_{i=1}^N \ind{x_i \neq y_i} \cdot \frac{2}{N} 
  = \frac{1}{N} \cdot \dist_\ham(x,y) = \dist_\reledit(\psi(\pi), \psi(\pi')) \,,
\end{align*}
which proves the claim.
\end{proof}

\begin{claim}
$\dist_\reledit(\psi(\pi), \psi(\pi')) \leq 2 \cdot \dist_\edit(\pi, \pi')$.
\end{claim}

\begin{proof}[Proof of claim] 
Let $(f,\cD)$ and $(g,\cE)$ be any labeled distributions with $\pi_{f,\cD} = \pi$ and $\pi_{g,\cE} =
\pi'$. We wish to show that
\[
  \dist_\reledit(\psi(\pi), \psi(\pi')) \lequestion 2 \cdot \dist_\TV( \cD_f, \cE_g ) \,.
\]
Using \cref{fact:edit-zero-interval}, followed by \cref{fact:edit-condense}, we may assume without
loss of generality that $(f,\cD)$ and $(g,\cE)$ satisfy the following conditions:
\begin{enumerate}
\item If $I$ is any interval such that $f$ and $g$ are both constant on $I$, and $f(x) \neq
g(x)$ on all $x \in I$, then either $\cD(I) = 0$ or $\cE(I) = 0$.
\item $f$ and $g$ have no alternation points in common.
\item $\cD$ and $\cE$ are supported on the set $C$, containing the alternation points of $f$ and
$g$.
\end{enumerate}
We will transform $\cD$ and $\cE$ into $\cD'$ and $\cE'$ that satisfy the
following properties:
\begin{enumerate}
\item $\pi_{f,\cD'} = \pi_{f,\cD} = \pi$ and $\pi_{g,\cE'} = \pi_{g,\cE} = \pi'$;
\item $\cD'$ and $\cE'$ are supported on $C$;
\item $\forall i \in \bZ$, $\cD'(i)$ and $\cE'(i)$ are integer multiples of $1/N$ (including 0);
\item $\forall i \in \bZ$, If $f(i) \neq g(i)$ then $\min\{ \cD'(i), \cE'(i) \} = 0$; and
\item $\dist_\TV( \cD'_f, \cE'_g ) \leq \dist_\TV( \cD_f, \cE_g )$.
\end{enumerate}
Let $a_1 < a_2 < \dotsm$ be the alternation points of $f$, and let $b_1 < b_2 < \dotsm$ be the
alternation points of $g$. Write $A = \{ a_1, a_2, \dotsc \}$ and $B = \{ b_1, b_2, \dotsc \}$; we
have $A \cap B = \emptyset$ and that $\cD$ and $\cE$ are supported on $C = A \cup B$.

We define $\cD'$ and $\cE'$ by performing the following transformation inside each interval
$(a_{i-1}, a_i]$ and $(b_{j-1}, b_j]$ in order of the endpoints $a_i$ and $b_j$; since $A \cap B =
\emptyset$, this is a well-defined ordering. We define the process for intervals $(a_{i-1}, a_i]$;
intervals $(b_{j-1}, b_j]$ are handled symmetrically. For each interval $(a_{i-1}, a_i]$. For each
iteration of the process, write $\cD$ and $\cE$ for the distributions before the iteration, and
$\cD'$ for the distribution after adjusting the mass in $(a_{i-1}, a_i]$. We will
guarantee that
\begin{equation}
\label{eq:edit-equiv-guarantee}
\begin{aligned}
  &\sum_{x \in (a_{i-1}, a_i]} \left(\ind{ f(x) = g(x) } | \cD'(x) - \cE(x) |
    + \ind{ f(x) \neq g(x) } ( \cD'(x) + \cE(x) ) \right) \\
  &\qquad\leq
  \sum_{x \in (a_{i-1}, a_i]} \left(\ind{ f(x) = g(x) } | \cD(x) - \cE(x) |
    + \ind{ f(x) \neq g(x) } ( \cD(x) + \cE(x) ) \right) \,.
\end{aligned}
\end{equation}
By \cref{prop:tv-distance-labeled}, this suffices to guarantee Property (5).
\begin{enumerate}
\item Let $C_i \define C \cap (a_{i-1}, a_i]$.
Let $C^+_i \define \{ x \in C_i : f(x) = g(x) \}$ and $C^-_i \define \{ x \in C_i : f(x) \neq g(x)
\}$.
\item If $C^+_i = \emptyset$, we are guaranteed that $(a_{i-1}, a_i] \cap B = \emptyset$ and either
$\cD(a_{i-1,a_i}] = 0$, or $\cE(a_{i-1}, a_i] = 0$ since $(a_{i-1}, a_i]$ is an interval where $f,g$
are constant and unequal. In this case, set $\cD'(x) = \cD(x)$ for all $x \in (a_{i-1}, a_i]$, so
$\cD'(a_i) = \cD(a_{i-1}, a_i]$ which is an integer multiple of $1/N$. This guarantees Property (3)
inside $(a_{i-1}, a_i]$, and the guarantee \eqref{eq:edit-equiv-guarantee} trivially holds.

\item Otherwise write $C^+_i = \{ c_1, \dotsc, c_m \}$ and consider the sequence $u_1, \dotsc, u_m$
where $u_j \define \cE(c_j)$. Note that for $c_j < a_i$, $\cE(c_j)$ has been defined earlier in this
process, since $c_j$ is the endpoint of an interval, and therefore $\cE(c_j)$ is an integer multiple
of $1/N$. Then define $\cD'$ on the points $c_1, \dotsc, c_m$ by distributing the mass $\cD(a_{i-1},
a_i]$ according to \cref{fact:edit-match-mass}. That fact guarantees $\cD'(a_{i-1}, a_i] =
\cD(a_{i-1}, a_i]$ and, by inspection of the proof, that each $\cD'(c_j)$ is an integer multiple of
$1/N$; we then have guarantee \eqref{eq:edit-equiv-guarantee}, because:
\begin{align*}
  &\sum_{x \in (a_{i-1}, a_i]} \left(\ind{ f(x) = g(x) } | \cD'(x) - \cE(x) |
  + \ind{ f(x) \neq g(x) } (\cD'(x) + \cE(x)) \right) \\
  &\qquad= \sum_{x \in C^+_i} |\cD'(x) - \cE(x)| + \sum_{x \in C^-_i} (\cD'(x) + \cE(x)) \\
  &\qquad= \sum_{x \in C^+_i} |\cD'(x) - \cE(x)| + \sum_{x \in C^-_i} \cE(x) 
  = \left|\cD(a_{i-1},a_i] - \sum_{x \in C^+_i} \cE(x) \right|  + \sum_{x \in C^-_i} \cE(x) \\
  &\qquad\leq \sum_{x \in C^+_i} | \cD(x) - \cE(x) | + \sum_{x \in C^-_i} (\cD(x) + \cE(x)) \\
  &\qquad= \sum_{x \in (a_{i-1}, a_i]} \left(\ind{ f(x) = g(x) } | \cD(x) - \cE(x) |
  + \ind{ f(x) \neq g(x) } (\cD(x) + \cE(x)) \right) \,.
\end{align*}
\end{enumerate}
Having obtained the desired labeled distributions $(f,\cD')$ and $(g,\cE')$, we conclude the proof
as follows. Write $C = \{ c_1, c_2, \dotsm \}$ such that $c_1 < c_2 < \dotsc$. Since $\cD'$ and
$\cE'$ have densities that are integer multiples of $1/N$, there is some $m$ such that $\cD'$ and
$\cE'$ are supported on $c_1, \dotsc, c_m$. For each $t \in [m]$, define $Z_t \define [ (t-1)N + 1 ,
tN ]$ so that $|Z_t| = N$. Let $p_t \define N \cdot \cD'(c_t)$ and $q_t \define N \cdot \cE'(c_t)$,
which are non-negative integers. Then we define the strings $x,y \in \{0,1,\bot\}^{m \cdot N}$ as
follows.  For each $t \in [m]$, define $x_i = f(c_t)$ for the first $p_t$ values of $i \in Z_t$, and
define $y_i = g(c_t)$ for the first $q_t$ values of $i \in Z_t$, and let the remaining characters in
$Z_t$ be $\bot$.

It is easily verified that $x \in \mathsf{ext}(\psi(\pi))$ and $y \in \mathsf{ext}(\psi(\pi'))$, so
\begin{align*}
  &\dist_\reledit(\psi(\pi), \psi(\pi')) = \frac{1}{N} \dist_{\stringedit}(\psi(\pi), \psi(\pi')) \\
  &\qquad\leq \frac{1}{N} \cdot \dist_\ham(x,y)\\
  &\qquad= \frac{1}{N} \sum_{t=1}^m \left(\ind{f(c_t) = g(c_t)} \cdot \left| p_t - q_t \right|
                                 + \ind{f(c_t) \neq g(c_t)} \max\{p_t , q_t \} \right) \\
  &\qquad= \sum_{t=1}^m \left(\ind{ f(c_t) = g(c_t) } \cdot | \cD'(c_t) - \cE'(c_t) |
                     + \ind{ f(c_t) \neq g(c_t) } \cdot \max\{ \cD'(c_t) + \cE'(c_t) \} \right) \\
  &\qquad= \sum_x \left(\ind{ f(x) = g(x) } \cdot | \cD'(x) - \cE'(x) |
                     + \ind{ f(x) \neq g(x) } \cdot \max\{ \cD'(x) + \cE'(x) \} \right) \\
  &\qquad= \sum_x \left(\ind{ f(x) = g(x) } \cdot | \cD'(x) - \cE'(x) |
                     + \ind{ f(x) \neq g(x) } \cdot ( \cD'(x) + \cE'(x) ) \right) \\
  &\qquad= 2 \cdot \dist_\TV(\cD'_f, \cE'_g) \,,
\end{align*}
which proves the claim.
\end{proof}

These two claims complete the proof.
\end{proof}

\subsection{Edit Distance for the Uniform Distribution}
\label{section:edit-distance-uniform-distribution}

\lemmaedittotvuniform*

\begin{proof}
    Recall that $\pi(i) = 1/k$ for each $i \in [k]$.
    Let $z \in \bR^{k}$ be the vector such that $\pi'(i) = \pi(i) + z_i$ for each $i \in [k]$.
    Note that, since $\pi$ and $\pi'$ are probability distributions, we have $\sum_i z_i = 0$ and
    \[
        \dist_\TV(\pi, \pi')
        = \frac{1}{2} \left( \sum_{i : z_i > 0} z_i + \sum_{i : z_i < 0} \abs*{z_i} \right)
        = \sum_{i : z_i > 0} z_i \,.
    \]
    Let $S \define \{i \in [k] : z_i > 0\}$, which we may assume is nonempty since, otherwise,
    the claim holds trivially. Now, our goal is to show that
    $\dist_\edit(\pi, \pi') \ge c \cdot \sum_{i \in S} z_i$.

    Let $\cD_f$ and $(g, \cE)$ be two $1$-proper labeled distributions such that
    $\pi_{f,\cD} = \pi$ and $\pi_{g,\cE} = \pi'$.
    Our goal is to show that $\dist_\TV(\cD_f, \cE_g) \ge c \cdot \sum_{i \in S} z_i$.
    We may assume that $g$ (also) alternates exactly $k-1$ times, because if it had fewer
    alternations, we could introduce extra alternations starting at a coordinate large enough that
    all but arbitrarily small mass of $\cE_g$ is affected.

    Let $a_1 < a_2 < \dotsm < a_{k-1}$ be the alternation sequence of $f$, and let
    $b_1 < b_2 < \dotsm < b_{k-1}$ be the alternation sequence of $g$.
    For convenience of notation, write $a_0 \define -\infty$, $a_k = \infty$, $b_0 \define -\infty$
    and $b_k \define \infty$, so that for each $t \in [k]$ we have
    $1/k = \pi(t) = \cD(a_{t-1}, a_t]$ and $\pi'(t) = \cE(b_{t-1}, b_t]$.

    Fix any $t \in S$. By \cref{prop:tv-distance-labeled}, it suffices to show the following:
    \[
        \sum_{i \in [b_{t-1}, b_t)} \ind{f(i) \neq g(i)}( \cD(i) + \cE(i) )
            +  \ind{f(i) = g(i)} | \cD(i) - \cE(i) |
        \gequestion 2c \cdot z_t \,.
    \]
    For convenience, let $F_t$ denote the left-hand side expression in this proposed inequality.

    Let $h \define \floor{\frac{\cE(b_{t-1}, b_t]}{1/k}}$ and
    $\epsilon \define \cE(b_{t-1}, b_t] - h/k$, so that $0 \le \epsilon < 1/k$ and
    $\cE(b_{t-1}, b_t] = h/k + \epsilon$. Note that $h \ge 1$ because $t \in S$,
    meaning that $\cE(b_{t-1}, b_t] > 1/k$. Moreover, recalling that
    $\cE(b_{t-1}, b_t] = 1/k + z_t$, we conclude that $z_t = (h-1)/k + \epsilon$.
    We now consider a number of cases.

    \textbf{Case 1.} Suppose
    $\cD(b_{t-1}, b_t] \le \frac{1}{k}\left(1 + \frac{h-1}{2}\right) + \frac{\epsilon}{2}$.
    Then we obtain
    \begin{align*}
        F_t
        &\ge \sum_{i \in (b_{t-1}, b_t]} \abs*{\cD(i) - \cE(i)}
        \ge \abs*{\cE(b_{t-1}, b_t] - \cD(b_{t-1}, b_t]}
        \ge \frac{h}{k} + \epsilon - \frac{1}{k} - \frac{(h-1)/2}{k} - \frac{\epsilon}{2} \\
        &= \frac{(h-1)/2}{k} + \frac{\epsilon}{2}
        = \frac{z_t}{2} \,,
    \end{align*}
    so we are done with this case.

    \textbf{Case 2.} Suppose that
    $\cD(b_{t-1}, b_t] > \frac{1}{k}\left(1 + \frac{h-1}{2}\right) + \frac{\epsilon}{2}$. We
    consider further sub-cases based on the value of $h$.
    Throughout the remaining analysis, we use the fact that $g$ is constant on $(b_{t-1}, b_t]$.

    \textbf{Case 2A.} $h = 1$. Let $t'$ be the smallest index such that $a_{t'} \in (b_{t-1}, b_t)$,
    which must exist because $\cD(b_{t-1}, b_t] > 1/k$, so $f$ must alternate in this interval.
    If $f(a_{t'}) = g(a_{t'})$, then $f$ and $g$ disagree from $a_{t'}+1$ up to just before the
    next alternation point $a_{t'+1}$ or $b_t$, whichever comes first. Moreover, since the
    minimality of $t'$ implies that $a_{t'-1} \le b_{t-1}$ and therefore
    $\cD(b_{t-1}, a_{t'}] \le 1/k$, we have
    $\cD(a_{t'}, b_t] = \cD(b_{t-1}, b_t] - \cD(b_{t-1}, a_{t'}]
    > \frac{(h-1)/2}{k} + \frac{\epsilon}{2} = \epsilon/2$. Therefore, recalling that
    $\epsilon < 1/k$, we obtain
    \[
        F_t \ge \sum_{i \in (a_{t'}, \min\{a_{t'+1}, b_t\}]} (\cD(i) + \cE(i))
        \ge \cD(a_{t'}, \min\{a_{t'+1}, b_t\}]
        \ge \min\{1/k, \epsilon/2\}
        = \epsilon/2
        = z_t/2 \,,
    \]
    as desired.

    Otherwise, suppose $f(a_{t'}) \ne g(a_{t'})$. The logic is similar, but now we argue that there
    must be substantial $\cD$-mass that is both in $(b_{t-1}, b_t]$ and either at most $a_{t'}$ or
    in $(a_{t'+1}, a_{t'+2}]$, \ie the regions where $f$ and $g$ disagree. Indeed, suppose
    $\cD(b_{t-1}, a_{t'}] < \epsilon/4$.
    Then $\cD(b_{t-1}, a_{t'+1}] = \cD(b_{t-1}, a_{t'}] + \cD(a_{t'}, a_{t'+1}]
    < \epsilon/4 + 1/k$, while
    $\cD(b_{t-1}, b_t] > 1/k + \epsilon/2$, implying that
    $\cD(a_{t'+1}, \min\{a_{t'+2}, b_t\}] \ge \min\{1/k, \epsilon/4\} = \epsilon/4$.
    Therefore $f$ and $g$ disagree in at least $\epsilon/4$ $\cD$-mass inside $(b_{t-1}, b_t]$,
    so $F_t \ge \epsilon/4 = z_t/4$, and we are done with this case.

    The cases with $h \ge 2$ follow similar logic, but now, the $(h-1)/k$ term in $z_t$ dominates
    the $\epsilon$ term, so we must adjust the argument accordingly.

    \textbf{Case 2B.} $2 \le h \le 13$. As above, let $t'$ be the smallest index such that
    $a_{t'} \in (b_{t-1}, b_t)$.
    If $f(a_{t'}) = g(a_{t'})$, then we are done as follows. Observe that, by the minimality of
    $t'$, we have $a_{t'-1} \le b_{t-1}$ and hence $\cD(b_{t-1}, a_{t'}] \le 1/k$. It follows that
    $\cD(a_{t'}, b_t] = \cD(b_{t-1}, b_t] - \cD(b_{t-1}, a_{t'}] > \frac{(h-1)/2}{k} \ge 1/2k$.
    Therefore, we obtain
    \[
        F_t \ge \sum_{i \in (a_{t'}, \min\{a_{t'+1}, b_t\}]} \left( \cD(i) + \cE(i) \right)
        \ge \cD(a_{t'}, \min\{a_{t'+1}, b_t\}]
        \ge \min\{1/k, 1/2k\}
        = 1/2k \,.
    \]
    Then, since $\epsilon < 1/k$ and $1 \le h-1 \le 12$, we get
    $F_t > \frac{1}{4k} + \frac{\epsilon}{4} \ge \frac{(h-1)/12}{4k} + \frac{\epsilon}{4}
    \ge \frac{z_t}{48}$, as needed.

    Otherwise, suppose $f(a_{t'}) \ne g(a_{t'})$. We proceed similarly to the previous cases by
    arguing that there
    must be substantial $\cD$-mass that is both in $(b_{t-1}, b_t]$ and either at most $a_{t'}$ or
    in $(a_{t'+1}, a_{t'+2}]$, \ie the regions where $f$ and $g$ disagree. Indeed, suppose
    $\cD(b_{t-1}, a_{t'}] < 1/4k$.
    Then $\cD(b_{t-1}, a_{t'+1}] = \cD(b_{t-1}, a_{t'}] + \cD(a_{t'}, a_{t'+1}] < 1/4k + 1/k$, while
    $\cD(b_{t-1}, b_t] > 1/k + 1/2k$ by assumption, implying that
    $\cD(a_{t'+1}, \min\{a_{t'+2}, b_t\}] \ge \min\{1/k, 1/4k\} = 1/4k$.
    Then, again using $\epsilon < 1/k$ and $1 \le h-1 \le 12$, we get
    $F_t \ge \frac{1}{4k} > \frac{1}{8k} + \frac{\epsilon}{8}
    \ge \frac{(h-1)/12}{8k} + \frac{\epsilon}{8} \ge \frac{z_t}{96}$, as needed.\footnote{We did
    not try to optimize the constant $c$.}

    \textbf{Case 2C.} $h \ge 14$. Let $\ell$ be the number of alternation points of $f$ in
    $(b_{t-1}, b_t]$; say they are $a_{t'}, a_{t'+1}, \dotsc, a_{t'+\ell-1}$.
    We claim that $\ell \ge (h-1)/2$. Indeed, suppose $\ell < (h-1)/2$. Then the total
    $\cD$-mass in $(b_{t-1}, b_t]$ is at most
    \[
        \cD(\max\{a_{t'-1}, b_{t-1}\}, a_{t'}]
            + \sum_{j=1}^{\ell-1} \cD(a_{t'+j-1}, a_{t'+j}]
            + \cD(a_{t'+\ell-1}, \min\{a_{t'+\ell}, b_t\}]
        \le \frac{1}{k} (\ell+1)
        < \frac{1}{k} \left( 1 + \frac{h-1}{2} \right) \,,
    \]
    contradicting our assumption about $\cD$. Therefore $\ell \ge (h-1)/2$.

    Now, consider the $\ell-1$ ranges of the form $(a_{t'+j-1}, a_{t'+j}]$ consisting of pairs of
    consecutive $f$ alternations inside $(b_{t-1}, b_t]$.
    Since $f$ is constant inside each of them, it disagrees with $g$ in at least $\floor{(\ell-1)/2} \ge \floor{(h-3)/4} \ge (h-7)/4 \ge h/8$ of them, where the last inequality
    holds because $h \ge 14$. Therefore $f$ and $g$ disagree on sufficient $\cD$-mass: recalling
    that $\epsilon < 1/k$, we have
    \[
        F_t \ge \frac{h}{8} \cdot \frac{1}{k}
        > \frac{1}{8} \left( \frac{h-1}{k} + \epsilon \right)
        = \frac{z_t}{8} \,,
    \]
    which concludes the proof.
\end{proof}

\subsection{Edit Distance for Labeled Distribution Support Size}
\label{section:edit-distance-support-size}
\propedittvdistancestosupportsizek*
\begin{proof}
Let $(f', \cD')$ and $(g, \cE) \in \Xi$ be such that $\pi_{f',\cD'} = \pi_{f,\cD}$ and
$\dist_\edit((f,\cD), \Xi) = \dist_\TV(\cD'_{f'}, \cE_g)$.

\textbf{Step 1.} We will show that there exists $(g', \cE') \in \Xi$ such that $\dist_\TV(\cD'_{f'},
\cE'_{g'}) \leq \dist_\TV(\cD'_{f'}, \cE_g)$, and the alternation sequence of $g'$ is a subset of
the alternation sequence of $f'$. By \cref{fact:edit-zero-interval}, we may assume that any interval
$I$ where $f'$ and $g$ are constant and unequal to each other has either $\cD(I) = 0$ or $\cE(I) =
0$.

Let $a'_1 < a'_2 < \dotsm$ and $b_1 < b_2 < \dotsm$ be the alternation sequences for $f'$ and $g$
respectively. Suppose there is $(a'_{i-1}, a'_i]$ such that there is $b_j \in (a'_{i-1}, a'_i)$.
Define $g'$ such that $g'(x) = f'(x)$ for all $x \in (a'_{i-1}, a'_i]$ and $g'(x) = g(x)$ otherwise.
By \cref{fact:edit-zero-mass}, $\dist_\TV(\cD'_{f'}, \cE_{g'}) = \dist_\TV(\cD'_{f'}, \cE_g)$.

We claim that this does not increase the number of alternation points, so $g'$ has at most the
number of alternations as $g$.  Let $z$ be the constant such that $f'(x) = g'(x) = z$ for all $x \in
(a'_i, a'_{i-1}]$.  Since there is an alternation point $b_j \in (a'_i, a'_{i-1})$, there is an
interval $I \subseteq (a'_i, a'_{i-1}]$ such that $g$ and $g'$ have constant value $z$ on $I$. When
we replace the values of $g$ with $z$ in $(a'_i, a'_{i-1}]$ to obtain $g'$, we cannot increase the
number of alternation points, since we simply expand the interval $I$.

Performing this operation in each interval $(a'_{i-1}, a'_i]$ where there exists an alternation
point $b_j \in (a'_{i-1}, a'_i)$, and simply setting $\cE' = \cE$ (for clarity of notation in step
2), we obtain $(g',\cE')$ with the desired property.

\textbf{Step 2.} We now have $(f', \cD')$ and $(g', \cE')$ where the alternation sequence $b'_1 <
b'_2 < \dotsm$ of $g'$ is a subset of the alternation sequence of $f'$. Let $a_1 < a_2 < \dotsm$ be
the alternation sequence of the original function $f$. We will define $(h, \cH)$ as follows. For
each interval $B'_j \define (b'_{j-1}, b'_j]$ in the alternation sequence of $g'$, let $a'_i = b'_{j-1} < a'_{i+1} < \dotsm < a'_{i+t} = b'_j$ be the alternation points of $f'$ contained in
$[b'_{j-1}, b'_j]$, and let $A_j = (a_i, a_{i+t}]$. Let $z_j$ be the value such that $g'(x) = z_j$
for all $x \in B'_j$. Let $T'_j \subseteq B'_j$ be the points $x$ such that $f'(x) = z_j$. We define
$h(x) = z_j$ for all $x \in A_j$. Note that the intervals $A_j$ partition the domain, so this fully
defines $h$.

Fix an interval $B'_j$.  If $f(x) \neq z_j$ for all $x \in A_j$, we set $\cH(a_{i+t}) = \cE'(B'_j)$
and $\cH(x) = 0$ for the remaining $x \in (a_i, a_{i+t})$. Then
\begin{align*}
  &\sum_{x \in A_j} \ind{f(x) \neq h(x)} (\cD(x) + \cH(x))
                + \ind{f(x) = h(x)} |\cD(x) + \cH(x)| \\
  &\qquad= \sum_{x \in A_j} (\cD(x) + \cH(x)) \\
  &\qquad= \cD(a_{i+t}) + \cH(a_{i+t}) + \sum_{x \in A_j \setminus \{a_{j+t}\}} \cD(x) \\
  &\qquad= \cE'(B'_j) + \cD(A_j) = \cE'(B'_j) + \cD'(B'_j) \\
  &\qquad= \sum_{x \in B'_j} \ind{f'(x) \neq g'(x)} (\cD'(x) + \cE'(x)) 
                        + \ind{f'(x) = g'(x)} |\cD'(x) - \cE'(x)| \,.
\end{align*}
\noindent
Otherwise, let $T_j \subseteq A_j$ be the coordinates such that $f(x) = z_j$ for $x \in T_j$. First
observe that
\begin{align*}
  &\sum_{x \in B'_j} \ind{ f'(x) \neq g'(x) } (\cD'(x) + \cE'(x))
                + \ind{ f'(x) =    g'(x) } | \cD'(x) - \cE'(x) | \\
  &\qquad= \cD'(B'_j \setminus T'_j) + \cE'(B'_j \setminus T'_j) + |\cE'(T'_j) - \cD'(T'_j)| \\
  &\qquad\geq \cD'(B'_j \setminus T'_j) + |\cE'(B'_j) - \cD'(T'_j)|  \\
  &\qquad= \cD(A_j \setminus T_j) + |\cE'(B'_j) - \cD'(T'_j)| \,.
\end{align*}
We assign values for $\cH$ to the coordinates in $T_j$ such that $\cH(T_j) = \cE'(B'_j)$, $\cH(x) =
0$ for all $x \in A_j \setminus T_j$, and
\[
  \sum_{x \in T_j} \left| \cH(x) - \cD(x) \right|
    = |\cH(T_j) - \cD(T_j)|
    = |\cH(A_j) - \cD(T_j)| 
    = |\cE'(B'_j) - \cD'(T'_j)| \,.
\]
which is possible due to \cref{fact:edit-match-mass}. Then
\begin{align*}
  &\sum_{x \in T_j} | \cH(x) - \cD(x) |
    + \sum_{x \in A_j \setminus T_j} ( \cH(x) + \cD(x) ) \\
  &\qquad= | \cE'(B'_j) - \cD'(T'_j) | + \cD( A_j \setminus T_j ) \\
  &\qquad\leq \sum_{x \in B'_j} \ind{ f'(x) \neq g'(x) } (\cD'(x) + \cE'(x))
                + \ind{ f'(x) =    g'(x) } | \cD'(x) - \cE'(x) | \,.
\end{align*}
Applying the same argument to each interval $B'_j = (b'_{j-1}, b'_j]$, we obtain $(h,\cH)$ with
the required properties, due to \cref{prop:tv-distance-labeled}:
\begin{align*}
&\sum_x \left( \ind{ f(x) = h(x) }\cdot | \cD(x) - \cH(x) |
    + \ind{ f(x) \neq h(x) }( \cD(x) + \cH(x) ) \right) \\
&= \sum_j \left( \sum_{x \in T_j}  | \cH(x) - \cD(x) |
    + \sum_{x \in A_j \setminus T_j} ( \cH(x) + \cD(x) ) \right)\\
&\leq \sum_j \left( \sum_{x \in B'_j} \ind{ f'(x) \neq g'(x) } (\cD'(x) + \cE'(x))
                + \ind{ f'(x) =    g'(x) } | \cD'(x) - \cE'(x) | \right) \\
&= \dist_\TV( \cD'_{f'}, \cE'_{g'} ) \leq \dist_\edit((f,\cD), \Xi) \,. \qedhere
\end{align*}
\end{proof}

\subsection{Edit Distance for Distribution Support Size}
\label{section:edit-distance-distribution-support-size}
\lemmadistfreereductionedit*

\begin{proof}
By a limit argument, and using the triangle inequality on the edit and TV distances, it suffices to
prove the claim for the case where $\pi$, and thus $\pi'$, have rational densities only. This will
allow us to minimize technical details by using the standard edit distance on \emph{strings}, which
is simpler to analyze, as follows. We may fix integer $N$ such that all densities of $\pi'$ are
integer multiples of $1/N$ (for example, we may take $N = \prod_i b_i$ where we write $\pi'(i) =
a_i/b_i$).

    Now, \cref{prop:dist-to-n-strings-distributions} applies: let $\Psi$ be the set of $2k$-block
    strings in $\zo^N$, and let $x' \define \psi(\pi')$;
    then $\dist_\reledit(x', \Psi) \le 2 \cdot \dist_\edit(\pi', \Pi_{2k})$.
    Let $\delta \define \dist_\reledit(x', \Psi)$. By definition of relative edit distance, there
    exists a sequence of $N\delta$ edit operations $O_1, \dotsc, O_{N\delta}$ such that
    $(O_{N\delta} \circ O_{N\delta-1} \circ \dotsm \circ O_1)(x') \in \Psi$, where each $O_j$
    is an \emph{insertion}, \emph{deletion}, or \emph{substitution} of a single character.

    We claim that there exists a sequence $O'_1, \dotsc, O'_{\ell}$ of operations, with
    $\ell \le N\delta$, such that
    \begin{enumerate}
        \item Each of $O'_1, \dotsc, O'_{\ell}$ is a deletion;
        \item $(O'_{\ell} \circ O'_{\ell-1} \circ \dotsm \circ O'_1)(x')$ is a $2k$-block string.
    \end{enumerate}
    To see why this is true, choose $s \in \Psi$ such that $\dist_\reledit(x', s) = \delta$ and,
    using \cref{fact:string-extension-hamming}, fix $u \in \ext(x'), v \in \ext(s)$ of equal
    length $M$ such that $\dist_\ham(u,v) = \dist_\stringedit(x',s) = N\delta$. We construct string
    $v' \in \{0,1,\bot\}^M$ as follows: for each $i \in [M]$,
    \begin{enumerate}
        \item If $u_i = \bot$, set $v'_i = \bot$.
        \item If $u_i \ne \bot$ and $u_i = v_i$, set $v'_i = v_i$.
        \item If $u_i \ne \bot$ and $u_i \ne v_i$, set $v'_i = \bot$.
    \end{enumerate}
    We make three observations. First, there exists a $2k$-block string $s' \in \zo^*$ such that
    $v' \in \ext(s')$; this is true because for each $i \in [M]$,
    either $v'_i = v_i$ or $v'_i = \bot$. Second, for every $i \in [M]$, we have the implication
    $u_i \ne v'_i \implies v'_i = \bot$; this holds by construction. Third,
    $\dist_\ham(u, v') \le \dist_\ham(u, v) = N\delta$, which is also clear by construction.
    Let $\ell \define \dist_\ham(u, v')$.

    We obtain our deletion operations as follows.
    Let $u^* \in \{0,1,\bot\}^M$ be given by $u^*_i = u_i$ when $u_i = v'_i$, and
    $u^*_i = \bot$ otherwise. Then $\dist_\ham(u^*, v') = 0$ by our second observation.
    Let $x^* \in \{0,1\}^*$ be obtained from $x'$ by deleting each of the $\ell$ characters
    corresponding to the case $u_i \ne v'_i$ above (\ie if $u_i \ne v'_i$, then this occurs at
    the $j$-th non-$\bot$ character of $u_i$, so delete the $j$-th character of $x'$).
    Then $\dist_\stringedit(x', x^*) = \ell$. Moreover, $u^* \in \ext(x^*)$, because we deleted
    characters from $x'$ to obtain $x^*$, and set to $\bot$ characters from $u$ to obtain $u^*$, in
    correspondence. Hence $\dist_\stringedit(x^*, s') \le \dist_\ham(u^*, v') = 0$,
    so that $x^* = s'$ is a $2k$-block string. Therefore $\ell = \dist_\stringedit(x', x^*)$
    deletion operations suffice to turn $x'$ into a $2k$-block string.
    This yields the desired $O'_1, \dotsc, O'_\ell$.

    We now use these operations to transform $\pi$ into a vector $\nu$ supported on at most
    $k$ elements, as follows. We set $\nu$ to zero everywhere outside the support of $\pi$.
    For each $i \in \supp(\pi)$,
    \begin{enumerate}
        \item Let $x'_p \dotsc x'_{p+\frac{1}{2}N\pi(i)-1} = 1^{N\pi'(2i-1)}$ be the block of 1s
            corresponding to the entry $\pi'(2i-1)$ in $x' = \psi(\pi')$. Similarly, let
            $x'_q \dotsc x'_{q+\frac{1}{2}N\pi(i)-1} = 0^{N\pi'(2i)}$ be the block of 0s
            corresponding to the entry $\pi'(2i)$ in $x'$.
        \item If all the characters in at least one of these two blocks were deleted by
            operations in $O'_1, \dotsc, O'_\ell$, set $\nu(i) = 0$. Otherwise,
            set $\nu(i) = \pi(i)$.
    \end{enumerate}
    First, note that
    \begin{align*}
        \|\pi - \nu\|_1
        &= \sum_{i \in \supp(\pi)} \pi(i) \cdot \ind{\text{all 0s or all 1s corresponding to $\pi(i)$ deleted}} \\
        &\le 2\sum_{i' \in \supp(\pi')} \pi'(i') \cdot \ind{\text{entire block corresponding to
$\pi'(i')$ deleted}}  \,.
    \end{align*}
If a block $b^{N\pi'(i')}$ was deleted, then there were $N\pi'(i')$ deletions required to delete the
characters in this block. Each deleted block corresponds to a disjoint set of deletion operations,
and there is a total of $\ell$ deletion operations.  Therefore we must have
\begin{align*}
  \sum_{i' \in \supp(\pi')} (N\pi'(i')) \cdot \ind{\text{entire block corresponding to $\pi'(i')$ deleted}} 
  \leq \ell \,.
\end{align*}
As a consequence,
\[
        \|\pi - \nu\|_1 \le 2 \ell / N \le 2\delta \,.
\]
Now, let $\nu^*$ be a probability distribution obtained from $\nu$ by adding the required
probability mass $1 - \|\nu\|_1$ to an arbitrary element with nonzero mass. Observe that $1 - \| \nu
\|_1 = \| \pi \|_1 - \| \nu \|_1 \leq 2\delta$, so $\|\pi - \nu^*\|_1 \leq \|\pi - \nu \|_1 + 2\delta
\leq 4\delta$. It follows that $\dist_\TV(\pi, \nu^*) \le 2\delta$. It remains to
show that $\nu^*$ is supported on at most $k$ elements.

    Let $K \define \abs{\supp(\pi)}$; then $\abs{\supp(\pi')} = 2K$ and $x'$ contains
    exactly $2K$ blocks. Since $(O'_\ell \circ \dotsc \circ O'_1)(x')$ contains at most $2k$
    blocks, it follows that at least $2(K-k)$ blocks are entirely deleted by the
    operations. Therefore at least $K-k$ distinct indices $i \in \supp(\pi)$ are such that
    $\nu(i) = 0$, by the construction above. Therefore
    $\abs{\supp(\nu^*)} = \abs{\supp(\nu)} \le \abs{\supp(\pi)} - (K-k) = k$, as desired.

    Putting everything together, we obtain
    $\dist_\TV(\pi, \Pi_k) \le \dist_\TV(\pi, \nu^*) \le 2\delta \le 4\cdot \dist_\edit(\pi', \Pi_{2k})$.
\end{proof}

\subsection{String Edit Distance for Support Size}
\label{section:string-edit-distance-support-size}

\propdisttonstringsdistributions*

\begin{proof}
    We first show that $\dist_\edit(\pi, \Pi) \le \dist_\reledit(x, \Psi)$.
    Pick some $y \in \Psi$ satisfying
    $\dist_\reledit(x,y) = \dist_\reledit(x, \Psi)$. Then $\psi^{-1}(y) \in \Pi$ since
    $\psi^{-1}(y)$ is supported on at most $n$ elements (by \cref{obs:n-block-support}),
    and therefore, using \cref{lemma:relative-edit-distance},
    $\dist_\edit(\pi, \Pi) \le \dist_\edit(\pi, \psi^{-1}(y)) \le \dist_\reledit(x, y)
    = \dist_\reledit(x, \Psi)$.

    We now show that $\dist_\edit(\pi, \Pi) \ge \frac{1}{2} \dist_\reledit(x, \Psi)$.
    If we can show that $\dist_\edit(\pi, \Pi) = \dist_\edit(\pi, \pi^*)$ for some
    $\pi^* \in \Pi$ whose densities are integer multiples of $1/N$, we will be done:
    such $\pi^*$ will satisfy $\psi(\pi^*) \in \Psi$, so the conclusion will follow from
    \cref{lemma:relative-edit-distance}.

    Let $\cD$ be the probability distribution over $\bZ$ with same densities as $\pi$
    (\ie we simply extend the domain from $\bN$ to $\bZ$), and let
    $f : \bZ \to \zo$ be a function such that $\pi = \pi_{f,\cD}$, which must exist.

    Let $\Xi$ be the property of labeled distributions $(h, \cF)$ such that $\pi_{h,\cF}$ has
    support size at most $n$; equivalently, such that $\pi_{h,\cF} \in \Pi$.
    By \cref{lemma:edit-tv-distances-to-support-size-k},
    $\dist_\edit((f,\cD), \Xi) \ge \dist_\TV((f,\cD), \Xi)$. We claim that there exists
    $(g,\cE) \in \Xi$ such that $\dist_\TV((f,\cD), \Xi) = \dist_\TV((f,\cD), (g,\cE))$ and,
    moreover, all densities of $\cE$ are integer multiples of $1/N$.

    Let $(g,\cE) \in \Xi$ be such that $\dist_\TV((f,\cD), (g,\cE)) = \dist_\TV((f,\cD), \Xi)$.
    We first claim that $\supp(\cD) \cap \supp(\cE) \ne \emptyset$ and, moreover,
    there exists $i^* \in \supp(\cD) \cap \supp(\cE)$ such that $f(i^*) = g(i^*)$. Indeed,
    suppose this is not the case. Then $\dist_\TV((f,\cD), (g,\cE)) = 1$. On the other hand,
    the labeled distribution $(f,\cD')$ where $\cD'$ is the singleton distribution supported on
    any $i \in \supp(\cD)$ satisfies $\dist_\TV((f,\cD), (f,\cD')) < 1$, contradicting our
    choice of $(g,\cE)$ since $\pi_{f,\cD'}$ is supported on a single element and thus
    $(f,\cD') \in \Xi$. Fix any such $i^*$.

    We first construct a distribution $\cE'$ from $\cE$ satisfying three conditions:
    \begin{enumerate}
        \item $(g,\cE') \in \Xi$.
        \item For every $i \in \supp(\cE')$, $g(i) = f(i)$; and
        \item $\dist_\TV((f,\cD), (g,\cE')) \le \dist_\TV((f,\cD), (g,\cE))$;
    \end{enumerate}
    We construct $\cE'$ as follows: for each $i \in \bN$,
    \begin{enumerate}
        \item If $i \ne i^*$ and $g(i) = f(i)$, set $\cE'(i) = \cE(i)$.
        \item If $g(i) \ne f(i)$, set $\cE'(i) = 0$.
        \item Set $\cE'(i^*) = 1 - \cE'(\supp(\cE) \setminus \{i^*\})$.
    \end{enumerate}
    By construction, $\cE'$ is a probability distribution. The first condition is easy to
    verify: note that $\supp(\cE') \subseteq \supp(\cE)$, and thus we have
    $|\supp(\pi_{g,\cE'})| \le |\supp(\pi_{g,\cE})| \le n$ and hence $(g,\cE') \in \Xi$.
    The second condition holds by construction of $\cE'$ and choice of $i^*$.
    Finally, we verify the third condition. By \cref{prop:tv-distance-labeled}, we have
    \begin{align*}
        &\dist_\TV((f,\cD), (g,\cE')) \\
        &\quad = \frac{1}{2} \sum_{i \in \bN} \ind{f(i) \ne g(i)}( \cD(i) + \cE'(i) )
            + \ind{f(i) = g(i)} \abs*{ \cD(i) - \cE'(i) } \\
        &\quad = \frac{1}{2} \sum_{i \ne i^*}\Bigg[ \ind{f(i) \ne g(i)} \cD(i)
            + \ind{f(i) = g(i)} \abs*{ \cD(i) - \cE(i) } \Bigg] \\
            &\qquad + \frac{1}{2} \Bigg[
                \abs*{\cD(i^*) - \left( 1 - \sum_{i \ne i^*} \ind{f(i) = g(i)} \cE(i) \right)}
                \Bigg] \\
        &\quad = \frac{1}{2} \sum_{i \ne i^*}\Bigg[ \ind{f(i) \ne g(i)} \cD(i)
            + \ind{f(i) = g(i)} \abs*{ \cD(i) - \cE(i) } \Bigg] \\
            &\qquad + \frac{1}{2} \Bigg[
                \abs*{\cD(i^*) - \left( \cE(i^*) + \sum_{i \in \bN} \ind{f(i) \ne g(i)} \cE(i) \right)}
                \Bigg] \\
        &\quad \le \frac{1}{2} \sum_{i \in \bN}\Bigg[ \ind{f(i) \ne g(i)} (\cD(i) + \cE(i))
            + \ind{f(i) = g(i)} \abs*{ \cD(i) - \cE(i) } \Bigg] \\
        &\quad = \dist_\TV((f,\cD), (g,\cE)) \,,
    \end{align*}
    the inequality being the triangle inequality.

    Now, we construct $\cE''$ from $\cE'$ to satisfy three conditions:
    \begin{enumerate}
        \item $(g,\cE'') \in \Xi$.
        \item Every density of $\cE''$ is an integer multiple of $1/N$; and
        \item $\dist_\TV((f,\cD), (g,\cE'')) \le \dist_\TV((f,\cD), (g,\cE'))$;
    \end{enumerate}
    We construct $\cE''$ as follows: for each $i \in \supp(\cE') \setminus \{i^*\}$,
    set $\cE''(i) = \cD(i)$; and set $\cE''(i^*) = 1 - \cE''(\supp(\cE') \setminus \{i^*\})$
    (and $0$ elsewhere).
    Again, $\cE''$ is a probability distribution by construction. One can check that
    $\supp(\cE'') \subseteq \supp(\cE)$, and hence $(g,\cE'') \in \Xi$, thus meeting the first
    condition. The second condition holds because each $\cD(i)$ is a multiple of $1/N$
    (recall $\cD$ has the same densities as $\pi$). One can also verify that
    \[
        \dist_\TV((f,\cD), (g,\cE''))
        = \sum_{i \in \bN \setminus (\supp(\cE') \cup \{i^*\}) } \cD(i)
        \le \dist_\TV((f,\cD), (g,\cE')) \,,
    \]
    satisfying the third condition.

    It follows that $\dist_\TV((f,\cD), \Xi) = \dist_\TV((f,\cD), (g,\cE''))$.
    Let $\pi^* \define \pi_{g,\cE''}$. Then $\pi^*$ has densities that are multiples of $1/N$ and,
    recalling that $\pi = \pi_{f,\cD}$,
    \begin{align*}
        \dist_\TV((f,\cD), (g,\cE''))
        &\ge \dist_\edit(\pi, \pi^*)
            & \text{(Definition of edit distance)} \\
        &= \dist_\edit((f,\cD), (g,\cE''))
            & \text{(Definition of edit distance)} \\
        &\ge \dist_\edit((f,\cD), \Xi)
            & \text{(Since $(g,\cE'') \in \Xi$)} \\
        &\ge \dist_\TV((f,\cD), \Xi)
            & \text{(\cref{lemma:edit-tv-distances-to-support-size-k})} \\
        &= \dist_\TV((f,\cD), (g,\cE''))
            & \text{(Conclusion above)} \,.
    \end{align*}
    Thus equality holds and $\dist_\edit(\pi, \pi^*) = \dist_\edit((f,\cD), \Xi)$.
    We claim that, in fact, $\dist_\edit((f,\cD), \Xi) = \dist_\edit(\pi, \Pi)$.
    Indeed, for any $(h,\cF) \in \Xi$ we have $\pi_{h,\cF} \in \Pi$ and thus
    $\dist_\edit(\pi, \Pi) \le \dist_\edit(\pi_{f,\cD}, \pi_{h,\cF})
    = \dist_\edit((f,\cD), (h,\cF))$, so $\dist_\edit(\pi, \Pi) \le \dist_\edit((f,\cD), \Xi)$.
    Similarly, for any $\pi' \in \Pi$ we may construct $(h,\cF)$ such that
    $\pi_{h,\cF} = \pi'$ and hence $(h,\cF) \in \Xi$, so that
    $\dist_\edit((f,\cD), \Xi) \le \dist_\edit((f,\cD), (h, \cF))
    = \dist_\edit(\pi_{f,\cD}, \pi_{h,\cF}) = \dist_\edit(\pi, \pi')$, and thus
    $\dist_\edit((f,\cD), \Xi) \le \dist_\edit(\pi, \Pi)$. We have constructed our desired
    $\pi^*$, concluding the proof.
\end{proof}

\section{Comparison of Labeled Distribution Testing vs.~the Parity Trace}
\label{section:relative-strength}

It remains to prove the $\not\to$ relations illustrated in \cref{fig:relative-strength}, which we
repeat here for convenience:

\begin{figure}[h]
    \centering
\fbox{
    \begin{tabular}{ccc}
        $(\LD, \TV)$   & $\substack{\centernot\longrightarrow \\ \longleftarrow}$ & $(\PT, \TV)$ \\
        $\downarrow \quad  \centernot\uparrow$ &      & $\downarrow \quad \centernot\uparrow$ \\
        $(\LD, \edit)$ & $\substack{\longrightarrow \\ \longleftarrow}$           & $(\PT, \edit)$ \\
    \end{tabular}}
\end{figure}

$(\LD, \TV) \centernot\rightarrow (\PT, \TV)$ is
\cref{prop:labeled-tv-tester-does-not-imply-parity-tv-tester}, and $(\LD, \edit) \centernot
\rightarrow (\LD, \TV)$ is \cref{prop:edit-labeled-tester-does-not-imply-tv-labeled-tester}.  The
remaining arrow follows by transitivity.

Recall that any density property $\Xi$ has an associated property of distributions $\Pi$, and vice
versa. 
\newcommand{\estD}{\widehat{\bm{D}}}
\newcommand{\estL}{\widehat{\bm{L}}}
\newcommand{\estR}{\widehat{\bm{R}}}

\begin{proposition}
    \label{prop:labeled-tv-tester-does-not-imply-parity-tv-tester}
    For every sufficiently small $\epsilon > 0$ and every $m \in \bN$,
    there exists a property $\Pi$ of distributions
    over $\bN$ and corresponding density property $\Xi = \Xi(\Pi)$ such that
    \begin{enumerate}
        \item There exists a $(\Xi, \far^\TV_{\epsilon}(\Xi), 3/4)$-labeled distribution tester
            with sample complexity $O(1/\epsilon^2)$;
        \item No $(\Pi, \far^\TV_{1/2}(\Pi), 2/3)$-distribution tester under the parity trace
            with sample complexity $m$ exists.
    \end{enumerate}
\end{proposition}
\begin{proof}[Proof sketch]
    Let $\delta = O(1/m)$, and let
    $\Pi \define \{\pi^*\}$ where $\pi^*$ is the property over $\bN$ with densities
    $\left( \frac{1-\delta}{2}, \delta, \frac{1-\delta}{2}, 0, 0, \dotsc \right)$.
    We now show that $\Pi$ and $\Xi = \Xi(\Pi)$ satisfy the two properties in the statement.

    \paragraph{Efficient labeled distribution tester.}
    We outline the construction of a $(\Xi, \far^\TV_{\epsilon}(\Xi), 3/4)$-labeled
    distribution tester using the
    testing-by-learning approach from \cref{prop:testing-by-learning}. For that, we need to
    give a learner-verifier pair for $\Xi$ with sample complexity $O(1/\epsilon^2)$.

    The learner $A$ takes a sample from $\cD_f$ of size $O(1/\epsilon^2)$, and uses it to get an
estimate $\bm{y}$ of the median of $\cD$. It then produces a function $\bm{g} : \bZ \to \zo$ given
by $\bm{g}(x) = 1$ for $x \ne \bm{y}$, and $\bm{g}(\bm{y}) = 0$. Note that, for appropriate
distribution $\cE$, we have $(\bm{g}, \cE) \in \Xi$. It remains to show that, when $(f, \cD) \in
\Xi$, the output $\bm{g}$ also satisfies the other conditions from \cref{def:learner-verifier-pair};
namely, that with high constant probability $\dist_\TV(\cD_f, \cD_{\bm{g}}) < \epsilon/4$ and, for
some $\cE$ satisfying $(\bm{g}, \cE) \in \Xi$, $\dist_\TV(\cD, \cE) < \epsilon/4$.

    The main idea is that, by Hoeffding's inequality, $\bm{y}$ will be $O(\epsilon)$ close to the
    true median of $\cD$, which is the point that should receive value zero and mass $\delta$
    as per the definition of $\Xi$. More precisely, by choosing $\bm{y}$ as close to the median
    of the sample as possible, but taking care \emph{not} to choose any of the 1-valued elements
    in the sample, we can guarantee the following conditions with sufficient probability:
    1) the total $\cD$-mass to the left and to the right of $\bm{y}$ only differ by $O(\epsilon)$;
    and 2) $\cD(\bm{y}) = O(\epsilon)$.
    The first condition is enough to ensure that, for some $\cE$ satisfying $(\bm{g}, \cE) \in \Xi$,
    $\dist_\TV(\cD, \cE) < \epsilon/4$. This is because the $\cD$-mass to the left and right of
    $\bm{y}$ are sufficiently close to the desired value $\frac{1-\delta}{2} \approx \frac{1}{2}$,
    which is also how much $\cE$-mass needs to be in either range to satisfy
    $(\bm{g}, \cE) \in \Xi$. An application of \cref{prop:labeled-distribution-construction}
    concludes that, as long as the masses to the left and right of $\bm{y}$ are correct to
    $O(\epsilon)$ tolerance, a specific $\cE$ can be chosen so that
    $\dist_\TV(\cD, \cE)$ is small.
    Then, the second condition ensures that $\cD(\bm{y})$ is sufficiently smaller than
    $\epsilon$ if $f(\bm{y})=1$ (because otherwise $\bm{y}$ would have been chosen differently),
    so that assigning $\bm{g}(\bm{y}) = 0$ does not make $\dist_\TV(\cD_f, \cE_{\bm{g}})$ too
    large (which would happen if $\cD_f$ and $\cE_{\bm{g}}$ disagreed on some element with
    $\Omega(\epsilon)$ $\cD$-mass). Together, these show that $A$ is a proper learner.

    We now outline the verifier $B_g$. Say $g = g_y$. Then on input $\cD$, which is
    a distribution on $\bZ$, $B_g$ must distinguish between the cases
    $\cD \in \close^\TV_{\epsilon/4}(\Pi_g)$ and $\cD \in \far^\TV_{\epsilon/2}(\Pi_g)$.
    Consider distribution $\pi_{g,\cD}$ on $\bN$.
    Note that $B_g$ is able to sample from $\pi_{g,\cD}$ by drawing a sample $\bm{x} \sim \cD$
    and mapping it to an index in $\{1,2,3\}$ depending on whether
    $\bm{x} < y$, $\bm{x} = y$ or $\bm{x} > y$.

    $B_g$ proceeds by sampling $O(1/\epsilon^2)$ points from $\pi_{g,\cD}$ and using them to learn
    $\pi_{g,\cD}$ to sufficiently small additive error $O(\epsilon)$.
    It follows that $B_g$ can distinguish,
    with high constant probability, between the cases $\dist_\TV(\pi_{g,\cD}, \pi^*) \le \epsilon/4$
    and $\dist_\TV(\pi_{g,\cD}, \pi^*) > \epsilon/2$.
    One can then show that this is equivalent to distinguishing between
    $\cD \in \close^\TV_{\epsilon/4}(\Pi_g)$ and $\cD \in \far^\TV_{\epsilon/2}(\Pi_g)$.
    Intuitively, this is because the only factor determining the distance of $\cD$ to
    $\Pi_g$ is how far its densities around $y$ are from the desired vector
    $\left(\frac{1-\delta}{2}, \delta, \frac{1-\delta}{2}\right)$. Formally, one of the directions
    requires another application of \cref{prop:labeled-distribution-construction}.

    It follows that $(A,B)$ is a learner-verifier pair for $\Xi$ with success probability $3/4$,
    error $\epsilon$, and sample complexity $O(1/\epsilon^2)$.
    By \cref{prop:testing-by-learning}, there exists a
    $(\Xi, \far^\TV_{\epsilon}(\Xi), 3/4)$-labeled distribution tester with sample complexity
    $O(1/\epsilon^2)$.

    \paragraph{Non-existence of efficient tester under the parity trace.}
Consider the distribution $\pi$ given by $\pi(1) = 1$, which has $\dist_\TV(\pi, \pi^*) > 1/2$.
However, under the parity trace, the only event that can distinguish $\pi$ from $\pi^*$ is a
0-valued symbol from $\pi^*$, which occurs for each sampled element with probability $\pi^*(2) =
\delta$. Therefore any tester that takes $o(1/\delta)$ samples cannot distinguish $\pi$ from $\pi^*$
with non-negligible probability.
\end{proof}

\begin{proposition}
    \label{prop:edit-labeled-tester-does-not-imply-tv-labeled-tester}
    For every sufficiently small $\epsilon > 0$ and every $m \in \bN$,
    there exists a property $\Pi$ of distributions
    over $\bN$ and corresponding density property $\Xi = \Xi(\Pi)$ such that
    \begin{enumerate}
        \item There exists a $(\Xi, \far^\edit_{\epsilon}(\Xi), 3/4)$-labeled distribution tester
            with sample complexity $O(1/\epsilon^2)$;
        \item No $(\Xi, \far^\TV_{\epsilon}(\Xi), 2/3)$-labeled distribution tester with sample
            complexity $m$ exists.
    \end{enumerate}
\end{proposition}
\begin{proof}[Proof sketch]
    Let $n = \Theta(m^2 \epsilon^4)$ be an integer.
    Let $\Pi$ be the class of all distributions $\pi$ supported on $\bN$ such that
    1) the total density on the odd numbers is exactly $1/2$; and
    2) for every $i \in \bN$, $\pi(i) \le 1/n$.
    Let $\Xi = \Xi(\Pi)$ be the corresponding density property.

    \textbf{Existence of efficient edit distance tester.}
    We observe that the second requirement of $\Pi$ has essentially no effect under the edit
    distance, as the following outline shows. Let $\cO$ denote the set of positive odd integers.
    Then $\dist_\edit(\pi, \Pi) > \epsilon$ implies that $\abs*{\pi(\cO) - \frac{1}{2}} > \epsilon$.
    The reason is that, given a labeled distribution $(f,\cD)$ such that
    $\pi_{f,\cD} = \pi$ with sufficiently small pointwise masses and sufficient space between
    nonzero entries (which can always be accomplished without affecting $\pi_{f,\cD}$),
    one may move $\abs*{\pi(\cO)-\frac{1}{2}}$ mass in $\cD$ between the even and odd elements
    so as to satisfy the first condition of $\Pi$, and then using the sufficient space between
    nonzero entries, one may change the values of $f$ at points of zero mass so as to break up
    any alternations with more than $1/n$ mass, so as to satisfy the second condition of $\Pi$
    at no additional cost.

    Therefore the following algorithm $A$ distinguishes $\Xi$ from $\far^\edit_{\epsilon}(\Xi)$ with
    high constant probability: take $O(1/\epsilon^2)$ samples and use the empirical frequency
    of $1$-valued sample points $\widehat{\bm{o}}$ as an estimate of $\pi(\cO)$
    to $\epsilon/4$ additive error. Then accept if and only if
    $\abs*{\widehat{\bm{o}} - \frac{1}{2}} < \epsilon/2$.

    \textbf{Non-existence of efficient TV distance tester.}
    We reduce the problem of testing uniformity of distributions over $[n]$ in the standard model,
    to $(\Xi, \far^\TV_{\epsilon}(\Xi))$-labeled distribution testing. Suppose
    algorithm $A$ is a $(\Xi, \far^\TV_{\epsilon}(\Xi), 2/3)$-labeled distribution tester.
    Then our algorithm $B$ to distinguish, in the standard model, between the uniform distribution
    over $[n]$ and distributions over $n$ that are $\epsilon$-far from uniform in TV distance
    works as follows.

    For input distribution $\pi$ over $[n]$, let $f,\cD$ be the labeled distribution given by
    $f(x) = \parity(x)$ on $x \ge 1$ and $f(x) = 1$ on $x \le 0$,
    and $\cD = \pi$. It follows that $\pi = \pi_{f,\cD}$. Therefore $B$, on input
    $\pi$, can simulate $A$ on input $(f,\cD)$ by sampling $\bm{x} \sim \pi$ and producing
    $(\bm{x}, \parity(\bm{x}))$ when $A$ requests a sample from $(f, \cD)$.

    If $\pi$ is uniform over $[n]$, it follows that $(f, \cD) \in \Xi$. On the other hand,
    if $\pi$ is supported on $[n]$ and $\epsilon$-far from uniform in TV distance, then its
    total density in excess of $1/n$ is $\sum_i \max\{0, \pi(i)-1/n\} > \epsilon$.
    Therefore $\dist_\TV(\cD_f, \Xi) > \epsilon$. It follows that $B$ correctly accepts/rejects
    with probability at least $2/3$. Since testing uniformity in the standard model requires
    $\Omega(\sqrt{n}/\epsilon^2)$ samples, the sample complexity of $A$ must
    be at least $\Omega(\sqrt{n}/\epsilon^2)$.
\end{proof}

\end{document}